\ifdefined\synctex\synctex=1\fi
\RequirePackage{fix-cm}
\documentclass[11pt,a4paper]{amsart}
\usepackage[foot]{amsaddr}
\usepackage[T1]{fontenc}

\usepackage{ifxetex}
\ifxetex
  \usepackage[no-math]{fontspec}
\fi
\usepackage{amsmath}
\usepackage{amsfonts}
\usepackage{amssymb}
\usepackage{amsthm}
\usepackage{enumitem}
\usepackage{fullpage}
\usepackage{microtype}
\usepackage{comment}
\RequirePackage{silence}
\ifxetex
  \usepackage[libertine]{newtxmath}
\else
  \usepackage{newtxmath}
\fi
\usepackage[tt=false]{libertine}
\usepackage{xcolor}
\usepackage{tikz}
\usetikzlibrary{arrows.meta,positioning,fit,calc}
\usepackage{hyperref}
\hypersetup{colorlinks=true,citecolor=blue, linkcolor=blue, urlcolor=blue,
  pdftitle={Zero-Freeness for Random Regular Graphs up to the Replica Symmetry Breaking Threshold}}
\usepackage[numbers]{natbib}
\usepackage{xifthen}
\usepackage{prettyref}
\makeatletter
\renewcommand\paragraph{\@startsection{paragraph}{4}%
 \z@{-\medskipamount}{-\fontdimen2\font}{\normalfont\bfseries}}
\makeatother

\renewcommand{\epsilon}{\varepsilon}

\newtheorem{theorem}{Theorem}[section]
\newtheorem{lemma}[theorem]{Lemma}
\newtheorem{proposition}[theorem]{Proposition}
\newtheorem{corollary}[theorem]{Corollary}
\newtheorem{fact}[theorem]{Fact}
\theoremstyle{definition}
\newtheorem{definition}[theorem]{Definition}
\newtheorem{condition}{Condition}
\theoremstyle{remark}
\newtheorem{remark}[theorem]{Remark}
\theoremstyle{plain}

\newrefformat{thm}{Theorem~\ref{#1}}
\newrefformat{lem}{Lemma~\ref{#1}}
\newrefformat{fact}{Fact~\ref{#1}}
\newrefformat{prop}{Proposition~\ref{#1}}
\newrefformat{cor}{Corollary~\ref{#1}}
\newrefformat{rem}{Remark~\ref{#1}}
\newrefformat{def}{Definition~\ref{#1}}
\newrefformat{cond}{Condition~\ref{#1}}
\newrefformat{sec}{Section~\ref{#1}}
\newrefformat{subsec}{Section~\ref{#1}}
\newrefformat{fig}{Figure~\ref{#1}}
\newrefformat{app}{Appendix~\ref{#1}}

\let\pref\prettyref
\newcommand{\Cref}[1]{\pref{#1}}

\newcommand{\dif}{\,\mathrm{d}}

\newcommand{\NP}{\ensuremath{\mathbf{NP}}}

\newcommand{\RP}{\ensuremath{\mathbf{RP}}}
\newcommand{\one}[1]{\mathbf{1}_{\{#1\}}}
\newcommand{\allone}{\mathbf{1}}
\newcommand{\norm}[1]{\left\lVert #1 \right\rVert}
\newcommand{\trans}{\intercal}
\newcommand{\diag}{\mathrm{diag}}
\newcommand{\e}{\mathrm{e}}

\newcommand{\DTV}[2]{\norm{#1-#2}_{\mathrm{TV}}}
\newcommand{\DKL}{\mathrm{D}_{\mathrm{KL}}}
\newcommand{\tp}[1]{\left(#1\right)}
\newcommand{\ol}{\overline}
\newcommand{\abs}[1]{\left\vert#1\right\vert}
\newcommand{\defeq}{\triangleq}

\newcommand{\eps}{\varepsilon}
\def\*#1{\boldsymbol{#1}}
\def\+#1{\mathcal{#1}}
\def\=#1{\mathbb{#1}}

\def\oPr{\mathop{\mathrm{Pr}}}
\renewcommand{\Pr}[2][]{ \ifthenelse{\isempty{#1}}
  {\oPr\left[#2\right]}
  {\oPr_{#1}\left[#2\right]} } % Use \Pr[a]{b} for \mathrm{Pr}_a[b], \Pr{b} for \mathrm{Pr}[b]

\def\oE{\mathop{\=E}}
\newcommand{\E}[2][]{ \ifthenelse{\isempty{#1}}
  {\oE\left[#2\right]}
  {\oE_{#1}\left[#2\right]} }

\newbool{doubleblind}
\setbool{doubleblind}{false}

\title[Zero-Freeness for Random Regular Graphs]{Zero-Freeness for Random Regular Graphs\\ up to the Replica Symmetry Breaking Threshold}

\ifdoubleblind
\author{Author(s)}
\else
\author{Zejia Chen, Chunyang Wang, Xinyuan Zhang}

\address[Zejia Chen, Xinyuan Zhang]
{School of Computer Science, Georgia Institute of Technology, Atlanta, GA, USA.
\textnormal{Emails:
\texttt{zchen3091@gatech.edu,normalvra@gmail.com}}}

\address[Chunyang Wang]
{National Institute of Informatics, Tokyo, Japan.
\textnormal{Email: \texttt{c\_wang@nii.ac.jp}}}

\fi

\begin{document}
\raggedbottom

\begin{abstract}
\vspace*{-5pt}
We study antiferromagnetic spin systems on random regular graphs through the lens of the phase transitions characterized by the complex zeros of the partition function.
Our motivating examples include the hardcore model and the antiferromagnetic Potts model (including proper $q$-colorings at {interaction parameter} $\beta=0$).
For both models on random $d$-regular graphs, we establish, with high probability over the graph, zero-free regions throughout the interior of the replica-symmetric regime. Asymptotically, this regime extends up to tree occupation density
\[
  \alpha < \alpha_{\mathrm{RSB}}(d)=\frac{2\log d-2\log\log d+2-2\log 2+o_d(1)}{d}
\]
for the hardcore model and up to degree
\[
  d < d_{\mathrm{RSB}}(q)=(2q-1)\log q-2\log2+o_q(1)
\]
for proper $q$-colorings.

Leveraging the standard algorithmic implications of zero-freeness, we obtain deterministic polynomial-time
algorithms that multiplicatively approximate the corresponding partition functions at parameters throughout this regime, with high probability over the graph.
Our algorithmic guarantees significantly
improve upon prior work and extend beyond the parameter ranges reached by previous search algorithms.

As a key intermediate step, we chart the replica-symmetric phase on random $d$-regular graphs by establishing the pointwise equivalence of three notions of replica symmetry away from the threshold:
\begin{enumerate}
    \item[(1)] \emph{Geometric replica symmetry}: averaged pairwise correlation decay under the Gibbs measure;
    \item[(2)] \emph{Thermodynamic replica symmetry}: convergence of the free-energy density to the Bethe prediction;
    \item[(3)] \emph{Variational replica symmetry}: nonpositivity of the Bethe functional.
\end{enumerate}

To establish zero-freeness across this regime, we combine this variational characterization with two technical ingredients. First, we develop \emph{reinforced small subgraph conditioning}, which resolves short-cycle fluctuations in the complex plane by lifting classical small subgraph conditioning on random regular graphs to random holomorphic functions. Second, to overcome the exponential blowup of the bare second moment inside the replica-symmetric phase, we use variational replica symmetry to upgrade planted overlap concentration to an \emph{exponentially small} bound, enabling a typical-graph truncation that controls complex second moments up to the replica symmetry breaking threshold.
\end{abstract}

\hypersetup{pageanchor=false}
\maketitle
\vspace*{-35pt}
\setcounter{tocdepth}{1}
\tableofcontents
\thispagestyle{empty}
\clearpage
\setcounter{page}{1}
\hypersetup{pageanchor=true}

% BEGIN INCLUDED FILE: chapters/01-introduction-v2.tex
\section{Introduction}
Diluted mean-field models are spin systems whose geometry of interactions is induced by a sparse random graph or
hypergraph. These models are extensively studied in statistical physics~\cite{monasson1998otpimization,mezard2001bethe,mezard2003cavity} in order to better understand disordered physical systems such as glasses or spin glasses. Unlike models on the complete graph, such as the Curie--Weiss or the Sherrington--Kirkpatrick model, they retain finite local connectivity~\cite{mezard1987spin,mezard2009information}. Moreover, the study of such models also plays a fundamental role in many other areas, including computer science through random $k$-SAT~\cite{mezard2002analytic,achlioptas2005rigorous,florent2007gibbs}, statistics through the stochastic block model~\cite{decelle2011asymptotic,Zdeborova2016statistical}, coding theory through low-density parity check codes~\cite{montanari2001glassy,richardson2008modern}, and probabilistic combinatorics, which studies random graphs themselves~\cite{dembo2010gibbs,dembo2013factor}.

\paragraph{Phase transitions and replica symmetry.}
A central focus of statistical physics research on diluted mean-field models is the study of \emph{phase transitions}, namely qualitative shifts in a system's macroscopic behavior as its parameters vary in the thermodynamic limit.
Conventionally, a (thermodynamic) phase transition occurs at a parameter where the limiting free energy fails to be real analytic~\cite{mezard2009information}; Ehrenfest's classical classification distinguishes transitions by the order of a singular derivative~\cite{ehrenfest1933phasenumwandlungen}.

To analyze the phase transitions of spin systems, statistical physicists developed the powerful yet non-rigorous ``cavity method''~\cite{mezard1986SK,mezard2001bethe}. This technique has been extensively applied to the aforementioned models and numerous others, driving an ongoing line of research that has generated profound predictions of phase transitions across a wide array of computational and physical problems~\cite{mezard2002analytic,florent2007gibbs,Zdeborova2016statistical}. There has also been substantial research effort devoted to putting the cavity method on a rigorous foundation, with successes including~\cite{dembo2013factor,bapst2014condensation,dembo2014replica,coja2014asymptotic,coja2017information,coja2019spin,ding2022satisfiability,gheissari2025fixed,jenssen2026replica}.

One of the most important predictions of the cavity method is that the Gibbs measures induced by such sparse random graph models undergo a \emph{replica symmetry breaking (RSB)} transition~\cite{mezard2001bethe,florent2007gibbs}. Given a random graph $G$ on $[n]=\{1,2,\dots,n\}$ and its Gibbs measure $\mu_G$, \emph{replica symmetry} is characterized by the absence of long-range correlations, such that two uniformly chosen vertices become asymptotically independent under the Gibbs measure of the sampled graph~\cite{coja2016belief}:

\begin{definition}[Replica symmetry]\label{def:replica-symmetry}
At a fixed model parameter and for a sequence of random graphs $G=G_n$ on $[n]$, we say that the model (described by $\mu_G$) is \emph{replica symmetric (RS)} if
    \begin{equation}\label{eq:def-replica-symmetry}
        \lim\limits_{n\to \infty}\frac{1}{n^2}\sum\limits_{1\leq u<v\leq n}\E[G]{\DTV{\mu_{G,u,v}}{\mu_{G,u}\otimes \mu_{G,v}}}=0,
    \end{equation}
    where $\mu_{G,u}$ and $\mu_{G,u,v}$ denote the marginal distributions of $\sigma_u$ and $(\sigma_u,\sigma_v)$ under $\sigma\sim\mu_G$, and $\DTV{\cdot}{\cdot}$ denotes total variation distance (\Cref{subsec:preliminaries}). For proper colorings, we define $\mu_G$ arbitrarily on uncolorable graphs; geometric replica symmetry additionally requires colorability with high probability.
\end{definition}

% As the model parameters vary, the cavity method predicts a sharp transition from a replica-symmetric regime to one in which \eqref{eq:def-replica-symmetry} fails, known as replica symmetry breaking.
Replica symmetry has deep connections with fundamental phase transitions in statistical physics. In the one-step-RSB (one-step replica symmetry breaking, or condensation) scenario predicted by the cavity method~\cite{mezard2001bethe,florent2007gibbs}, the breakdown of replica symmetry is predicted to coincide with the \emph{condensation phase transition}, beyond which the Gibbs measure concentrates on a sub-exponential number of dominant clusters with macroscopic weights; {key variational free-energy and overlap-characterization aspects of this transition have been rigorously established in a broad class of diluted spin systems~\cite{bapst2014condensation,coja2017information,coja2018charting,coja2019spin,coja2020inference}}. In statistical physics, this is also predicted to correspond to the \emph{Kauzmann transition}~\cite{kauzmann1948nature,florent2007gibbs} in mean-field glass theory, characterizing the threshold where the equilibrium configurational entropy density vanishes. Within this scenario, this phase transition also precisely delineates the regime in which the limiting free energy is governed by the replica-symmetric Bethe formula; beyond this critical point, the RS description breaks down~\cite{coja2016belief,coja2018charting,coja2020inference}.

Intuitively, within the replica-symmetric phase, the absence of long-range correlations supports a relatively simple local description of the equilibrium behavior of an exponentially large configuration space. This description underlies cavity predictions and analyses of local message-passing algorithms~\cite{coja2017information,coja2018charting,coja2022belief}.

%We use \emph{free energy} for $n^{-1}\log Z$. For $Z>0$, its quenched
%and annealed versions are $n^{-1}\E{\log Z}$ and
%$n^{-1}\log\E{Z}$, respectively. The homogeneous tree measure predicts
%their common limit in the replica-symmetric region. Local convergence
%to a tree alone does not establish this prediction.

%For hardcore at large degree and colorings with many colors, the cavity
%method predicts distinct clustering and condensation transitions
%\cite{barbier2013hardcore,zdeborova2007phase}. At clustering, most Gibbs
%weight lies in many groups separated by large Hamming distances. At
%condensation, a subexponential number of these groups carries most of
%the Gibbs weight. In these fixed-degree models, the prediction is that,
%beyond condensation, the typical
%normalized log partition function is strictly smaller than its limiting
%first-moment rate, when both limits exist
%\cite{barbier2013hardcore,zdeborova2007phase}.
%These boundaries of the replica-symmetric free-energy formula are not in
%general the clustering threshold or the local Kesten--Stigum stability
%threshold. The latter tests the growth of the sum of squared
%correlations on the tree.

\paragraph{Lee--Yang theory of phase transitions.}
In their seminal work~\cite{lee1952statistical,lee1952lattice}, Lee and Yang pioneered the study of phase transitions through the lens of complex partition function zeros. Their framework explicitly connects thermodynamic singularities to the asymptotic distribution of these zeros. In particular, establishing uniform zero-free regions under standard thermodynamic-limit assumptions guarantees the complex analyticity of the limiting free-energy density, thereby ensuring real analyticity on the physical axis and the absence of phase transitions.  While the original Lee--Yang theory was motivated by analyzing spontaneous magnetization of the Ising model, it has been subsequently developed~\cite{asano1970lee,ruelle1971extension} and extended for studying ferromagnets with single-spin measures satisfying the Lee--Yang property~\cite{newman1974zeros,lieb1981general}, monomer-dimer systems~\cite{heilmann1972theory}, weighted subgraphs~\cite{wagner2009weighted}, and recently for mean-field spin glasses with Gaussian interactions~\cite{bencs2026zeros}.

Apart from providing evidence for the absence of phase transitions, zero-free regions also have direct probabilistic and algorithmic consequences. One particular application includes deterministic approximate counting algorithms through Barvinok's interpolation method~\cite{barvinok2016combinatorics,patel2017deterministic,liu2025correlation,bencs2026zeros}, thereby offering a bridge connecting phase transitions in statistical physics with \emph{computational phase transitions}. Other applications include establishing concentration bounds~\cite{pemantle2014concentration,kyng2018matrix}, central limit theorems~\cite{lebowitz2016central, vishesh2022approximate,marcus2024central}, and asymptotic normality~\cite{kahn2000normal}. Real stability and the resulting negative dependence have also led to improved approximation algorithms~\cite{shayan2011randomized,karlin2021improved}. Under suitable self-reducibility and interpolation assumptions, zero-freeness has been connected to an asymptotic form of strong spatial mixing~\cite{gamarnik2023correlation}. Under additional graph-family and pinning hypotheses, it can imply the usual uniform strong spatial mixing property~\cite{regts2023absence}. Related multivariate stability and independence properties have also been used to prove rapid mixing of Markov chains~\cite{alimohammadi2021fraction,chen2022spectral,vishesh2023optimal,aiya2025fast}. Our univariate high-probability zero-free regions do not by themselves supply these additional hypotheses.

\paragraph{Challenges in understanding diluted mean-field spin systems.}
Establishing rigorous zero-free regions remains a challenge for disordered mean-field spin systems. Very recently, Bencs et al.~\cite{bencs2026zeros} made significant progress by establishing zero-free regions in the complex inverse-temperature plane for finite Gaussian mixtures of $p$-spin interactions with $p\ge2$, throughout the entire second-moment regime. Notably, their result implies zero-free regions for the zero-field Sherrington--Kirkpatrick model throughout the interior of the replica-symmetric phase.
For diluted mean-field spin systems, related investigations include an exactly solvable infinite-coupling limit~\cite{laumann2008griffiths} and numerical cavity calculations on the interior of a Cayley tree~\cite{matsuda2010distribution}, yet these do not establish a fixed zero-free neighborhood throughout the replica-symmetric phase of the random regular models considered here.

A similar gap between physical predictions and rigorous results persists on the algorithmic side: known efficient algorithms for approximate counting and sampling cover parameter ranges well short of the predicted thresholds.
For example, consider the hardcore model on random $d$-regular graphs with fugacity $\lambda$. On the infinite $d$-regular tree, the \emph{tree uniqueness threshold} $\lambda_c(d)=(d-1)^{d-1}/(d-2)^d=\Theta(1/d)$ marks the onset of multiple Gibbs measures (worst-case influence of worst-case boundary conditions), whereas the \emph{reconstruction (extremality) threshold} $\lambda_{\mathrm{rec}}(d)$ marks the onset of non-vanishing mutual information between the root and typical distant boundary configurations~\cite{mezard2006reconstruction,Bhatnagar2016decay}. Rigorous bounds for $\lambda_{\mathrm{rec}}(d)$~\cite{Bhatnagar2016decay} and the cavity prediction for condensation (replica symmetry breaking)~\cite{barbier2013hardcore} give the respective asymptotics:
\[
\lambda_{\mathrm{rec}}(d) = (\log d)^{2+{o_d(1)}}, \quad \lambda_{\mathrm{cond}}(d) = \left(\frac{\e^2}{2}+{o_d(1)}\right)\frac{d}{\log d}.
\]
Yet, it was only recently that Chen et al.~\cite{chen2025rapidrandom} achieved algorithmic progress beyond the tree uniqueness threshold $\lambda_c(d)$ for sufficiently large $d$, demonstrating rapid mixing of the Glauber dynamics for fugacities $\lambda < 1/(2\sqrt{d-1}-1)$---a bound that remains drastically lower than either $\lambda_{\mathrm{rec}}(d)$ or $\lambda_{\mathrm{cond}}(d)$.

The dual lack of analytical and algorithmic understanding above motivates the following central question, answering which promises twofold benefits:

\begin{center}
    \emph{Can we establish zero-free regions for random regular graphs within the
replica-symmetric phase?}
\end{center}

\subsection{Our Results}\label{subsec:our-results}
In this work, we answer the above question affirmatively by establishing zero-free regions for two important antiferromagnetic spin systems on random regular graphs: the hardcore model and the antiferromagnetic Potts model, throughout the interior of the replica-symmetric regime, thereby greatly advancing our understanding of phase transitions for such models. Throughout, for integers $d \geq 3$ and $n \ge d+1$ with $d n$ even (which we call \emph{admissible} $n$), let $\+G_{n,d}$ denote the uniform distribution over all simple $d$-regular graphs on $[n]=\{1,\dots,n\}$.

Through Barvinok's interpolation method~\cite{barvinok2016combinatorics,patel2017deterministic}, we directly obtain polynomial-time deterministic approximation algorithms for the partition function of each model within the replica-symmetric regime. Our results substantially extend the previous algorithmic thresholds, and provide polynomial-time deterministic algorithms for random regular graphs in regimes beyond worst-case results.
% A somewhat surprising finding is that our regimes work even beyond where efficient local search algorithms can exist.
Notably, our algorithms apply even in regimes beyond the scope of currently known algorithms for finding a (typical) configuration.

\subsubsection{Zero-freeness and approximate counting}
\paragraph{Hardcore model.}

The hardcore model is a fundamental model in statistical mechanics for analyzing the equilibrium of lattice gases, serving as a canonical discrete analogue for the behavior of hard spheres. In computer science, the partition function of the hardcore model corresponds exactly to the weighted sum of independent sets in a graph~\cite{domb1958some,burley1960lattice}. Given a graph $G=(V,E)$, let $\+I(G)$ be its collection of independent sets (which we identify with indicator configurations $\sigma=(\one{v\in I})_{v\in V}\in\{0,1\}^V$ satisfying $\sigma_u\sigma_v=0$ for all $\{u,v\}\in E$). Given a fugacity parameter $\lambda>0$, the hardcore \emph{Gibbs distribution} $\mu=\mu_{G,\mathrm{fug},\lambda}$ over $\+I(G)$ and \emph{partition function} $Z={Z_{G,\mathrm{fug}}(\lambda)}$ are given by
\[
 \mu(I)\defeq \frac{\lambda^{|I|}}{Z}\quad (I\in \+I(G)),\qquad
 Z={Z_{G,\mathrm{fug}}(\lambda)}\defeq \sum_{I\in \+I(G)}\lambda^{|I|}.
\]

{We parameterize the hardcore model by the tree occupation density $\alpha\in(0,1/2)$, defined as the occupation probability at the root under the translation-invariant Markov Gibbs measure on the infinite $d$-regular tree. In the interior of the replica-symmetric regime, the expected occupied fraction on a random $d$-regular graph converges to this parameter. The fugacity is related to $\alpha$ by~\cite{Bhatnagar2016decay}}
\[
  \lambda(\alpha)\defeq \frac{\alpha(1-\alpha)^{d-1}}{(1-2\alpha)^d}.
\]
From now on, we write $\mu_{G,\alpha}\defeq\mu_{G,\mathrm{fug},\lambda(\alpha)}$ and $Z_{G}(\alpha){\defeq Z_{G,\mathrm{fug}}(\lambda(\alpha))}$ for the Gibbs distribution and the hardcore partition function at fugacity $\lambda(\alpha)$, respectively.

For the hardcore model on $G\sim\+G_{n,d}$ with any fixed integer $d\ge3$, there exists a threshold $\alpha_{\mathrm{RSB}}(d)\in(0,1/2)$ separating the replica-symmetric regime $\alpha<\alpha_{\mathrm{RSB}}(d)$ from the replica-symmetry-broken regime $\alpha>\alpha_{\mathrm{RSB}}(d)$ (established in \Cref{thm:hardcore-rs-transition} informally and \Cref{thm:hardcore-three-characterizations} formally). Furthermore, $\alpha_{\mathrm{RSB}}(d)$ has the following asymptotic expansion, in agreement with the cavity method prediction~\cite{barbier2013hardcore}:
\begin{equation}\label{eq:replica-symmetry-breaking-hardcore-asymptotic}
  \alpha_{\mathrm{RSB}}(d)=\frac{2\log d-2\log\log d+2-2\log 2+o_d(1)}{d}.
\end{equation}

 For a subset $S\subseteq\=C$, let $D(S,\eta)$ denote the open $\eta$-neighborhood of $S$ in $\=C$. We are now ready to state our first main theorem, which is a zero-free strip for the hardcore model on random regular graphs up to the replica symmetry breaking threshold.
% Our zero-free region for the hardcore model is established for $Z_{G}(\alpha) \defeq Z_G(\lambda(\alpha))$ as a function of $\alpha$.
% % (therefore also as the function of fugacity $\lambda$ by \eqref{eq:hardcore-density-parameter}).
% Complex zeros with respect to the fugacity parameter are often referred to as \emph{Lee-Yang zeros}. Here and below, $D(S,r)=\{z:\operatorname{dist}(z,S)<r\}$ for a nonempty
% set $S\subset\=C$ and $r>0$, and $D(t,r)=D(\{t\},r)$.

\begin{theorem}[Zero-free region for the hardcore model]\label{thm:hardcore-zero}
  Fix an integer $d\ge3$ and let $0<\delta<\alpha_{\mathrm{RSB}}(d)$. There exists a constant
  {$\eta=\eta(d,\delta)>0$} such that with probability $1-o_n(1)$ over $G\sim\+G_{n,d}$, the partition function {$Z_G$} is zero-free on {$D([0,\alpha_{\mathrm{RSB}}(d)-\delta],\eta)$}.
\end{theorem}

We remark that the parameter range in \Cref{thm:hardcore-zero} is optimal in the sense that the zero-free region cannot be extended across $\alpha_{\mathrm{RSB}}(d)$ (see \Cref{rem:sharpness-zero-free} in \Cref{sec:rs-identification-local}).
Also, to the best of our knowledge, \Cref{thm:hardcore-zero} is the first zero-free result for the hardcore model on random regular graphs that extends beyond the tree uniqueness threshold.

% We remark that the real-parameter range in \Cref{thm:hardcore-zero} is optimal in the sense that it cannot be extended across $\alpha_{\mathrm{RSB}}(d)$. Indeed, for $\alpha<\alpha_{\mathrm{RSB}}(d)$, the limiting quenched free energy (\Cref{def:free-energies}) equals the real-analytic Bethe prediction $\phi_d(\alpha)$ (\Cref{def:bethe-canonical-form}), whereas for every $\alpha\in(\alpha_{\mathrm{RSB}}(d),1/2)$, the limit superior of the normalized quenched free energy is strictly below $\phi_d(\alpha)$; see \Cref{thm:hc-rs-rsb-free-energy}. As shown in \Cref{sec:rs-identification-local}, this nonanalyticity prevents extending the zero-free guarantee across $\alpha_{\mathrm{RSB}}(d)$.

Barvinok's interpolation method~\cite{barvinok2016combinatorics,patel2017deterministic}
then gives a deterministic polynomial-time approximation algorithm for the hardcore partition function throughout the interior of the replica-symmetric regime, whose guarantee holds {with high probability over $G\sim\+G_{n,d}$}. In the counting corollaries below, the model parameter $\alpha$ or $\beta$ is a fixed rational number, and the algorithm receives only $G$ and the desired accuracy as inputs.

\begin{corollary}[Approximate counting for random regular hardcore]\label{cor:hardcore-counting}
     Fix an integer $d\ge3$. Let $0<\delta<\alpha_{\mathrm{RSB}}(d)$ and fix a rational $\alpha\in (0,\alpha_{\mathrm{RSB}}(d)-\delta]$. There is a deterministic algorithm that, given as input a $d$-regular graph $G$ with $n$ vertices and an accuracy parameter $0<\varepsilon<1$, runs in time $(n/\varepsilon)^{O_{d,\delta}(1)}$ and outputs an estimate $\hat{Z}$ satisfying
     \[
     (1-\eps)Z_{G}(\alpha)\leq \hat{Z}\leq (1+\eps)Z_{G}(\alpha)
     \]
     with probability $1-o_n(1)$ over $G\sim\+G_{n,d}$.
\end{corollary}

% TODO: zero-freeness within uniqueness.
To build intuition, it is instructive to situate these results alongside prior work and other critical thresholds. For approximate counting and sampling in the hardcore model, Weitz~\cite{weitz06counting} provided the first polynomial-time algorithm throughout the region below the tree uniqueness threshold $\alpha\sim \frac{1}{d}$ for graphs of maximum degree $d$. Beyond this critical threshold, these tasks are computationally intractable on worst-case graphs of maximum degree $d$ unless $\NP=\RP$~\cite{sly2014computational}.

For random regular graphs, Chen et al.~\cite{chen2025rapidrandom} showed that, for sufficiently large $d$, the Glauber dynamics mix rapidly with high probability beyond the uniqueness regime, functioning up to the density $\alpha\sim \frac{\log d}{2d}$. Subsequently, G\"{o}bel et al.~\cite{gobel2026simple} provided an alternative, significantly simpler proof of this rapid-mixing behavior using Bochner's identity. Meanwhile, it is a classic result that the normalized independence number (i.e., the maximum independent set size scaled by $1/n$) of a random $d$-regular graph concentrates around $\alpha_{\max}(d) \sim \frac{2\log d}{d}$ with high probability for large $d$~{\cite{frieze1992independence}}. The reconstruction threshold has density $\alpha_{\mathrm{rec}}(d)\sim \frac{\log d}{d}$~\cite{Bhatnagar2016decay}; its relation to Glauber dynamics on trees depends on the boundary conditions~\cite{martinelli2004glauber,martinelli2007fast,sly2017glauber}. This leading-order density also matches the best-known efficient algorithms for finding large independent sets in random $d$-regular graphs. For every fixed $\eps>0$ and sufficiently large $d$, factor-of-i.i.d.\ local algorithms cannot find independent sets of density $(1+\eps)\log d/d$~\cite{gamarnik2017limits,rahman2017local}.

In contrast, our zero-free region in \Cref{thm:hardcore-zero} and the deterministic approximate counting algorithm in \Cref{cor:hardcore-counting} work at every fixed parameter strictly below the replica symmetry breaking threshold $\alpha_{\mathrm{RSB}}(d)\sim \frac{2\log d}{d}$ in~\eqref{eq:replica-symmetry-breaking-hardcore-asymptotic}.
Therefore, our leading-order bound matches the normalized independence number, substantially improves over previous algorithmic guarantees, and extends beyond the densities reached by the best-known efficient algorithms for finding large independent sets.

%Note that the asymptotic in~\eqref{eq:replica-symmetry-breaking-asymptotic} was previous obtained by physics prediction, for which we also rigorously prove later in this paper.

\paragraph{Antiferromagnetic Potts model.}

The Potts model is another fundamental spin system widely studied in statistical physics and theoretical computer science, serving as a canonical framework for modeling interacting multistate spins and proper $q$-colorings~\cite{Potts1952generalized}.  Fix an integer $q \ge 2$, representing the number of available spin states (or colors). The Potts model is parameterized here by an edge weight (or interaction parameter) $\beta \geq 0$. For a graph $G = (V, E)$, configurations of the Potts model are all possible assignments of states to the vertices of the graph. Each assignment $\sigma : V \to [q]$ has a weight $w(\sigma)=\beta^{m_G(\sigma)}$ which is determined by the number $m_G(\sigma)$ of monochromatic edges under $\sigma$.
{The partition function and, whenever $Z_{G,q}(\beta)>0$, the Gibbs distribution $\mu=\mu_{G,q,\beta}$ are given by}
\[
 \mu_{G,q,\beta}(\sigma)\defeq\frac{\beta^{m_G(\sigma)}}{Z_{G,q}(\beta)},
 \qquad Z_{G,q}(\beta)\defeq\sum_{\sigma:V(G)\to[q]}\beta^{m_G(\sigma)}.
\]

The behavior of the Potts model differs fundamentally depending on whether $\beta$ is less than or greater than $1$. When $\beta < 1$, neighboring vertices prefer to adopt different states, and the model is called \emph{antiferromagnetic}. Conversely, when $\beta > 1$, adjacent vertices favor the same state, and the model is called \emph{ferromagnetic}.

Furthermore, the Potts model serves as a natural generalization of two foundational models: the Ising model, which corresponds precisely to the special case $q=2$; and the $q$-coloring model, which emerges in the ``zero-temperature'' antiferromagnetic limit ($\beta = 0$) as we use $0^0=1$. {Whenever $G$ is $q$-colorable, the Gibbs distribution at $\beta=0$ is uniform on its proper $q$-colorings.}

We focus on the antiferromagnetic regime $\beta\in[0,1]$. For the Potts model on $G\sim\+G_{n,d}$ with fixed integers $d\ge3$ and $q\ge2$, there exists a threshold $\beta_{\mathrm{RSB}}(d,q)\in[0,1)$ separating the replica-symmetric regime $\beta>\beta_{\mathrm{RSB}}(d,q)$ from the replica-symmetry-broken regime $0<\beta<\beta_{\mathrm{RSB}}(d,q)$.
For each integer $q\ge3$, proper $q$-colorings ($\beta=0$) similarly admit a degree threshold $d_{\mathrm{RSB}}(q)\in[3,\infty)$, below which replica symmetry holds and above which the Bethe free-energy prediction fails. These phase characterizations are established informally in \Cref{thm:hardcore-rs-transition}, and formally in \Cref{thm:potts-three-characterizations} and  \Cref{thm:coloring-three-characterizations}. Furthermore, $d_{\mathrm{RSB}}(q)$ has the following asymptotic expansion, in agreement with the cavity method prediction~\cite{zdeborova2007phase,florent2007gibbs}:
\begin{equation}\label{eq:replica-symmetry-breaking-potts-asymptotic}
  d_{\mathrm{RSB}}(q)=(2q-1)\log q-2\log2+o_q(1).
\end{equation}
Our zero-free regions for $Z_{G,q}(\beta)$ are established in the complex interaction parameter $\beta$, where zeros are known as \emph{Fisher zeros}~\cite{fisher1965nature}.

\begin{theorem}[Zero-free region for random regular antiferromagnetic Potts and $q$-colorings]\label{thm:potts-zero}
  Fix integers $d\ge3$ and $q\ge2$.
  \begin{enumerate}
    \item \emph{Antiferromagnetic Potts ($\beta>0$):} For $0<\delta<1-\beta_{\mathrm{RSB}}(d,q)$, there is a constant $\eta=\eta(d,q,\delta)>0$ such that, with probability $1-o_n(1)$ over $G\sim\+G_{n,d}$, $Z_{G,q}$ is zero-free on $D([\beta_{\mathrm{RSB}}(d,q)+\delta,1],\eta)$.
    \item {\emph{Proper $q$-colorings ($\beta=0$):} If $q\ge3$ and $d<d_{\mathrm{RSB}}(q)$, there is a constant $\eta=\eta(d,q)>0$ such that, with probability $1-o_n(1)$ over $G\sim\+G_{n,d}$, $Z_{G,q}$ is zero-free on $D([0,1],\eta)$.}
  \end{enumerate}
\end{theorem}

We remark that the parameter range in \Cref{thm:potts-zero} is also optimal. When $\beta_{\mathrm{RSB}}(d,q)>0$, the zero-free region in \Cref{thm:potts-zero}(1) cannot extend across $\beta_{\mathrm{RSB}}(d,q)$ (see \Cref{rem:sharpness-zero-free} in \Cref{sec:rs-identification-local}).

Via Barvinok's interpolation method, \Cref{thm:potts-zero} yields a deterministic polynomial-time approximation algorithm for both the antiferromagnetic Potts partition function and the number of proper $q$-colorings $Z_{G,q}(0)$, whose guarantee holds {with high probability over $G\sim\+G_{n,d}$}.

\begin{corollary}[Approximate counting for random regular antiferromagnetic Potts and $q$-colorings]\label{cor:potts-counting}
  Fix integers $d\ge3$ and $q\ge2$ and a rational parameter $\beta$, and assume either:
  \begin{enumerate}
    \item $\beta\in[\beta_{\mathrm{RSB}}(d,q)+\delta,1]$ for some $0<\delta<1-\beta_{\mathrm{RSB}}(d,q)$, or
    \item {$q\ge3$, $d<d_{\mathrm{RSB}}(q)$, and $\beta\in[0,1]$ (including proper $q$-colorings at $\beta=0$).}
  \end{enumerate}
  There is a deterministic algorithm that, given as input a $d$-regular graph $G$ with $n$ vertices and an accuracy parameter $0<\varepsilon<1$, runs in time $(n/\varepsilon)^{O_{d,q,\delta}(1)}$ ({or $(n/\varepsilon)^{O_{d,q}(1)}$ in case~2}) and outputs an estimate $\hat{Z}$ satisfying
  \[
    (1-\varepsilon)Z_{G,q}(\beta)\le\hat Z\le(1+\varepsilon)Z_{G,q}(\beta)
  \]
  with probability $1-o_n(1)$ over $G\sim\+G_{n,d}$.
\end{corollary}

Previously, the characterization of Fisher zero-free regions for $q$-colorings was established for general graphs of maximum degree $d$.
For $q \geq 2d$, Liu, Sinclair, and Srivastava~\cite{liu2025correlation} established a zero-free neighborhood around the real interval $[0,1]$ on the complex plane. This threshold was subsequently improved to $q \geq 1.998d$ by Bencs, Berrekkal, and Regts~\cite{bencs2026deterministic}, and very recently pushed to $q \geq 1.809d$ for $d\ge125$ by Shao and Shi~\cite{shao2026coupling}. For these degrees, this latest bound matches the state-of-the-art threshold for {almost uniform sampling of} $q$-colorings on bounded-degree graphs~\cite{carlson2025flip}. Through Barvinok's polynomial interpolation method, all of these zero-free regions directly yield fully polynomial-time approximation schemes (FPTAS) for counting the number of valid $q$-colorings.

From an algorithmic perspective (approximate sampling and counting), prior work has also predominantly focused on graphs with bounded maximum degree. Jerrum~\cite{jerrum1995simple} established the rapid mixing of Glauber dynamics for $q>2d$ by a coupling argument. Vigoda~\cite{vigoda1999improved} subsequently improved this threshold to $q>\frac{11}{6}d$ by analyzing flip dynamics. Subsequent work improved the constant below $11/6$~\cite{chen2019improved}, reaching $q\ge1.809d$ for $d\ge125$~\cite{carlson2025flip}.

{For random $d$-regular graphs, $q\ge3$, and positive interaction $\beta\in(0,1)$, Blanca et al.~\cite{blanca2020sampling} gave a polynomial-time sampling algorithm whenever the model has uniqueness on the $(d-1)$-ary tree and $\beta\ne(d-q)/d$.} Their sampler has total-variation error $O(n^{-\delta})$ for a constant $\delta>0$, with high probability over the graph. Very recently, Chen and Wang~\cite{chen2026high} proved rapid mixing of Glauber dynamics for fixed $q$, any fixed $\zeta\in(0,1/2)$, and sufficiently large $d$ when $1-\beta\le(1/2-\zeta)\frac{\log d}{d}$, extending beyond tree uniqueness. Neither result includes proper $q$-colorings at $\beta=0$.

It is also useful to compare with other thresholds for $q$-colorings on random regular graphs. A line of work~\cite{frieze1992independence,achlioptas2004chromatic,kemkes2010chromatic,coja2016chromatic} shows that there is a nonnegative sequence $\varepsilon_q\to0$ such that, for all sufficiently large $q$ and each fixed integer $d\ge3$, the following implications hold as $n\to\infty$ along admissible values:
\begin{equation}\label{eq:colorability-bounds}
\begin{aligned}
d&\le(2q-1)\log q-2\log2-\varepsilon_q
&&\Longrightarrow\quad G\sim \+G_{n,d}\text{ is \(q\)-colorable with high probability},\\
d&\ge(2q-1)\log q-1+\varepsilon_q
&&\Longrightarrow\quad G\sim \+G_{n,d}\text{ is not \(q\)-colorable with high probability}.
\end{aligned}
\end{equation}

For every fixed $\eps>0$ and sufficiently large $q$, algorithms are known to find a proper $q$-coloring in random $d$-regular graphs with high probability when $d\le (1-\eps)q\log q$; this follows from the locally sparse coloring result of~\cite{AIS19}. In contrast, factor-of-i.i.d.\ local algorithms face a barrier at the scale $d=(1+\eps)q\log q$: a color class would give a local independent set of density at least $1/q$, contradicting the independent-set bounds in~\cite{rahman2017local,gamarnik2017limits}. Finding colorings beyond this scale by general efficient algorithms remains an open challenge~\cite{molloy2019list}.

{In contrast, our zero-free region in \Cref{thm:potts-zero} and the deterministic approximate counting algorithm in \Cref{cor:potts-counting} extend to proper colorings whenever $d<d_{\mathrm{RSB}}(q)$. By~\eqref{eq:replica-symmetry-breaking-potts-asymptotic}, for every fixed $\delta>0$ and sufficiently large $q$, this includes $d\le(2q-1)\log q-2\log2-\delta$.}
The counting regime reaches the asymptotic sufficient colorability bound in~\eqref{eq:colorability-bounds}, substantially improves upon previous algorithmic bounds, and extends beyond the known guarantees for finding a proper coloring.

\paragraph{Antiferromagnetic Ising model.}

We also highlight the (zero-field) antiferromagnetic Ising model,
which is the special case $q=2$ of the antiferromagnetic Potts
model. Its replica symmetry breaking transition on random
regular graphs admits a particularly explicit description:
the threshold is known in closed form for every fixed $d\ge3$ and coincides with the Kesten--Stigum threshold~\cite{coja2022ising}:
\begin{equation}\label{eq:ising-rsb-threshold}
 \beta_{\mathrm{RSB}}(d,2)
 =\frac{\sqrt{d-1}-1}{\sqrt{d-1}+1}.
\end{equation}
Furthermore, the limiting
fluctuations of the partition function have been characterized
throughout the interior of the replica-symmetric
regime~\cite{fabian2021ising}.

Specializing \Cref{thm:potts-zero} and \Cref{cor:potts-counting}
to $q=2$ gives the following results.

\begin{corollary}[Zero-free region for random regular antiferromagnetic Ising]\label{cor:ising-zero}
 Fix an integer $d\ge3$, and let $0<\delta<\frac{2}{\sqrt{d-1}+1}$. There is a constant {$\eta=\eta(d,\delta)>0$} such that, with probability $1-o_n(1)$ over $G\sim\+G_{n,d}$, $Z_{G,2}$ is zero-free on {$D\bigl(\bigl[\frac{\sqrt{d-1}-1}{\sqrt{d-1}+1}+\delta,1\bigr],\eta\bigr)$}.
\end{corollary}

\begin{corollary}[Approximate counting for random regular antiferromagnetic Ising]\label{cor:ising-counting}
 Fix an integer $d\ge3$, let $0<\delta<\frac{2}{\sqrt{d-1}+1}$, and fix a rational $\beta\in\bigl[\frac{\sqrt{d-1}-1}{\sqrt{d-1}+1}+\delta,1\bigr]$. There is a deterministic algorithm that, given a $d$-regular graph $G$ with $n$ vertices and an accuracy parameter $0<\varepsilon<1$, runs in time $(n/\varepsilon)^{O_{d,\delta}(1)}$ and outputs an estimate $\hat Z$ satisfying
 \[
  (1-\varepsilon)Z_{G,2}(\beta)\le\hat Z\le(1+\varepsilon)Z_{G,2}(\beta)
 \]
 with probability $1-o_n(1)$ over $G\sim\+G_{n,d}$.
\end{corollary}

For graphs of maximum degree $d$, Fisher zero-free regions and
deterministic approximate counting are known throughout the
tree uniqueness regime $\beta>(d-2)/d$
\cite{sinclair2014approximation,liu2025correlation}. On random $d$-regular graphs, {the spectral-diameter algorithms of} Koehler, Lee, and Risteski~\cite{koehler2022sampling} {yield} polynomial-time
algorithms for both sampling and randomized approximate counting
when $1-\beta\le c/\sqrt d$, for a sufficiently small absolute
constant $c>0$.
In contrast, our results provide deterministic approximate counting on random regular graphs
throughout the interior of the replica-symmetric regime.
%A recent talk by Liu~\cite{liu2026disorderedtalk} also announces Fisher zero-freeness and an FPTAS for the antiferromagnetic Ising model on random regular graphs up to RSB, jointly with Bencs, Huang, Lee, and Regts. This announced result concerns the Ising special case, rather than the hardcore or general Potts models.

Finally, we emphasize that our algorithms are restricted to \emph{approximate counting} and do not directly yield an efficient algorithm for \emph{approximate sampling from the Gibbs measure}.
Conceptually, approximate counting can be viewed as the quantitative analogue of the decision (or existence) problem, whereas approximate sampling is the probabilistic counterpart of the search problem.
In worst-case complexity, randomized approximate counting and almost-uniform generation are polynomial-time interreducible for suitable efficiently self-reducible counting problems~\cite{jerrum1986random}; on random regular graphs, however, the standard vertex-pinning self-reduction leaves the $d$-regular ensemble.
Whether efficient approximate sampling also extends up to the replica symmetry breaking threshold{---}or whether counting and sampling exhibit an average-case computational gap analogous to that between existence thresholds and known search algorithms{---}remains an open question.

\subsubsection{Charting the replica-symmetric phase}
{
As a key intermediate step of our proofs, we give a characterization of the replica-symmetric phase on random $d$-regular graphs by rigorously establishing the equivalence of three notions of replica symmetry.

Note that we haven't formally defined the replica symmetry breaking thresholds $\alpha_{\mathrm{RSB}}(d),\beta_{\mathrm{RSB}}(d,q)$ and $d_{\mathrm{RSB}}(q)$. Historically, such a threshold has been conceptualized and defined through three distinct lenses across statistical physics, probability theory, and theoretical computer science\footnote{We remark that the naming choices, \emph{geometric}, \emph{variational} and \emph{thermodynamic} are not canonical in the literature, and are proposed by the authors for a clear comparison.}:
\begin{enumerate}
  \item \emph{(Geometric) replica symmetry:} 
  In Parisi's foundational theory of spin glasses \cite{parisi1979infinite,parisi1983order,mezard1987spin,
  talagrand2010mean}, replica symmetry is governed by the concentration of the empirical overlap between two independent Gibbs samples (replicas) around a deterministic matrix, which also implies averaged pairwise decorrelation (\Cref{def:replica-symmetry}) we use here:
  \[
   \lim\limits_{n\to \infty}\frac{1}{n^2}\sum\limits_{1\leq u<v\leq n}\E[G]{\DTV{\mu_{G,u,v}}{\mu_{G,u}\otimes \mu_{G,v}}}=0,
  \]
  i.e., two uniformly chosen vertices become asymptotically independent under the Gibbs measure~\cite{dembo2010gibbs,coja2016belief,coja2018charting}.
  \item \emph{Thermodynamic replica symmetry:} 
    The Bethe free-energy approximation has its origins in the Bethe--Peierls approach to lattice models \cite{bethe1935statistical}. It plays a central role in the cavity method for sparse disordered systems \cite{mezard2001bethe,mezard2003cavity,mezard2009information}. Here we call a model thermodynamically replica symmetric if its random free-energy density $n^{-1}\log Z_G$ converges in probability to the replica-symmetric Bethe prediction $\phi$ (formally defined in~\Cref{def:bethe-canonical-form} and~\Cref{def:free-energies}):
    \[
    n^{-1}\log Z_G\xrightarrow{\mathrm p} \phi.
    \]
    Rigorous Bethe free-energy results have been obtained in a variety of settings, including tree-uniqueness regimes on locally tree-like graphs \cite{dembo2013factor}.
    \item \emph{Variational replica symmetry:} 
  Motivated by the one-step replica-symmetry-breaking cavity picture \cite{monasson1998otpimization,mezard2001bethe,florent2007gibbs}, we consider the \emph{Bethe functional} $\mathcal{B}(\nu)$ on the space $\mathcal{P}_{\*\pi}$ of trial cavity message distributions with mean ${\*\pi}$ (formally defined in \Cref{def:rsb-bethe-functional}). Physically, instead of fixing cavity boundary marginals to the homogeneous tree marginal ${\*\pi}$, $\mathcal{B}(\nu)$ evaluates the local vertex-and-edge cavity functional when boundary marginals fluctuate across clusters according to $\nu$, relative to the homogeneous baseline $\mathcal{B}(\delta_{\*\pi})=0$. Mathematically, $\mathcal{B}(\nu)$ is the $m=1^-$ derivative of a normalized fractional-moment trial functional: any witness with $\mathcal{B}(\nu)>0$ yields a fractional-moment upper bound \cite{lelarge2018replica,ayre2022lower} strictly below $\phi$, forcing a macroscopic free-energy gap and signaling condensation. Here we call a model \emph{variationally replica symmetric} if
  \[\mathcal{B}(\nu)\le 0 \quad \text{for all } \nu\in\mathcal{P}_{\*\pi},\] so that $\delta_{\*\pi}$ globally maximizes $\mathcal{B}$; related variational criteria have been established across several sparse random models \cite{bapst2014condensation,coja2017information,coja2020inference}.
\end{enumerate}
}

In the one-step replica-symmetry-breaking condensation scenario analyzed for these models (within the satisfiable regime for constraint satisfaction problems), the cavity picture associates the condensation transition with the simultaneous onset of a subexponential family of dominant states, the loss of averaged pairwise decorrelation, and the failure of the replica-symmetric prediction $\phi$~\cite{florent2007gibbs,zdeborova2007phase,florent2008Potts,barbier2013hardcore}. This physical picture also motivates the variational criterion $\mathcal{B}\le 0$ used here~\cite{mezard2001bethe}, suggesting a common threshold for all three perspectives.

Rigorous work has established substantial connections among these notions, depending on the model and ensemble: on {Erd\H{o}s--R\'enyi} graphs, the Poisson Bethe criterion {identifies the condensation threshold separating replica-symmetric free-energy behavior from a regime where that prediction fails} for $q$-colorings and the Potts antiferromagnet \cite{bapst2014condensation,coja2017information}, and {below this threshold, averaged pairwise decorrelation and overlap concentration have also been established} \cite{coja2018charting,coja2020replica}. On random regular graphs, prior work relates geometric replica symmetry to the Bethe free energy {under additional assumptions} \cite{coja2016belief,coja2019spin}, equates ${\mathcal{B}}\le 0$ with the replica-symmetric quenched free-energy {identity} for positive-weight Potts \cite{coja2020inference}, establishes a hardcore variational formula over kernel distributions \cite{coja2019spin} and fractional-moment upper bounds \cite{lelarge2018replica,ayre2022lower}, and determines the antiferromagnetic Ising ($q=2$) free-energy transition \cite{coja2022ising}.

In our work, we establish an equivalence of geometric, variational, and thermodynamic replica symmetry on $\mathcal{G}_{n,d}$ across all three models {at fixed parameters away from their respective critical thresholds}. We state a unified
informal version here and defer the model-by-model formal statements {to \Cref{sec:rs-identification-local} (Theorems~\ref{thm:hardcore-three-characterizations}, \ref{thm:potts-three-characterizations}, and~\ref{thm:coloring-three-characterizations}}).

%Building on these tools, we establish a unified, pointwise equivalence of all three perspectives on $\mathcal{G}_{n,d}$ away from the critical thresholds across all three models (\Cref{thm:hardcore-three-characterizations,thm:potts-three-characterizations,thm:coloring-three-characterizations} in \Cref{sec:charting-rs}).

%Building on these tools, we establish a unified, pointwise equivalence of all three perspectives on $\mathcal{G}_{n,d}$ away from the critical thresholds across all three models (\Cref{thm:hardcore-three-characterizations,thm:potts-three-characterizations,thm:coloring-three-characterizations} in \Cref{sec:charting-rs}), including the hardcore model over $\mathcal{P}_{\pi_\alpha}$ and the hard-coloring endpoint $\beta=0$ ($d,q\ge 3$).

%The theorem below shows that these three perspectives identify the same phase transition, using a variational threshold in the spirit of the cavity condensation criterion~\cite{florent2007gibbs,bapst2014condensation,coja2017information}. 

\begin{theorem}[Charting the replica-symmetric phase, informal]\label{thm:hardcore-rs-transition}\label{thm:potts-rs-transition}\label{thm:coloring-rs-transition}
  For random regular graphs $G\sim\+G_{n,d}$, the three forms of replica symmetry (geometric, thermodynamic, variational) above have the following thresholds:
  \begin{enumerate}
    \item \textbf{Hardcore model.} For every integer $d\ge3$, there exists a threshold $\alpha_{\mathrm{RSB}}(d)\in(0,1/2)$ satisfying~\eqref{eq:replica-symmetry-breaking-hardcore-asymptotic}. All three forms of replica symmetry hold for $0<\alpha<\alpha_{\mathrm{RSB}}(d)$ and fail for $\alpha_{\mathrm{RSB}}(d)<\alpha<1/2$.
    \item \textbf{Antiferromagnetic Potts model.} For every pair of integers $d\ge3$ and $q\ge2$, there exists a threshold $\beta_{\mathrm{RSB}}(d,q)\in[0,1)$. All three forms of replica symmetry hold for $\beta_{\mathrm{RSB}}(d,q)<\beta\le1$ and fail for $0<\beta<\beta_{\mathrm{RSB}}(d,q)$.
    \item \textbf{Proper $q$-colorings.} For every integer $q\ge3$, there exists a (real-valued) degree threshold $d_{\mathrm{RSB}}(q)\in[3,\infty)$ satisfying~\eqref{eq:replica-symmetry-breaking-potts-asymptotic}. For every integer $d\ge3$, all three forms of replica symmetry hold if $d<d_{\mathrm{RSB}}(q)$ and fail if $d>d_{\mathrm{RSB}}(q)$. In the former case, $G$ is $q$-colorable with high probability.
  \end{enumerate}
\end{theorem}

\subsection{Related Work}

\paragraph{Fully connected mean-field spin glasses.}
Alongside diluted systems, \emph{fully connected} spin glasses constitute the other principal family of disordered mean-field models. In these systems, random interactions couple every pair of spins---or, in the case of $p$-body interactions, every $p$-tuple. Canonical examples include the Sherrington--Kirkpatrick (SK) model~\cite{kirkpatrick1975solvable,kirpatrick1978infinite}, the Ising and spherical $p$-spin models~\cite{gross1984simplest,gardner1985spin,crisanti1992spherical}, and fully connected vector and Potts spin glasses~\cite{deAlmeida1978infinite,gabay1981coexistence,gross1985mean}.

The most prominent example of fully connected mean-field spin glasses is the SK model, which is central to the development of the modern cavity method and theory of replica symmetry breaking~\cite{parisi1979infinite,mezard1986SK}. The SK model is a fully connected mean-field model of disordered magnets, in which every pair of Ising spins interacts through Gaussian random couplings that can favor either alignment or opposition.  In the absence of an external field, the SK model
undergoes a replica symmetry breaking phase transition at inverse temperature $\beta = 1$ under the standard variance-$1/n$ normalization of the couplings~\cite{talagrand2010mean,panchenko2013Sherrington}.

Substantial recent work has developed efficient sampling and counting algorithms throughout the replica-symmetric phase of the zero-field SK model. El Alaoui, Montanari, and Sellke gave polynomial-time sampling algorithms throughout $\beta < 1$ with vanishing normalized $2$-Wasserstein distance error, based on algorithmic stochastic localization~\cite{celetano2024sudakov,alaoui2025sampling}. Bencs et al. proved zero-freeness throughout $\beta < 1$ and obtained deterministic quasipolynomial-time approximation algorithms~\cite{bencs2026zeros}, whose running time was recently improved to deterministic polynomial time~\cite{schmidhuber2026polynomial,harbuzova2026color}.

There has also been a long line of work focusing on efficient sampling algorithms with vanishing {total-variation} error, with the tractable regime sequentially improved from $\beta<1/4$~\cite{bauerschmidt2019simple,eldan2022spectral,anari2022entropic,CE22} to $\beta<\beta^*\approx0.295$~\cite{anari2024trickle}, $\beta<1/2$~\cite{davies2026potential,davies2026weak,wang2026optimal}, $\beta<1/2+c$ for some small constant $c>0$~\cite{boban2026rank}, and finally very recently to $\beta<1$~\cite{lee2026potentialhessianascentiv,hou2026simulated,du2026sampling,OpenAISKCutoff26,OpenAISKSpectralGap26}. Recent manuscripts also claim polynomial-time mixing at the critical point
$\beta=1$~\cite{OpenAISKCriticalMixing26}, together with complementary
lower bounds on critical slowing down~\cite{OpenAISKCriticalSlowing26}.

\paragraph{Erd\H{o}s--R\'enyi random graphs.} Another prominent model of sparse random graphs, which also serves as a standard interaction geometry for diluted mean-field models in statistical physics, is the \emph{Erd\H{o}s--R\'enyi random graph} $G(n,d/n)$ with a fixed average degree $d$. This geometry models fluctuating connectivity, characterized by asymptotically Poisson-distributed degrees. It provides the classical setting for Poisson-diluted spin-glass models, including the well-known Viana--Bray model~\cite{viana1985phase}.

Cavity analyses of the hardcore and antiferromagnetic Potts models predict analogous replica-symmetry-breaking phenomena on both Erd\H{o}s--R\'enyi and random regular graphs, albeit with differing phase transition thresholds~{\cite{florent2008Potts,zhang2009stability,barbier2013hardcore}}. Rigorous results include establishing Bethe free-energy formulas for the hardcore model within uniqueness regimes~\cite{dembo2013factor}, characterizing the variational condensation threshold for proper $q$-colorings and the antiferromagnetic Potts model on Erd\H{o}s--R\'enyi graphs~\cite{coja2017information}, {and determining critical Ising magnetization and free-energy fluctuations~\cite{prodromidis2026distribution}}. For broader context, we also direct the reader to results concerning existence thresholds~\cite{frieze1990independence,achlioptas2005two,coja2013chasing}, {solution-space geometry and algorithmic barriers~\cite{achlioptas2008algorithmic},} search algorithms~\cite{mcdiarmid1984colouring,achlioptas1997analysis}, and approximate counting and sampling~{\cite{sinclair2017spatial,yin2016sampling,efthymiou2022sampling,efthymiou2023mixing,li2026rapid}}.

\paragraph{Random regular bipartite graphs.}
The study of spin systems has also been extended to \emph{random regular bipartite graphs}, a setting that has garnered substantial attention from the theoretical computer science community. On these structures, antiferromagnetic spin systems such as the hardcore and antiferromagnetic Potts models exhibit \emph{phase coexistence} in certain regimes, where the Gibbs distribution supports multiple macroscopically distinct ordered states. In certain regimes, this leads to exponentially slow mixing of single-site Glauber dynamics~{\cite{dyer2002counting,mossel2007hardness,li2026counting}} and supplies gadgets for approximate-counting hardness reductions~\cite{sly2014computational,galanis2015inapproximability,cai2016hardness}.

Consequently, this phase coexistence yields a distinctly different computational landscape on random regular bipartite graphs. To bypass the failure of local dynamics at high fugacities, Jenssen, Keevash, and Perkins~\cite{jenssen2020algorithms} developed an FPTAS and an efficient sampling algorithm for $\lambda > 50(\log d)^2/d$ (assuming sufficiently large $d$). Chen et al.~\cite{chen2022sampling} subsequently extended the range to $\lambda > 100\log d/d$, obtaining a fully polynomial-time randomized approximation scheme (FPRAS) and an efficient sampler and noting that an FPTAS also follows by interpolation. These guarantees hold with high probability over the random regular bipartite graph. More recently, bridging the gap from the low-fugacity regime, Kocurek, Oveis Gharan, and Tjowasi~\cite{kocurek2026sampling} obtained an FPRAS for $0 < \lambda \le c/\sqrt{d}$ for an absolute constant $c > 0$, thereby successfully covering all positive fugacities $\lambda > 0$ for sufficiently large $d$.  We refer the reader to~{\cite{liao2019counting,galanis2021fast,li2026counting,li2026fpras}} for some other related developments.

\paragraph{Lifting probabilistic arguments to the complex plane.}
As we will soon see in the next section, the zero-freeness results established in this work fall into the broader paradigm of \emph{lifting probabilistic arguments to the complex plane}.

There has been a recent surge of results adopting this methodology. For instance, Liu, Wang, Yin, and Yu~\cite{liu2025phase} lifted the information-percolation framework originally designed to prove the rapid mixing of Glauber dynamics to establish zero-freeness for hypergraph independent sets, an approach subsequently extended to hypergraph colorings~\cite{liu2026zero}.

{In random matrix theory, complex second-moment bounds on characteristic polynomials have been combined with tightness and trace limits to study the spectral radius of matrices with i.i.d.\ entries~\cite{bordenave2022convergence}, and with Jensen's formula to obtain non-asymptotic eigenvalue bounds for random matrices, including random regular graphs~\cite{mohanty2025eigenvalue}.} Additionally, Bencs et al.~\cite{bencs2026zeros} combined reweighted second-moment estimates with Jensen's formula to obtain zero-free regions for finite Gaussian mixtures of $p$-spin interactions with $p\ge2$ and Ising or spherical spins.

Further complementing this line of research, several recent works have successfully extended real-valued contraction arguments to the complex plane~{\cite{shao2020contraction,liu2025correlation,bencs2026deterministic,chen2026connective,shao2026coupling}}.

% END INCLUDED FILE: chapters/01-introduction-v2.tex
% BEGIN INCLUDED FILE: chapters/02-technical-overview.tex
\section{Technical Overview and Organization}
We now give a technical overview of our proof strategy and outline the organization of the paper. We focus on how we establish the zero-free regions for the hardcore and antiferromagnetic Potts models (\Cref{thm:hardcore-zero} and \Cref{thm:potts-zero}), as well as {the replica-symmetry characterizations (\Cref{thm:hardcore-rs-transition}; see \Cref{thm:hardcore-three-characterizations}, \Cref{thm:potts-three-characterizations}, and \Cref{thm:coloring-three-characterizations})} established along the way. General graph and probabilistic notation used below (including the $d$-regular configuration model $\mathrm{CM}=\mathrm{CM}_{n,d}$ and the normalized partition function $\mathcal Z_G=Z_G/\e^{n\phi}$) is collected in \Cref{subsec:preliminaries}.

At a high level, our approach combines the complex-analytic perspective of Bencs et al.~\cite{bencs2026zeros}---who established zero-freeness for the Sherrington--Kirkpatrick (SK) model up to its RSB threshold via complex second moments---with the variational cavity and planted overlap truncation framework developed for sparse random graphs~\cite{coja2017information,coja2018charting,coja2019spin,coja2020inference,coja2022ising}.
To establish complex zero-freeness throughout the replica-symmetric phase on random regular graphs in the presence of hard constraints (the hardcore model and proper $q$-colorings at $\beta=0$), we organize the proof around three components (see \Cref{fig:proof-structure}):
\begin{enumerate}
 \item \emph{Reinforced small subgraph conditioning (\Cref{subsec:overview-rssc}):}
  In the SK setting, Bencs et al.~\cite{bencs2026zeros} use an analytic random reweighting together with Jensen's formula and complex moment estimates to handle partition-function fluctuations. For the random regular graph models and parameter regimes considered here, the relevant fluctuations are governed by discrete short-cycle counts. We formulate \emph{reinforced small subgraph conditioning (RSSC)}, which lifts the classical Robinson--Wormald small subgraph conditioning method~\cite{robinson1992almost,robinson1994almost,janson1995random} to random holomorphic functions and deduces complex zero-freeness from real-line cycle conditioning and locally uniform complex second-moment bounds.
  \item \emph{Charting the replica-symmetric phase (\Cref{subsec:charting-rs}):} To connect our zero-free regions to the replica-symmetric phase, we show that averaged pairwise decorrelation (geometric replica symmetry), convergence of the free-energy density to the Bethe prediction (thermodynamic replica symmetry), and the Bethe variational criterion (variational replica symmetry) identify a common replica symmetry breaking threshold in each model (\Cref{thm:hardcore-rs-transition}). The proof combines cavity arguments and fractional-moment bounds with entropy estimates and planted-overlap concentration. We obtain this three-way phase characterization for the hardcore and antiferromagnetic Potts models and extend it to proper colorings, linking the variational condition used in our zero-freeness proof to the geometric and thermodynamic descriptions of the phase.
  \item \emph{Planted overlap concentration and graph truncation under variational replica symmetry (\Cref{subsec:overview-variational}):}
  On sparse random graphs, the unrestricted second moment can grow exponentially relative to the square of the first moment before RSB~\cite{coja2013chasing,coja2016chromatic}. Following the cavity-method framework of~\cite{coja2017information,coja2018charting,coja2020inference,coja2022ising}, we replace the bare second-moment criterion by the variational replica symmetry condition, i.e., the nonpositivity of the \emph{Bethe functional}, and {establish exponential concentration of planted overlaps (\Cref{lem:exponential-overlap}) and an exponentially accurate typical-graph truncation (\Cref{lem:graph-truncation}), accommodating hard constraints to obtain complex zero-freeness throughout the interior of the variational replica-symmetric regime (\Cref{thm:variational-zero-free}).}
\end{enumerate}

We now discuss each of these three steps in greater detail.

\subsection{Second Moment Method and Reinforced Small Subgraph Conditioning}\label{subsec:overview-rssc}
The second moment method is a classical strategy for proving that a real, nonnegative partition function satisfies $Z > 0$ with high probability: if one establishes the sharp second-moment bound
\[
\E{Z^2}=(1+o_n(1))\E{Z}^2,
\]
then the Paley--Zygmund inequality immediately yields
\[
  \Pr{Z>0}\geq \frac{\E{Z}^2}{\E{Z^2}}=1-o_n(1).
\]
On random $d$-regular graphs, one performs moment calculations in Bollob\'as's \emph{configuration model} $\mathrm{CM}$~\cite{bollobas1980probabilistic} (formally defined in \Cref{subsec:preliminaries}). However, a naive calculation in $\mathrm{CM}$ typically yields only the weaker bound
\[
  \E{Z^2}=O(1)\cdot \E{Z}^2,
\]
because the partition function suffers from $O(1)$ multiplicative fluctuations driven by the counts $X_{\ell,n}$ of \emph{short cycles} of length $\ell$ in the underlying graph {(\Cref{def:short-cycle-factors})}.

This obstacle was famously resolved on the real line by Robinson and Wormald via the \emph{small subgraph conditioning method}~\cite{robinson1992almost,robinson1994almost,janson1995random}, originally introduced to prove that random regular graphs are asymptotically almost surely Hamiltonian. Their key insight is that the $O(1)$ multiplicative fluctuations of $Z/\E{Z}$ are completely explained by the counts of short cycles in the graph: under the full small subgraph conditioning hypotheses---when the short-cycle counts are asymptotically independent Poisson random variables, {the Poisson-weighted squared short-cycle factors are summable}, and the mixed factorial moments and second moment match the resulting cycle-factor predictions---$Z/\E{Z}$ converges in distribution to a random infinite product of per-cycle factors, which is strictly positive almost surely as long as no individual cycle factor vanishes.

To establish zero-free regions in the complex plane, we introduce \emph{reinforced small subgraph conditioning (RSSC)} (\Cref{thm:reinforced-ssc} in \Cref{sec:reinforced-ssc}), which lifts small subgraph conditioning from a single real parameter to random elements of the space $H(D)$ of holomorphic functions on a complex domain $D\subseteq\=C$, equipped with the topology of locally uniform convergence {(\Cref{subsec:preliminaries})}.

The idea is to combine classical small subgraph conditioning along a real interval $I\subset D\cap\=R$ with a complex second-moment bound $\E{|Y_n(z)/\E{Y_n(z)}|^2}=O(1)$ uniform on each compact subset of $D$, where $Y_n$ is the typical-graph truncation used below.
By Montel's theorem and the identity theorem, the complex second-moment bound provides tightness while the real-line limit uniquely pins down the limiting random holomorphic function across all of $D$ as a holomorphic infinite product of cycle factors.

Conditioning on the absence of self-loops and parallel edges (to pass from $\mathrm{CM}$ to simple $d$-regular graphs) removes the $\ell=1,2$ power factors, so under locally uniform summability $\sum_{\ell\ge1}\lambda_\ell|\delta_\ell(z)|^2<\infty$ of the Poisson means $\lambda_\ell=(d-1)^\ell/(2\ell)$ and short-cycle factors {$\delta_\ell(z)$ (\Cref{def:short-cycle-factors})} and provided the remaining {per-cycle} factors $1+\delta_\ell(z)$ ($\ell\ge3$) are nowhere zero on $D$, this limiting function is almost surely nowhere zero on $D$. Continuity of the minimum modulus on each fixed nonempty compact subset $K\subset D$ then implies that $Y_n$ is zero-free on $K$ with high probability. This forces the typical-graph event, on which $Y_n$ and the full partition function have the same zeros.

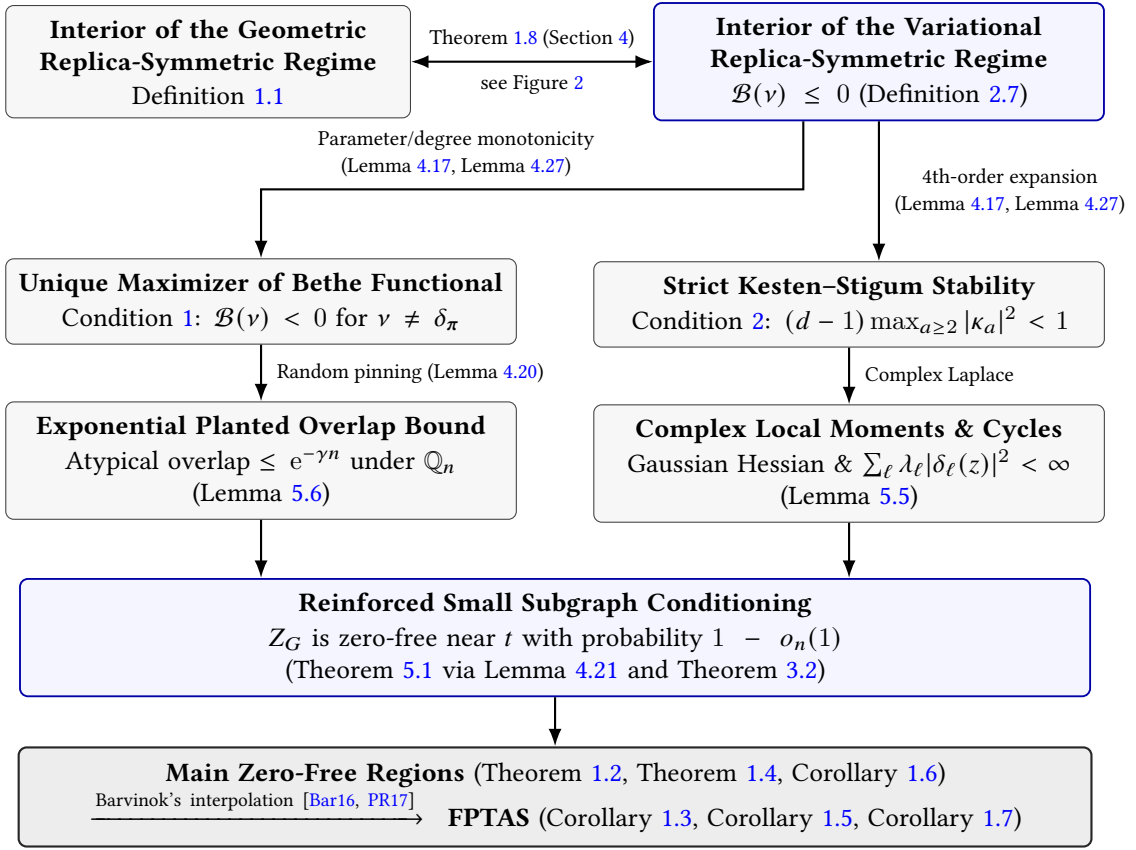
\begin{figure}[t]
\centering
\begin{tikzpicture}[
  >=Latex,
  box/.style={
    draw=black!75,
    rounded corners=3pt,
    fill=gray!6,
    inner sep=5pt,
    align=center,
    font=\small
  },
  corebox/.style={
    draw=blue!65!black,
    line width=0.7pt,
    rounded corners=3pt,
    fill=blue!4,
    inner sep=5pt,
    align=center,
    font=\small
  },
  goalbox/.style={
    draw=black!85,
    line width=0.8pt,
    rounded corners=3pt,
    fill=gray!14,
    inner sep=5.5pt,
    align=center,
    font=\small
  },
  lbl/.style={
    font=\scriptsize,
    text=black,
    align=center
  }
]
  % Top row: interiors of the geometric and variational replica-symmetric regimes
  \node[box, text width=5.0cm] (rs) {
    \textbf{{Interior of the Geometric\\Replica-Symmetric Regime}}\\[1pt]
    {\Cref{def:replica-symmetry}}
  };
  \node[corebox, text width=5.6cm, right=32mm of rs] (cond1) {
    \textbf{{Interior of the Variational\\Replica-Symmetric Regime}}\\[1pt]
    {$\mathcal B(\nu)\le0$ (\Cref{def:rsb-bethe-functional})}
  };
  \draw[<->, thick] (rs) --
    node[above=1pt, lbl] {{\Cref{thm:hardcore-rs-transition} (\Cref{sec:rs-identification-local})}}
    node[below=1pt, lbl] {{see \Cref{fig:charting-rs}}}
    (cond1);

  % Second row: aligned with left edge of rs and right edge of cond1
  \node[box, text width=6.4cm, anchor=north west] (cond2) at ([yshift=-18.5mm]rs.south west) {
    \textbf{{Unique Maximizer of Bethe Functional}}\\[1pt]
    \Cref{cond:unique-variational-maximizer}: $\mathcal B(\nu)<0$ for $\nu\ne\delta_{\*\pi}$
  };
  \node[box, text width=6.4cm, anchor=north east] (cond3) at ([yshift=-18.5mm]cond1.south east) {
    \textbf{Strict Kesten--Stigum Stability}\\[1pt]
    \Cref{cond:strict-KS}: $(d-1)\max_{a\ge2}|\kappa_a|^2<1$
  };

  \draw[->, thick] ([xshift=-10mm]cond1.south) -- ++(0,-9mm) -| node[pos=0.32, above=1pt, lbl] {Parameter/degree monotonicity\\(\Cref{lem:model-variational-conditions}, \Cref{lem:coloring-variational-conditions})} (cond2.north);
  \draw[->, thick] (cond1.south) -- (cond1.south |- cond3.north) node[midway, right=2pt, lbl] {4th-order expansion\\(\Cref{lem:model-variational-conditions}, \Cref{lem:coloring-variational-conditions})};

  % Third row: Global overlap bound & Local moments
  \node[box, text width=6.4cm, below=7.5mm of cond2] (global) {
    \textbf{Exponential Planted Overlap Bound}\\[1pt]
    Atypical overlap $\le \e^{-\gamma n}$ under $\mathbb Q_n$\\
    (\Cref{lem:exponential-overlap})
  };
  \node[box, text width=6.4cm, below=7.5mm of cond3] (local) {
    \textbf{Complex Local Moments \& Cycles}\\[1pt]
    Gaussian Hessian \& $\sum_\ell\lambda_\ell|\delta_\ell(z)|^2<\infty$\\
    (\Cref{lem:local-profile-moments})
  };

  \draw[->, thick] (cond2) -- node[right=2pt, lbl] {Random pinning (\Cref{lem:revealed-entropy-bound})} (global);
  \draw[->, thick] (cond3) -- node[right=2pt, lbl] {Complex Laplace} (local);

  % Fourth row: RSSC -> Zero-Freeness
  \coordinate (mid34) at ($(global.south)!0.5!(local.south)$);
  \node[corebox, text width=13.8cm, anchor=north] (rssc) at ([yshift=-7.5mm]mid34) {
    \textbf{Reinforced Small Subgraph Conditioning}\\[1pt]
    $Z_G$ is zero-free near $t$ with probability $1-o_n(1)$\\ (\Cref{thm:variational-zero-free} via \Cref{lem:graph-truncation} and \Cref{thm:reinforced-ssc})
  };

  \draw[->, thick] (global.south) -- (global.south |- rssc.north);
  \draw[->, thick] (local.south) -- (local.south |- rssc.north);

  % Bottom row: Main Zero-Free Theorems & Counting Algorithms
  \node[goalbox, text width=13.8cm, below=6.5mm of rssc] (main) {
    \textbf{Main Zero-Free Regions} (\Cref{thm:hardcore-zero}, \Cref{thm:potts-zero}, \Cref{cor:ising-zero})
    $\;\xrightarrow{\text{Barvinok's interpolation~\cite{barvinok2016combinatorics,patel2017deterministic}}}\;$
    \textbf{FPTAS} (\Cref{cor:hardcore-counting}, \Cref{cor:potts-counting}, \Cref{cor:ising-counting})
  };

  \draw[->, thick] (rssc) -- (main);
\end{tikzpicture}
\caption{Overview of the proof architecture across \Cref{sec:reinforced-ssc}, {\Cref{sec:rs-identification-local}, and \Cref{sec:variational-rs}}. The implications to the strict conditions apply at nontrivial parameters; the independent-spin endpoints are handled directly. The counting guarantees hold with high probability.}
\label{fig:proof-structure}
\end{figure}

\subsection{Charting the Replica-Symmetric Phase}\label{subsec:charting-rs}
{The second component of our framework bridges geometric replica symmetry, thermodynamic replica symmetry, and variational replica symmetry.
While reinforced small subgraph conditioning reduces zero-freeness to second-moment control, verifying complex second-moment bounds on sparse random graphs encounters a fundamental obstruction: rare graph fluctuations cause the bare second moment to blow up exponentially inside the replica-symmetric phase.
As discussed in \Cref{subsec:overview-variational}, overcoming this barrier requires truncating the graph ensemble via the exponential concentration of replica overlaps, a property governed technically by the \emph{variational} condition $\mathcal B(\nu)\le0$.
However, the natural physical and algorithmic boundaries of the replica-symmetric phase are defined by the decay of spatial correlations (\emph{geometric replica symmetry}) and the convergence of the free-energy density to the Bethe formula (\emph{thermodynamic replica symmetry}).
To ensure that our zero-free regions extend up to the true replica symmetry breaking threshold, we prove that these three distinct perspectives identify a common critical threshold (\Cref{thm:hardcore-rs-transition}; formally \Cref{thm:hardcore-three-characterizations} and \Cref{thm:potts-three-characterizations}).}

To formalize this equivalence, we consider replica symmetry through three distinct physical and probabilistic lenses:
\begin{itemize}[leftmargin=1.8em,itemsep=0.25em,topsep=0.25em]
  \item \textbf{Geometric replica symmetry (\Cref{def:replica-symmetry}):}
  At the level of spatial correlations, spins at two uniformly chosen vertices become asymptotically independent under the Gibbs measure of the sampled graph, in the sense of averaged pairwise decorrelation.
  \item \textbf{Thermodynamic replica symmetry (\Cref{def:free-energies}):} At the analytic level, the free-energy density converges in probability to the real-analytic Bethe prediction: $n^{-1}\log Z_G \xrightarrow{\mathrm p} \phi$ (\Cref{def:bethe-canonical-form}).
  \item \textbf{Variational replica symmetry (\Cref{def:rsb-bethe-functional}):} At the cavity-field level, the Bethe functional $\mathcal B(\nu)$---originating in the one-step replica-symmetry-breaking cavity method~\cite{bethe1935statistical,mezard2001bethe,coja2019spin} as the expected free-energy shift under local graph increments with boundary messages drawn from $\nu\in\mathcal P_{\*\pi}$---satisfies $\mathcal B(\nu)\le0$ for all trial distributions $\nu$, meaning that the replica-symmetric fixed point is a global maximizer.
\end{itemize}

\begin{figure}[t]
\centering
\begin{tikzpicture}[
  >=Latex,
  box/.style={
    draw=black!75,
    rounded corners=3pt,
    fill=gray!6,
    inner sep=4pt,
    align=center,
    font=\small
  },
  corebox/.style={
    draw=blue!65!black,
    line width=0.7pt,
    rounded corners=3pt,
    fill=blue!4,
    inner sep=4pt,
    align=center,
    font=\small
  },
  lbl/.style={
    font=\scriptsize,
    text=black,
    align=center
  }
]
  \node[box, text width=3.2cm, minimum height=14.5mm] (geom) {
    \textbf{{Geometric RS}}\\[1pt]
    {\footnotesize\Cref{def:replica-symmetry}:\\[-1pt]correlation decay}
  };
  \node[box, text width=3.2cm, minimum height=14.5mm, right=26mm of geom] (thermo) {
    \textbf{{Thermodynamic RS}}\\[1pt]
    {\footnotesize\Cref{def:free-energies}:\\[-1pt]$\frac1n\log Z_G\xrightarrow{\mathrm p}\phi$}
  };
  \node[corebox, text width=3.2cm, minimum height=14.5mm, right=26mm of thermo] (var) {
    \textbf{{Variational RS}}\\[1pt]
    {\footnotesize\Cref{def:rsb-bethe-functional}:\\[-1pt]$\mathcal B(\nu)\le0$}
  };

  \draw[->, thick, draw=blue!70!black] (geom.east) --
    node[above=1.5pt, lbl] {Cavity arguments\\[-1pt](\Cref{subsec:local-rs-free-energy})}
    (thermo.west);

  \draw[->, thick, draw=blue!70!black] (thermo.east) --
    node[above=1.5pt, lbl] {Free-energy gaps\\[-1pt](\Cref{subsec:local-bethe-consequences})}
    (var.west);
  \draw[->, thick, draw=blue!70!black] (var.south) --
    ([yshift=-6mm]var.south) --
    node[below=1.5pt, lbl] {Entropy bound and truncation\\[-1pt](\Cref{subsec:local-variational-geometric})}
    ([yshift=-6mm]geom.south) -- (geom.south);
\end{tikzpicture}
\caption{{Proof structure for charting the replica-symmetric phase in \Cref{sec:rs-identification-local}.}}
\label{fig:charting-rs}
\end{figure}
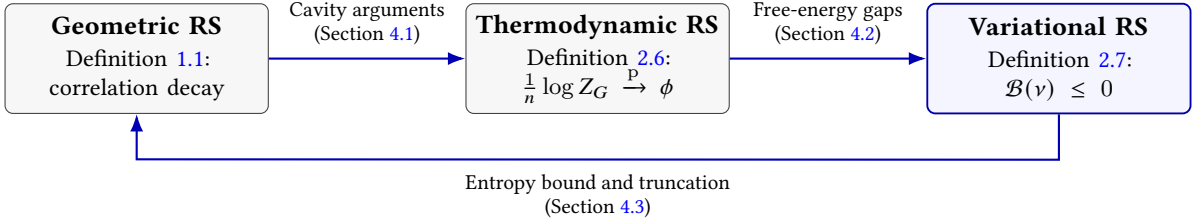

{
We highlight the key ideas and technical innovations in the proof outlined in \Cref{fig:charting-rs}.

\paragraph{From geometric to thermodynamic replica symmetry (\Cref{subsec:local-rs-free-energy}).}
Under geometric replica symmetry, spins at any fixed number of uniformly sampled vertices asymptotically factorize. This supports cavity calculations of free-energy increments under local graph modifications~\cite[Section~1.2]{coja2019spin}.
Specifically, following the cavity approach of~\cite{coja2019spin}, we use a \emph{two-root coupling} for the hardcore model to relate free-energy increments to the BP Bethe functional. The limiting cavity-message distributions satisfy the distributional BP equation, yielding a lower bound on the quenched free energy by the minimum of this functional over BP fixed points. Together with a fractional-moment interpolation upper bound, this identifies the limiting free energy with that functional minimum. A further fractional-moment argument shows that every minimizing message law is a point mass at the homogeneous fixed point, yielding $n^{-1}\E{\log Z_G}\to\phi_d(\alpha)$; concentration then gives convergence in probability.

\paragraph{From thermodynamic to variational replica symmetry (\Cref{subsec:local-bethe-consequences}).}
We establish the second implication by contraposition: if variational replica symmetry fails, there exists a witness distribution $\nu\in\mathcal P_{\*\pi}$ such that $\mathcal B(\nu)>0$, representing a trial 1RSB phase separation.
To transfer this witness to the macroscopic free energy, we connect the fractional Bethe functional $\mathcal F_m(\eta_m)$, evaluated at a carefully constructed family of message distributions $\eta_m$, to $\mathcal B(\nu)$ by showing that its left derivative at $m=1$ equals $\mathcal B(\nu)$. Consequently, the presence of a positive witness drives the fractional moment strictly below the Bethe prediction, forcing a \emph{strict free-energy gap}: $\limsup n^{-1}\E{\log Z_G}\le\phi-\varepsilon$ for some $\varepsilon>0$. Thus, convergence to $\phi$ forces $\mathcal B\le0$ everywhere.

\paragraph{From variational to geometric replica symmetry (\Cref{subsec:local-variational-geometric}).}
At positive parameters, we close the implication cycle by deducing geometric replica symmetry from concentration of the overlap of two Gibbs replicas.
We first analyze the overlap under the \emph{planted graph distribution}, which reweights the random graph measure by the partition function.
In the interior of the variational replica-symmetric regime, the Bethe functional is uniquely maximized at the homogeneous message law; an entropy bound for revealed spins converts this uniqueness into exponential concentration of the planted overlap.
To return to the quenched ensemble via change of measure, we combine this entropy bound with the concentration of $\log Z_G$ to obtain exponential lower-tail control of the partition function.
This control ensures that the planted overlap bound survives the change of measure, yielding overlap concentration in the original ensemble and establishing geometric replica symmetry.
At the hard-coloring endpoint, graph truncation and classical small subgraph conditioning supply the lower-tail estimate needed to obtain overlap concentration.}

\subsection{Zero-Freeness under {Variational Replica Symmetry}}\label{subsec:overview-variational}
With the phase characterization established, we return to the complex second-moment obstruction described above. Rare graphs with correlated replica clusters can inflate $\E[\mathrm{CM}]{Z_G^2}$ by an exponential factor even strictly inside the replica-symmetric phase~\cite{achlioptas2004chromatic,achlioptas2005two,kemkes2010chromatic,coja2013chasing,coja2016chromatic,coja2017information}.

\paragraph{Truncating graphs via overlap concentration.}
{We overcome this barrier in \Cref{sec:variational-rs} by restricting RSSC to a high-probability event $E_n$ that places an upper bound on the normalized partition function at a real parameter. We set $Y_n(z)=\mathcal Z_G(z)$ on $E_n$ and $Y_n(z)=0$ otherwise, where $\mathcal Z_G$ is defined in \Cref{def:bethe-canonical}. To control the second moment, we study the \emph{overlap} (\Cref{def:local-second-moment}) of two independent configurations drawn from the Gibbs measure $\mu_G$~\cite{coja2017information,coja2018charting,coja2020inference}.}
For the fixed-degree model $\mathrm{CM}_{n,d}$ used in the zero-free proof, this concentration is analyzed under the \emph{planted distribution}, which reweights $G$ proportionally to $Z_G$.

While an $o_n(1)$ bound on atypical overlaps under the planted distribution suffices at real parameters, extending $Z_G(z)$ to a complex neighborhood of a positive parameter can change its magnitude by an exponential factor in $n$, which would overwhelm an $o_n(1)$ truncation error.
We therefore strengthen the random-pinning method of~\cite{coja2020inference} to prove that atypical overlaps occur with \emph{exponentially small} probability under the planted distribution (\Cref{lem:exponential-overlap}).

{At positive parameters, combining this planted overlap bound with the upper cutoff gives an exponentially small total contribution from atypical pairs on $E_n$; it does not require every graph in $E_n$ to have a small atypical-overlap weight. At the coloring endpoint, the truncation estimate is obtained by continuation from a sufficiently small fixed positive parameter. The resulting truncation $Y_n(z)$ (\Cref{lem:graph-truncation}) has a controlled second moment across a small complex disk, so RSSC applies.}

\paragraph{From {variational replica symmetry} to zero-freeness and approximate counting.}
{Our unified zero-freeness criterion (\Cref{thm:variational-zero-free}, stated and proved in \Cref{sec:variational-rs}) assumes the following two strict conditions at a nontrivial parameter $t$; the independent-spin endpoints $\alpha=0$ and $\beta=1$ are handled directly. On the replica-symmetric side of the respective thresholds, the conditions follow from parameter slack at positive interaction and real-degree slack at the coloring endpoint (\Cref{lem:model-variational-conditions} and \Cref{lem:coloring-variational-conditions}):}
\begin{enumerate}
  \item \emph{Unique maximizer of the Bethe functional (\Cref{cond:unique-variational-maximizer}: $\mathcal B(\nu)<0$ for $\nu\ne\delta_{\*\pi}$):}
  Obtained via strict parameter or degree monotonicity.
  Unlike bare second-moment bounds~\cite{kemkes2010chromatic,coja2016chromatic,Bhatnagar2016decay}, which require {independent replicas around $\*\pi\otimes\*\pi$ to dominate the unrestricted second moment}, \Cref{cond:unique-variational-maximizer} only requires $\delta_{\*\pi}$ to uniquely maximize {the Bethe functional} $\mathcal B(\nu)$; via random pinning (\Cref{lem:revealed-entropy-bound}), this suffices to establish exponential planted overlap concentration (\Cref{lem:exponential-overlap}) and construct $E_n$.
  \item \emph{Strict Kesten--Stigum stability (\Cref{cond:strict-KS}: $(d-1)\max_{a\ge2}|\kappa_a|^2<1$):}
  Obtained via a fourth-order Taylor expansion around $\delta_{\*\pi}$~\cite{coja2018charting}, where $(\kappa_a)_{a\ge2}$ are the nonconstant eigenvalues of the tree broadcasting operator (\Cref{def:bethe-canonical}). On the central overlap box around $\*\pi\otimes\*\pi$, \Cref{cond:strict-KS} ensures a negative-definite local second-moment Hessian and a convergent short-cycle series $\sum_{\ell\ge1}\lambda_\ell|\delta_\ell(z)|^2<\infty$, while the explicit hardcore and Potts eigenvalues ensure $1+\delta_\ell(z)\ne0$ for $\ell\ge3$ ({\Cref{def:short-cycle-factors}, \Cref{subsec:pointwise-conditions}, \Cref{subsec:common-moments}, and} \Cref{lem:local-profile-moments}).
\end{enumerate}
{Finally, in \Cref{subsec:main-zero-free-counting-proofs}, combining \Cref{thm:variational-zero-free} across compact parameter intervals with the phase characterizations of \Cref{sec:rs-identification-local} (\Cref{thm:hardcore-three-characterizations}, \Cref{thm:potts-three-characterizations}, and \Cref{thm:coloring-three-characterizations}) and invoking Barvinok's Taylor polynomial interpolation~\cite{barvinok2016combinatorics} and Patel--Regts connected-subgraph enumeration~\cite{patel2017deterministic} (\Cref{lem:barvinok-interpolation} and \Cref{lem:patel-regts}) yields \Cref{thm:hardcore-zero}, \Cref{thm:potts-zero}, \Cref{cor:ising-zero}, and the deterministic counting algorithms in \Cref{cor:hardcore-counting}, \Cref{cor:potts-counting}, and \Cref{cor:ising-counting}.}

\subsection{Preliminaries}\label{subsec:preliminaries}

\paragraph{Notations.} We first present some notations that will be used throughout the paper.
\begin{itemize}[leftmargin=1.5em,itemsep=0.25em,topsep=0.25em,parsep=0pt]
  \item {\emph{Sets, combinatorics, and asymptotics.}}
  {For an integer $m\ge1$, we write $[m]\defeq\{1,\ldots,m\}$ (with $[0]\defeq\varnothing$).} We write $\one{A}$ for the indicator of $A$, and for $\*x,\*y\in S^n$, {$d_{\mathrm{Ham}}(\*x,\*y)\defeq\sum_{v=1}^n\one{x_v\ne y_v}$ for their Hamming distance. For integers $m\ge1$ and $k\ge0$, we write $m!!\defeq m(m-2)\cdots 1$ for odd $m$, $m!!\defeq m(m-2)\cdots 2$ for even $m$ (with $0!=0!!=(-1)!!=1$), and $(x)_k\defeq x(x-1)\cdots(x-k+1)$ (with $(x)_0\defeq 1$) for the falling factorial.} All logarithms are natural, with $0\log0=0$, $\log0=-\infty$, and $\exp(-\infty)=0$. {We write $f\asymp g$ when $f=\Theta(g)$, and subscripts on asymptotic notation such as $O_\theta(\cdot)$ indicate parameters on which the implicit constants may depend.}

  \item {\emph{Vectors, matrices, and norms.}}
  We write {$\allone=\allone_m$ and $\*0=\*0_m$ for the all-ones and all-zeros vectors,} $\*U_m\defeq m^{-1}\allone_m$ for the uniform distribution on $[m]$, {and $\*e_a\in\=R^m$ ($a\in[m]$) for the $a$th standard basis vector.} For arrays and matrices, $\|\cdot\|_p$ {($p\in[1,\infty]$)} denotes the entrywise $\ell_p$ norm {(including $\|\cdot\|_\infty$ for the entrywise maximum norm)}, {$\|\cdot\|_{\mathrm{op}}$ and $\|\cdot\|_{\mathrm F}$ denote the operator ($\ell_2\to\ell_2$) and Frobenius norms, $\operatorname{tr}(\cdot)$ denotes the matrix trace, $I=I_m$ denotes the identity matrix,} $\diag(\*v)$ denotes the diagonal matrix with diagonal $\*v$, and inequalities such as $P\ge0$ or $\*\pi>0$ on vectors or matrices are understood entrywise.

  \item {\emph{Graphs and multigraphs.}}
  We use $d\ge3$ for an integer regular degree. All $d$-regular graph limits are taken along \emph{admissible} $n$ ($n\ge d+1$ with $dn$ even), as in \Cref{subsec:our-results}. Real $r\ge0$ is used only to interpolate the Bethe functional in \Cref{def:rsb-bethe-functional}. {For a multigraph $G$, write $N_G(u)$ for the neighborhood of $u\in V(G)$, $\operatorname{dist}(u,v)=\operatorname{dist}_G(u,v)$ for the shortest-path distance between $u,v\in V(G)$, $G-v$ and $G-S$ for the subgraphs obtained by deleting a vertex $v$ or vertex subset $S$, $G+F$ for the multigraph obtained by adding a multiset of edges $F$ to $G$, and $\operatorname{mult}_G(e)$ for the multiplicity of an edge $e$.}

  \item {\emph{Probability and information measures.}}
  We write {$\Delta_\Omega$ for the probability simplex on a finite set $\Omega$, $\operatorname{Mult}(S,\*p)$ for the multinomial distribution of type counts in $S\ge1$ independent draws from $\*p\in\Delta_\Omega$, $\mathcal P(S)$ for the space of probability distributions on $S$ (with support $\operatorname{supp}(\cdot)$), $C(S)$ for the space of continuous real-valued functions on $S$,} $\otimes$ for tensor products, {$X_n\xrightarrow{\mathrm p}X$ for convergence in probability,} $X_n\xrightarrow{d}X$ for convergence in distribution, and $X\overset{d}{=}Y$ for equality in distribution. {For deterministic $a_n>0$, $X_n=o_{\mathrm p}(a_n)$ means $X_n/a_n\xrightarrow{\mathrm p}0$.}

  {For a probability measure $\mu$ on $\Omega^V$ and vertices $v_1,\ldots,v_k\in V$, we write $\mu_{v_1,\ldots,v_k}\in\Delta_{\Omega^k}$ (and $\mu_{v_i}\in\Delta_\Omega$) for the marginal distribution of $(\sigma_{v_1},\ldots,\sigma_{v_k})$ under $\sigma\sim\mu$. For configurations $\sigma,\sigma'\in\Omega^V$ on a finite vertex set $V$ ($n=|V|$), we write $\*\rho(\sigma)\in\Delta_\Omega$ and $\*\rho(\sigma,\sigma')\in\Delta_{\Omega\times\Omega}$ for the \emph{empirical vertex marginal} and \emph{empirical overlap} with entries $\rho_a(\sigma)\defeq\frac1n\sum_{v\in V}\one{\sigma_v=a}$ and $\rho_{ab}(\sigma,\sigma')\defeq\frac1n\sum_{v\in V}\one{\sigma_v=a,\,\sigma'_v=b}$.} {A symmetric distribution $Q=(Q_{ab})_{a,b\in\Omega}\in\Delta_{\Omega\times\Omega}$ (called an \emph{edge marginal}, viewed as an $|\Omega|\times|\Omega|$ matrix) has \emph{vertex marginal} $\*\rho\in\Delta_\Omega$ if $Q\allone=\*\rho$.}

  For probability laws $\*p,\*q$, total variation distance and KL divergence are defined by
  \[
    \DTV{\*p}{\*q}\defeq\frac12\sum_x|p_x-q_x|,\qquad
    \DKL(\*p\Vert\*q)\defeq\sum_x p_x\log\frac{p_x}{q_x},
  \]
  with $0\log(0/q)=0$ also at $q=0$ and $p\log(p/0)=+\infty$ for $p>0$.
\end{itemize}

\paragraph{Configuration model and random regular graphs.}
On multigraphs, we adopt the convention that a self-loop forbids its vertex from an independent set, and each edge occurrence contributes to the monochromatic edge count ${m_G(\sigma)\defeq|\{uv\in E(G):\sigma_u=\sigma_v\}|}$.

\begin{definition}[Configuration model]\label{def:config-model}
  For a deterministic vector $\*d=(d_v)_{v\in[n]}\in\=Z_{\ge0}^n$, let $\mathrm{CM}_{\*d}$ be the law of the multigraph induced by a uniformly chosen maximum pairing of the $\sum_{v\in[n]}d_v$ labeled half-edges. If the total is odd, one half-edge remains unmatched.

  For an integer $d\ge3$, write $\mathrm{CM}=\mathrm{CM}_{n,d}\defeq\mathrm{CM}_{d\allone_n}$. Along admissible $n$, this is Bollob\'as's $d$-regular configuration model~\cite{bollobas1980probabilistic}.
\end{definition}
The maximum-pairing convention also permits intermediate graph sizes with odd $dn$ in the cavity arguments.

\begin{definition}[Random regular graphs]\label{def:random-regular-graphs}
  For an integer $d\ge3$ and admissible $n$, let $\+G_{n,d}$ be the uniform law on simple $d$-regular graphs on $[n]$. Equivalently, it is the law of $G\sim\mathrm{CM}_{n,d}$ conditioned on $G$ being simple.
\end{definition}

\begin{lemma}[Simplicity probability~\cite{bollobas1980probabilistic,janson2014simple}]\label{lem:simplicity-prob}
  For every integer $d\ge3$, there exists $c_d>0$ such that for all sufficiently large admissible $n$,
  \[
    \Pr[\mathrm{CM}_{n,d}]{G\text{ is simple}}\ge c_d.
  \]
\end{lemma}

\paragraph{Interaction kernels and normalized partition function.}
To analyze the hardcore and antiferromagnetic Potts models in a unified framework, we use the vertex marginal $\*\pi(t)$ and symmetric edge marginal matrix $P(t)$ of the homogeneous Markov Gibbs measure on the infinite $d$-regular tree $\=T_d$~\cite{mezard2001bethe,mezard2009information,dembo2013factor}, and their holomorphic extensions near $t$. Reweighting edges by the ratio between $P_{ij}$ and the product marginal $\pi_i\pi_j$ yields the normalized interaction kernel, tree broadcasting operator, and normalized partition function used throughout our moment calculations.

\begin{definition}[Interaction kernels and normalized partition function]\label{def:bethe-canonical}
  Fix {a finite spin space $\Omega$ ($|\Omega|\ge2$) and a real parameter $t$.} Let $P(z)$ be a symmetric $|\Omega|\times|\Omega|$ holomorphic matrix near $t$ with $\sum_{i,j\in\Omega}P_{ij}(z)=1$, $P(t)\allone>0$, and $P(s)\ge0$ on a nonempty open real interval whose closure contains $t$. On a neighborhood of $t$ where the \emph{vertex marginal} $\*\pi(z)\defeq P(z)\allone$ has nonzero entries, define the \emph{normalized interaction kernel} $M(z)$ and the \emph{tree broadcasting operator} $K_z$ by
  \[
    M_{ij}(z)\defeq\frac{P_{ij}(z)}{\pi_i(z)\pi_j(z)},\qquad
    K_z\defeq \diag(\*\pi(z))^{-1}P(z)=M(z)\diag(\*\pi(z)).
  \]
  On a multigraph $G$, the \emph{normalized partition function} is defined by
  \[
    \mathcal Z_G(z)
    \defeq\sum_{\sigma:V(G)\to\Omega}\prod_{v\in V(G)}\pi_{\sigma_v}(z)\prod_{uv\in E(G)}M_{\sigma_u,\sigma_v}(z).
  \]
\end{definition}

{We write $\*\pi_z=\*\pi(z)$, $P_z=P(z)$, $M_z=M(z)$, and $K(z)=K_z$ interchangeably when convenient, and define the configuration-model first moment
\[
  a_n(z)\defeq\E[\mathrm{CM}_{n,d}]{\mathcal Z_G(z)}.
\]}
Note that $M(z)\*\pi(z)=\allone$, $K_z\allone=\allone$, {and $\*\pi(z)^\trans K_z=\*\pi(z)^\trans$, so $K_z$ has trivial eigenvalue $\kappa_1(z)=1$ on $\mathrm{span}(\allone)$ and leaves $\{f\in\=C^\Omega:\*\pi(z)^\trans f=0\}$ invariant; we denote its eigenvalues on this subspace (counted with algebraic multiplicity) by $\kappa_2(z),\ldots,\kappa_{|\Omega|}(z)$ (the \emph{nonconstant eigenvalues}, which are real at $z=t$ by reversibility~\cite{kesten1966additional,mossel2003information}). When $\mathcal Z_G(t)>0$ at a real parameter $t$ with $M(t)\ge0$, the corresponding \emph{canonical Gibbs distribution} on $\Omega^{V(G)}$ is
\[
  \mu^{\mathrm{can}}_{G,t}(\sigma)\defeq\mathcal Z_G(t)^{-1}\prod_{v\in V(G)}\pi_{\sigma_v}(t)\prod_{uv\in E(G)}M_{\sigma_u,\sigma_v}(t).
\]} For our two models, $(\*\pi,P)$ {and the nonconstant eigenvalues of $K$} are given explicitly by:
\begin{enumerate}
  \item[(i)] \emph{Hardcore model} ($\Omega=\{0,1\}$, $t=\alpha\in(0,1/2)$):
  \begin{equation}\label{eq:hardcore_kappa}
    \*\pi_\alpha=\begin{pmatrix}1-\alpha\\\alpha\end{pmatrix},\qquad
    P_\alpha=\begin{pmatrix}1-2\alpha&\alpha\\\alpha&0\end{pmatrix},\qquad
    {\kappa_2(\alpha)=-\theta_\alpha\defeq-\frac{\alpha}{1-\alpha}.}
  \end{equation}
  \item[(ii)] \emph{Antiferromagnetic Potts model} ($\Omega=[q]$, $t=\beta\in[0,1]$, $B_\beta\defeq\allone_q\allone_q^\trans+(\beta-1)I_q$):
  \begin{equation}\label{eq:Potts_kappa}
    \*\pi=\*U_q=\frac1q\allone_q,\qquad
    P_\beta=\frac{B_\beta}{q(q-1+\beta)},\qquad
    {\kappa_2(\beta)=\cdots=\kappa_q(\beta)=\kappa_\beta\defeq-\frac{1-\beta}{q-1+\beta},}
  \end{equation}
  so that $\mathcal Z_G(\beta)=q^{-n}((q-1+\beta)/q)^{-|E(G)|}Z_{G,q}(\beta)$ on any multigraph $G$.
\end{enumerate}
For hardcore at fixed $\lambda=\lambda(\alpha)$, {we use $\mu_{G,\mathrm{fug},\lambda}$ for the ordinary hardcore law on any graph $G$}. On $d$-regular graphs it equals $\mu^{\mathrm{can}}_{G,\alpha}$, but these laws generally differ on irregular graphs. For Potts the canonical and ordinary laws agree on every graph with positive partition function. In unified notation, $Z_G(z)$ denotes either partition function and $\mu_G$ denotes $\mu^{\mathrm{can}}_{G,t}$, with the model and Potts spin count $q$ fixed by context {(in particular, $q$ is suppressed in $\mathcal Z_G$)}; the hardcore cavity arguments explicitly use ${\mu_{G,\mathrm{fug},\lambda}}$. {By a \emph{nontrivial parameter} $t$ we mean $t=\alpha\in(0,1/2)$ for the hardcore model or $t=\beta\in[0,1)$ for the antiferromagnetic Potts model (excluding the trivial independent-spin endpoints $\alpha=0$ and $\beta=1$).}

\paragraph{Bethe free-energy density and Bethe canonical form.}
On $d$-regular graphs, the unnormalized partition function $Z_G(z)$ differs from the normalized partition function $\mathcal Z_G(z)$ only by an exponential factor $C(z)^n=\e^{n\phi(z)}$ governed by the (replica-symmetric) Bethe free-energy density $\phi(z)$~\cite{mezard2001bethe,mezard2009information,dembo2013factor}.

\begin{definition}[Bethe free-energy density and Bethe canonical form]\label{def:bethe-canonical-form}
  Given a holomorphic \emph{(replica-symmetric) Bethe free-energy density} $\phi(z)$ near $t$ and scaling factor
  \[
    C(z)\defeq\e^{\phi(z)},
  \]
  the partition function on an $n$-vertex $d$-regular graph $G$ in \emph{Bethe canonical form} is
  \begin{equation}\label{eq:finite-spin-partition}
    Z_G(z)=C(z)^n\sum_{\sigma:V(G)\to\Omega}
         \frac{\prod_{uv\in E(G)}P_{\sigma_u,\sigma_v}(z)}
              {\prod_{v\in V(G)}\pi_{\sigma_v}(z)^{d-1}}
    =C(z)^n\mathcal Z_G(z).
  \end{equation}
\end{definition}

{Notice from~\eqref{eq:finite-spin-partition} that scaling by $C(z)^n$ leaves $Z_G(z)/\E{Z_G(z)}$ (and $Z_G(z)\one{E_n}/\E{Z_G(z)\one{E_n}}=Y_n(z)/\E{Y_n(z)}$ {for any graph truncation $Y_n=\mathcal Z_G\one{E_n}$}) unchanged. For our two models, $\phi$ and $C$ are given explicitly by}
\[
  \phi_d(\alpha)\defeq(d-1)\log(1-\alpha)-\frac d2\log(1-2\alpha),\qquad
  C_d(\alpha)\defeq\e^{\phi_d(\alpha)}=(1-\alpha)^{d-1}(1-2\alpha)^{-d/2},
\]
\[
  \phi_{d,q}(\beta)\defeq\log q+\frac d2\log\frac{q-1+\beta}{q},\qquad
  {C_{d,q}(\beta)\defeq\e^{\phi_{d,q}(\beta)}=q\left(\frac{q-1+\beta}{q}\right)^{d/2}.}
\]

\paragraph{Quenched and annealed free energies.}
\begin{definition}[Quenched, annealed, and thermodynamic replica-symmetric free energies]\label{def:free-energies}
  For a random graph $G$ on $n$ vertices at a real parameter $t$, the \emph{free-energy density} (with $\log0=-\infty$ when $Z_G(t)=0$), \emph{quenched free energy} (when $Z_G(t)>0$ almost surely), and \emph{annealed free energy} are defined by
  \[
    \frac1n\log Z_G(t),\qquad
    \frac1n\E[G]{\log Z_G(t)},\qquad
    \frac1n\log\E[G]{Z_G(t)}.
  \]
  {We say that the model satisfies \emph{thermodynamic replica symmetry} at $t$ if $\frac1n\log Z_G(t)\xrightarrow{\mathrm p}\phi(t)$; for $G\sim\+G_{n,d}$ (or $\mathrm{CM}_{n,d}$) and $t>0$, this is equivalent to $\lim_{n\to\infty}\frac1n\E[G]{\log Z_G(t)}=\phi(t)$ by the concentration bound~\eqref{eq:regular-logz-concentration} below.}
\end{definition}

The quenched free energy averages the log partition function over the graph disorder. {At a fixed positive parameter, concentration of the normalized log partition function identifies the quenched free energy with the typical exponential growth rate of $Z_G(t)$.} The annealed free energy takes the logarithm of the first moment and is at least the quenched free energy by Jensen's inequality. {Under $\mathrm{CM}_{n,d}$ at any fixed positive model parameter, there exists $c>0$ depending only on the parameters such that
\begin{equation}\label{eq:regular-logz-concentration}
 \Pr[\mathrm{CM}_{n,d}]{|\log Z_G-\E[\mathrm{CM}_{n,d}]{\log Z_G}|>u}\le2\e^{-cu^2/n}
 \quad(u>0),\qquad
 \operatorname{Var}_{\mathrm{CM}_{n,d}}(\log Z_G)=O(n).
\end{equation}
Indeed, adding or deleting an edge changes $\log Z_G$ by at most $\log(1+\lambda)$ for hardcore and $|\log\beta|$ for Potts. A pairing switch replaces two edges, requiring at most two deletions and two insertions, so the increments of the pairing-exposure martingale are bounded by a parameter-dependent constant. Azuma--Hoeffding gives the tail bound and integration gives the variance bound. By Cauchy--Schwarz and \Cref{lem:simplicity-prob}, conditioning on simplicity changes $\E{\log Z_G}$ by $O(\sqrt n)$ (so $|\E[\+G_{n,d}]{\log Z_G}-\E[\mathrm{CM}_{n,d}]{\log Z_G}|=O(\sqrt n)$) and preserves the tail bound in~\eqref{eq:regular-logz-concentration} up to the constant factor $c_d^{-1}$.}

\paragraph{{Cavity message distributions and Bethe functional}.}
\begin{definition}[{Cavity message distributions, Bethe functional, and variational replica symmetry}]\label{def:rsb-bethe-functional}
  Fix a real parameter $t$. The space of {\emph{cavity message distributions}} with mean $\*\pi(t)$ is defined by
  \[
    \mathcal P_{\*\pi(t)}\defeq\Bigl\{\nu\in\mathcal P(\Delta_\Omega):\E[\*\mu\sim\nu]{\*\mu}=\*\pi(t)\Bigr\},
  \]
  and $\delta_{\*\pi(t)}\in\mathcal P_{\*\pi(t)}$ denotes the point mass at the constant message $\*\pi(t)$.
  For independent $\*\mu_1,\*\mu_2,\ldots\sim\nu\in\mathcal P_{\*\pi(t)}$ and every integer $k\ge0$, the \emph{cavity vertex and edge factors} and the {\emph{Bethe functional}} are defined by
  \begin{equation}\label{eq:marked-functional}
    {\begin{gathered}
    V_{k,t}\defeq\sum_{a\in\Omega}\pi_a(t)\prod_{i=1}^k(M(t)\*\mu_i)_a,\qquad
    E_t\defeq\*\mu_1^\trans M(t)\*\mu_2,\\
    \mathcal B_{k,t}(\nu)\defeq\E{V_{k,t}\log V_{k,t}}-\frac k2\E{E_t\log E_t}.
    \end{gathered}}
  \end{equation}
  with empty products equal to $1$, and extended to real $r\ge0$ by deterministic piecewise-linear interpolation:
  \[
    \mathcal B_{r,t}(\nu)\defeq(1-r+\lfloor r\rfloor)\mathcal B_{\lfloor r\rfloor,t}(\nu)+(r-\lfloor r\rfloor)\mathcal B_{\lfloor r\rfloor+1,t}(\nu).
  \]
  {We say that the model satisfies \emph{variational replica symmetry} at $(r,t)$ if $\mathcal B_{r,t}(\nu)\le0$ for all $\nu\in\mathcal P_{\*\pi(t)}$.}
\end{definition}
{For Potts, the dependence on $q$ is suppressed in $\mathcal B_{r,t}$.}
{We write $\mathcal P_{\*\pi}=\mathcal P_{\*\pi(t)}$, $\delta_{\*\pi}=\delta_{\*\pi(t)}$, $V_k=V_{k,t}$, $E=E_t$, and $\mathcal B=\mathcal B_{r,t}$ when $(r,t)$ are clear from context. Since $M\*\pi=\allone$, the constant message distribution satisfies $V_k=E=1$ almost surely, so $\mathcal B(\delta_{\*\pi})=0$.}

{
\paragraph{Complex analysis basics.}
{For a subset $S\subseteq\=C$, we write $\ol S$, $\mathrm{int}(S)$, and $\partial S$ for its closure, interior, and boundary (and $\ol z$ for the complex conjugate of $z\in\=C$), and $D(S,\eta)\defeq\{z\in\=C:\inf_{w\in S}|z-w|<\eta\}$ (with $D(z_0,\eta)=D(\{z_0\},\eta)$).} Let $H(D)$ denote the space of holomorphic functions on an open connected region $D\subseteq\=C$, equipped with the topology of locally uniform convergence: a sequence $f_n\in H(D)$ converges to $f\in H(D)$ if $\norm{f_n-f}_K\defeq\sup_{z\in K}\abs{f_n(z)-f(z)}\to0$ for every compact set $K\subset D$.

We first record two standard properties of holomorphic functions in $H(D)$.

\begin{proposition}[Identity theorem]\label{prop:identity}
  Let $f,g\in H(D)$ for some connected region $D$. If $f(z)=g(z)$ for all $z$ in some set $S\subset D$ with an accumulation point in $D$, then $f\equiv g$.
\end{proposition}

\begin{proposition}[Montel's theorem]\label{prop:montel}
  Let $D\subseteq\=C$ be a connected region. A closed subset $\+F\subseteq H(D)$ is compact if $\sup_{f\in\+F}\norm{f}_K<\infty$ for every compact set $K\subset D$.
\end{proposition}

For a sequence of random elements $F_n\in H(D)$, we say that $F_n$ converges \emph{almost surely} to $F$ in $H(D)$ if $\Pr{\norm{F_n-F}_K\to0\text{ for all compact }K\subset D}=1$, and converges \emph{in distribution} ($F_n\xrightarrow{d}F$) if $\E{\varphi(F_n)}\to\E{\varphi(F)}$ for every bounded continuous $\varphi:H(D)\to\=R$.

To establish convergence in distribution in $H(D)$ from second-moment bounds, we will need Prokhorov's theorem and Tonelli's theorem{; we write $\dif A$ for planar area measure on $\=C$ and $\operatorname{area}(K)\defeq\int_K\dif A(z)$ for a measurable set $K\subset\=C$}.

\begin{proposition}[Prokhorov's theorem~{\cite[Theorem~5.1]{billingsley2013convergence}}]\label{prop:prokhorov}
  A sequence of random elements $\{F_n\}$ in {a Polish space $\mathcal X$} is \emph{tight} if for every $\varepsilon>0$, there exists a compact set $K_\varepsilon\subseteq{\mathcal X}$ such that $\inf_n\Pr{F_n\in K_\varepsilon}\ge1-\varepsilon$. Every tight sequence in {$\mathcal X$} has a subsequence converging in distribution.
\end{proposition}

\begin{proposition}[Tonelli's theorem]\label{prop:tonelli}
  For any jointly measurable nonnegative random function $g(z)\ge0$ on a compact set $K\subset\=C$ and any nonnegative random variables {$\zeta_\ell\ge0$},
  \[
    \E{\int_K g(z)\dif A(z)}=\int_K\E{g(z)}\dif A(z)
    \qquad\text{and}\qquad
    \E{\sum_{\ell\ge1}{\zeta_\ell}}=\sum_{\ell\ge1}\E{{\zeta_\ell}}.
  \]
\end{proposition}
}

% END INCLUDED FILE: chapters/02-technical-overview.tex
% BEGIN INCLUDED FILE: chapters/03-reinforced-ssc.tex
\section{Reinforced Small Subgraph Conditioning}\label{sec:reinforced-ssc}

In this section, we establish our reinforced small subgraph conditioning theorem (\Cref{thm:reinforced-ssc}) by combining Janson's real-line small subgraph conditioning~\cite{robinson1992almost,robinson1994almost,janson1995random} with the \emph{tightness method} in the space $H(D)$ of holomorphic functions on an open connected set $D\subseteq\=C$ (recall from \Cref{subsec:preliminaries} that $H(D)$ is equipped with the topology of locally uniform convergence).

{In our application to the $d$-regular configuration model $\mathrm{CM}_{n,d}$ (\Cref{def:config-model}), the multiplicative fluctuations of the normalized partition function $\mathcal Z_G(z)$ (\Cref{def:bethe-canonical}) are governed by the counts and relative weights of short cycles:

\begin{definition}[Short-cycle counts, Poisson means, and short-cycle factors]\label{def:short-cycle-factors}
  For a $d$-regular multigraph $G$ on $[n]$ and $\ell\ge1$, let $X_{\ell,n}$ denote the number of $\ell$-cycles in $G$ (including self-loops for $\ell=1$ and parallel-edge pairs for $\ell=2$). In the setup of \Cref{def:bethe-canonical}, let $C_\ell$ denote a single cycle of length $\ell$; the \emph{Poisson means} $\lambda_\ell$ and \emph{short-cycle factors} $\delta_\ell(z)$ are defined by
  \[
    \lambda_\ell\defeq\frac{(d-1)^\ell}{2\ell},\qquad
    \delta_\ell(z)\defeq\mathcal Z_{C_\ell}(z)-1
    =\sum_{a_1,\ldots,a_\ell\in\Omega}\prod_{r=1}^\ell\pi_{a_r}(z)M_{a_r,a_{r+1}}(z)-1\quad(a_{\ell+1}\defeq a_1).
  \]
\end{definition}

We state our reinforced small subgraph conditioning theorem for general random holomorphic functions $Y_n\in H(D)$, cycle counts $X_{\ell,n}\in\=Z_{\ge0}$, means $\lambda_\ell>0$, and weights $\delta_\ell\in H(D)$. Set $A_n(z)\defeq\E{Y_n(z)}$ and $F_n\defeq Y_n/A_n$ whenever $A_n$ is finite and nowhere zero on $D$, and for independent $Z_\ell\sim\mathrm{Poisson}(\lambda_\ell)$, define the formal product
\[
  W(z)\defeq\prod_{\ell\ge1}\bigl(1+\delta_\ell(z)\bigr)^{Z_\ell}\exp\{-\lambda_\ell\delta_\ell(z)\}
  \qquad(0^0\defeq1).
\]}

\begin{theorem}[Reinforced small subgraph conditioning (RSSC)]\label{thm:reinforced-ssc}
Let $D\subseteq\=C$ be an open connected set, $I\subseteq D\cap\=R$ be a nonempty open interval, $Y_n$ be a random element of $H(D)$, $X_{\ell,n}\in\=Z_{\ge0}$ ($\ell\ge1$) be random variables, $\lambda_\ell>0$, and $\delta_\ell\in H(D)$. Assume:
\begin{enumerate}
\item[(B0)] $A_n$ is finite, holomorphic, and nowhere zero on $D$ for all sufficiently large $n$, and $Y_n(t)\ge0$ almost surely for every $t\in I$;
\item[(B1)] For each fixed $m\ge1$, $(X_{1,n},\ldots,X_{m,n})\xrightarrow{d}(Z_1,\ldots,Z_{m})$;
\item[(B2)] For fixed $m\ge1$, $j_1,\ldots,j_m\in\=Z_{\ge0}$, and $t\in I$,
  \[
    \E{F_n(t)\prod_{\ell=1}^m(X_{\ell,n})_{j_\ell}}
      \longrightarrow\prod_{\ell=1}^m[\lambda_\ell(1+\delta_\ell(t))]^{j_\ell};
  \]
\item[(B3)] $\sum_{\ell\ge1}\lambda_\ell\sup_{z\in K}|\delta_\ell(z)|^2<\infty$ for every compact $K\subset D$;
\item[(B4)] For every $t\in I$, $\displaystyle\limsup_{n\to\infty}\E{F_n(t)^2}\le\exp\Bigl(\sum_{\ell\ge1}\lambda_\ell\delta_\ell(t)^2\Bigr)$;
\item[(B5)] $\limsup_{n\to\infty}\sup_{z\in K}\E{|F_n(z)|^2}<\infty$ for every compact $K\subset D$.
\end{enumerate}
Then $W$ converges almost surely in $H(D)$, and $(F_n,(X_{\ell,n})_{\ell\ge1})\xrightarrow{d}(W,(Z_\ell)_{\ell\ge1})$ in $H(D)\times\=Z_{\ge0}^{\=Z_{\ge1}}$.
In particular, if $1+\delta_\ell$ is nowhere zero on $D$ for every $\ell\ge3$, then conditional on $X_{1,n}=X_{2,n}=0$, $Y_n$ is zero-free on each fixed compact $K\subset D$ with probability $1-o_n(1)$.
\end{theorem}

{Classical small subgraph conditioning~\cite{robinson1992almost,robinson1994almost,janson1995random} applies at each real parameter. Here \textup{(B2)} and \textup{(B4)} are needed only on $I$: \textup{(B3)} constructs the holomorphic product $W$, while \textup{(B5)} supplies tightness in $H(D)$. Real-line small subgraph conditioning and the identity theorem (\Cref{prop:identity}) then identify the complex limit; compare the functional tightness approach of~\cite{bordenave2022convergence}.}

To bridge {locally uniform} second-moment bounds and the locally uniform topology of $H(D)$, we begin with the standard area mean-value inequality, writing $\dif A$ for planar area measure (and recalling from \Cref{subsec:preliminaries} that $\norm{h}_K=\sup_{z\in K}|h(z)|$ for a compact set $K\subset D$ and $h\in H(D)$).

\begin{lemma}[Mean-value inequality]\label{lem:mean-value}
  Let $K$ and $L$ be compact sets with $K\subset\mathrm{int}(L)\subset L\subset D$. There exists a constant $C_{K,L}>0$ such that for any $h\in H(D)$,
  \begin{equation}\label{eq:mean-value}
    \norm{h}_K^2 \le C_{K,L}\int_L\abs{h(z)}^2\dif A(z).
  \end{equation}
\end{lemma}

{
\begin{proof}
  Let ${\eta}\defeq\operatorname{dist}(K,\partial L)>0$. For every $h\in H(D)$ and $z_0\in K$, the disk ${D(z_0,\eta)}$ is contained in $L$.
  Therefore, by the area mean-value identity and Cauchy--Schwarz, we have that
  $$
    |h(z_0)|^2 = \abs{\frac 1{\pi {\eta^2}}\int_{{D(z_0,\eta)}}h(z)\dif A(z)}^2 \le \frac{1}{\pi {\eta^2}}\int_{{D(z_0,\eta)}}|h(z)|^2\dif A(z) \le \frac{1}{\pi {\eta^2}}\int_L|h(z)|^2\dif A(z).
  $$
  Taking the supremum over $z_0\in K$ proves the claim.
\end{proof}}

Write $W_N\defeq\prod_{\ell=1}^N(1+\delta_\ell)^{Z_\ell}\exp\{-\lambda_\ell\delta_\ell\}$ for the partial products of $W$; the event $Z_1=Z_2=0$ corresponds to simplicity in the configuration-model application.

\begin{lemma}[Well-definedness and nonvanishing of $W$]\label{lem:W-well-defined}
  Under~$(B3)$, $W_N\to W$ almost surely in $H(D)$ as $N\to\infty$. Moreover, if $1+\delta_\ell$ is nowhere zero on $D$ for every $\ell\ge3$, then conditional on $Z_1=Z_2=0$, $W$ is nowhere zero on $D$ almost surely.
\end{lemma}

\begin{proof}
  Fix compact subsets $K$ and $L$ with $K\subset\mathrm{int}(L)\subset L\subset D$, and define $r_\ell\defeq\norm{\delta_\ell}_L$.
  By~$(B3)$, $\sum_{\ell\ge1}\lambda_\ell r_\ell^2<\infty$.
  On the index set $J\defeq\{\ell\ge1:r_\ell>1/2\}$, we have $1<4r_\ell^2$ and $r_\ell<2r_\ell^2$; therefore,
  $$
    \sum_{\ell\in J}\Pr{Z_\ell\ne0}\le 4\sum_{\ell\in J}\lambda_\ell r_\ell^2<\infty\qquad\text{and}\qquad\sum_{\ell\in J}\lambda_\ell r_\ell\le 2\sum_{\ell\in J}\lambda_\ell r_\ell^2<\infty.
  $$
  By Borel--Cantelli and uniform convergence of $\sum_{\ell\in J}\lambda_\ell\delta_\ell(z)$ on $L$, the factor over $\ell\in J$ almost surely reduces to
  \[
    \exp\!\left(-\sum_{\ell\in J}\lambda_\ell\delta_\ell(z)\right)
    \prod_{\substack{\ell\in J\\Z_\ell>0}}(1+\delta_\ell(z))^{Z_\ell},
  \]
  where the remaining product has only finitely many factors on $\mathrm{int}(L)$.

  For $\ell\notin J$, we have $|\delta_\ell(z)|\le r_\ell\le1/2$ on $L$, and we decompose the principal branch of each logarithmic factor into its centered linear part and quadratic Taylor remainder:
  \[
    Z_\ell\log(1+\delta_\ell(z))-\lambda_\ell\delta_\ell(z)
    = \xi_\ell(z)+\eta_\ell(z),
  \]
  where $\xi_\ell(z)\defeq(Z_\ell-\lambda_\ell)\delta_\ell(z)$ and $\eta_\ell(z)\defeq Z_\ell\bigl[\log(1+\delta_\ell(z))-\delta_\ell(z)\bigr]$.
  Since $|\log(1+w)-w|\le2|w|^2$ for $|w|\le1/2$, Tonelli's theorem (\Cref{prop:tonelli}) and~$(B3)$ give $\E{\sum_{\ell\notin J}\norm{\eta_\ell}_L}\le2\sum_{\ell\notin J}\lambda_\ell r_\ell^2<\infty$, so $\sum_{\ell\notin J}\eta_\ell(z)$ converges absolutely and uniformly on $L$ almost surely.
  Meanwhile, the centered linear terms $\{\xi_\ell=(Z_\ell-\lambda_\ell)\delta_\ell\}_{\ell\notin J}$ are independent, mean-zero random elements in the separable Hilbert space $L^2(L,\dif A)$ with finite total variance:
  \[
    \sum_{\ell\notin J}\E{\norm{\xi_\ell}_{L^2(L)}^2}
    =\sum_{\ell\notin J}\operatorname{Var}(Z_\ell)\int_L|\delta_\ell(z)|^2\dif A(z)
    \le\operatorname{area}(L)\sum_{\ell\notin J}\lambda_\ell r_\ell^2<\infty.
  \]
  {
    By the $L^2$-martingale convergence theorem in Hilbert space, the
partial sums
$S_N\defeq\sum_{\ell\le N,\,\ell\notin J}\xi_\ell$
converge almost surely in $L^2(L,\dif A)$. By~\Cref{lem:mean-value},
$$
\|S_N-S_M\|_K^2
\le C_{K,L}\int_L|S_N(z)-S_M(z)|^2\,\dif A(z)
\longrightarrow0
\qquad(N,M\to\infty),
$$
which, together with the uniform convergence of $\sum_{\ell\notin J}\eta_\ell$, implies that $\sum_{\ell\notin J}(\xi_\ell+\eta_\ell)$ converges uniformly on $K$.
Exponentiating and combining the two factors gives
$$
W_N(z)=
\left[
\prod_{\substack{\ell\le N\\\ell\in J}}
(1+\delta_\ell(z))^{Z_\ell}\exp\tp{-\lambda_\ell\delta_\ell(z)}
\right]
\exp\left\{
\sum_{\substack{\ell\le N\\\ell\notin J}}
(\xi_\ell(z)+\eta_\ell(z))
\right\}
\longrightarrow W(z)
$$
uniformly on $K$ almost surely.
Exhausting $D$ by compact sets shows that $W_N\to W$ almost surely in $H(D)$.
  }

  Finally, on $\{Z_1=Z_2=0\}$, the $\ell=1,2$ factors $(1+\delta_\ell(z))^{Z_\ell}$ equal $1$, and on each compact subset $K\subset D$ the above decomposition represents $W$ as a finite product of nonzero factors times an exponential, since $1+\delta_\ell$ is nowhere zero for $\ell\ge3$.
  {Since the series in the exponent converges uniformly on $K$, the exponent is finite and the exponential factor is nonzero.}
  Thus $W$ is nowhere zero on $D$ almost surely conditional on $Z_1=Z_2=0$.
\end{proof}

\begin{lemma}[Tightness]\label{lem:step-2}
  Under conditions~$(B0)$ and~$(B5)$, {the tail $\{F_n\}_{n\ge n_0}$ is tight in $H(D)$ for every sufficiently large $n_0$}.
\end{lemma}

\begin{proof}
  Fix $n_0$ so that $A_n$ is nowhere zero on $D$ and $F_n\in H(D)$ for all $n\ge n_0$.
  Let $K_1\subset\mathrm{int}(K_2)\subset K_2\subset\cdots$ be a compact exhaustion of $D$.
{
For every fixed $j\ge1$, we have that
\begin{align*}
\limsup_{n\to\infty}\E{\|F_n\|_{K_j}^2}
&\le C_{K_j,K_{j+1}}\limsup_{n\to\infty}\,
\E{\int_{K_{j+1}}|F_n(z)|^2\,\dif A(z)}&& (\text{\Cref{lem:mean-value}})\\
&= C_{K_j,K_{j+1}}
\limsup_{n\to\infty}\int_{K_{j+1}}\E{|F_n(z)|^2}\,\dif A(z)&&(\text{Tonelli's theorem})\\
&\le C_{K_j,K_{j+1}}\operatorname{area}(K_{j+1})
\limsup_{n\to\infty}\sup_{z\in K_{j+1}}\E{|F_n(z)|^2} < \infty. &&(B5)
\end{align*}
}
  Thus there exists $n_1(j)\ge n_0$ such that $\sup_{n\ge n_1(j)}\E{\norm{F_n}_{K_j}^2}<\infty$.
  Since $\norm{F_n}_{K_j}<\infty$ almost surely for each of the finitely many indices $n_0\le n<n_1(j)$, Markov's inequality for $n\ge n_1(j)$ together with tightness of those finitely many random variables provides, for each $\eps>0$ and $j\ge1$, a threshold $R_j(\eps)<\infty$ such that
  \[
    \sup_{n\ge n_0}\Pr{\norm{F_n}_{K_j}>R_j(\eps)}\le\eps 2^{-j}.
  \]
  The set $\+C_\eps\defeq\{f\in H(D):\norm{f}_{K_j}\le R_j(\eps)\text{ for all }j\ge1\}$ is closed and locally uniformly bounded, and hence compact in $H(D)$ by Montel's theorem (\Cref{prop:montel}).
  A union bound gives
  \[\sup_{n\ge n_0}\Pr{F_n\notin\+C_\eps}\le\eps,\]
  so $\{F_n\}_{n\ge n_0}$ is tight in $H(D)$.
\end{proof}

{We identify the limit on the real interval $I$ using Janson's theorem. For $t\in I$, nonnegativity of $Y_n(t)$ and $X_{\ell,n}$ with \textup{(B0)} and \textup{(B2)} (for $j_\ell=1$) implies $1+\delta_\ell(t)\ge0$, so $\delta_\ell(t)$ is real and at least $-1$.}

\begin{theorem}[Real-line small subgraph conditioning~{\cite[Theorem~1 and Lemma~1]{janson1995random}}]\label{thm:janson-ssc}
  {In the notation of \Cref{thm:reinforced-ssc}, fix $t\in I$ and suppose $Y_n(t)\ge0$ almost surely, $0<A_n(t)=\E{Y_n(t)}<\infty$ for all sufficiently large $n$, conditions~\textup{(B1)}--\textup{(B2)} and \textup{(B4)} hold at $t$, and $\sum_{\ell\ge1}\lambda_\ell\delta_\ell(t)^2<\infty$. Then, for every $m\ge1$,
  \[
    \bigl(F_n(t),X_{1,n},\ldots,X_{m,n}\bigr)
    \xrightarrow{d}\bigl(W(t),Z_1,\ldots,Z_m\bigr).
  \]}
\end{theorem}

\begin{proof}[Proof of~\Cref{thm:reinforced-ssc}]
  By~\Cref{lem:W-well-defined}, the partial products $W_N$ converge almost surely to $W$ in $H(D)$.
  By~\Cref{lem:step-2} and~$(B1)$, $\{(F_n,(X_{\ell,n})_{\ell\ge1})\}_{n\ge n_0}$ is tight in $H(D)\times\=Z_{\ge0}^{\=Z_{\ge1}}$.
  By Prokhorov's theorem (\Cref{prop:prokhorov}), let $(F,Z)$ be any subsequential limit in distribution of $(F_n,(X_{\ell,n})_{\ell\ge1})$, where $Z=(Z_\ell)_{\ell\ge1}$ consists of independent $\mathrm{Poisson}(\lambda_\ell)$ variables by~$(B1)$, and form $W$ from $Z$.
  For each $t\in I$, \Cref{thm:janson-ssc} gives \[(F_n(t),X_{1,n},\ldots,X_{m,n})\xrightarrow{d}(W(t),Z_1,\ldots,Z_m)\]
  for every $m\ge1$, whence $(F(t),Z)\overset{d}{=}(W(t),Z)$. Since $W(t)$ is a measurable function of $Z$, this equality of joint laws forces $F(t)=W(t)$ almost surely.
  Applying this simultaneously to a countable sequence $\{t_k\}_{k\ge1}\subset I$ with an accumulation point in $I$, the identity theorem (\Cref{prop:identity}) on the connected domain $D$ forces $F=W$ almost surely in $H(D)$, which uniquely identifies every subsequential limit and proves $(F_n,(X_{\ell,n})_{\ell\ge1})\xrightarrow{d}(W,Z)$.

  Put $\+E_n^{\mathrm{simp}}\defeq\{X_{1,n}=X_{2,n}=0\}$ and $\+E^{\mathrm{simp}}\defeq\{Z_1=Z_2=0\}$, so $\Pr{\+E_n^{\mathrm{simp}}}\to\Pr{\+E^{\mathrm{simp}}}=\e^{-\lambda_1-\lambda_2}>0$.
  Assume $1+\delta_\ell$ is nowhere zero on $D$ for every $\ell\ge3$, and fix a nonempty compact $K\subset D$.

  Continuity of $f\mapsto\min_{z\in K}|f(z)|$ makes $\+U_K\defeq\{f\in H(D):\min_{z\in K}|f(z)|>0\}$ open in $H(D)$, so $\+U_K\times\{(0,0)\}$ is open in $H(D)\times\=Z_{\ge0}^2$.
  Since $\Pr{W\in\+U_K\mid \+E^{\mathrm{simp}}}=1$ by~\Cref{lem:W-well-defined} and $(F_n,X_{1,n},X_{2,n})\xrightarrow{d}(W,Z_1,Z_2)$, the Portmanteau theorem gives
  \[
    \liminf_{n\to\infty}\Pr{F_n\in\+U_K\mid \+E_n^{\mathrm{simp}}}
    =\frac{\liminf_{n\to\infty}\Pr{F_n\in\+U_K,\,\+E_n^{\mathrm{simp}}}}{\Pr{\+E^{\mathrm{simp}}}}
    \ge\Pr{W\in\+U_K\mid \+E^{\mathrm{simp}}}=1.
  \]
  Hence $F_n$, and therefore $Y_n$, is zero-free on $K$ with probability $1-o_n(1)$ conditional on $\+E_n^{\mathrm{simp}}$.
\end{proof}

% END INCLUDED FILE: chapters/03-reinforced-ssc.tex
% BEGIN INCLUDED FILE: chapters/sec5-deduplicated.tex
\section{Charting the Replica-Symmetric Phase}\label{sec:rs-identification-local}
In this section, we establish the formal counterparts of the phase characterizations announced in \Cref{thm:hardcore-rs-transition} (\Cref{thm:hardcore-three-characterizations}, \Cref{thm:potts-three-characterizations}, and \Cref{thm:coloring-three-characterizations}), showing that the geometric (\Cref{def:replica-symmetry}), variational (\Cref{def:rsb-bethe-functional}), and thermodynamic (\Cref{def:free-energies}) descriptions share a common replica symmetry breaking threshold in each model.

{
Recall from \Cref{subsec:preliminaries} the replica-symmetric Bethe free-energy predictions $\phi_d(\alpha)$ and $\phi_{d,q}(\beta)$ (\Cref{def:bethe-canonical-form}), the Bethe functional $\mathcal B_{d,t}$, and the space $\mathcal P_{\*\pi}$ of cavity message distributions on $\Delta_\Omega$ with mean $\*\pi$ (\Cref{def:rsb-bethe-functional}, where $\*\pi_\alpha=(1-\alpha,\alpha)^\trans$ for hardcore and $\*\pi=\*U_q$ for Potts). For proper $q$-colorings ($\beta=0$), we define $\mu_{G,q,0}$ arbitrarily on uncolorable graphs and take $\log 0=-\infty$; geometric replica symmetry for colorings additionally requires colorability with high probability. All limits in $n$ below are taken along admissible values.

Our three main theorems of this section is given as follows.

\begin{theorem}[Charting the replica-symmetric phase for the hardcore model]\label{thm:hardcore-three-characterizations}
  For every integer $d\ge3$, there exists a threshold $\alpha_{\mathrm{RSB}}(d)\in(0,1/2)$ such that, for every $\alpha\in(0,1/2)$, the following three properties hold if $\alpha<\alpha_{\mathrm{RSB}}(d)$ and fail if $\alpha>\alpha_{\mathrm{RSB}}(d)$:
  \begin{enumerate}
    \item \textbf{Geometric replica symmetry:} The hardcore model on $G\sim\+G_{n,d}$ is replica symmetric in the sense of \Cref{def:replica-symmetry}.
    \item \textbf{Variational replica symmetry:} For every $\nu\in\mathcal P_{\*\pi_\alpha}$,
    \[
      \mathcal B_{d,\alpha}(\nu)\le0.
    \]
    \item \textbf{Thermodynamic replica symmetry:} The free-energy density and limiting quenched free energy equal the Bethe prediction in \Cref{def:bethe-canonical-form}:
    \[
      \frac1n\log Z_G(\alpha)\xrightarrow{\mathrm p}\phi_d(\alpha)
      \qquad\text{and}\qquad
      \lim_{n\to\infty}\frac1n\E{\log Z_G(\alpha)}=\phi_d(\alpha).
    \]
  \end{enumerate}
  {Moreover, the threshold has the asymptotic expansion~\eqref{eq:replica-symmetry-breaking-hardcore-asymptotic} as $d\to\infty$.}
\end{theorem}

\begin{theorem}[Charting the replica-symmetric phase for the antiferromagnetic Potts model]\label{thm:potts-three-characterizations}
  For every pair of integers $d\ge3$ and $q\ge2$, there exists a threshold $\beta_{\mathrm{RSB}}(d,q)\in[0,1)$ such that, for every $\beta\in(0,1]$, the following three properties hold if $\beta>\beta_{\mathrm{RSB}}(d,q)$ and fail if $\beta<\beta_{\mathrm{RSB}}(d,q)$:
  \begin{enumerate}
    \item \textbf{Geometric replica symmetry:} The Potts model on $G\sim\+G_{n,d}$ is replica symmetric in the sense of \Cref{def:replica-symmetry}.
    \item \textbf{Variational replica symmetry:} For every $\nu\in\mathcal P_{\*U_q}$,
    \[
      \mathcal B_{d,\beta}(\nu)\le0.
    \]
    \item \textbf{Thermodynamic replica symmetry:} The free-energy density and limiting quenched free energy equal the Bethe prediction in \Cref{def:bethe-canonical-form}:
    \[
      \frac1n\log Z_{G,q}(\beta)\xrightarrow{\mathrm p}\phi_{d,q}(\beta)
      \qquad\text{and}\qquad
      \lim_{n\to\infty}\frac1n\E{\log Z_{G,q}(\beta)}=\phi_{d,q}(\beta).
    \]
  \end{enumerate}
\end{theorem}

\begin{theorem}[Charting the replica-symmetric phase for proper $q$-colorings]\label{thm:coloring-three-characterizations}
  For every integer $q\ge3$, there exists a threshold $d_{\mathrm{RSB}}(q)\in[3,\infty)$ such that, for every integer $d\ge3$, the following three properties hold if $d<d_{\mathrm{RSB}}(q)$, and fail if $d>d_{\mathrm{RSB}}(q)$:
  \begin{enumerate}
    \item \textbf{Geometric replica symmetry:} With high probability $G\sim\+G_{n,d}$ is $q$-colorable and its coloring Gibbs measure is replica symmetric in the sense of \Cref{def:replica-symmetry}. We define this measure arbitrarily on uncolorable graphs.
    \item \textbf{Variational replica symmetry:} For every $\nu\in\mathcal P_{\*U_q}$,
    \[
      \mathcal B_{d,0}(\nu)\le0.
    \]
    \item \textbf{Thermodynamic replica symmetry:} The free-energy density equals the Bethe prediction in \Cref{def:bethe-canonical-form}:
    \[
      \frac1n\log Z_{G,q}(0)\xrightarrow{\mathrm p}\phi_{d,q}(0).
    \]
    % {\color{green!40!black}(which, when $d<d_{\mathrm{RSB}}(q)$, implies $\lim_{n\to\infty}\frac1n\E{\log Z_{G,q}(0)\mid Z_{G,q}(0)>0}=\phi_{d,q}(0)$).}
  \end{enumerate}
  {Moreover, the threshold has the asymptotic expansion~\eqref{eq:replica-symmetry-breaking-potts-asymptotic} as $q\to\infty$.}
\end{theorem}
}

\begin{remark}[Sharpness of the zero-free region]\label{rem:sharpness-zero-free}
An immediate consequence of Theorems~\ref{thm:hardcore-three-characterizations}--\ref{thm:coloring-three-characterizations} is the sharpness of our zero-free regions: the zero-free neighborhood cannot be extended across $\alpha_{\mathrm{RSB}}(d)$ for the hardcore model or across $\beta_{\mathrm{RSB}}(d,q)>0$ for the antiferromagnetic Potts model. Indeed, $n^{-1}\log Z_G$ converges in probability to the real-analytic Bethe prediction $\phi$ on the replica-symmetric side of the threshold, whereas on the replica-symmetry-breaking side it lies strictly below $\phi$ with high probability by \Cref{lem:strict-free-energy-gaps}(i)--(ii).

Suppose for contradiction that $Z_G$ were zero-free on a sufficiently small open disk $D_0$ centered at the threshold with probability bounded away from zero along a subsequence of admissible $n$. Fix a positive real basepoint $t_0\in D_0$ on the replica-symmetric side; on each zero-free graph $G$, the holomorphic branch $F_G(z)\defeq n^{-1}\log Z_G(z)$ normalized to be real at $t_0$ has $\Re F_G$ uniformly bounded above on $D_0$ and $F_G(t_0)$ uniformly bounded on both sides, so $\{F_G\}$ is locally uniformly bounded on $D_0$ by the Borel--Carath\'eodory theorem. Intersecting the zero-free event with the high-probability convergence $F_G(t)\to\phi(t)$ at a countable dense set of real points $t\in D_0$ on the replica-symmetric side and the high-probability strict gap $F_G(t_1)\le\phi(t_1)-\varepsilon$ at a real point $t_1\in D_0$ beyond the threshold, diagonalization selects a deterministic sequence of zero-free graphs on which both properties hold. By Montel's theorem and the identity theorem, any subsequential limit along this selected sequence must equal the analytic continuation of $\phi$ across $D_0$, contradicting the strict gap at $t_1$.
\end{remark}

\paragraph{Proof strategy and organization.}
The proofs of Theorems~\ref{thm:hardcore-three-characterizations}--\ref{thm:coloring-three-characterizations} establish the cyclic implications connecting geometric, thermodynamic, and variational replica symmetry depicted in \Cref{fig:charting-rs}.

In the cavity framework~\cite{mezard2001bethe}, limiting quenched free energies are formulated through variational principles over local messages. On random regular graphs, Coja-Oghlan and Perkins~\cite[Theorem~2.7]{coja2019spin} established such a principle for positive-weight factor-graph models satisfying their positivity condition \textup{\textbf{POS}}, including the positive-weight antiferromagnetic Potts model~\cite[Lemma~7.3 and Theorem~7.6]{coja2019spin}. However, this framework does not directly cover the hardcore model, where hard-core exclusion breaks strict positivity and requires a separate soft-constraint limit~\cite[Theorem~1.4 and Section~7.4]{coja2019spin}. Here we establish the hardcore free-energy dichotomy using a fractional-moment bound and a two-root cavity coupling (\Cref{lem:local-hc-fractional-upper} and \Cref{lem:local-hc-rs-cavity-lower}), and treat Potts using pointwise variational and cavity arguments. For proper $q$-colorings, the real-valued threshold $d_{\mathrm{RSB}}(q)$ is defined via the deterministic piecewise-linear interpolation of the integer-degree Bethe functionals (\Cref{def:rsb-bethe-functional}), while all random graphs below have integer regular degree~$d$.

The remainder of this section is organized around the cyclic implications of \Cref{fig:charting-rs}: in \Cref{subsec:local-rs-free-energy}, we deduce thermodynamic replica symmetry from geometric replica symmetry; in \Cref{subsec:local-bethe-consequences}, we prove that thermodynamic replica symmetry implies variational replica symmetry via strict free-energy gaps (\Cref{lem:strict-free-energy-gaps}); in \Cref{subsec:local-variational-geometric}, we complete the cycle by deducing geometric replica symmetry from strict variational conditions; in \Cref{subsec:local-coloring-phase}, we extend these implications to proper colorings; and in \Cref{subsec:rsb-threshold-asymptotics}, we establish the threshold asymptotics.

\subsection{{From Geometric to Thermodynamic Replica Symmetry}}\label{subsec:local-rs-free-energy}

We prove the replica-symmetric free-energy convergence in~{\Cref{thm:hardcore-three-characterizations}(3) and~\Cref{thm:potts-three-characterizations}(3)} via geometric replica symmetry by the cavity method, which computes free-energy increments through the local cost of adding vertices and edges~\cite[Section~1.2]{coja2019spin}.

We begin with the hardcore model at fugacity $\lambda=\lambda(\alpha)$ and, in the hardcore arguments below, write $\mu_H={\mu_{H,\mathrm{fug},\lambda}}$ for every graph $H$ (\Cref{def:bethe-canonical}). {Write $\bar F_n\defeq\E[G\sim\+G_{n,d}]{\log Z_G(\alpha)}$}. For {a directed} edge $u\to v$, the \emph{BP cavity message} {$m_{u \to v}=m_{u\to v,G}\defeq\mu_{G-v}(\sigma_u=1)$ is defined as
the probability that $u$ is occupied under the hardcore measure on $G-v$. We drop the subscript $G$ when the graph is clear from context.}
On a tree, the BP cavity messages satisfy the \emph{belief propagation (BP) recurrence} exactly:
\begin{align}\label{eq:BP}
  {m_{u\to v}
  =\frac{\lambda\prod_{w\in N_G(u)\setminus\{v\}}(1-m_{w\to u})}
         {1+\lambda\prod_{w\in N_G(u)\setminus\{v\}}(1-m_{w\to u})}.}
\end{align}
On a sparse random graph $G \sim \+G_{n,d}$, which is asymptotically locally tree-like, replica symmetry allows us to approximate the BP cavity messages using the fixed points of belief propagation.

The collective behavior of all BP cavity messages across $G$ is recorded by the empirical measure 
\[
  \eta_G\defeq{\frac1{dn}}\sum_{u\to v}{\delta_{m_{u\to v,G}}}\in\mathcal P(\mathcal I_\lambda),
\]
where $\mathcal I_\lambda\defeq[\lambda(1+\lambda)^{-d},\lambda/(1+\lambda)]$ provides the uniform bounds for BP cavity messages.
Analogously to~\eqref{eq:BP}, a self-consistent BP cavity message distribution $\eta$ is a fixed point of the distributional belief propagation operator $T : \mathcal P(\mathcal I_\lambda) \to \mathcal P(\mathcal I_\lambda)$, defined by
\[
  T\eta\defeq\mathcal L\bigl(\mathcal T(Q_1,\ldots,Q_{d-1})\bigr),\qquad
  \mathcal T(Q_1,\ldots,Q_{d-1})
  \defeq\frac{\lambda\prod_{i=1}^{d-1}(1-Q_i)}
                  {1+\lambda\prod_{i=1}^{d-1}(1-Q_i)},
\]
where $Q_1,\ldots,Q_{d-1}$ are independent with law $\eta\in\mathcal P(\mathcal I_\lambda)$ {and $\mathcal L(\cdot)$ denotes the law of a random variable}.

The BP cavity messages also predict the free-energy shifts resulting from local graph assembly. For independent random variables $Q_1,\ldots,Q_d\sim\eta$, we define the unnormalized BP vertex and edge factors:
\[
  V^{\mathrm{BP}}\defeq1+\lambda\prod_{i=1}^d(1-Q_i),\qquad E^{\mathrm{BP}}\defeq1-Q_1Q_2,
\]
which predict the ratios of partition functions after adding a single vertex or a single edge; these formulas are exact when the relevant boundary spins are independent.
For $\eta$ compactly supported in $[0,1)$, define the {\emph{BP Bethe functional}} $\mathcal F_0$. For $0<m\le1$, the \emph{fractional Bethe free-energy functional} $\mathcal F_m$ is defined for every probability law $\eta$ on $[0,1]$ with $\eta(\{1\})<1$:
\[
  {\mathcal F_0(\eta)}\defeq\E{\log V^{\mathrm{BP}}}-\frac d2\E{\log E^{\mathrm{BP}}},\qquad
  {\mathcal F_m(\eta)}\defeq\frac1m\log\E{(V^{\mathrm{BP}})^m}-\frac d{2m}\log\E{(E^{\mathrm{BP}})^m}.
\]
Furthermore, the variational minimum $b_\star$ is defined by
\[
  b_\star\defeq\min_{\eta\in\mathcal P(\mathcal I_\lambda):\,T\eta=\eta}{\mathcal F_0(\eta)}.
\]
{The BP fixed-point set is nonempty, since it contains a point mass at a fixed point of the scalar BP recurrence. It is compact because $T$ is continuous, and $\mathcal F_0$ is continuous on $\mathcal P(\mathcal I_\lambda)$, so the minimum defining $b_\star$ is attained. Under geometric replica symmetry, the following two lemmas identify the limiting quenched free energy with $b_\star$}{, where the upper bound in~\Cref{lem:local-hc-fractional-upper} adapts the fractional-moment interpolation bound of~\cite[Theorem~2 and Section~3.2]{lelarge2018replica}}.

\begin{lemma}\label{lem:local-hc-fractional-upper}
For every probability law $\eta$ compactly supported in $[0,1)$ and every $m\in(0,1)$,
\[
  \limsup_{n\to\infty} \frac{\bar F_n}{n}\le{\mathcal F_m(\eta)}.
\]
In particular, $\limsup_{n\to\infty} \bar F_n/n \le b_\star$.
\end{lemma}
\begin{proof}
Fix $\delta>0$ with $\operatorname{supp}(\eta)\subseteq[0,1-\delta]$. For $0<\varepsilon<1$, let $Z_{G,\varepsilon}$ have vertex weights $\lambda^{x_v}$ and {strictly positive} edge weights $\psi_\varepsilon(x_u,x_v)=1-(1-\varepsilon)x_ux_v$, for spins in $\{0,1\}$. Thus $Z_G\le Z_{G,\varepsilon}$. {For $0<\delta'<\delta/2$, let $\eta_{\delta'}$ be the law of $(1-\delta')Q+\delta'/2$ for $Q\sim\eta$. This law is supported in a compact subset of $(0,1)$. Specializing the configuration-model fractional-moment interpolation bound of~\cite[Theorem~2 and Section~3.2]{lelarge2018replica} to $Z_{G,\varepsilon}$ with trial marginal law $\eta_{\delta'}$ yields, under $\mathrm{CM}_{n,d}$,}
\[
 \limsup_{n\to\infty}\frac1n\E{\log Z_{G,\varepsilon}}
 \le\frac1m\log\E{\bigl(1+\lambda\prod_{i=1}^d(1-(1-\varepsilon)Q_i)\bigr)^m}
 -\frac d{2m}\log\E{\bigl(1-(1-\varepsilon)Q_1Q_2\bigr)^m},
\]
where the $Q_i$ are independent with law $\eta_{\delta'}$. The expectations on the right come from the vertex and edge terms, respectively.

{Since the softened edge weights $\psi_\varepsilon(x_u,x_v)$ lie in $[\varepsilon,1]$, replacing a single edge changes $\log Z_{G,\varepsilon}$ by at most $\log(1/\varepsilon)$, so the same configuration-model pairing martingale and conditioning on simplicity as in~\eqref{eq:regular-logz-concentration} (\Cref{lem:simplicity-prob}) give $|\E[\+G_{n,d}]{\log Z_{G,\varepsilon}}-\E[\mathrm{CM}_{n,d}]{\log Z_{G,\varepsilon}}|=O_{d,\varepsilon}(\sqrt n)$.}
Hence the displayed interpolation bound also bounds $\limsup \bar F_n/n$. Letting first $\varepsilon\downarrow0$ and then $\delta'\downarrow0$ yields the first assertion: the vertex factors are bounded between $1$ and $1+\lambda$, while the edge factors stay uniformly bounded away from zero, so bounded convergence applies.

Finally, for $\eta\in\mathcal P(\mathcal I_\lambda)$ both log factors are bounded and $\mathcal F_m(\eta)\to\mathcal F_0(\eta)$ as $m\downarrow0$. Taking the infimum over the BP fixed points proves the second assertion.
\end{proof}

\begin{lemma}\label{lem:local-hc-rs-cavity-lower}
Under replica symmetry for the hardcore model,
\[
  {\liminf_{n\to\infty}\frac{\bar F_n}{n} \ge b_\star.}
\]
\end{lemma}

{To prove~\Cref{lem:local-hc-rs-cavity-lower}, we} first introduce the two-root cavity coupling described in~\cite{coja2019spin}.
Given $G\sim\+G_{n,d}$, choose an ordered pair of distinct vertices $a,a'$ with $\operatorname{dist}(a,a') \ge 4$ uniformly at random, and independently order their neighbors uniformly as $B=(b_i)_{i\le d}$ and $C=(c_i)_{i\le d}$. 
Set
\begin{align}\label{eq:cavity-graph}
  H_0 \defeq G-\{a,a'\},\qquad
  G' \defeq H_0+\{b_ic_i:i\le d\},\qquad
  {\*e \defeq ((b_i,c_i))_{i\le d}.}
\end{align}

After uniformly relabeling the remaining vertices onto $[n-2]$, let $\mathbb M_n$ denote the law of the marked graph $(G',\*e)$. 
We compare $\mathbb M_n$ with the distribution $\mathbb R_{n-2}$ on pairs $(H,\*e)$, under which $H\sim\+G_{n-2,d}$ and, conditional on $H$, the $d$ {directed} edges in $\*e$ are sampled independently and uniformly. {Let $A_{n-2}$ be the set of pairs $(H,\*e)$ consisting of a simple $d$-regular graph $H$ on $[n-2]$ and an ordered tuple $\*e=((b_i,c_i))_{i\le d}$ of directed edges of $H$ whose $2d$ endpoints are pairwise distinct and satisfy $b_ic_j\notin E(H)$ for all $i\ne j$.}

\begin{lemma}\label{lem:cavity-increment-comparison}
{The law $\mathbb M_n$ is supported on $A_{n-2}$, $\mathbb R_{n-2}(A_{n-2})=1-O(n^{-1})$, and uniformly for $(H,\*e)\in A_{n-2}$,}
\begin{equation}\label{eq:marked-cavity-comparison}
  \frac{\dif\mathbb M_n}{\dif\mathbb R_{n-2}}(H,\*e) = 1+O(n^{-1}).
\end{equation} 
\end{lemma}

\begin{proof}[Proof of~\Cref{lem:local-hc-rs-cavity-lower}]
We begin with expressing the free-energy increment $\bar F_n - \bar F_{n-2}$ using the two-root cavity coupling. Let $\mathbb M_n^\star$ denote the marginal distribution of the graph $G'$ defined in~\eqref{eq:cavity-graph}, and set $\rho_{n-2} \defeq \frac{\dif \mathbb M_n^\star}{\dif \+G_{n-2,d}}$. By~\Cref{lem:cavity-increment-comparison}, $\rho_{n-2} = 1 + O(n^{-1})$ uniformly. Combining this with the variance bound $\operatorname{Var}_{G\sim\+G_{n-2,d}}(\log Z_G) = O(n)$ {by~\eqref{eq:regular-logz-concentration} and~\Cref{lem:simplicity-prob}}, Cauchy--Schwarz yields
\begin{align*}
  \abs{\E[G' \sim \mathbb M_n^\star]{\log Z_{G'}} - \E[G \sim \+G_{n-2,d}]{\log Z_G}}
  &= \abs{\E[G \sim \+G_{n-2,d}]{\bigl(\rho_{n-2}(G) - 1\bigr)\log Z_G}} \\
  &\le \sqrt{\E[G \sim \+G_{n-2,d}]{\bigl(\rho_{n-2}(G) - 1\bigr)^2}} \sqrt{\operatorname{Var}_{G \sim \+G_{n-2,d}}(\log Z_G)} \\
  &= O(n^{-1/2}).
\end{align*}
Consequently, the free-energy increment satisfies
\begin{equation}\label{eq:diff-energy}
  \bar F_n - \bar F_{n-2} = \E[G \sim \+G_{n,d}]{\log \frac{Z_G}{Z_{G'}}} + o(1),
\end{equation}
where $G'$ is coupled to $G$ via~\eqref{eq:cavity-graph}.

Next, we express the ratio $Z_G / Z_{G'}$ in terms of local boundary marginals. {Since $\operatorname{dist}_G(a,a')\ge4$, the neighborhoods $B=N_G(a)$ and $C=N_G(a')$ are disjoint in $H_0=G-\{a,a'\}$. Any independent set of $G$ with $\sigma_a=0$ and $\sigma_{a'}=1$ restricts to an independent set of $H_0$ with $\sigma_C=0$, and conversely any independent set of $H_0$ with $\sigma_C=0$ extends uniquely to an independent set of $G$ with $\sigma_a=0,\sigma_{a'}=1$ and also satisfies the matching constraints $\{b_ic_i:i\le d\}$ of $G'$ because $\sigma_{c_i}=0$ for all $i\le d$. Thus both $Z_G \mu_G(\sigma_a=0, \sigma_{a'}=1)$ and $\lambda Z_{G'} \mu_{G'}(\sigma_C=0)$ equal $\lambda Z_{H_0-C}$.} Hence, we have the exact identity
\begin{equation}\label{eq:diff-energy-2}
  \frac{Z_G}{Z_{G'}} = \frac{\lambda \mu_{G'}(\sigma_C = 0)}{\mu_G(\sigma_a = 0, \sigma_{a'} = 1)}.
\end{equation}
Substituting this identity into~\eqref{eq:diff-energy} yields
\[
  \bar F_n - \bar F_{n-2} = \log \lambda + \E[G \sim \+G_{n,d}]{\log \mu_{G'}(\sigma_C = 0)} - \E[G \sim \+G_{n,d}]{\log \mu_G(\sigma_a = 0, \sigma_{a'}=1)} + o(1).
\]
To evaluate these boundary terms, write $p_{u,H}\defeq\mu_H(\sigma_u=1)$ and $p_u\defeq p_{u,G}$.
{Let $v\in V(G)$ be a uniformly chosen vertex, with expectations below averaging over $(G,a,a',B,C)$ in~\eqref{eq:cavity-graph}, the uniform relabeling of the remaining vertices, and $v$.} We use the following boundary estimates.
\begin{fact}\label{fact:cavity-marginal-estimates}
  Under replica symmetry and the coupling in~\eqref{eq:cavity-graph},
  \begin{align*}
    \E{\log\mu_{G'}(\sigma_C=0)}
      &=d\E{\log(1-p_v)}+o(1),\\
    \E{\log\mu_G(\sigma_a=0,\sigma_{a'}=1)}
      &=\E{\log(1-p_v)}+\E{\log p_v}+o(1).
  \end{align*}
\end{fact}

Applying~\Cref{fact:cavity-marginal-estimates}, we obtain
\begin{equation}\label{eq:cavity-marginal-increment}
  \bar F_n-\bar F_{n-2}
  =\log\lambda+(d-1)\E{\log(1-p_v)}-\E{\log p_v}+o(1).
\end{equation}

For $G\sim\+G_{n,d}$, {denote the law of the empirical BP cavity message distribution $\eta_G\in\mathcal P(\mathcal I_\lambda)$ by $\Lambda_n\in\mathcal P(\mathcal P(\mathcal I_\lambda))$}. Since $\mathcal I_\lambda$ is compact, so is $\mathcal P(\mathcal P(\mathcal I_\lambda))$ in the weak topology. Thus every subsequence has a further subsequence, denoted by $n_k$, along which $\Lambda_{n_k}\xrightarrow{d}\Lambda$ for some $\Lambda\in\mathcal P(\mathcal P(\mathcal I_\lambda))$.
{For $q_1,\ldots,q_d\in[0,1]$, set
\[
  \mathcal U(q_1,\ldots,q_d)\defeq\frac{\lambda\prod_{i=1}^d(1-q_i)}{1+\lambda\prod_{i=1}^d(1-q_i)}.
\]
Let $G_k\sim\+G_{n_k,d}$ with a uniformly chosen vertex $v_k\in V(G_k)$, and for $\eta\sim\Lambda$ and conditional i.i.d.\ draws $Q_1,\ldots,Q_d\sim\eta$, set $p_\star\defeq\mathcal U(Q_1,\ldots,Q_d)$.}
{
\begin{fact}\label{fact:cavity-prob}
  Assume replica symmetry. Along the subsequence for which $\Lambda_{n_k}\xrightarrow{d}\Lambda$, we have $(\eta_{G_k},p_{v_k})\xrightarrow{d}(\eta,p_\star)$.
\end{fact}}

Set $V^{\mathrm{BP}}=1+\lambda\prod_{i=1}^d(1-Q_i)$. Since $1-p_\star=(V^{\mathrm{BP}})^{-1}$ and $p_\star=\lambda\prod_i(1-Q_i)/V^{\mathrm{BP}}$, \Cref{fact:cavity-prob} gives
\[
\begin{aligned}
  \bar F_{n_k}-\bar F_{n_k-2}
  &\longrightarrow\log\lambda+(d-1)\E{\log(1-p_\star)}-\E{\log p_\star}\\
  &=(2-d)\E{\log V^{\mathrm{BP}}}-d\E{\log(1-Q_1)}.
\end{aligned}
\]
The limiting cavity message law also satisfies the BP fixed-point equation.
{
\begin{fact}\label{fact:fixed-point}
  Assume replica symmetry. Then $T\eta=\eta$ for $\Lambda$-almost every $\eta$.
\end{fact}}

{For $R\defeq\mathcal T(Q_1,\ldots,Q_{d-1})$, \Cref{fact:fixed-point} implies that, conditional on $\eta$, $R$ and $Q_d$ are independent with common law $\eta$. The definition of $\mathcal T$ gives the algebraic relation
\begin{equation}\label{eq:bp-vertex-edge-relation}
  1-RQ_d=\frac{V^{\mathrm{BP}}}{1+\lambda\prod_{i=1}^{d-1}(1-Q_i)}=(1-R)V^{\mathrm{BP}}.
\end{equation}
Taking logarithms, conditioning on $\eta$ (under which $(R,Q_d)\overset{d}{=}(Q_1,Q_2)$ for $\Lambda$-almost every $\eta$), and averaging over $\eta\sim\Lambda$ yields} $\E{\log V^{\mathrm{BP}}}=\E{\log(1-Q_1Q_2)}-\E{\log(1-Q_1)}$, and substitution gives
\[
  \bar F_{n_k}-\bar F_{n_k-2} \to 2\E{\log V^{\mathrm{BP}}}-d\E{\log(1-Q_1Q_2)}=2\int_{\mathcal P(\mathcal I_\lambda)}{\mathcal F_0(\eta)}\,\Lambda(\dif\eta)\ge2b_\star,
\]
where the last inequality uses {$\mathcal F_0(\eta)\ge b_\star$} on the BP fixed points.
{Since every subsequence of admissible $n$ has a further subsequence along which $\Lambda_{n_k}\xrightarrow{d}\Lambda$, we conclude that $\liminf_{n\to\infty}(\bar F_n-\bar F_{n-2})\ge 2b_\star$, and telescoping over $n$ in steps of $2$ gives} $\liminf_{n\to\infty} \bar F_n/n \ge b_\star$.
\end{proof}

The proofs of \Cref{lem:cavity-increment-comparison} and Facts~\ref{fact:cavity-marginal-estimates}--\ref{fact:fixed-point} are deferred to \Cref{subsec:app-hc-cavity}.

{{Combining~\Cref{lem:local-hc-fractional-upper} and~\Cref{lem:local-hc-rs-cavity-lower} identifies the limiting quenched free energy with $b_\star$. To obtain the following implication, we show that $b_\star=\phi_d(\alpha)$ and then apply concentration.}

\begin{lemma}\label{lem:hc-pointwise-rs-free-energy}
  Fix an integer $d\ge3$ and $\alpha\in(0,1/2)$. If the hardcore model on $G\sim\+G_{n,d}$ satisfies geometric replica symmetry (\Cref{def:replica-symmetry}) at $\alpha$, then it also satisfies thermodynamic replica symmetry (\Cref{def:free-energies}) at $\alpha$, i.e.,
  \[
    \frac1n\log Z_G(\alpha)\xrightarrow{\mathrm p}\phi_d(\alpha)
    \qquad\text{and}\qquad
    \lim_{n\to\infty}\frac1n\E{\log Z_G(\alpha)}=\phi_d(\alpha).
  \]
\end{lemma}}

\begin{proof}[Proof of~{\Cref{lem:hc-pointwise-rs-free-energy}}]
Combining~\Cref{lem:local-hc-fractional-upper} and~\Cref{lem:local-hc-rs-cavity-lower} establishes the existence of the limiting free energy {$f\defeq\lim_{n\to\infty}\bar F_n/n$} and identifies it with the {BP Bethe variational value}:
\[
  f=\lim_{n\to\infty}\frac{\bar F_n}{n}=b_\star.
\]
It therefore remains to identify the minimizers of the {BP Bethe functional} {$\mathcal F_0$} over the set of BP fixed points. We begin by showing that every non-Dirac fixed point $\eta = T\eta \in \mathcal P(\mathcal I_\lambda)$ satisfies the slope condition
\[
  {\left.\partial_m\mathcal F_m(\eta)\right|_{m=0^+}}<0.
\]
Fix $\eta=T\eta$ and let $Q_1,\ldots,Q_d$ be independent with law $\eta$.
Applying the cumulant expansion
\[
  \frac{1}{m}\log\E{X^m} = \E{\log X} + \frac{m}{2}\operatorname{Var}(\log X) + O(m^2)
\]
to the definition of {$\mathcal F_m(\eta)$}, the derivative at $m = 0^+$ satisfies
\begin{align}\label{eq:partial-Phi}
  {\left.\partial_m\mathcal F_m(\eta)\right|_{m=0^+}} = \frac{1}{2}\left(\operatorname{Var}(\log V^{\mathrm{BP}}) - \frac{d}{2}\operatorname{Var}(\log E^{\mathrm{BP}})\right).
\end{align}
We first apply the {Hoeffding orthogonal decomposition} to $\log V^{\mathrm{BP}}$. Let $\{A_S\}_{S \subseteq [d]}$ be functions {with $A_S$ depending only on $Q_S\defeq(Q_i)_{i\in S}$} such that for any ${I \subseteq [d]}$,
\[
  \sum_{S \subseteq {I}} A_S = \E{\log V^{\mathrm{BP}} \mid Q_{{I}}}
  \qquad{\Bigl(\text{so }A_S=\sum_{I\subseteq S}(-1)^{|S\setminus I|}\E{\log V^{\mathrm{BP}}\mid Q_I}\Bigr).}
\]
A simple calculation yields $\E{A_S A_{S'}} = 0$ for any $S \neq S'$.
Therefore, the variance of $\log V^{\mathrm{BP}}$ can be expressed as
\begin{align}\label{eq:var-log-V}
  \operatorname{Var}(\log V^{\mathrm{BP}}) = \sum_{\varnothing \neq S \subseteq [d]} \E{A^2_S}.
\end{align}

Write $R\defeq\mathcal T(Q_1,\ldots,Q_{d-1})$.
Since $T\eta=\eta$ and $Q_1,\ldots,Q_{d-1}$ are independent, $R\sim\eta$ independently of $Q_d$. 
Hence, we have $\log E^{\mathrm{BP}} \overset{d}{=} \log (1-R Q_d)$, {and~\eqref{eq:bp-vertex-edge-relation} gives} $\log V^{\mathrm{BP}}=-\log(1-R)+\log(1-RQ_d)$.
The variance of $\log E^{\mathrm{BP}}$ can be expressed as
\begin{align}\label{eq:var-log-E}
  \nonumber\operatorname{Var}(\log E^{\mathrm{BP}}) &= \operatorname{Var}(\log (1-R Q_d))\\
  \nonumber &=\operatorname{Var}(\E{\log (1-RQ_d) \mid R}) + \E{\operatorname{Var}(\log (1-R Q_d) \mid R)}\\
  \nonumber &=\operatorname{Var}(\E{\log (1-RQ_d) \mid Q_d}) + \E{\operatorname{Var}(\log (1-R Q_d) \mid R)}\\
  &=\operatorname{Var}(\E{\log V^{\mathrm{BP}} \mid Q_d}) + \E{\operatorname{Var}(\log V^{\mathrm{BP}} \mid Q_1,\ldots,Q_{d-1})}.
\end{align}
By orthogonal decomposition, $\E{\log V^{\mathrm{BP}} \mid Q_d} = A_{\varnothing} + A_{ \{d\}}$, implying 
\[
  \operatorname{Var}(\E{\log V^{\mathrm{BP}} \mid Q_d}) = \E{A_{\{d\}}^2} = \frac{1}{d} a_1,
\]
where $a_k \defeq \sum_{S:\abs{S} = k} \E{A^2_S}$.
Similarly, $\log V^{\mathrm{BP}} - \E{\log V^{\mathrm{BP}} \mid Q_1,Q_2,\ldots,Q_{d-1}} = \sum_{S:d \in S} A_S$, implying 
\[
  \E{\operatorname{Var}(\log V^{\mathrm{BP}} \mid Q_1,Q_2,\ldots,Q_{d-1})} = \sum_{S : d\in S} \E{A_S^2} = \frac{1}{d} \sum_{k=1}^d k a_k.
\]
Combining~\eqref{eq:var-log-E}, the variance of $\log E^{\mathrm{BP}}$ is $\frac{1}{d}\bigl(2 a_1+\sum_{k=2}^d ka_k\bigr)$.
Together with~\eqref{eq:var-log-V} and~\eqref{eq:partial-Phi}, this gives
\begin{equation}\label{eq:local-hc-fixed-point-slope}
\begin{aligned}
  {\left.\partial_m\mathcal F_m(\eta)\right|_{m=0^+}}
  &=\frac12\left(\operatorname{Var}(\log V^{\mathrm{BP}})-\frac d2\operatorname{Var}(\log E^{\mathrm{BP}})\right)\\
  &=\frac12\left(\sum_{k=1}^d a_k-\frac12\left(2a_1+\sum_{k=2}^d k a_k\right)\right)\\
  &=-\frac14\sum_{k=3}^d(k-2)a_k\le0.
\end{aligned}
\end{equation}

We next show that this inequality is strict whenever $\eta$ is non-Dirac. We first observe that the support of any non-Dirac fixed point contains at least three points: if $\operatorname{supp}(\eta)$ contains two distinct points {$q_1<q_2$}, then because $d-1\ge2$ and $\mathcal T$ is strictly monotone in each coordinate, evaluating $\mathcal T$ on tuples containing zero, one, or two copies of {$q_2$} produces three strictly distinct output values. Since $T\eta=\eta$, all three values belong to $\operatorname{supp}(\eta)$.

Fix $Q_4,\ldots,Q_{d} \sim \eta$. The function $\log V^{\mathrm{BP}}$ can be written as $g(1-Q_1,1-Q_2,1-Q_3)$, where
\[
  g(x,y,z)=\log(1+cxyz),\qquad {\text{with }}
  c=\lambda\prod_{i=4}^d (1-Q_i)>0.
\]
Suppose, for contradiction, that $a_k=0$ for all $k\ge3$. {{The Hoeffding orthogonal decomposition} then expresses $\log V^{\mathrm{BP}}$ as a sum of functions of at most two coordinates almost surely. Each component is a finite combination of integrals of the continuous function $\log V^{\mathrm{BP}}$ over the other coordinates, and hence is continuous. The identity therefore extends to the support of the product law. Fix distinct $q_1,q_2\in\operatorname{supp}(\eta)$, put $a=1-q_1$ and $b=1-q_2$, and fix values of $Q_4,\ldots,Q_d$ in $\operatorname{supp}(\eta)$. The resulting three-coordinate identity forces the difference}
\[
  D(z)\defeq g(b,b,z)-g(b,a,z)-g(a,b,z)+g(a,a,z)
\]
{to be constant in $z=1-q$ for $q\in\operatorname{supp}(\eta)$, for these fixed $a,b$.} However, direct calculation gives
\[
  \exp(D(z))
  =\frac{(1+cb^2z)(1+ca^2z)}{(1+cabz)^2}
  =1+\frac{c(b-a)^2z}{(1+cabz)^2}.
\]
For any constant $C>1$, the level equation $1+\frac{c(b-a)^2z}{(1+cabz)^2}=C$ has at most two solutions, violating $|\operatorname{supp}(\eta)| \ge 3$.

Let $\eta_\star$ be a minimizer of {$\mathcal F_0$} with $\eta_\star = T \eta_\star$. If $\eta_\star$ were non-Dirac, \Cref{lem:local-hc-fractional-upper} together with the above claim would imply, for all sufficiently small $m>0$,
\[
  b_\star=f\le{\mathcal F_m(\eta_\star)<\mathcal F_0(\eta_\star)}=b_\star,
\]
which is impossible. Hence the minimizer must be a Dirac mass {$\eta_\star=\delta_{m_\star}$}. For a point mass, the BP fixed-point equation {$T\delta_{m_\star}=\delta_{m_\star}$} reduces to the fixed point of the tree recursion:
\[
  {\frac{m_\star}{(1-m_\star)^d}}=\lambda,
\]
which admits a unique solution {$m_\star = \theta_\alpha = \frac{\alpha}{1-\alpha}$ in $(0,1)$ (\eqref{eq:hardcore_kappa})}. Therefore, the {BP Bethe functional} gives
\[
  f=b_\star={\mathcal F_0(\delta_{\theta_\alpha})}
  =\log\bigl(1+\lambda(1-\theta_\alpha)^d\bigr)-\frac d2\log(1-\theta_\alpha^2)
  =\phi_d(\alpha),
\]
{and concentration~\eqref{eq:regular-logz-concentration} upgrades $n^{-1}\bar F_n\to\phi_d(\alpha)$ to $n^{-1}\log Z_G(\alpha)\xrightarrow{\mathrm p}\phi_d(\alpha)$,} which completes the proof.
\end{proof}

{For the antiferromagnetic Potts model, we follow a similar proof strategy, using the Bethe-state decomposition of~\cite{coja2019spin} and the concentration bound~\eqref{eq:regular-logz-concentration}. The proof of \Cref{lem:potts-pointwise-rs-free-energy} is given in \Cref{subsec:app-potts-rs}.}

{
\begin{lemma}\label{lem:potts-pointwise-rs-free-energy}
  Fix integers $d\ge3$ and $q\ge2$, and let $\beta\in(0,1]$. If the Potts model on $G\sim\+G_{n,d}$ satisfies geometric replica symmetry (\Cref{def:replica-symmetry}) at $\beta$, then it also satisfies thermodynamic replica symmetry (\Cref{def:free-energies}) at $\beta$, i.e.,
  \[
    \frac1n\log Z_{G,q}(\beta)\xrightarrow{\mathrm p}\phi_{d,q}(\beta)
    \qquad\text{and}\qquad
    \lim_{n\to\infty}\frac1n\E{\log Z_{G,q}(\beta)}=\phi_{d,q}(\beta).
  \]
\end{lemma}
}

{By Dobrushin uniqueness, the Potts model satisfies geometric replica symmetry on some interval $[\beta_0,1]$ with $\beta_0\in(0,1)$. Hence \Cref{lem:potts-pointwise-rs-free-energy} gives thermodynamic replica symmetry throughout this interval.}

\subsection{From Thermodynamic to Variational Replica Symmetry}\label{subsec:local-bethe-consequences}\label{subsec:local-strict-gap}

In this subsection, we prove that thermodynamic replica symmetry implies variational replica symmetry. Equivalently, we establish the following contrapositive statement.
%  failure of variational replica symmetry implies failure of thermodynamic replica symmetry.

\begin{lemma}\label{lem:thermodynamic-implies-variational}\label{lem:strict-free-energy-gaps}
Fix integers $d\ge3$ and $q\ge2$, and let $G\sim\+G_{n,d}$. The following holds.
\begin{enumerate}[label=(\roman*),font=\normalfont,leftmargin=2.5em,labelsep=0.6em,topsep=0.25em,itemsep=0.2em,parsep=0pt]
  \item \emph{Hardcore:} {If $\alpha\in(0,1/2)$ and $\mathcal B_{d,\alpha}(\nu)>0$ for some $\nu\in\mathcal P_{\*\pi_\alpha}$, then there exists $\varepsilon=\varepsilon(d,\alpha)>0$ such that}
  \[
    \lim_{n\to\infty}\Pr{\frac1n\log Z_G(\alpha)\le\phi_d(\alpha)-\varepsilon}=1
    \qquad\text{and}\qquad
    \limsup_{n\to\infty}\frac1n\E{\log Z_G(\alpha)}\le\phi_d(\alpha)-\varepsilon.
  \]
  \item \emph{Antiferromagnetic Potts:} {If $\beta\in(0,1)$ and $\mathcal B_{d,\beta}(\nu)>0$ for some $\nu\in\mathcal P_{\*U_q}$, then there exists $\varepsilon=\varepsilon(d,q,\beta)>0$ such that}
  \[
    \lim_{n\to\infty}\Pr{\frac1n\log Z_{G,q}(\beta)\le\phi_{d,q}(\beta)-\varepsilon}=1
    \quad\text{and}\quad
    \limsup_{n\to\infty}\frac1n\E{\log Z_{G,q}(\beta)}\le\phi_{d,q}(\beta)-\varepsilon.
  \]
\end{enumerate}
\end{lemma}

{Taking the contrapositive of \Cref{lem:strict-free-energy-gaps} yields the desired implication at nontrivial parameters. At the Potts endpoint $\beta=1$, variational replica symmetry holds directly because $\mathcal B_{d,1}\equiv0$.}

\begin{corollary}\label{cor:thermodynamic-implies-variational}
Fix integers $d\ge3$ and $q\ge2$, and let $G\sim\+G_{n,d}$. If the hardcore model at $\alpha\in(0,1/2)$ or the antiferromagnetic Potts model at $\beta\in(0,1]$ satisfies thermodynamic replica symmetry (\Cref{def:free-energies}), then it also satisfies variational replica symmetry (\Cref{def:rsb-bethe-functional}) at the same parameter.
\end{corollary}

We prove \Cref{lem:strict-free-energy-gaps}(i) below and defer the proof of part~(ii) to \Cref{subsec:app-variational-bridges}.
We first relate the fractional Bethe free-energy functional {$\mathcal F_m(\eta)$} to the Bethe functional $\mathcal B_{d,\alpha}(\nu)$ (\Cref{def:rsb-bethe-functional}). Fix $d\ge3$, $\alpha\in(0,1/2)$, and a {Bethe trial law} $\nu\in\mathcal P_{\*\pi_\alpha}$, identified with the law of $u\in[0,1]$ through $\*\mu=(1-u,u)^\trans$ with $\E{u}=\alpha$. For $0<m\le1$, put $\theta=\theta_\alpha=\alpha/(1-\alpha)$ and $a(u)\defeq 1-\theta+\theta u$, and define $\eta_m$ as the law of $Q\defeq u/a(u)$ under the tilted distribution $\nu_m(\dif u)\propto a(u)^m\nu(\dif u)$. Let $V_{d,\alpha},E_\alpha$ denote the cavity vertex and edge factors in~\eqref{eq:marked-functional} evaluated at independent messages with law $\nu$.
% Since $\E{u}=\alpha<1$ and $a(u)>0$, we have $\eta_m(\{1\})<1$, so $\mathcal F_m(\eta_m)$ is well-defined.

\begin{proposition}\label{prop:nu-eta-fractional-correspondence}
Every $0<m\le1$ satisfies
\begin{equation}\label{eq:Phi-normalized-fractional-correspondence}
  {\mathcal F_m(\eta_m)}
  =\phi_d(\alpha)
   +\frac1m\log\E{V_{d,\alpha}^m}
   -\frac d{2m}\log\E{E_\alpha^m}.
\end{equation}
\end{proposition}
{The algebraic verification of~\eqref{eq:Phi-normalized-fractional-correspondence} is given in \Cref{subsec:app-variational-bridges}.}

\begin{proof}[{Proof of~\Cref{lem:strict-free-energy-gaps}(i)}]
{Take a witness $\nu\in\mathcal P_{\*\pi_\alpha}$ with $\mathcal B_{d,\alpha}(\nu)>0$.}
By continuity, we may first replace each message $\*\mu$ by $(1-\delta)\*\mu+\delta\*\pi_\alpha$ for a sufficiently small fixed $\delta>0$, preserving this strict inequality and ensuring that the corresponding $\eta_m$ is compactly supported in $[0,1)$.

{For any bounded nonnegative $X$ with $\E{X}=1$,
\begin{equation}\label{eq:fractional-moment-derivative}
  \left.\frac{\dif}{\dif m}\left(\frac1m\log\E{X^m}\right)\right|_{m=1^-}
  =\E{X\log X}.
\end{equation}
Since $\E{V_{d,\alpha}}=\E{E_\alpha}=1$, applying~\eqref{eq:fractional-moment-derivative} to $V_{d,\alpha}$ and $E_\alpha$ in~\eqref{eq:Phi-normalized-fractional-correspondence} gives $\mathcal F_1(\eta_1)=\phi_d(\alpha)$ and identifies the left derivative of $\mathcal F_m(\eta_m)$ at $m=1$ with $\mathcal B_{d,\alpha}(\nu)$.} Consequently,
\[
  {\mathcal F_{1-m_0}(\eta_{1-m_0})}
  =\phi_d(\alpha)-m_0\mathcal B_{d,\alpha}(\nu)+o(m_0)
  <\phi_d(\alpha)
\]
for a sufficiently small fixed $m_0>0$, so~\Cref{lem:local-hc-fractional-upper} {and concentration~\eqref{eq:regular-logz-concentration} establish both the expected and high-probability strict free-energy gaps in~(i).}
\end{proof}

\subsection{From Variational to Geometric Replica Symmetry}\label{subsec:local-variational-geometric}\label{subsec:pointwise-conditions}

We establish geometric replica symmetry throughout the interior of the variational replica-symmetric regime. In this subsection, the interior is taken at fixed integer $d$ in the positive-parameter interval $(0,1/2)$ for hardcore or $(0,1)$ for Potts; the coloring endpoint $\beta=0$ is treated separately using real-degree slack.

\begingroup
\begin{lemma}\label{lem:variational-rs-converse}
Fix integers $d\ge3$ and $q\ge2$, and let $G\sim\+G_{n,d}$. Consider the hardcore model at $t=\alpha\in(0,1/2)$ or the antiferromagnetic Potts model at $t=\beta\in(0,1)$. If $t$ lies in the interior of the variational replica-symmetric regime, then the model satisfies geometric replica symmetry (\Cref{def:replica-symmetry}).
\end{lemma}
\endgroup

{To prove \Cref{lem:variational-rs-converse}, we first establish the strict variational conditions in \Cref{lem:model-variational-conditions}: unique maximization of the Bethe functional and strict Kesten--Stigum stability. We then establish geometric replica symmetry using the overlap characterization in \Cref{lem:rs-overlap-characterization}, the entropy bound in \Cref{lem:revealed-entropy-bound}, and the planted-overlap bound in \Cref{lem:exponential-overlap}. These conditions also contribute to the zero-freeness proof in \Cref{sec:variational-rs}.}

\begin{condition}\label{cond:unique-variational-maximizer}
The Bethe functional satisfies $\mathcal B_{d,t}(\nu)<0$ for all $\nu\in\mathcal P_{\*\pi(t)}\setminus\{\delta_{\*\pi(t)}\}$, where $\mathcal B_{d,t}$ and the space $\mathcal P_{\*\pi(t)}$ of cavity message distributions are defined in \Cref{def:rsb-bethe-functional}.
\end{condition}

\begin{condition}\label{cond:strict-KS}
The inequality $(d-1)\max_{a\ge2}|\kappa_a(t)|^2<1$ holds, where $\kappa_a(t)$, $a\ge2$, are the nonconstant eigenvalues of the tree broadcasting operator $K_t$ from \Cref{def:bethe-canonical}.
\end{condition}

By~\eqref{eq:hardcore_kappa} and~\eqref{eq:Potts_kappa}, \Cref{cond:strict-KS} reads $(d-1)\theta_\alpha^2<1$ for hardcore and $(d-1)\kappa_\beta^2<1$ for Potts. {The next lemma shows that both conditions hold at nontrivial parameters in the interior of the {variational replica-symmetric regime},} building on~\cite[Section~5 and Lemma~5.4]{coja2018charting}.

\begin{lemma}\label{lem:model-variational-conditions}
Fix integers $d\ge3$ and $q\ge2$. Consider the hardcore model at $t=\alpha\in(0,1/2)$ or the antiferromagnetic Potts model at $t=\beta\in(0,1)$. If $t$ lies in the interior of the {variational replica-symmetric regime}, then both \Cref{cond:unique-variational-maximizer} and \Cref{cond:strict-KS} hold at $(d,t)$.

{Moreover, the property of variational replica symmetry is monotone. Specifically, it is preserved under decreasing $\alpha$ for hardcore and increasing $\beta$ for Potts.}
\end{lemma}

\begin{proof}[Proof of~\Cref{lem:model-variational-conditions}]
{
For hardcore, put $\theta=\theta_\alpha$ from~\eqref{eq:hardcore_kappa} and encode $\*\mu=(1-u,u)^\trans\in\Delta_{\{0,1\}}$ by $h=(1+\theta)u-\theta\in[-\theta,1]$, so $\*\mu\sim\nu\in\mathcal P_{\*\pi_\alpha}$ corresponds to a mean-zero law of $h$ on $[-\theta,1]$. Since $(M_\alpha\*\mu)_0=1+\theta h$ and $(M_\alpha\*\mu)_1=1-h$, the factors in~\eqref{eq:marked-functional} for i.i.d.\ $\*\mu_1,\ldots,\*\mu_d\sim\nu$ with $h_i=(1+\theta)u_i-\theta$ are
\begin{equation}\label{eq:hc-message-functional}
  V_{d,\alpha}=\frac{\prod_{i=1}^d(1+\theta h_i)+\theta\prod_{i=1}^d(1-h_i)}{1+\theta},\qquad
  E_\alpha=1-h_1h_2.
\end{equation}
}

\paragraph{Strict Kesten--Stigum stability (\Cref{cond:strict-KS}).}
In the setup of~\Cref{def:rsb-bethe-functional}, let ${\*w}\in\=R^\Omega$ be an eigenvector of $K_t$ with nonconstant eigenvalue $\kappa_a$, {normalized by $\sum_i\pi_i w_i^2=1$. Then $\sum_i\pi_i w_i=0$ and $M_t\diag(\*\pi)\*w=\kappa_a\*w$}. The symmetric two-point cavity law
\[
  \nu_\varepsilon\defeq\frac12\bigl(\delta_{\*\pi+\varepsilon\diag(\*\pi){\*w}}+\delta_{\*\pi-\varepsilon\diag(\*\pi){\*w}}\bigr)\in\mathcal P_{\*\pi}
\]
satisfies a fourth-order Taylor expansion around $\varepsilon=0$: averaging over i.i.d.\ uniform signs $\xi_i\in\{\pm1\}$ gives
\[
  \E{V_{k,t}\log V_{k,t}}=\frac12\binom k2\kappa_a^4\varepsilon^4+O(\varepsilon^6)\quad(k\ge0),
  \qquad
  \E{E_t\log E_t}=\frac12\kappa_a^2\varepsilon^4+O(\varepsilon^8),
\]
and hence
\begin{equation}\label{eq:unified-fourth-order}
  \mathcal B_{d,t}(\nu_\varepsilon)
  =\frac d4\kappa_a^2\bigl((d-1)\kappa_a^2-1\bigr)\varepsilon^4+O(\varepsilon^6).
\end{equation}

\smallskip
Thus nonpositivity of $\mathcal B_{d,t}$ implies $(d-1)\max_{a\ge2}|\kappa_a(t)|^2\le1$. Applying this bound at a parameter $\alpha'\in(\alpha,1/2)$ or $\beta'\in[0,\beta)$ available in the interior of the variational replica-symmetric regime makes the inequality strict at $t$, since $\theta_\alpha$ strictly increases with $\alpha$ and $|\kappa_\beta|$ strictly decreases with $\beta$. This proves \Cref{cond:strict-KS} for both models.
\paragraph{{Unique maximizer for the} hardcore model.}
{Suppose $\alpha\in(0,1/2)$ has $\alpha'\in(\alpha,1/2)$ with $\mathcal B_{d,\alpha'}\le0$,} with $\theta_0\defeq\theta_\alpha<\theta'\defeq\theta_{\alpha'}<1$, and take any $\nu\in\mathcal P_{\*\pi_\alpha}\setminus\{\delta_{\*\pi_\alpha}\}$. For $\*\mu=(1-u,u)^\trans\sim\nu$, the variable $h=(1+\theta_0)u-\theta_0$ has a nonconstant mean-zero law $\eta\ne\delta_0$ on $[-\theta_0,1]$. For each $\theta\in[\theta_0,1)$, set $\alpha(\theta)\defeq\theta/(1+\theta)$ and let $\nu_\theta\in\mathcal P_{\*\pi_{\alpha(\theta)}}$ be the law of $(1-u_\theta,u_\theta)^\trans$, where
\[
  u_\theta\defeq\frac{h+\theta}{1+\theta}\quad(h\sim\eta),
  \qquad
  F(\theta)\defeq\E[\nu_\theta^{\otimes d}]{V_{d,\alpha(\theta)}\log V_{d,\alpha(\theta)}},
\]
so that $\nu_{\theta_0}=\nu$ and $\nu'\defeq\nu_{\theta'}\in\mathcal P_{\*\pi_{\alpha'}}$.

\smallskip
Let $h_1,\ldots,h_d$ be i.i.d.\ under the tilted probability law $(1+\theta h)\eta(\dif h)$ and put $W_i\defeq(1-h_i)/(1+\theta h_i)\ge0$. These variables satisfy $\E{W_i}=1$. Couple independent pairs $(W_i,W_i^*)$ with $W_i^*\ge W_i$, where $W_i^*\sim w\Pr{W_i\in\dif w}$ is the \emph{size-biased} version of $W_i$. The size-biasing identity gives $\E{f(W_i^*)}=\E{W_i f(W_i)}$. For $g(w_1,\ldots,w_d)\defeq\log(1+\theta\prod_i w_i)$, differentiating under $\eta^{\otimes d}$, rewriting under the tilted law, and applying this identity give
\[
\begin{aligned}
 (1+\theta)^2F'(\theta)
 &=\E{\Bigl(\prod_{i=1}^dW_i-\sum_{i=1}^dW_i+d-1\Bigr)g(W_1,\ldots,W_d)}\\
 &=\E{g(W_1^*,\ldots,W_d^*)-\sum_{i=1}^d g(W_1,\ldots,W_i^*,\ldots,W_d)+(d-1)g(W_1,\ldots,W_d)}>0,
\end{aligned}
\]
since $\partial_i\partial_jg=\theta\prod_{k\ne i,j}w_k/(1+\theta\prod_k w_k)^2\ge0$, with strict inequality on the positive-probability event where $W_1^*>W_1$, $W_2^*>W_2$, and $W_3,\ldots,W_d>0$. Since the edge factor $E_{\alpha(\theta)}=1-h_1h_2$, {evaluated under the original law $\eta^{\otimes2}$}, is independent of $\theta$, the function $\theta\mapsto\mathcal B_{d,\alpha(\theta)}(\nu_\theta)$ strictly increases on $[\theta_0,1)$. Since $\mathcal B_{d,\alpha'}\le0$ at $\alpha'\in(\alpha,1/2)$,
\[
  \mathcal B_{d,\alpha}(\nu)<\mathcal B_{d,\alpha'}(\nu')\le0,
\]
proving \Cref{cond:unique-variational-maximizer}. {This comparison only assumes nonpositivity at $\alpha'$, so it also proves that variational replica symmetry is preserved under decreasing $\alpha$.}

\paragraph{{Unique maximizer for the} antiferromagnetic Potts model.}
Suppose $\beta\in(0,1)$ has $\beta'\in[0,\beta)$ with $\mathcal B_{d,\beta'}\le0$.
With $\kappa_\beta$ from~\eqref{eq:Potts_kappa} and the uniform marginal $\*U_q=q^{-1}\allone_q$ (\Cref{subsec:preliminaries}), the identity $M_\beta\*\mu=\allone+q\kappa_\beta(\*\mu-\*U_q)$ allows us to contract each cavity message toward $\*U_q$: set $c\defeq \kappa_\beta/\kappa_{\beta'}\in(0,1)$ and let $\nu_c\in\mathcal P_{\*U_q}$ be the law of $\*\mu'=(1-c)\*U_q+c\*\mu$ for $\*\mu\sim\nu$. Then $M_{\beta'}\*\mu'=M_\beta\*\mu$, so {for every $k\ge0$,}
\[
  V_{{k},\beta'}(\*\mu'_1,\ldots,\*\mu'_{k})=V_{{k},\beta}(\*\mu_1,\ldots,\*\mu_{k}),
  \qquad
  E_{\beta'}(\*\mu'_1,\*\mu'_2)=1-c+cE_\beta(\*\mu_1,\*\mu_2).
\]
Convexity of $x\mapsto x\log x$ gives
\begin{equation}\label{eq:potts-message-contraction}
  \mathcal B_{d,\beta}(\nu)
  \le \mathcal B_{d,\beta'}(\nu_c)
      -\frac{d}2(1-c)\E{E_\beta\log E_\beta}.
\end{equation}
For any $\nu\in\mathcal P_{\*U_q}\setminus\{\delta_{\*U_q}\}$, the covariance matrix $\operatorname{Cov}(\*\mu)\defeq\E{(\*\mu-\*U_q)(\*\mu-\*U_q)^\trans}$ is nonzero, so $\E{E_\beta}=1$ and $\operatorname{Var}(E_\beta)=q^2\kappa_\beta^2\operatorname{tr}(\operatorname{Cov}(\*\mu)^2)>0$; strict convexity therefore gives $\E{E_\beta\log E_\beta}>0$ and~\eqref{eq:potts-message-contraction} yields $\mathcal B_{d,\beta}(\nu)<0$, proving \Cref{cond:unique-variational-maximizer}. {Since the comparison only requires nonpositivity at $\beta'$, it also proves preservation of variational replica symmetry under increasing $\beta$; the endpoint $\beta=1$ follows from $\mathcal B_{d,1}\equiv0$.}
\end{proof}

{
We use the following characterization to pass from overlap concentration to geometric replica symmetry.
For random probability measures $\nu_n$ on $\Omega^n$, the \emph{pairwise-decorrelation} condition is
\begin{equation}\label{eq:measure-pairwise-decorrelation}
  \lim_{n\to\infty}\frac1{n^2}\sum_{1\le u<v\le n}\E{\DTV{\nu_{n,u,v}}{\nu_{n,u}\otimes\nu_{n,v}}}=0.
\end{equation}
In particular, replica-symmetric Gibbs measures satisfy this condition (\Cref{def:replica-symmetry}).
\begin{lemma}\label{lem:rs-overlap-characterization}
Let $\nu_n$ be random probability measures on $\Omega^n$ for a fixed finite spin space $\Omega$, and let $\Sigma_{\nu_n}$ be the covariance matrix of the indicator vector $(\one{\sigma_v=a})_{v\in[n],a\in\Omega}$ under $\sigma\sim\nu_n$.
The sequence $\nu_n$ satisfies~\eqref{eq:measure-pairwise-decorrelation} if and only if $\E{\|\Sigma_{\nu_n}\|_{\mathrm F}^2}/n^2\allowbreak=o(1)$, which implies $\E{\|\Sigma_{\nu_n}\|_{\mathrm{op}}}/n=o(1)$.
In particular, $\nu_n$ satisfies~\eqref{eq:measure-pairwise-decorrelation} whenever the overlap $\*\rho(\sigma,\sigma')$ under $(\sigma,\sigma')\sim\nu_n^{\otimes2}$ concentrates in probability around a deterministic matrix.
\end{lemma}
The proof is given together with the auxiliary measure properties of \Cref{lem:rs-measure-properties} in \Cref{subsec:app-rs-measures}.

We next introduce the planted law and two estimates needed for the proof. Fix an integer $d\ge3$, write $\mathrm{CM}=\mathrm{CM}_{n,d}$, and set $a_n(z)=\E[\mathrm{CM}]{\mathcal Z_G(z)}$. Unless otherwise specified, expectations are under $\mathrm{CM}$. For the pinned log-likelihood bound below, take $t=\alpha\in(0,1/2)$ for hardcore or $t=\beta\in(0,1)$ for Potts, so that $\mathcal Z_G(t)>0$ on every multigraph. The planted law is also defined at the coloring endpoint whenever $a_n(0)>0$.

\begin{definition}\label{def:planted-overlap}
  {For a real parameter $t$ with $M=M(t)\ge0$ and $a_n(t)>0$,} the \emph{planted distribution} $\mathbb Q_n$ under $\mathrm{CM}$ is the joint law on $(G,\sigma)$ defined by
  \[
    \frac{\dif\mathbb Q_n}{\dif(\mathrm{CM}\otimes\*\pi^{\otimes n})}(G,\sigma)
    \defeq a_n^{-1}\prod_{uv\in E(G)}M_{\sigma_u,\sigma_v}.
  \]
  {Given $G$ with $\mathcal Z_G>0$, two independent samples $(\sigma,\sigma')\sim\mu_G^{\otimes2}$ are called \emph{Gibbs replicas}.}
\end{definition}
{The graph marginal of $\mathbb Q_n$ has density $\mathcal Z_G/a_n$ with respect to $\mathrm{CM}$, and the conditional law of $\sigma$ given $G$ is $\mu_G$. For functions of $G$ alone, expectations under $\mathbb Q_n$ refer to its graph marginal.}

{
{Fix $s\in[0,1]$ and form $U$ by including each vertex independently with probability $s$, independently of $(G,\sigma)$. Reveal the planted spins on $U$. For $(G,\sigma)\sim\mathbb Q_n$, define the \emph{pinned log-likelihood ratio}}
\[
 \ell_s=\ell_s(\sigma)\defeq\log\mu_G(\sigma_U)-\sum_{v\in U}\log\pi_{\sigma_v},
\]
{where $\mu_G(\sigma_U)\defeq\sum_{\sigma'\in\Omega^{V(G)}:\,\sigma'_U=\sigma_U}\mu_G(\sigma')$ is the marginal probability of $\sigma_U=(\sigma_v)_{v\in U}$ under $\mu_G$.}
The following lemma is adapted from~\cite[Section~14]{coja2020inference}; a proof for the present fixed-degree model is given in~\Cref{app:revealed-entropy}.

\begin{lemma}\label{lem:revealed-entropy-bound}
{$|\log a_n(t)|=O(\log(n+1))$. Moreover, for every $s\in[0,1]$ and $b>0$, there exists $c>0$ such that $\Pr[\mathbb Q_n]{\ell_s-\E[\mathbb Q_n]{\ell_s}>bn}\le\e^{-cn}$ for all large $n$, and}
\begin{equation}\label{eq:planted-free-energy-bound}
 \limsup_{n\to\infty}\frac1n\E[\mathbb Q_n]{\ell_s}
 \le\limsup_{n\to\infty}\frac1n\E[\mathbb Q_n]{\log\mathcal Z_G+\ell_s}
 \le\max_{\nu\in\mathcal P_{\*\pi}}
 \left\{\mathcal B_{d,t}(\nu)
       +d s\E[\nu]{\sum_{a\in\Omega}\pi_a((M\*\mu)_a-1)^2}\right\}.
\end{equation}
\end{lemma}
}

For an overlap box $B$, let $\mathcal Z_{G,B}^{(2)}(z,w)$ denote the normalized total weight of pairs whose empirical overlap lies in $B$:
\begin{equation}\label{eq:local-second-moment-sum}
 \mathcal Z_{G,B}^{(2)}(z,w)=\sum_{\substack{\sigma,\sigma'\in\Omega^{[n]}\\\*\rho(\sigma,\sigma')\in B}}
 \prod_{v\in V(G)}\pi_{\sigma_v}(z)\pi_{\sigma'_v}(w)
 \prod_{uv\in E(G)}M_{\sigma_u,\sigma_v}(z)M_{\sigma'_u,\sigma'_v}(w).
\end{equation}
Here an overlap box is the preimage of an open axis-aligned box in $|\Omega|^2-1$ coordinates of $\Delta_{\Omega\times\Omega}$.
The following truncation estimate also applies at the coloring endpoint $t=0$.

{Because the local second-moment sum $\mathcal Z_{G,B}^{(2)}(z,w)$ restricts pairs of spin configurations rather than the graph $G$ itself, we next construct a typical-graph event $E_n$ such that the graph truncation} $Y_n(z)\defeq\mathcal Z_G(z)\one{E_n}$ {matches $\mathcal Z_G(z)$ in first moment and $\mathcal Z_{G,B}^{(2)}(z,w)$ in second moment up to exponentially small errors.}

\begin{lemma}\label{lem:graph-truncation}
  Assume \Cref{cond:unique-variational-maximizer} and \Cref{cond:strict-KS} hold at a nontrivial parameter $t$. For a sufficiently small overlap box $B$ about $\*\pi(t)\otimes\*\pi(t)$, there exist graph events $E_n$, an open disk ${D_0}$ about $t$, and $\gamma>0$ such that ${\E{(\mathcal Z_G(t)^2-\mathcal Z_{G,B}^{(2)}(t,t))\one{E_n}}=O(\e^{-\gamma n}a_n(t)^2)}$ and, uniformly for $z,w\in {D_0}$,
  \begin{equation}\label{eq:holomorphic-truncation}
   \begin{aligned}
    \E{Y_n(z)}&=a_n(z)\bigl(1+O(\e^{-\gamma n})\bigr),\\
    \E{Y_n(z)Y_n(w)}&=\E{\mathcal Z_{G,B}^{(2)}(z,w)}
      +O\bigl(\e^{-\gamma n}|a_n(z)a_n(w)|\bigr).
   \end{aligned}
  \end{equation}
\end{lemma}
{{The proof of \Cref{lem:graph-truncation} is given in \Cref{app:coloring-endpoint}, using the local moment estimates of \Cref{lem:local-profile-moments} and the planted overlap bound of \Cref{lem:exponential-overlap}. These estimates are proved from the stated variational conditions without using \Cref{lem:variational-rs-converse}.}}

We can now prove \Cref{lem:variational-rs-converse}, using the entropy bound in \Cref{lem:revealed-entropy-bound} and the planted-overlap bound in \Cref{lem:exponential-overlap}. We treat the hard-coloring endpoint separately in \Cref{subsec:local-coloring-phase}.
}

\begin{proof}[Proof of~\Cref{lem:variational-rs-converse}]
By \Cref{lem:model-variational-conditions}, both \Cref{cond:unique-variational-maximizer} and \Cref{cond:strict-KS} hold at $t$.
Write $\mathbb P_n=\mathrm{CM}_{n,d}$, $F_n=\log\mathcal Z_G(t)$, and $a_n=\E[\mathbb P_n]{\mathcal Z_G(t)}$. Define the planted graph law by $\dif\mathbb Q_n^G/\dif\mathbb P_n=\mathcal Z_G(t)/a_n$. At reveal density $s=0$, the pinned log-likelihood in \Cref{lem:revealed-entropy-bound} is zero. Since \Cref{cond:unique-variational-maximizer} gives $\max_\nu\mathcal B_{d,t}(\nu)=0$, that lemma yields
\[
 |\log a_n|=O(\log n),\qquad
 0\le D_n\defeq\DKL(\mathbb Q_n^G\Vert\mathbb P_n)
   =\E[\mathbb Q_n^G]{F_n}-\log a_n=o(n).
\]
Jensen's inequality gives $\E[\mathbb P_n]{F_n}\le\log a_n=o(n)$. For the reverse bound, suppose that $\E[\mathbb P_n]{F_n}\le-\varepsilon n$ along a subsequence, for some $\varepsilon>0$, and put $H_n=\{F_n> -\varepsilon n/2\}$. Concentration~\eqref{eq:regular-logz-concentration} gives $\mathbb P_n(H_n)\le\e^{-c_\varepsilon n}$, whereas the planted density gives
\[
 \mathbb Q_n^G(H_n^c)
 =a_n^{-1}\E[\mathbb P_n]{\e^{F_n}\one{H_n^c}}
 \le a_n^{-1}\e^{-\varepsilon n/2}=o(1).
\]
Applying the log-sum inequality to the partition $\{H_n,H_n^c\}$ yields
\[
 D_n\ge\mathbb Q_n^G(H_n)\log\frac1{\mathbb P_n(H_n)}-\log2
 \ge(1-o(1))c_\varepsilon n-\log2,
\]
contradicting $D_n=o(n)$. Thus $n^{-1}\E[\mathbb P_n]{F_n}\to0$, and~\eqref{eq:regular-logz-concentration} gives $F_n/n\xrightarrow{\mathrm p}0$ under $\mathbb P_n$.

Fix $\eta>0$ and let $A_\eta$ be the event that the overlap of two Gibbs replicas differs from $\*\pi(t)\otimes\*\pi(t)$ by more than $\eta$ in $\ell_\infty$ norm. By \Cref{lem:exponential-overlap}, $q_n\defeq\E[\mathbb Q_n^G]{\mu_G^{\otimes2}(A_\eta)}\le\e^{-c_\eta n}$ for some $c_\eta>0$ and all large $n$. At these positive parameters, $\mathcal Z_G(t)>0$ for every configuration-model graph, so $\dif\mathbb P_n/\dif\mathbb Q_n^G=a_n\e^{-F_n}$.

Set $L_n=\{F_n\ge-c_\eta n/2\}$. Since $\E[\mathbb P_n]{F_n}=o(n)$, the concentration bound~\eqref{eq:regular-logz-concentration} gives $\mathbb P_n(L_n^c)\le\e^{-c'_\eta n}$ for some $c'_\eta>0$ and all large $n$. Splitting according to $L_n$ and using the same conditional Gibbs law under both graph measures gives
\[
 \E[\mathbb P_n]{\mu_G^{\otimes2}(A_\eta)}
 \le\mathbb P_n(L_n^c)+a_n\e^{c_\eta n/2}q_n
 \le\e^{-c'_\eta n}+a_n\e^{-c_\eta n/2}
 =\e^{-\Omega(n)},
\]
where the last bound uses $|\log a_n|=O(\log n)$. Conditioning on simplicity preserves both the overlap and free-energy convergence, since its probability is bounded away from zero by \Cref{lem:simplicity-prob}. Thus $n^{-1}\log Z_G(t)\xrightarrow{\mathrm p}\phi(t)$ under $\+G_{n,d}$, and \Cref{lem:rs-overlap-characterization} turns the overlap concentration into geometric replica symmetry. The hard-coloring endpoint is handled separately in \Cref{subsec:local-coloring-phase}.
\end{proof}

{To define the replica symmetry breaking thresholds, introduce the positive-witness sets and their extremal boundaries:
\begin{equation}\label{eq:variational-thresholds-def}
\begin{aligned}
  S_d&\defeq\bigl\{\alpha\in(0,1/2):\exists\,\nu\in\mathcal P_{\*\pi_\alpha}\text{ with }\mathcal B_{d,\alpha}(\nu)>0\bigr\},
  &\alpha_{\mathrm{RSB}}(d)&\defeq\inf S_d,\\
  T_{d,q}&\defeq\bigl\{\beta\in(0,1):\exists\,\nu\in\mathcal P_{\*U_q}\text{ with }\mathcal B_{d,\beta}(\nu)>0\bigr\},
  &\beta_{\mathrm{RSB}}(d,q)&\defeq\sup T_{d,q},
\end{aligned}
\end{equation}
with $\inf\varnothing=1/2$ and $\sup\varnothing=0$, respectively. By \Cref{lem:model-variational-conditions}, $S_d$ is upward closed in $(0,1/2)$ and $T_{d,q}$ is downward closed in $(0,1)$. Thus $\mathcal B$ is nonpositive at every parameter strictly on the replica-symmetric side of the threshold and has a positive witness at every parameter strictly on the other side.}

\begin{proof}[Proof of \Cref{thm:hardcore-three-characterizations} and \Cref{thm:potts-three-characterizations}, except for the hardcore threshold bounds and~\eqref{eq:replica-symmetry-breaking-hardcore-asymptotic}]
For the thresholds defined in~\eqref{eq:variational-thresholds-def}, we make no assertion at the boundary.

The hardcore threshold bounds $0<\alpha_{\mathrm{RSB}}(d)<1/2$ and the asymptotic expansion~\eqref{eq:replica-symmetry-breaking-hardcore-asymptotic} are proved in \Cref{subsec:rsb-threshold-asymptotics}. For Potts, the high-temperature free-energy formula established in \Cref{subsec:local-rs-free-energy} and the implication proved in \Cref{subsec:local-bethe-consequences} show that the Bethe functional is nonpositive for $\beta$ sufficiently close to $1$. Thus $\beta_{\mathrm{RSB}}(d,q)<1$.

If $\alpha<\alpha_{\mathrm{RSB}}(d)$, choose $\alpha'\in(\alpha,\alpha_{\mathrm{RSB}}(d))$; if $\beta\in(\beta_{\mathrm{RSB}}(d,q),1)$, choose $\beta'\in(\beta_{\mathrm{RSB}}(d,q),\beta)$. Nonpositivity at these stronger parameters places the original parameter in the interior of the variational replica-symmetric regime, so \Cref{lem:model-variational-conditions} gives \Cref{cond:unique-variational-maximizer} and \Cref{cond:strict-KS}. Hence \Cref{lem:variational-rs-converse} gives geometric replica symmetry, and \Cref{lem:hc-pointwise-rs-free-energy} or \Cref{lem:potts-pointwise-rs-free-energy} then gives thermodynamic replica symmetry. At $\beta=1$, the Potts Gibbs measure is a product measure, $Z_{G,q}(1)=q^n$, and $\mathcal B_{d,1}\equiv0$, so all three properties hold directly.

On the other side of either threshold, a positive Bethe witness exists, so variational replica symmetry fails. By \Cref{lem:strict-free-energy-gaps}, thermodynamic replica symmetry fails as well. Geometric replica symmetry would imply thermodynamic replica symmetry by \Cref{lem:hc-pointwise-rs-free-energy} or \Cref{lem:potts-pointwise-rs-free-energy}, so it also fails. This completes both phase characterizations.
\end{proof}

\subsection{Extension to Proper Colorings}\label{subsec:local-coloring-phase}

We extend the three implications to proper $q$-colorings at the hard-constraint endpoint $\beta=0$. The converse from variational to geometric replica symmetry holds in the interior of the variational replica-symmetric regime, with slack in the interpolated degree.

\paragraph{From geometric to thermodynamic replica symmetry.}
Geometric replica symmetry identifies the limiting coloring free energy as follows.
\begin{lemma}\label{lem:coloring-pointwise-rs-free-energy}\label{lem:integer-hard-coloring-entropy}
Fix integers $d,q\ge3$. If $G\sim\+G_{n,d}$ is $q$-colorable with high probability and its uniform proper-coloring measure satisfies geometric replica symmetry (\Cref{def:replica-symmetry}), then it also satisfies thermodynamic replica symmetry (\Cref{def:free-energies}), i.e.,
\[
 \frac1n\log Z_{G,q}(0)\xrightarrow{\mathrm p}\phi_{d,q}(0).
\]

Moreover,
\[
 \lim_{n\to\infty}\frac1n\E{\log Z_{G,q}(0)\mid Z_{G,q}(0)>0}=\phi_{d,q}(0).
\]
\end{lemma}
The proof in \Cref{app:hard-coloring-entropy} compares the law obtained by choosing a uniform colorable graph and then a uniform proper coloring with the annealed planted law. Color-preserving half-edge switches and geometric replica symmetry show that the relative entropy of the former law with respect to the latter is $o(n)$, which identifies the conditional quenched free energy. The annealed upper bound then yields convergence in probability. This argument treats the hard constraint directly.

\paragraph{From thermodynamic to variational replica symmetry.}
The next implication follows by ruling out a positive Bethe witness whenever the free energy attains the Bethe prediction.
\begin{corollary}\label{cor:coloring-thermodynamic-implies-variational}
Fix integers $d,q\ge3$ and let $G\sim\+G_{n,d}$. If the proper-coloring model satisfies thermodynamic replica symmetry (\Cref{def:free-energies}), then it also satisfies variational replica symmetry (\Cref{def:rsb-bethe-functional}), i.e., $\mathcal B_{d,0}(\nu)\le0$ for every $\nu\in\mathcal P_{\*U_q}$.
\end{corollary}
To prove this corollary, we establish the following strict-gap statement.
\begin{lemma}\label{lem:coloring-free-energy-gap}
Fix integers $d,q\ge3$ and let $G\sim\+G_{n,d}$. If $\mathcal B_{d,0}(\nu)>0$ for some $\nu\in\mathcal P_{\*U_q}$, then there exists $\varepsilon=\varepsilon(d,q)>0$ such that
\[
 \lim_{n\to\infty}\Pr{\frac1n\log Z_{G,q}(0)\le\phi_{d,q}(0)-\varepsilon}=1.
\]
\end{lemma}
We use the fractional-moment upper bound from~\cite[Corollary~2.11]{ayre2022lower} at a fixed positive Potts parameter and then compare its partition function with the number of proper colorings.
\begin{lemma}\label{lem:local-potts-fractional-upper}
Fix integers $d,q\ge3$, a real $\beta\in(0,1)$, and $\nu\in\mathcal P_{\*U_q}$. For every $m\in(0,1)$,
\begin{equation}\label{eq:local-potts-fractional-upper}
 \limsup_{n\to\infty}\frac1n\E[G\sim\+G_{n,d}]{\log Z_{G,q}(\beta)}
 \le\phi_{d,q}(\beta)
       +\frac1m\log\E{V_{d,\beta}^{\,m}}
       -\frac d{2m}\log\E{E_\beta^{\,m}}.
\end{equation}
\end{lemma}
\begin{proof}[Proof of~\Cref{lem:coloring-free-energy-gap}]
Fix a witness $\nu$ with $\mathcal B_{d,0}(\nu)>0$, and write $H_m(\beta)$ for the right-hand side of~\eqref{eq:local-potts-fractional-upper}. The cavity factors are uniformly bounded with mean~$1$, so $H_m$ extends continuously to $\beta=0$ for each fixed $m\in(0,1]$. Moreover, $H_1(0)=\phi_{d,q}(0)$ and~\eqref{eq:fractional-moment-derivative} gives
\[
 \left.\frac{\partial}{\partial m}H_m(0)\right|_{m=1^-}
 =\mathcal B_{d,0}(\nu)>0.
\]
Choose $m<1$ with $H_m(0)<\phi_{d,q}(0)$. By continuity, there are fixed $\beta_0\in(0,1)$ and $\varepsilon>0$ such that $H_m(\beta_0)<\phi_{d,q}(0)-3\varepsilon$. The fractional-moment bound gives
\[
 \frac1n\E{\log Z_{G,q}(\beta_0)}\le\phi_{d,q}(0)-2\varepsilon
\]
for all sufficiently large $n$. Concentration~\eqref{eq:regular-logz-concentration} at this fixed positive parameter and $Z_{G,q}(0)\le Z_{G,q}(\beta_0)$ now imply the claimed gap.
\end{proof}
\begin{proof}[Proof of~\Cref{cor:coloring-thermodynamic-implies-variational}]
A positive Bethe witness would give the strict gap in \Cref{lem:coloring-free-energy-gap}, contradicting convergence in probability to $\phi_{d,q}(0)$.
\end{proof}

\paragraph{From variational to geometric replica symmetry.}
For the converse, the strict variational conditions from \Cref{subsec:local-variational-geometric} give both high-probability colorability and geometric replica symmetry.
\begin{lemma}\label{lem:coloring-variational-rs-converse}
Fix integers $d,q\ge3$. If \Cref{cond:unique-variational-maximizer} and \Cref{cond:strict-KS} hold at $(d,0)$, then $G\sim\+G_{n,d}$ is $q$-colorable with high probability and its uniform proper-coloring measure satisfies geometric replica symmetry (\Cref{def:replica-symmetry}).
\end{lemma}
To apply this lemma in the interior of the variational replica-symmetric regime, we first show that nonpositivity of the Bethe functional at a real degree $r>d$ supplies both strict conditions.

\begin{lemma}\label{lem:coloring-variational-conditions}
Fix integers $d,q\ge3$. If there exists a real $r>d$ such that $\mathcal B_{r,0}(\nu)\le0$ for every $\nu\in\mathcal P_{\*U_q}$, then both \Cref{cond:unique-variational-maximizer} and \Cref{cond:strict-KS} hold at $(d,0)$.

{Moreover, variational replica symmetry is preserved under decreasing the interpolated degree. For every $\nu\in\mathcal P_{\*U_q}\setminus\{\delta_{\*U_q}\}$, the function $s\mapsto\mathcal B_{s,0}(\nu)$ is convex on $[0,\infty)$, with $\mathcal B_{0,0}(\nu)=0$ and $\mathcal B_{1,0}(\nu)<0$.}
\end{lemma}

\begin{proof}
For each integer $k\ge0$ and $\nu\in\mathcal P_{\*U_q}$, set $\mathcal A_k\defeq\E{V_{k,0}\log V_{k,0}}$.
If $A$ is uniform on $[q]$ and, given $A=a$, the observations $\xi_1,\xi_2,\ldots$ are independent with density $\*\mu\mapsto q(1-\mu_a)/(q-1)$ relative to $\nu$, then each observation has marginal $\nu$ and the joint law $\nu^{(k)}$ of $(\xi_1,\ldots,\xi_k)$ has density $V_{k,0}$ relative to $\nu^{\otimes k}$. Writing $I(X;Z)\defeq\DKL(P_{X,Z}\Vert P_X\otimes P_Z)$ for the mutual information of $(X,Z)$ with joint law $P_{X,Z}$ and marginals $P_X,P_Z$, we have
\[
  \mathcal A_k=\DKL\bigl(\nu^{(k)}\big\Vert\nu^{\otimes k}\bigr),
  \qquad
  \mathcal A_{k+1}-\mathcal A_k=I\bigl(\xi_{k+1};(\xi_1,\ldots,\xi_k)\bigr).
\]

\smallskip
By the chain rule for relative entropy and exchangeability of $(\xi_1,\xi_2,\ldots)$, these increments satisfy $I(\xi_{k+2};(\xi_1,\ldots,\xi_{k+1}))\ge I(\xi_{k+2};(\xi_1,\ldots,\xi_k))=I(\xi_{k+1};(\xi_1,\ldots,\xi_k))$, so $(\mathcal A_k)_{k\ge0}$ is a convex sequence with $\mathcal A_0=\mathcal A_1=0$. Subtracting the linear edge term $\frac{r'}{2}\E{E_0\log E_0}$ from its piecewise-linear extension shows that $r'\mapsto\mathcal B_{r',0}(\nu)$ is convex on $[0,\infty)$ with $\mathcal B_{0,0}(\nu)=0$. For any $\nu\in\mathcal P_{\*U_q}\setminus\{\delta_{\*U_q}\}$, $\operatorname{Var}(E_0)>0$ gives $\mathcal B_{1,0}(\nu)=-\frac12\E{E_0\log E_0}<0$. Since $\mathcal B_{0,0}(\nu)=0$, convexity implies $\mathcal B_{s,0}(\nu)\le(s/r)\mathcal B_{r,0}(\nu)$ for $r>0$ and $0\le s\le r$, proving the claimed monotonicity. Now suppose that $r>d$ and $\mathcal B_{r,0}(\nu)\le0$ for all $\nu\in\mathcal P_{\*U_q}$. Interpolating between degrees $1$ and $r$ gives
\[
 \mathcal B_{d,0}(\nu)
 \le\frac{r-d}{r-1}\mathcal B_{1,0}(\nu)+\frac{d-1}{r-1}\mathcal B_{r,0}(\nu)
 \le\frac{r-d}{r-1}\mathcal B_{1,0}(\nu)<0,
\]
proving \Cref{cond:unique-variational-maximizer} at $(d,0)$.

It remains to prove \Cref{cond:strict-KS}, which at $\beta=0$ reads $(d-1)/(q-1)^2<1$. For the interpolated functional at degree $r$, the expansion~\eqref{eq:unified-fourth-order} and convexity of $k\mapsto k(k-1)$ give a fourth-order coefficient of at least $\frac r4\kappa_0^2((r-1)\kappa_0^2-1)$. Hence $\mathcal B_{r,0}\le0$ implies $(r-1)\kappa_0^2\le1$. Since $r>d$ and $\kappa_0=-1/(q-1)$, we obtain $(d-1)/(q-1)^2<1$.
\end{proof}

We now prove the geometric conclusion. The truncation estimate from \Cref{lem:graph-truncation} controls the weight of replica pairs away from the independent overlap. A standard lower-tail consequence of small subgraph conditioning supplies the denominator bound needed to turn this estimate into an overlap probability.
\begin{proof}[Proof of~\Cref{lem:coloring-variational-rs-converse}]
Let $A$ be the complement of a fixed neighborhood of $\*U_q\otimes\*U_q$, and choose the overlap box $B\subseteq A^c$ in \Cref{lem:graph-truncation}. Write
\[
 \mathcal Z=\mathcal Z_G(0),\qquad
 a_n=\E[\mathrm{CM}_{n,d}]{\mathcal Z},\qquad
 Y_n=\mathcal Z\one{E_n},\qquad
 A_n=\E[\mathrm{CM}_{n,d}]{Y_n}.
\]
By \Cref{lem:graph-truncation}, $A_n/a_n\to1$. The same truncation estimate, combined with a standard local second-moment calculation as in~\cite[Proposition~5.1]{coja2016chromatic} (see also \Cref{lem:local-profile-moments}), yields
\[
 \frac{\E[\mathrm{CM}_{n,d}]{Y_n^2}}{A_n^2}
 \longrightarrow\exp\!\left(\sum_{\ell\ge1}\lambda_\ell\delta_\ell(0)^2\right).
\]
The sum is finite by \Cref{cond:strict-KS}. The cycle counts have independent Poisson limits, and their mixed factorial moments with $Y_n$ agree asymptotically with those for $\mathcal Z$: every fixed product of cycle counts is polynomially bounded, whereas $\E[\mathrm{CM}_{n,d}]{\mathcal Z-Y_n}=O(\e^{-\gamma n}a_n)$. These are the hypotheses of classical small subgraph conditioning.

Since $\delta_1(0)=-1$ and $\delta_\ell(0)>-1$ for $\ell\ge2$, the standard lower-tail bound~\cite[Theorem~4.1 and Corollary~4.2]{wormald1999models} implies that $Y_n/a_n$ is bounded away from zero in probability conditional on $X_{1,n}=0$. Conditioning further on $X_{2,n}=0$ preserves this conclusion, since $\Pr[\mathrm{CM}_{n,d}]{X_{2,n}=0\mid X_{1,n}=0}\to\e^{-\lambda_2}>0$. Thus, for every deterministic sequence $\varepsilon_n\downarrow0$,
\[
 \Pr[\+G_{n,d}]{Y_n/a_n<\varepsilon_n}=o(1).
\]
In particular, $Y_n>0$ with high probability, so $G$ is $q$-colorable with high probability.

The overlap error in \Cref{lem:graph-truncation} also gives
\[
 \E[\mathrm{CM}_{n,d}]{Y_n^2\mu_G^{\otimes2}(A)}
 \le\E[\mathrm{CM}_{n,d}]{\bigl(\mathcal Z^2-\mathcal Z_{G,B}^{(2)}(0,0)\bigr)\one{E_n}}
 =O(\e^{-\gamma n}a_n^2).
\]
Conditioning on simplicity changes this bound by only a constant factor. Therefore
\[
 \E[\+G_{n,d}]{\mu_G^{\otimes2}(A)}
 \le\Pr[\+G_{n,d}]{Y_n/a_n<\e^{-\gamma n/4}}
      +O(\e^{-\gamma n/2})=o(1).
\]
Here $\mu_G$ may be defined arbitrarily on uncolorable graphs, whose probability is $o(1)$. The overlap characterization in \Cref{lem:rs-overlap-characterization} proves geometric replica symmetry.
\end{proof}

These three implications complete the coloring phase characterization. Define
\begin{equation}\label{eq:coloring-variational-threshold-def}
 \mathcal U_q\defeq\bigl\{r\in[3,\infty):\exists\,\nu\in\mathcal P_{\*U_q}\text{ with }\mathcal B_{r,0}(\nu)>0\bigr\},
 \qquad d_{\mathrm{RSB}}(q)\defeq\inf\mathcal U_q,
\end{equation}
with $\inf\varnothing=\infty$. By \Cref{lem:coloring-variational-conditions}, $\mathcal U_q$ is upward closed.
\begin{proof}[Phase characterization in \Cref{thm:coloring-three-characterizations}]
The one-hot trial in \Cref{subsec:rsb-threshold-asymptotics} gives a positive Bethe witness at sufficiently large integer degrees, so $d_{\mathrm{RSB}}(q)\in[3,\infty)$.

For an integer $d<d_{\mathrm{RSB}}(q)$, the definition of the threshold gives
\begin{equation}\label{eq:coloring-real-degree-nonpositivity}
 d<d_{\mathrm{RSB}}(q)\quad\Longrightarrow\quad
 \exists\,r\in(d,d_{\mathrm{RSB}}(q)):\quad
 \mathcal B_{r,0}(\nu)\le0\quad\forall\nu\in\mathcal P_{\*U_q}.
\end{equation}
Degree monotonicity gives variational replica symmetry at $d$, and \Cref{lem:coloring-variational-conditions} supplies the two strict conditions. Hence \Cref{lem:coloring-variational-rs-converse} gives high-probability colorability and geometric replica symmetry. By \Cref{lem:coloring-pointwise-rs-free-energy}, thermodynamic replica symmetry and the conditional quenched free-energy identity follow as well.

For an integer $d>d_{\mathrm{RSB}}(q)$, upward closure of $\mathcal U_q$ supplies a positive witness, so variational replica symmetry fails. The strict gap in \Cref{lem:coloring-free-energy-gap} rules out thermodynamic replica symmetry. High-probability colorability together with geometric replica symmetry would imply thermodynamic replica symmetry by \Cref{lem:coloring-pointwise-rs-free-energy}, so the geometric property also fails. No assertion is made when $d=d_{\mathrm{RSB}}(q)$.
\end{proof}
\subsection{Asymptotics of the Replica Symmetry Breaking Thresholds}\label{subsec:rsb-threshold-asymptotics}
In this subsection, we prove~\eqref{eq:replica-symmetry-breaking-hardcore-asymptotic} and~\eqref{eq:replica-symmetry-breaking-potts-asymptotic}. {We also establish $\alpha_{\mathrm{RSB}}(d)\in(0,1/2)$ for every $d\ge3$ and $d_{\mathrm{RSB}}(q)\in[3,\infty)$ for every $q\ge3$. These results complete the proofs of~\Cref{thm:hardcore-three-characterizations}, \Cref{thm:coloring-three-characterizations}, and~\Cref{thm:hardcore-rs-transition}.}
Let $N_n(\alpha)$ denote the number of independent sets of size $\lfloor\alpha n\rfloor$ in $G\sim\mathrm{CM}_{n,d}$. We record two consequences of moment and variational estimates.

The following second-moment estimate follows from~\cite[Theorems~1.3 and~11.1, Section~2.1, eqs.~(3) and~(7)--(9)]{gerencser2025boosted}.
\begin{lemma}\label{lem:local-gerencser-second-moment}
For all sufficiently large $d$, there exists $\alpha\in(0,1/2)$ with $d\alpha=2\log d-2\log\log d+2-2\log2+o_d(1)$ such that 
\[
  \frac{1}{n}\log\E{N_n(\alpha)}=\phi_d(\alpha)-\alpha\log\lambda(\alpha)+o_n(1) \qquad \text{and} \qquad \log\bigl(\E{N_n(\alpha)^2}/(\E{N_n(\alpha)})^2\bigr)=o(n).
\]
\end{lemma}

For colorings, we use the following variational bound from~\cite[Theorem~1.2, eq.~(1.7), and Section~1.2]{coja2017information}.
\begin{lemma}\label{lem:local-poisson-coloring-bethe}
There is a sequence $\varepsilon_q\to0$ such that for all sufficiently large $q$ and real $3\le r<(2q-1)\log q-2\log2+\varepsilon_q$,
\[
  \max_{\nu\in\mathcal P_{\*U_q}}\E[K\sim\mathrm{Poisson}(r)]{\mathcal B_{K,0}(\nu)}=0.
\]
\end{lemma}

\begin{proof}[Proof of~\eqref{eq:replica-symmetry-breaking-hardcore-asymptotic}]
{{For every $d\ge3$, Weitz's strong spatial mixing bound~\cite[Corollary~2.6]{weitz06counting} holds throughout $0<\alpha<1/d$ and implies geometric replica symmetry. By \Cref{lem:hc-pointwise-rs-free-energy}, thermodynamic replica symmetry also holds there. The contrapositive of \Cref{lem:strict-free-energy-gaps}(i) therefore gives $\mathcal B_{d,\alpha}\le0$ on $(0,1/d)$, and~\eqref{eq:variational-thresholds-def} gives $\alpha_{\mathrm{RSB}}(d)\ge1/d>0$.}}
For the asymptotic lower bound, let $\alpha$ and $N_n=N_n(\alpha)$ be given by~\Cref{lem:local-gerencser-second-moment}. The Paley--Zygmund inequality and $Z_G(\alpha)\ge\lambda(\alpha)^{\lfloor\alpha n\rfloor}N_n$ give
\[
  \Pr[\mathrm{CM}_{n,d}]{\log Z_G(\alpha)\ge n\phi_d(\alpha)-o(n)}
  \ge\Pr[\mathrm{CM}_{n,d}]{N_n\ge\tfrac12\E{N_n}}\ge\exp\{-o(n)\}.
\]
The concentration and model-comparison bounds in~{\eqref{eq:regular-logz-concentration}} give $\liminf_n n^{-1}\E[\+G_{n,d}]{\log Z_G(\alpha)}\ge\phi_d(\alpha)$. Otherwise, the preceding event would have exponentially small probability. Applying~\Cref{lem:local-hc-fractional-upper} with $\eta=\delta_{\theta_\alpha}$ and $m\uparrow1$ (\Cref{prop:nu-eta-fractional-correspondence}) gives the matching upper bound, yielding
\[
  \lim_{n\to\infty}\frac1n\E[\+G_{n,d}]{\log Z_G(\alpha)}=\phi_d(\alpha).
\]
{A positive Bethe witness at $\alpha$ would contradict this equality by \Cref{lem:strict-free-energy-gaps}(i); thus~\eqref{eq:variational-thresholds-def} gives $\alpha\le\alpha_{\mathrm{RSB}}(d)$, proving the lower bound.}

For the upper bound, take the {hard cavity message distribution} $\nu_{\mathrm{hard}}\defeq(1-\alpha)\delta_{(1,0)^\trans}+\alpha\delta_{(0,1)^\trans}\in\mathcal P_{\*\pi_\alpha}$. Applying $(x+y)\log(x+y)\ge x\log x+y\log y$ to the two summands of $V_{d,\alpha}$ and using $M_\alpha\*\pi_\alpha=\allone$ gives
\begin{equation}\label{eq:local-hardcore-hard-trial}
  \mathcal B_{d,\alpha}(\nu_{\mathrm{hard}})
  \ge\sum_{a\in\{0,1\}}(\*\pi_\alpha)_a\log(\*\pi_\alpha)_a
    +\frac d2\sum_{a,b\in\{0,1\}}(P_\alpha)_{ab}\log(M_\alpha)_{ab}
  =\alpha\log\lambda(\alpha)-\phi_d(\alpha).
\end{equation}
{For every $d\ge3$, the right-hand side of~\eqref{eq:local-hardcore-hard-trial} tends to $(d/2-1)\log 2>0$ as $\alpha\uparrow1/2$, so {\eqref{eq:variational-thresholds-def} yields $\alpha_{\mathrm{RSB}}(d)<1/2$}.}
For $\alpha=t/d$ with $t=\Theta(\log d)$, the right-hand side expands as
\[
  \alpha\log\lambda(\alpha)-\phi_d(\alpha)
  =\alpha\Bigl(\log\alpha-1+\frac t2+O(t^2/d)\Bigr),
\]
which is strictly positive for $t=2\log d-2\log\log d+2-2\log2+{\varepsilon}$ and all sufficiently large $d$. {By~\eqref{eq:variational-thresholds-def}, this positivity forces $\alpha\ge\alpha_{\mathrm{RSB}}(d)$.} Since ${\varepsilon>0}$ is arbitrary, we obtain $\alpha_{\mathrm{RSB}}(d)\le(2\log d-2\log\log d+2-2\log2+o_d(1))/d$, matching the lower bound.
\end{proof}

\begin{proof}[Proof of~\eqref{eq:replica-symmetry-breaking-potts-asymptotic}]
{By~\eqref{eq:coloring-variational-threshold-def}, $\mathcal U_q\subseteq[3,\infty)$ gives $d_{\mathrm{RSB}}(q)\ge3$ for every $q\ge3$.} For sufficiently large $q$ and real $3\le r<(2q-1)\log q-2\log2+\varepsilon_q$, degree convexity {of $s\mapsto\mathcal B_{s,0}(\nu)$ on $[0,\infty)$ (\Cref{lem:coloring-variational-conditions})}, Jensen's inequality, and~\Cref{lem:local-poisson-coloring-bethe} give, for every $\nu\in\mathcal P_{\*U_q}$,
\[
  \mathcal B_{r,0}(\nu)
  \le\E[K\sim\mathrm{Poisson}(r)]{\mathcal B_{K,0}(\nu)}\le0.
\]
The expectation is finite: $V_{k,0}\le(q/(q-1))^k$ and $\E{V_{k,0}}=1$ give $0\le\E{V_{k,0}\log V_{k,0}}\le k\log(q/(q-1))$, while the edge term is linear in $k$. {By~\eqref{eq:coloring-variational-threshold-def}, this nonpositivity gives $d_{\mathrm{RSB}}(q)\ge(2q-1)\log q-2\log2+o_q(1)$.}

{
For the matching upper bound, fix $c>0$ and set $r_q=(2q-1)\log q-2\log2+c$, $c_q=q/(q-1)$, and $\nu_{\mathrm{hard}}=q^{-1}\sum_{a=1}^q\delta_{\*e_a}\in\mathcal P_{\*U_q}$. For $k$ independent messages, let $A_k$ count the absent colors. The vertex factor is $V_{k,0}=c_q^kA_k/q$, and the edge factor is $c_q$ for distinct labels and zero otherwise. By color symmetry, the expectation of $\log V_{k,0}$ under vertex-factor reweighting can be evaluated by conditioning a fixed color to be absent, giving
\[
 \mathcal B_{k,0}(\nu_{\mathrm{hard}})
 =\frac k2\log c_q-\log q+\E{\log(1+J_k)},
\]
where $J_k$ counts absent colors after $k$ independent uniform draws from $q-1$ colors. {In particular, since $J_k\ge0$, for every fixed $q\ge3$ we have $\mathcal B_{k,0}(\nu_{\mathrm{hard}})\ge\frac k2\log c_q-\log q>0$ for all sufficiently large $k$, so $d_{\mathrm{RSB}}(q)<\infty$ by~\eqref{eq:coloring-variational-threshold-def}.} Moreover, the factorial moments of $J_k$ are
\[
 \E{J_k}=(q-1)\left(1-\frac1{q-1}\right)^k,\qquad
 \E{J_k(J_k-1)}=(q-1)(q-2)\left(1-\frac2{q-1}\right)^k.
\]
Using $0\le j\log2-\log(1+j)\le(\log2)j(j-1)$ for integers $j\ge0$, these give, uniformly for integers $k=r_q+O(1)$,
\[
 \E{\log(1+J_k)}=\frac{\log2}{q}+O\!\left(\frac{\log q}{q^2}\right),\qquad
 \mathcal B_{k,0}(\nu_{\mathrm{hard}})
 =\frac{k-(2q-1)\log q+2\log2}{2q}+O\!\left(\frac{\log q}{q^2}\right).
\]
The piecewise-linear interpolation between $k=\lfloor r_q\rfloor$ and $k=\lceil r_q\rceil$ in \Cref{def:rsb-bethe-functional} gives
\[
 \mathcal B_{r_q,0}(\nu_{\mathrm{hard}})=\frac{c}{2q}+O\!\left(\frac{\log q}{q^2}\right)>0
\]
for sufficiently large $q$. {By~\eqref{eq:coloring-variational-threshold-def}, this positive witness gives $d_{\mathrm{RSB}}(q)\le r_q$; letting $c\downarrow0$ gives the matching bound.}
}
\end{proof}

% END INCLUDED FILE: chapters/sec5-deduplicated.tex
% BEGIN INCLUDED FILE: chapters/04-zero-freeness.tex
\section{\texorpdfstring{{Zero-Freeness and Counting in the Variational Replica-Symmetric Regime}}{Zero-Freeness and Counting in the Variational Replica-Symmetric Regime}}
\label{sec:variational-rs}

{
{In this section, we prove a general zero-freeness theorem under strict variational conditions (\Cref{thm:variational-zero-free}), and combine it with the threshold characterizations of \Cref{sec:rs-identification-local} to deduce our main zero-free and approximate counting results (\Cref{subsec:main-zero-free-counting-proofs}).}
}

\begingroup
\begin{theorem}[Zero-freeness under strict variational conditions]\label{thm:variational-zero-free}
  Fix an integer $d\ge3$ and a nontrivial parameter $t$ for the hardcore or antiferromagnetic Potts model (\Cref{subsec:preliminaries}). If \Cref{cond:unique-variational-maximizer} and \Cref{cond:strict-KS} hold at $(d,t)$, then $Z_G$ is zero-free on a fixed complex neighborhood of $t$ with probability $1-o_n(1)$ over $G\sim\+G_{n,d}$. The same conclusion holds at the independent-spin endpoints $\alpha=0$ and $\beta=1$.
\end{theorem}
\endgroup

{
{The phase characterizations in \Cref{sec:rs-identification-local} identify the three replica-symmetry properties away from their respective thresholds. At nontrivial positive parameters, \Cref{lem:model-variational-conditions} supplies the strict conditions on the replica-symmetric side of the threshold. At the coloring endpoint, \Cref{lem:coloring-variational-conditions} supplies them when the interpolated Bethe functional is nonpositive at some real degree $r>d$; this degree slack holds when $d<d_{\mathrm{RSB}}(q)$ by~\eqref{eq:coloring-real-degree-nonpositivity}.} At the independent-spin endpoints $\alpha=0$ and $\beta=1$, zero-freeness holds deterministically on all graphs of maximum degree at most $d$ (\Cref{lem:endpoint-zero-free}).
}

\subsection{Local Moments and Cycle Factors}\label{subsec:common-moments}

Throughout this subsection, $d\ge3$ is a fixed integer, {$t$ is a nontrivial parameter}, and unsubscripted expectations and probabilities are taken under $\mathrm{CM}_{n,d}$. Oriented edge counts in a multigraph $G$ on $[n]$ include both orientations of each edge. Each self-loop contributes two orientations.

We begin by defining the {first-moment} empirical profile and Bethe rate function.

\begin{definition}[Empirical profile and Bethe rate function]\label{def:profiles-bethe-rates}
  Let $G$ be a $d$-regular multigraph on $V(G)=[n]$ and $\sigma\in\Omega^{[n]}$ be a spin assignment with $V_a\defeq\sigma^{-1}(a)$.
  \begin{itemize}
    \item The \emph{profile} $(\*\rho,Q)\in\Delta_\Omega\times\Delta_{\Omega\times\Omega}$ of $(G,\sigma)$ has empirical vertex marginal $\*\rho=\*\rho(\sigma)$ with $\rho_a(\sigma)\defeq|V_a|/n$ and $dnQ_{ab}$ equal to the number of oriented edges from $V_a$ to $V_b$. A pair $(\*\rho,Q)$ {of a vertex marginal $\*\rho$ and an edge marginal $Q$ with vertex marginal $\*\rho$ (\Cref{subsec:preliminaries})} is \emph{feasible} if $n\rho_a\in\=Z$, $dnQ_{ab}\in\=Z$ for $a\ne b$, and $dnQ_{aa}\in2\=Z$.
    \item For real $s$ with {$\*\pi(s)>0$} and $P(s)\ge0$, the \emph{Bethe rate function} on $\Delta_\Omega\times\Delta_{\Omega\times\Omega}$ is defined by
    \begin{equation}\label{eq:planted-profile-gap}
      \Phi_s(\*\rho,Q)
      \defeq (d-1)\DKL\bigl(\*\rho\big\Vert\*\pi(s)\bigr)-\frac d2\DKL\bigl(Q\big\Vert P(s)\bigr).
    \end{equation}
  \end{itemize}
\end{definition}

\begin{proposition}[First-moment Bethe upper bound]\label{prop:profile-moments}
  \leavevmode
  \begin{enumerate}
    \item {For real $s$ with $\*\pi(s)>0$ and $P(s)\ge0$,} uniformly over all feasible profiles $(\*\rho,Q)$ on $\Omega$, the expected contribution to $\mathcal Z_G(s)$ from assignments with profile $(\*\rho,Q)$ is $\exp\bigl\{n\Phi_s(\*\rho,Q)+O_{d,|\Omega|}(\log n)\bigr\}$ (equal to $0$ when $\Phi_s(\*\rho,Q)=-\infty$).
    \item If $\kappa_a(t)\le0$ for all $a\ge2$, then there exists $c>0$ such that for every real $s$ near $t$ with $P(s)\ge0$ and every {edge marginal} $Q$ with vertex marginal $\*\rho$,
    \[
      \Phi_s(\*\rho,Q)\le-c\bigl(\|\*\rho-\*\pi(s)\|_2^2+\|Q-P(s)\|_{\mathrm F}^2\bigr),
    \]
    and consequently $\log\E{\mathcal Z_G(s)}\le O(\log n)$.
  \end{enumerate}
\end{proposition}

{For the second moment, we consider pairs of spin configurations ${(\sigma,\sigma')}\in\Omega^{[n]}\times\Omega^{[n]}$, viewed as single configurations on the product alphabet $\Omega\times\Omega$.

\begin{definition}[{Empirical overlap, local second moment, and paired Bethe rate function}]\label{def:local-second-moment}
  Let $G$ be a $d$-regular multigraph on $V(G)=[n]$.
  \begin{itemize}
    \item We use the empirical overlap $\*\rho(\sigma,\sigma')$ and overlap boxes from \Cref{subsec:local-variational-geometric}, and the local second-moment sum $\mathcal Z_{G,B}^{(2)}(z,w)$ defined in~\eqref{eq:local-second-moment-sum}.
    \item For real $s,u$ with {$\*\pi(s),\*\pi(u)>0$} and $P(s),P(u)\ge0$, the \emph{paired Bethe rate function} on $\Delta_{\Omega\times\Omega}\times\Delta_{(\Omega\times\Omega)\times(\Omega\times\Omega)}$ is
    \[
      \Phi_{s,u}(\*\rho,Q)
      \defeq (d-1)\DKL\bigl(\*\rho\big\Vert\*\pi(s)\otimes\*\pi(u)\bigr)-\frac d2\DKL\bigl(Q\big\Vert P(s)\otimes P(u)\bigr).
    \]
  \end{itemize}
\end{definition}

Because the paired kernel $K_t\otimes K_t$ on $\Omega\times\Omega$ can have positive product eigenvalues $\kappa_a(t)\kappa_b(t)>0$, the paired rate function $\Phi_{t,t}$ can become positive away from $\*\pi(t)\otimes\*\pi(t)$. Restricting the second moment to a small overlap box $B$ around $\*\pi(t)\otimes\*\pi(t)$ ensures a negative-definite local Hessian under \Cref{cond:strict-KS}. Moreover, since $(K_z^\trans)_{ab}=\pi_a(z)M_{ab}(z)$ and $K_z$ has trivial eigenvalue $1$ and nonconstant eigenvalues $\kappa_2(z),\ldots,\kappa_{|\Omega|}(z)$, the short-cycle factors $\delta_\ell(z)$ from \Cref{def:short-cycle-factors} satisfy the trace identity
\begin{equation}\label{eq:cycle-factor-trace}
  \delta_\ell(z)=\mathcal Z_{C_\ell}(z)-1=\operatorname{tr}K_z^\ell-1=\sum_{a\ge2}\kappa_a(z)^\ell,
\end{equation}
so \Cref{cond:strict-KS} directly controls $\sum_{\ell\ge1}\lambda_\ell|\delta_\ell(z)|^2$ near $t$.
The following lemma supplies the resulting local moment and short-cycle inputs for \Cref{thm:reinforced-ssc}; {the proofs of \Cref{prop:profile-moments} and \Cref{lem:local-profile-moments}}, adapted from~\cite{bollobas1980probabilistic,robinson1992almost,janson1995random,coja2018charting}, are given in \Cref{app:appendix-moments}.

\begin{lemma}[Local moments and short-cycle estimates]\label{lem:local-profile-moments}
  Assume that {$t$ is nontrivial}, $\kappa_a(t)\le0$ for all $a\ge2$, and \Cref{cond:strict-KS} holds at $t$. For every sufficiently small open overlap box $B$ around $\*\pi(t)\otimes\*\pi(t)$, there exist $c>0$ and an open complex disk ${D_0}$ about $t$ such that:
  \begin{enumerate}
  \item $a_n(z)$ and $\E{\mathcal Z_{G,B}^{(2)}(z,w)}/[a_n(z)a_n(w)]$ converge locally uniformly on ${D_0}$ and ${D_0\times D_0}$ to holomorphic functions, with the limit {of $a_n$} nowhere zero on ${D_0}$;
  \item For all real $s,u\in {D_0}$ with $P(s),P(u)\ge0$ and every {edge marginal} $Q$ with vertex marginal $\*\rho\in B$,
  \[
    \Phi_{s,u}(\*\rho,Q)\le-c\bigl(\|\*\rho-\*\pi(s)\otimes\*\pi(u)\|_2^2
                      +\|Q-P(s)\otimes P(u)\|_{\mathrm F}^2\bigr);
  \]
  \item For every fixed $m\ge1$, $(X_{1,n},\ldots,X_{m,n})\xrightarrow{d}(Z_1,\ldots,Z_m)$ with independent $Z_\ell\sim\mathrm{Poisson}(\lambda_\ell)$ {(\Cref{def:short-cycle-factors})}, and for every real $s\in {D_0}$ with $P(s)\ge0$ and fixed $j_1,\ldots,j_m\in\=Z_{\ge0}$,
  \begin{equation}\label{eq:common-cycle-moments}
  \begin{aligned}
    \frac{\E{\mathcal Z_{G,B}^{(2)}(s,s)}}{a_n(s)^2}
      &\longrightarrow
      \exp\!\left(\sum_{\ell\ge1}\lambda_\ell\delta_\ell(s)^2\right),\\
    \frac{\E{\mathcal Z_G(s)
                  \prod_{\ell=1}^m(X_{\ell,n})_{j_\ell}}}{a_n(s)}
      &\longrightarrow
      \prod_{\ell=1}^m
          \bigl[\lambda_\ell\bigl(1+\delta_\ell(s)\bigr)\bigr]^{j_\ell};
  \end{aligned}
  \end{equation}
  \item $\sum_{\ell\ge1}\lambda_\ell\sup_{z\in {D_0}}|\delta_\ell(z)|^2<\infty$, and $1+\delta_\ell(z)\ne0$ for all $\ell\ge3$ and $z\in {D_0}$.
  \end{enumerate}
\end{lemma}
}

\subsection{Exponential Concentration of the {Planted} Overlap}
\label{subsec:variational-overlap}

While \Cref{lem:local-profile-moments} controls replica pairs with overlap inside $B$, unique maximization of the Bethe functional gives exponential concentration under the planted distribution of \Cref{def:planted-overlap}. Throughout this subsection, $d\ge3$ is fixed and $t=\alpha\in(0,1/2)$ for hardcore or $t=\beta\in(0,1)$ for Potts. We use the notation $a_n(t)$, $\mathbb Q_n$, and $\ell_s$ introduced in \Cref{subsec:local-variational-geometric}.
\begin{lemma}[Exponential overlap bound]\label{lem:exponential-overlap}
Under \Cref{cond:unique-variational-maximizer} at $(d,t)$, for every $\eta>0$ there exists $c>0$ such that, for all sufficiently large $n$,
\[
 \E[\mathbb Q_n]{\mu_G^{\otimes2}
 \bigl(\|\*\rho({\sigma,\sigma'})-\*\pi\otimes\*\pi\|_\infty>\eta\bigr)}
 \le\e^{-cn}.
\]
\end{lemma}

\begin{proof}[Proof of~\Cref{lem:exponential-overlap}]
All probabilities and expectations in this proof average over the graph marginal of $\mathbb Q_n$, independent Gibbs replicas ${(\sigma,\sigma')}\sim\mu_G^{\otimes2}$ {(so $(G,\sigma)$ and $(G,\sigma')$ each have marginal law $\mathbb Q_n$)}, and the independent reveal set $U\sim\operatorname{Bernoulli}(s)^{\otimes V(G)}$. Under \Cref{cond:unique-variational-maximizer}, the upper bound in~\eqref{eq:planted-free-energy-bound} is $o(s)$ as $s\downarrow0$: by compactness of $\mathcal P_{\*\pi}$, continuity of the functional in~\eqref{eq:planted-free-energy-bound}, and unique maximization at $\delta_{\*\pi}$, any maximizer $\nu_s$ converges weakly to $\delta_{\*\pi}$, where $(M\*\pi)_a=1$, so
\[
  \lim_{s\downarrow0}\frac1s\limsup_{n\to\infty}\frac1n\E{\ell_s}=0.
\]
Fix $0<\eta\le1$, set $b\defeq s\eta^2/12$, and choose $s>0$ small enough that $\E{\ell_s}\le bn/8$ for all large $n$. By \Cref{lem:revealed-entropy-bound}, there exists $c_0>0$ such that
\[
  \Pr{\ell_s>bn/4}\le\e^{-c_0n}.
\]

\smallskip
Conditional on $(G,U)$, the revealed spin pairs ${(\sigma_U,\sigma'_U)}$ have density $\exp\{\ell_s(\sigma)+\ell_s({\sigma'})\}$ relative to the product measure $(\*\pi\otimes\*\pi)^{\otimes U}$, and its log-density $\ell_s(\sigma)+\ell_s({\sigma'})$ exceeds $bn/2$ with unconditional probability at most $2\e^{-c_0n}$. Define the revealed-overlap deviation event
\[
 H\defeq\left\{\max_{a,c\in\Omega}\left|\sum_{v\in U}
 \bigl(\one{\sigma_v=a,\,{\sigma'_v}=c}-\pi_a\pi_c\bigr)\right|\ge \frac{s\eta n}{2}\right\}.
\]

\smallskip
Conditional on $(G,\sigma,{\sigma'})$, the sum in $H$ for coordinate $(a,c)$ is a sum over $v\in V(G)$ of independent bounded signed random variables $\one{v\in U}\bigl(\one{\sigma_v=a,\,{\sigma'_v}=c}-\pi_a\pi_c\bigr)\in[-1,1]$ with total mean $sn(\rho_{ac}(\sigma,{\sigma'})-\pi_a\pi_c)$. Bernstein's inequality under $(G,\sigma,{\sigma'})$ and under $(\*\pi\otimes\*\pi)^{\otimes U}$ therefore gives
\[
  \Pr{H^c,\,\|\*\rho(\sigma,{\sigma'})-\*\pi\otimes\*\pi\|_\infty>\eta}\le2\e^{-bn},
  \qquad
  \E[U]{(\*\pi\otimes\*\pi)^{\otimes U}(H)}\le2|\Omega|^2\e^{-bn}.
\]
On $H\cap\{\ell_s(\sigma)+\ell_s({\sigma'})\le bn/2\}$, the conditional density relative to $(\*\pi\otimes\*\pi)^{\otimes U}$ is at most $\e^{bn/2}$. Decomposing $\{\|\*\rho(\sigma,{\sigma'})-\*\pi\otimes\*\pi\|_\infty>\eta\}$ into $H^c\cap\{\|\*\rho(\sigma,{\sigma'})-\*\pi\otimes\*\pi\|_\infty>\eta\}$, $H\cap\{\ell_s(\sigma)+\ell_s({\sigma'})\le bn/2\}$, and $\{\ell_s(\sigma)+\ell_s({\sigma'})>bn/2\}$ yields
\[
  \Pr{\|\*\rho(\sigma,{\sigma'})-\*\pi\otimes\*\pi\|_\infty>\eta}
  \le 2\e^{-bn}+2|\Omega|^2\e^{-bn/2}+2\e^{-c_0n},
\]
completing the proof.
\end{proof}

\subsection{Proof of \Cref{thm:variational-zero-free}}\label{subsec:variational-zero-free-proof}

Unless another law is specified, expectations and probabilities in this subsection are taken under $\mathrm{CM}=\mathrm{CM}_{n,d}$ for a fixed integer $d\ge3$.
We apply the graph truncation from \Cref{lem:graph-truncation} together with the local moment and cycle estimates of \Cref{lem:local-profile-moments}.

{
\begin{lemma}[Deterministic zero-freeness near trivial endpoints~{\cite{shearer85,scott2005repulsive,barvinok2017computing}}]\label{lem:endpoint-zero-free}
  For every simple graph $G$ of maximum degree at most $d$, the hardcore partition function $Z_{G,\mathrm{fug}}(\lambda)$ is zero-free in fugacity on $|\lambda|<(d-1)^{d-1}/d^d$, and the Potts partition function $Z_{G,q}(\beta)$ is zero-free on $|\beta-1|<0.35/d$ for every $q\ge2$.
\end{lemma}
}

\begin{proof}[Proof of~\Cref{thm:variational-zero-free}]
  {Since $\lambda(\alpha)=\alpha(1-\alpha)^{d-1}/(1-2\alpha)^d$ maps a neighborhood of $\alpha=0$ into $|\lambda|<(d-1)^{d-1}/d^d$, \Cref{lem:endpoint-zero-free} handles the trivial endpoints $\alpha=0$ and $\beta=1$.}

  \smallskip
  At a nontrivial parameter $t$ satisfying \Cref{cond:unique-variational-maximizer} and \Cref{cond:strict-KS}, let $Y_n=\mathcal Z_G\one{E_n}$ and the open disk $D_0$ be supplied by \Cref{lem:local-profile-moments} and \Cref{lem:graph-truncation}, put $A_n=\E{Y_n}$, {choose a smaller open complex disk $D_1$ centered at $t\in\=R$ with $\ol{D_1}\subset D_0$, and choose a nonempty open real interval $I\subseteq D_1\cap\=R$ on which $P(s)\ge0$. Since the locally uniform limit of $a_n$ is nowhere zero on $\ol{D_1}$, $A_n$ is nowhere zero on $D_1$ for all large $n$, verifying \textup{(B0)}. The moment and cycle estimates of \Cref{lem:local-profile-moments} and \Cref{lem:graph-truncation} verify \textup{(B1)} and \textup{(B3)} on $D_1$ together with the real-line hypotheses \textup{(B2)} and \textup{(B4)} on $I$. The mixed cycle moments transfer to $Y_n$ because fixed cycle-count products are polynomially bounded. Finally, setting $w=\ol z$ in~\eqref{eq:holomorphic-truncation} and using $|Y_n(z)|^2=Y_n(z)Y_n(\ol z)$ verifies the locally uniform second-moment bound \textup{(B5)} on the complex neighborhood $D_1$.}

  \smallskip
  The simplicity event is $X_{1,n}=X_{2,n}=0$, and conditioning on this event gives $G\sim\+G_{n,d}$ by \Cref{def:random-regular-graphs}. Under this conditional law, \Cref{thm:reinforced-ssc} makes $Y_n$ zero-free on a closed disk about $t$ with probability $1-o_n(1)$. Since $Y_n\equiv0$ on $E_n^c$, this forces $E_n$, on which $Z_G=C^nY_n$ by the Bethe canonical form (\Cref{def:bethe-canonical-form}).
\end{proof}

\subsection{{Proofs of the Main Zero-Free and Counting Results}}\label{subsec:main-zero-free-counting-proofs}
\begingroup
\begin{proof}[Proof of \Cref{thm:hardcore-zero}, \Cref{thm:potts-zero}, and \Cref{cor:ising-zero}]
Set $J=[0,\alpha_{\mathrm{RSB}}(d)-\delta]$ for the hardcore model, $J=[\beta_{\mathrm{RSB}}(d,q)+\delta,1]$ for positive-weight Potts, and $J=[0,1]$ for the coloring case $q\ge3$ and $d<d_{\mathrm{RSB}}(q)$, with $\delta$ as in the respective theorem.

For the hardcore and positive-weight Potts intervals, the phase characterizations and \Cref{lem:model-variational-conditions} give the two strict conditions at every nontrivial $t\in J$. In the coloring case,~\eqref{eq:coloring-real-degree-nonpositivity} provides a real $r>d$ with $\mathcal B_{r,0}\le0$, so \Cref{lem:coloring-variational-conditions} gives both strict conditions at $\beta=0$. Degree monotonicity also gives $\mathcal B_{d,0}\le0$; the Potts contraction argument in the proof of \Cref{lem:model-variational-conditions}, applied with $\beta'=0$, then gives $\mathcal B_{d,\beta}\le0$ for every $\beta\in(0,1)$, so $\beta_{\mathrm{RSB}}(d,q)=0$ by~\eqref{eq:variational-thresholds-def}. Thus \Cref{lem:model-variational-conditions} gives the two strict conditions throughout $\beta\in(0,1)$. The independent-spin endpoints $\alpha=0$ and $\beta=1$ are covered directly by \Cref{lem:endpoint-zero-free}.

By \Cref{thm:variational-zero-free}, each $t\in J$ has a fixed open zero-free neighborhood with probability $1-o_n(1)$. A finite collection of these neighborhoods covers the compact set $J$ and contains $D(J,\eta)$ for some $\eta>0$. A union bound proves the hardcore and Potts theorems, and $q=2$ with~\eqref{eq:ising-rsb-threshold} gives the Ising corollary.
\end{proof}
\endgroup

\begingroup
\begin{lemma}[Taylor polynomial interpolation~{\cite{barvinok2016combinatorics}}]\label{lem:barvinok-interpolation}
  Let $P$ be a degree-$N$ polynomial with $P(0)>0$ that is zero-free on $D([0,z_0],\eta)$ ($z_0\ge0$, $\eta>0$). For every $\varepsilon\in(0,1)$, $P(z_0)$ can be deterministically approximated within relative error $\varepsilon$ in time $\mathrm{poly}(m)$ from the degree-$m$ Taylor polynomial of $\log P$ at $0$, for some $m=O_{z_0,\eta}(\log((N+1)/\varepsilon))$.
\end{lemma}

\begin{lemma}[Connected-subgraph enumeration~\cite{patel2017deterministic}]\label{lem:patel-regts}
  On any $n$-vertex graph $G$ of maximum degree at most $d$, the exact rational degree-$m$ Taylor polynomials of $\log\bigl((1-2z)^{dn}Z_G(z)\bigr)$ and {$\log\bigl(q^{-n}Z_{G,q}(1-z)\bigr)$} at $0$ can be computed deterministically in time ${n^{O(1)}\e^{O_{d,q}(m)}}$.
\end{lemma}

\begin{proof}[Proof of \Cref{cor:hardcore-counting}, \Cref{cor:potts-counting}, and \Cref{cor:ising-counting}]
For the hardcore model, the target $\alpha$ lies in the interval $[0,\alpha_{\mathrm{RSB}}(d)-\delta]\subset[0,1/2)$. Shrink the zero-free neighborhood from \Cref{thm:hardcore-zero}, if necessary, so that it excludes $1/2$; then $(1-2z)^{dn}Z_G(z)$ is a zero-free polynomial there. For the Potts model, use $q^{-n}Z_{G,q}(1-z)$ at the target $1-\beta$.

Applying \Cref{lem:barvinok-interpolation} and \Cref{lem:patel-regts} and multiplying by $(1-2\alpha)^{-dn}$ or $q^n$, respectively, gives \Cref{cor:hardcore-counting} and \Cref{cor:potts-counting} from \Cref{thm:hardcore-zero} and \Cref{thm:potts-zero}. {The normalized polynomials have degrees at most $dn$ and $dn/2$, respectively. A fixed rational polynomial change of variable in the interpolation method maps a disk of radius greater than one into the zero-free neighborhood; its truncation error is $O(n\theta^m)$ for some $\theta<1$. Taking $m$ to be a sufficiently large multiple of $\log((n+1)/\varepsilon)$ makes this error $O(\varepsilon)$. The change of variable and its constants can be chosen uniformly over each parameter interval by scaling one fixed map to its target point.} Specializing to $q=2$ via~\eqref{eq:ising-rsb-threshold} gives \Cref{cor:ising-counting}.
\end{proof}
\endgroup

% END INCLUDED FILE: chapters/04-zero-freeness.tex
% BEGIN INCLUDED FILE: chapters/06-conclusions.tex
\section{Conclusions and Open Problems}\label{sec:conclusions}

We establish complex zero-free regions and, for fixed rational model parameters, deterministic polynomial-time approximation algorithms for the partition functions of the hardcore model, the antiferromagnetic Potts model, and proper $q$-colorings on random $d$-regular graphs throughout the interior of the replica-symmetric regime, with high probability over the graph. Our approach combines reinforced small subgraph conditioning on the configuration model with a variational second-moment truncation based on exponential planted overlap concentration, lifting cavity-method characterizations of replica symmetry to the complex plane even where the unrestricted second moment diverges exponentially.

Our work raises several open questions:
\begin{itemize}
  \item First, can complex zero-freeness and deterministic approximate counting be extended up to the replica symmetry breaking threshold for Erd\H{o}s--R\'enyi random graphs $G(n,d/n)$ and other diluted mean-field models? Unlike random regular graphs, $G(n,d/n)$ exhibits $O(\sqrt{n})$ fluctuations of $\log Z_G$ from degree variations as well as vertices of unbounded degree, both of which obstruct a direct application of the current method in this work.

  \item Second, what is the threshold for efficient approximate sampling on random regular graphs, and is there a computational separation between counting and sampling in diluted spin glasses where vertex pinning leaves the random regular graph ensemble, so our average-case counting guarantee does not immediately support the standard self-reduction~\cite{jerrum1986random}? While local algorithms cannot find independent sets of density $(1+\varepsilon)\log d/d$ for fixed $\varepsilon>0$ and sufficiently large $d$~\cite{gamarnik2017limits,rahman2017local}, this does not rule out non-local global algorithms; for instance, in the Lov\'asz Local Lemma {(LLL)} regime, recent non-local algorithms match the counting LLL threshold~\cite{liu2026counting} for sampling despite the failure of vertex-pinning self-reduction~\cite{liu2026counting,achlioptas2026sampling,jain2026general}.

  \item Finally, our deterministic counting algorithms run in time  $(n/\varepsilon)^{O_{d,q,\delta}(1)}$, where the exponent of $n$ depends on the model parameters and the gap to criticality. For the Sherrington--Kirkpatrick model, color coding recently improved the quasipolynomial Barvinok interpolation runtime~\cite{bencs2026zeros} to polynomial time with a degree-independent exponent on $n$~\cite{harbuzova2026color}; a separate polynomial-time algorithm uses the cavity method~\cite{schmidhuber2026polynomial}. Can a universal exponent on $n$ be achieved up to the RSB threshold on random regular graphs, or does requiring a parameter-independent exponent induce a distinct computational phase transition?
\end{itemize}

\section*{Concurrent Work and Acknowledgements}

\paragraph{Concurrent Work.} 
In concurrent and independent work, Ferenc Bencs, Brice Huang, {Daniel Z.\ Lee}, Kuikui Liu, and Guus Regts~\cite{bencs2026ising} established Fisher zero-free neighborhoods and deterministic polynomial-time approximate counting for the zero-field antiferromagnetic Ising model on random $d$-regular graphs throughout the interior of the replica-symmetric regime, up to the Kesten--Stigum threshold. Their results cover the same parameter range as \Cref{cor:ising-zero} and \Cref{cor:ising-counting}.

At a technical level, both approaches combine the second moment method with small subgraph conditioning, applied to a suitable truncation of the partition function. To preclude zeros, their proof constructs a random analytic reweighting and applies Jensen's formula, whereas our reinforced small subgraph conditioning framework uses tightness and real-line small subgraph conditioning to identify a nonvanishing limiting random holomorphic function.

Their Jensen argument additionally proves a stronger conclusion of an $o_n(1)$ expected zero count than our high-probability zero-free result for the zero-field antiferromagnetic Ising model. Our result additionally derives a Bethe variational condition from replica symmetry and uses this condition to obtain exponential
planted overlap concentration and construct the required graph
truncation, allowing for broader applications to the hardcore model and the antiferromagnetic Potts model (including proper $q$-colorings).

\paragraph{Acknowledgements.} We thank Ferenc Bencs, Brice Huang, {Daniel Z.\ Lee}, Kuikui Liu, and Guus Regts for early discussion and coordination of the concurrent work. We thank Sidhanth Mohanty and Amit Rajaraman for helpful discussions and pointing out related references.

\section*{AI Disclosure}
During the preparation of this manuscript, the authors used ChatGPT (5.6 Sol, 6 Sol, 6 Astra, 6.1 Sol, and 6.1 Terra) and Gemini 4 Argon as AI assistants. The work started in the first week of August 2026. Early in this work, the authors obtained an initial draft for the hardcore model extending beyond the tree reconstruction threshold $\alpha_{\mathrm{rec}}(d)\sim\frac{\log d}{d}$ using ChatGPT 5.6 Sol.

Through extensive iterative interactions with ChatGPT 5.6 Sol and subsequent versions (ChatGPT 6 Sol/Astra and 6.1 Sol), during which the authors contributed many key ideas and conceptual inputs, this approach was progressively strengthened. Specifically, the analysis first matched the second-moment regime by restricting to typical configurations (which controls real moments up to the second-moment threshold, such as $d\le(2q-2)\log q+\varepsilon_q$ for balanced $q$-colorings for some $\varepsilon_q=o_q(1)$), and ultimately reached the replica symmetry breaking threshold across both the hardcore and antiferromagnetic Potts models (including proper $q$-colorings) by instead truncating to typical graphs $Y_n(z)=\mathcal Z_G(z)\one{E_n}$ via exponential planted overlap concentration. AI tools also assisted in searching the literature and adapting published cavity-method bounds to our unified setting.

The authors devoted substantial effort to understanding, digesting, and reorganizing the manuscript to make the ideas and proofs clear and accessible. In particular, the authors independently wrote the Introduction and Technical Overview, unified the separate model-specific drafts into the modular architecture, wrote parts of the technical arguments, and rewrote the rest. Gemini 4 Argon was used to assist in digesting the technical proofs, streamlining arguments, and polishing the exposition. The authors {take} full responsibility for the correctness, originality, and integrity of the work, including its proofs and references.

% Bibliography generated from the current references.bib using BibTeX (alpha).
\newcommand{\etalchar}[1]{$^{#1}$}

% END INCLUDED FILE: chapters/06-conclusions.tex
\appendix
% BEGIN INCLUDED FILE: appendices/app-sec4-cavity.tex
\section{Cavity and Variational Proofs for the Replica-Symmetric Phase}\label{app:potts-fractional-proof}

In this appendix, we give the deferred measure, cavity, and variational proofs from \Cref{sec:rs-identification-local}.

\Cref{subsec:app-rs-measures} collects the auxiliary properties of replica-symmetric measures (\Cref{lem:rs-measure-properties}) and proves the covariance and overlap characterization (\Cref{lem:rs-overlap-characterization}).

In~\Cref{subsec:app-hc-cavity}, we establish the graph-law comparison (\Cref{lem:cavity-increment-comparison}) and the three auxiliary facts (\Cref{fact:cavity-marginal-estimates}, \Cref{fact:cavity-prob}, and \Cref{fact:fixed-point}) used in the proof of~\Cref{lem:local-hc-rs-cavity-lower}.

\Cref{subsec:app-potts-rs} proves the pointwise replica-symmetric free-energy lemma for the Potts model (\Cref{lem:potts-pointwise-rs-free-energy}) and \Cref{subsec:app-variational-bridges} connects the fractional-moment and bipartite factor-graph functionals to the canonical Bethe functional (\Cref{prop:nu-eta-fractional-correspondence} and \Cref{lem:local-potts-free-energy-criterion}).

\subsection{Properties of Replica-Symmetric Measures}\label{subsec:app-rs-measures}\label{app:rs-measures}

{
We record two consequences of pairwise decorrelation~\eqref{eq:measure-pairwise-decorrelation} used in the cavity proofs and prove them together with \Cref{lem:rs-overlap-characterization}.
\begin{lemma}[{\cite[Corollary~2.2]{coja2016belief}, \cite[Proposition~3.11, Lemmas~3.12 and~3.17, and Corollary~3.18]{coja2019spin}}]\label{lem:rs-measure-properties}
Let $\nu_n$ be random probability measures on $\Omega^n$ for a fixed finite spin space $\Omega$.
\begin{enumerate}[label=(\roman*),font=\normalfont,leftmargin=2.5em,labelsep=0.6em,topsep=0.25em,itemsep=0.2em,parsep=0pt]
  \item \emph{Multi-point factorization:} If $\nu_n$ satisfies~\eqref{eq:measure-pairwise-decorrelation}, then for any fixed $k\ge1$ and independent uniform vertices $v_1,\ldots,v_k\in [n]$, $\E{\DTV{\nu_{n,v_1,\ldots,v_k}}{\bigotimes_{i=1}^k\nu_{n,v_i}}}=o(1)$.
  \item \emph{Stability under bounded-density tilts:} If $\nu_n$ satisfies~\eqref{eq:measure-pairwise-decorrelation} and $\nu'_n$ is a random probability measure on $\Omega^n$ with $\dif\nu'_n/\dif\nu_n\le C$ for a constant $C<\infty$, then $n^{-1}\sum_{v\in[n]}\DTV{\nu'_{n,v}}{\nu_{n,v}}\xrightarrow{\mathrm p}0$.
\end{enumerate}
\end{lemma}

\begin{proof}[Proof of \Cref{lem:rs-overlap-characterization} and \Cref{lem:rs-measure-properties}]
Write $s=|\Omega|$ and, conditional on $\nu=\nu_n$, let $X=(\one{\sigma_v=a})_{v,a}$, $p=\E[\nu]{X}$, and $\Sigma=\operatorname{Cov}_\nu(X)$.
Each entry of $\Sigma$ is $\nu_{u,v}(a,b)-\nu_u(a)\nu_v(b)$. Comparing the entrywise $\ell_1$ and $\ell_2$ norms, using that the entries are bounded and that the terms with $u=v$ contribute only $O(n)$, proves
\[
  \eqref{eq:measure-pairwise-decorrelation}\quad\Longleftrightarrow\quad
  \frac1{n^2}\E{\|\Sigma\|_{\mathrm F}^2}=o(1).
\]
In particular, under~\eqref{eq:measure-pairwise-decorrelation}, $r_n\defeq\|\Sigma\|_{\mathrm{op}}/n$ satisfies $\E{r_n}=o(1)$.

For the overlap assertion, take two independent replicas $X,X'$ corresponding to $(\sigma,\sigma')\sim\nu^{\otimes2}$ conditional on $\nu$, and put
\[
  T\defeq\frac1n X^\trans X'=\operatorname{tr}\*\rho(\sigma,\sigma'),
  \qquad D\defeq\frac1n(X-p)^\trans(X'-p).
\]
The variable $D=T-\E{T\mid\nu,X}-\E{T\mid\nu,X'}+\E{T\mid\nu}$ is the orthogonal double-centered component of $T$. Hence, for every constant $c$,
\[
  \frac{\|\Sigma\|_{\mathrm F}^2}{n^2}
  =\E{D^2\mid\nu}\le\E{(T-c)^2\mid\nu}.
\]
If $\*\rho(\sigma,\sigma')\xrightarrow{\mathrm p}R$ for a deterministic matrix $R$, take $c=\operatorname{tr}R$. Since $T$ is bounded, averaging the last inequality over $\nu$ proves \Cref{lem:rs-overlap-characterization}.

To prove the two auxiliary properties, Cauchy--Schwarz gives, for any square-integrable $f$,
\[
  \|\operatorname{Cov}_\nu(f,X)\|_2^2
  \le\operatorname{Var}_\nu(f)\|\Sigma\|_{\mathrm{op}}.
\]
For a partition $A_1,\ldots,A_m$ of $\Omega^n$, apply this bound to $f=\one{A_j}$ and sum over $j$. With terms of zero probability omitted, this yields
\begin{equation}\label{eq:state-partition-marginals}
  \frac1n\sum_{v=1}^n\sum_{j=1}^m\nu(A_j)
  \DTV{\nu_v(\cdot\mid A_j)}{\nu_v}
  \le\frac12\sqrt{s m r_n}.
\end{equation}
For \Cref{lem:rs-measure-properties}(i), take the partition according to the spins at $v_1,\ldots,v_{k-1}$, which has at most $s^{k-1}$ classes, and average over the independent uniform vertex $v_k$. The triangle inequality for total variation, followed by induction on $k$, bounds the expected $k$-point factorization error by $\frac12\sum_{j=2}^k s^{j/2}\E{\sqrt{r_n}}=o(1)$.

For \Cref{lem:rs-measure-properties}(ii), take $f=\dif\nu'_n/\dif\nu_n$. Since $0\le f\le C$ and $\E[\nu]{f}=1$, we have $\operatorname{Var}_\nu(f)\le C-1$, and the same covariance bound gives
\[
  \frac1n\sum_{v=1}^n\DTV{\nu'_{n,v}}{\nu_{n,v}}
  \le\frac12\sqrt{s(C-1)r_n}.
\]
The right-hand side tends to zero in expectation, and hence also in probability.
\end{proof}
}

\subsection{Two-Root Cavity Coupling for the Hardcore Model}\label{subsec:app-hc-cavity}\label{app:hc-cavity}

Throughout this subsection, write $\mu_H={\mu_{H,\mathrm{fug},\lambda}}$ for the ordinary hardcore law on any graph $H$, including graphs obtained by deleting vertices. On $d$-regular graphs it agrees with $\mu^{\mathrm{can}}_{H,\alpha}$. We use the graph laws $\mathbb M_n$ and $\mathbb R_{n-2}$ from \Cref{lem:cavity-increment-comparison}.

\begin{proof}[Proof of~\Cref{lem:cavity-increment-comparison}]
Let $\mathfrak G_{m,d}$ denote the set of simple $d$-regular graphs on $[m]$.
{If $(G',\*e)$ is generated by~\eqref{eq:cavity-graph} from a simple $d$-regular graph $G$ with $\operatorname{dist}_G(a,a')\ge4$, then the $2d$ neighbors $b_1,\ldots,b_d$ and $c_1,\ldots,c_d$ are pairwise distinct and no cross-edge $b_ic_j$ ($i\ne j$) lies in $G$ (otherwise $\operatorname{dist}_G(a,a')\le3$), so $(G',\*e)\in A_{n-2}$ after relabeling. Conversely, each $(H,\*e)\in A_{n-2}$ is obtained from uniquely determined preimages once the labels of $a,a'\in[n]$ and the bijection $[n]\setminus\{a,a'\}\to[n-2]$ are chosen, giving exactly $n(n-1)(n-2)! = n!$ preimages $(G,a,a',B,C,\text{relabeling})$. Because a simple $d$-regular graph $G$ has $n(n-1)\bigl(1-O(n^{-1})\bigr)$ ordered pairs $(a,a')$ at distance at least~$4$ and $(d!)^2$ neighbor orderings $(B,C)$, each preimage has probability $\bigl(|\mathfrak G_{n,d}|\,n!\,(d!)^2\bigr)^{-1}\bigl(1+O(n^{-1})\bigr)$, so $\mathbb M_n(H,\*e)=\bigl(|A_{n-2}|\bigr)^{-1}\bigl(1+O(n^{-1})\bigr)$ uniformly on $A_{n-2}$. Under $\mathbb R_{n-2}$, $H$ is uniform on $\mathfrak G_{n-2,d}$ and $\*e$ is uniform on the $(dn-2d)^d$ $d$-tuples of directed edges of $H$; in any simple $d$-regular $H$, the fraction of such tuples with $2d$ distinct endpoints and no cross-edges $b_ic_j$ is $1-O(n^{-1})$, so $\mathbb R_{n-2}(A_{n-2})=1-O(n^{-1})$ and $\mathbb R_{n-2}(\cdot\mid A_{n-2})$ is uniform on $A_{n-2}$. Dividing the two point masses on $A_{n-2}$ gives~\eqref{eq:marked-cavity-comparison}.}
\end{proof}

{Next, we show that under replica symmetry of $\mu_G$, the two-root cavity switch perturbs the single-site occupation probabilities and the empirical BP cavity message distribution by $o(1)$ on average. Let
\[
  d_{\mathrm{BL}}(\mu,\nu)\defeq\sup\left\{\left|\int f\dif\mu-\int f\dif\nu\right|:\sup_{\mathcal I_\lambda}|f|+\sup_{x\ne y}\frac{|f(x)-f(y)|}{|x-y|}\le1\right\}
\]
denote the bounded-Lipschitz metric on $\mathcal P(\mathcal I_\lambda)$, where $\eta_H\in\mathcal P(\mathcal I_\lambda)$ is the empirical BP cavity message distribution of a regular graph $H$.}

\begin{fact}\label{fact:cavity-switch-stability}
Under replica symmetry, the coupled graphs satisfy
\[
  \E{\frac1n\sum_{u\in V(H_0)}|p_{u,G}-p_{u,G'}|}=o(1),
  \qquad
  \E{d_{\mathrm{BL}}(\eta_G,\eta_{G'})}=o(1).
\]
\end{fact}

\begin{proof}[Proof of~\Cref{fact:cavity-switch-stability}]
Let $\nu$ be the restriction of $\mu_G$ to $V(H_0)$, and for a vertex subset or tuple $S$, write $\sigma_S=0$ for the event $\{\sigma_u=0\ (\forall u\in S)\}$. Summing over
the two nonadjacent roots gives
\[
  \frac{\dif\nu}{\dif\mu_{H_0}}(\sigma)
  =\frac{(1+\lambda\one{\sigma_B=0})(1+\lambda\one{\sigma_C=0})}
         {\E[\mu_{H_0}]{(1+\lambda\one{\sigma_B=0})(1+\lambda\one{\sigma_C=0})}}.
\]
Thus $\mu_{H_0}$ has density at most $(1+\lambda)^2$ with respect
to $\nu$. Moreover, $\mu_{G'}$ is $\mu_{H_0}$ conditioned on the
absence of occupied matched pairs. That event contains $\{\sigma_B=0\}$,
whose probability is at least $(1+\lambda)^{-d}$ by successively
conditioning on the spins outside each vertex in $B$. Hence
$\mu_{G'}$ has density at most $(1+\lambda)^d$ with respect to
$\mu_{H_0}$.

{Restricting $\mu_G$ to the coordinates in $V(H_0)$ changes its averaged pairwise total-variation correlation by only $O(n^{-1})$, so the pairwise-decorrelation condition~\eqref{eq:measure-pairwise-decorrelation} passes to $\nu$. The preceding density bounds give $\dif\mu_{G'}/\dif\nu\le(1+\lambda)^{d+2}$, so \Cref{lem:rs-measure-properties}(ii) implies that the mean absolute difference of the single-site marginals of $\nu$ and $\mu_{G'}$ is $o(1)$ in probability and hence in expectation.} This proves the first assertion.

For every {directed} edge $u\to v$ in a simple graph $H$, the hardcore
constraint gives
\[
 {m_{u\to v,H}}=\mu_H(\sigma_u=1\mid\sigma_v=0)
 =\frac{p_{u,H}}{1-p_{v,H}}.
\]
{All single-site occupation probabilities lie in $[0,\lambda/(1+\lambda)]$, on which $(p,q)\mapsto p/(1-q)$ is $(1+\lambda)^2$-Lipschitz. The directed-edge sets of $G$ (total $dn$) and $G'$ (total $d(n-2)$) share all $d(n-2)-2d$ directed edges of $H_0$, so for any test function $f$ with $\sup_{\mathcal I_\lambda}|f|+\sup_{x\ne y}\frac{|f(x)-f(y)|}{|x-y|}\le1$,
\[
\begin{aligned}
  \Bigl|\int f\dif\eta_G-\int f\dif\eta_{G'}\Bigr|
  &\le \frac{4d}{dn}+\frac{(1+\lambda)^2}{d(n-2)}\sum_{u\to v\text{ in }H_0}\bigl(|p_{u,G}-p_{u,G'}|+|p_{v,G}-p_{v,G'}|\bigr)\\
  &\le \frac{4}{n}+\frac{2(1+\lambda)^2}{n-2}\sum_{u\in V(H_0)}|p_{u,G}-p_{u,G'}|.
\end{aligned}
\]
Taking the supremum over $f$ and then the expectation gives the second assertion from the first.}
\end{proof}

We first combine the graph-law comparison with replica symmetry to evaluate the logarithmic marginal terms and prove \Cref{fact:cavity-marginal-estimates}.

\begin{proof}[Proof of~\Cref{fact:cavity-marginal-estimates}]
{Conditioning on $\sigma_{V(G)\setminus\{a,a'\}}$ shows that $\mu_G(\sigma_a=0)\ge(1+\lambda)^{-1}$ and $\mu_G(\sigma_a=0,\sigma_{a'}=1)\ge\lambda(1+\lambda)^{-(d+2)}$, while conditioning successively on the spins outside each vertex of $C$ gives $\mu_{G'}(\sigma_C=0)\ge(1+\lambda)^{-d}$. Likewise $p_u\in[\lambda(1+\lambda)^{-(d+1)},\lambda/(1+\lambda)]$ for every vertex $u$, so all logarithms below are uniformly bounded and Lipschitz in the underlying probabilities.}
Let $(u,w)$ be a uniformly chosen ordered pair of distinct vertices of $G$. Conditioning on $\operatorname{dist}(u,w)\ge4$ changes its law by $O(n^{-1})$ in total variation. Hence, by replica symmetry,
\begin{align*}
  \E{\log\mu_G(\sigma_a=0,\sigma_{a'}=1)}
  &=\E[G,u,w]{\log\mu_G(\sigma_u=0,\sigma_w=1)}+o(1)\\
  &=\E[G,u,w]{\log(1-p_u)+\log p_w}+o(1)\\
  &=\E{\log(1-p_v)}+\E{\log p_v}+o(1).
\end{align*}

For the first identity, let $x$ be a uniformly chosen vertex of $H$. Under $\mathbb R_{n-2}$, the endpoints $c_1,\ldots,c_d$ are independent uniform vertices. Furthermore, by~{\Cref{lem:rs-measure-properties}(i)}, replica symmetry implies asymptotic factorization at any fixed number of uniformly sampled vertices.
{Thus~\Cref{lem:cavity-increment-comparison} (changing the measure from $\mathbb M_n$ to $\mathbb R_{n-2}$ at cost $O(n^{-1})$ since $\log\mu_H(\sigma_C=0)$ is bounded) and replica symmetry under $H\sim\+G_{n-2,d}$ give the first two equalities below. Rewriting the third line back under $\mathbb M_n$ via~\Cref{lem:cavity-increment-comparison} and then applying the first bound of~\Cref{fact:cavity-switch-stability} (together with the uniform Lipschitz continuity of $\log(1-p)$ on $[0,\lambda/(1+\lambda)]$) gives the last equality:}
\begin{align*}
  \E[\mathbb M_n]{\log\mu_{G'}(\sigma_C=0)}
  &=\E[\mathbb R_{n-2}]{\log\mu_H(\sigma_C=0)}+o(1)\\
  &=d\E[H\sim\+G_{n-2,d},x]{\log(1-p_{x,H})}+o(1)\\
  &=d\E[G\sim\+G_{n,d},v]{\log(1-p_v)}+o(1). \qedhere
\end{align*}
\end{proof}

We next identify the joint limiting law of the cavity message distribution and the single-site occupation probability by establishing asymptotic factorization of the neighbor-spin marginal in the graph obtained by deleting a typical root.

\begin{proof}[Proof of~\Cref{fact:cavity-prob}]
Let $\sigma \sim \mu_{G'}$, ${\mu_{\*e}}$ be the law of $(\sigma_{b_i},\sigma_{c_i})_{i \le d}$, and ${\mu_{\*e}^{\mathrm{prod}}}$ be the product distribution $\bigotimes_{i=1}^d {\mu_{G',b_i,c_i}}$.
{Under $\mathbb R_{n-2}$, the marked {directed} edges are sampled independently and uniformly. For an edge $bc$, the hardcore constraint gives $\one{\sigma_b=1,\sigma_c=0}=\sigma_b$, $\one{\sigma_b=0,\sigma_c=1}=\sigma_c$, $\one{\sigma_b=\sigma_c=0}=1-\sigma_b-\sigma_c$, and $\one{\sigma_b=\sigma_c=1}=0$. Expand the indicators of the finitely many joint configurations on the edges $b_ic_i$ into products involving at most one endpoint of each edge. The selected endpoints are independent uniform vertices, so {multi-point factorization (\Cref{lem:rs-measure-properties}(i))} yields the estimate below under $\mathbb R_{n-2}$. The marked comparison in \Cref{lem:cavity-increment-comparison} transfers it to $\mathbb M_n$:}
\[
  \E[\mathbb{M}_n]{\DTV{{\mu_{\*e}}}{{\mu_{\*e}^{\mathrm{prod}}}}} = o(1).
\]
Define $Q_i\defeq\mu_{G'-c_i}(\sigma_{b_i}=1)
=\mu_{G'}(\sigma_{b_i}=1\mid\sigma_{c_i}=0)${ $=m_{b_i\to c_i,G'}$}.
Conditioning on $\sigma_C=0$ removes the matching edges, and its probability
is uniformly bounded below {by $(1+\lambda)^{-d}$}. Thus the preceding factorization across the edges $b_ic_i$ gives
\[
  \E{\DTV{\mu_{H_0}(\sigma_B\in\cdot\mid\sigma_C=0)}
       {\bigotimes_{i=1}^d\operatorname{Bernoulli}(Q_i)}}=o(1).
\]
{We claim that the same product approximation holds without conditioning on $\sigma_C=0$ after adding $a'$ back:
\begin{equation}\label{eq:fact-hc-product}
  \E{\DTV{{\mu_{G-a}(\sigma_B\in\cdot)}}
       {\bigotimes_{i=1}^d\operatorname{Bernoulli}(Q_i)}}=o(1).
\end{equation}
Indeed, let $\Sigma_G$ be the covariance matrix of the occupation indicators under $\mu_G$; by~\Cref{lem:rs-overlap-characterization}, replica symmetry implies $\E{\|\Sigma_G\|_{\mathrm{op}}}/n=o(1)$. Fix $G,a$, and write $\nu=\mu_{G-a}$ for the restriction of $\mu_G(\cdot\mid\sigma_a=0)$ to the remaining vertices. Since $\mu_G(\sigma_a=0)\ge(1+\lambda)^{-1}$, the variational characterization of variance gives $\|\Sigma_\nu\|_{\mathrm{op}}\le(1+\lambda)\|\Sigma_G\|_{\mathrm{op}}$. For any ${g_B}\in[0,1]$ depending only on $\sigma_B$, the joint covariance matrix gives
\[
  \sum_{v\ne a}\operatorname{Cov}_\nu({g_B},\sigma_v)^2
  \le\operatorname{Var}_\nu({g_B})\|\Sigma_\nu\|_{\mathrm{op}}.
\]
Given $G,a$, the second root $a'$ is uniform outside $O_d(1)$ vertices; thus its covariance with ${g_B}$ is $o(1)$ on average, uniformly in ${g_B}$. Moreover, $\nu(\sigma_{a'}=1)\ge\lambda(1+\lambda)^{-(d+1)}$. Applying the covariance bound to each of the finitely many configurations on $B$ shows
\[
 \E{\DTV{{\nu(\sigma_B\in\cdot)}}{\nu(\sigma_B\in\cdot\mid\sigma_{a'}=1)}}=o(1).
\]
Finally, occupying $a'$ forces $C$ vacant, so $\nu(\sigma_B\in\cdot\mid\sigma_{a'}=1)=\mu_{H_0}(\sigma_B\in\cdot\mid\sigma_C=0)$. Together with the preceding conditional product estimate, this proves~\eqref{eq:fact-hc-product}.}

{With $\mathcal U(\*m)=\lambda\prod_i(1-m_i)/\bigl(1+\lambda\prod_i(1-m_i)\bigr)$ as in \Cref{fact:cavity-prob}, since conditional on $G$ every vertex has only $O_d(1)$ vertices within distance~$3$, the marginal law of the root $a$ is within $O_d(n^{-1})$ in total variation of the uniform vertex $v$; together with $p_a=\lambda\mu_{G-a}(\sigma_B=0)/\bigl(1+\lambda\mu_{G-a}(\sigma_B=0)\bigr)$, the approximation~\eqref{eq:fact-hc-product}, and the second bound of~\Cref{fact:cavity-switch-stability}, this gives the first equality below for every $f\in C(\mathcal P(\mathcal I_\lambda)\times[0,1])$. Changing the measure of $(G',\*e)$ from $\mathbb M_n$ to $\mathbb R_{n-2}$ via~\Cref{lem:cavity-increment-comparison} (under which $Q_i=m_{b_i\to c_i,H}$ are i.i.d.\ draws from $\eta_H$) gives the second equality, and changing back from $\mathbb R_{n-2}$ to $\mathbb M_n$ via~\Cref{lem:cavity-increment-comparison} followed by~\Cref{fact:cavity-switch-stability} gives the third:}
\begin{align*}
  \E{f(\eta_G,p_v)}
  &=\E{f\bigl(\eta_{G'},\mathcal U(Q_1,\ldots,Q_d)\bigr)}+o(1)\\
  &={\E[H\sim\+G_{n-2,d}]{\int f\bigl(\eta_H,\mathcal U(\*m)\bigr)
                   \,\eta_H^{\otimes d}(\dif\*m)}+o(1)}\\
  &={\E{\int f\bigl(\eta_G,\mathcal U(\*m)\bigr)
                   \,\eta_G^{\otimes d}(\dif\*m)}+o(1).}
\end{align*}
Along $n_k$, the last expression converges to
{$\int\!\int f(\eta,\mathcal U(\*m))\,\eta^{\otimes d}(\dif\*m)\,\Lambda(\dif\eta)$}, proving the claim.
\end{proof}

Finally, we apply the factorization established in the proof of~\Cref{fact:cavity-prob} and prove~\Cref{fact:fixed-point}.

\begin{proof}[Proof of~\Cref{fact:fixed-point}]
Let $A=\mu_{G-a}(\sigma_B=0)$ and $r=\mu_{G-a}(\sigma_{b_1}=1)$. Adding $a$ back gives
\[
  p_a=\frac{\lambda A}{1+\lambda A},\qquad
  p_{b_1}=\frac{r}{1+\lambda A},\qquad
  {m_{a\to b_1,G}}=\frac{p_a}{1-p_{b_1}}
                =\frac{\lambda A}{1+\lambda A-r}.
\]
By~\eqref{eq:fact-hc-product}, we may replace $A$ by $\prod_i(1-Q_i)$ and $r$ by $Q_1$ with $o(1)$ expected absolute error. All denominators are bounded away from zero; canceling $1-Q_1$ yields
\[
  \E{\bigl|{m_{a\to b_1,G}}-\mathcal T(Q_2,\ldots,Q_d)\bigr|}=o(1).
\]

For $h\in C(\mathcal I_\lambda)$, define the continuous functional
\[
  D_h(\nu)=\int h\dif\nu-\int h\dif(T\nu).
\]
{Conditional on $G$, the approximate uniformity of $a$ established above and the uniform neighbor ordering make the directed edge $a\to b_1$ uniform up to $O(n^{-1})$ total-variation error, which gives the first equality below. The second equality replaces $\eta_G$ by $\eta_{G'}$ via~\Cref{fact:cavity-switch-stability} and $m_{a\to b_1,G}$ by $\mathcal T(Q_2,\ldots,Q_d)$ via the preceding $L^1$ bound. The third equality changes the law of $(G',\*e)$ from $\mathbb M_n$ to $\mathbb R_{n-2}$ via~\Cref{lem:cavity-increment-comparison} (under which $Q_2,\ldots,Q_d$ are i.i.d.\ draws from $\eta_H$, so the conditional expectation of $h(\mathcal T(Q_2,\ldots,Q_d))$ is $\int h\dif(T\eta_H)$ by definition of $T$) and then returns from $\mathbb R_{n-2}$ to $\mathbb M_n$. The fourth equality replaces $\eta_{G'}$ back by $\eta_G$ via~\Cref{fact:cavity-switch-stability}:}
\begin{align*}
  \E{D_h(\eta_G)\int h\dif\eta_G}
  &=\E{D_h(\eta_G)h({m_{a\to b_1,G}})}+o(1)\\
  &=\E{D_h(\eta_{G'})h(\mathcal T(Q_2,\ldots,Q_d))}+o(1)\\
  &=\E{D_h(\eta_{G'})\int h\dif(T\eta_{G'})}+o(1)\\
  &=\E{D_h(\eta_G)\int h\dif(T\eta_G)}+o(1).
\end{align*}
Hence $\E{D_h(\eta_G)^2}=o(1)$, and $\Lambda_{n_k}\xrightarrow{d}\Lambda$ implies
\[
  \int D_h(\eta)^2\,\Lambda(\dif\eta)=0.
\]
Taking $h$ in a countable dense subset of $C(\mathcal I_\lambda)$ shows that $T\eta=\eta$ for $\Lambda$-almost every $\eta$.
\end{proof}

\subsection{Replica-Symmetric Free-Energy Proofs for the Potts Model}\label{subsec:app-potts-rs}\label{app:potts-rs}
We now prove the replica-symmetric free-energy formula for the Potts model.

\begin{proof}[Proof of~\Cref{lem:potts-pointwise-rs-free-energy}]
At $\beta=1$, both conclusions are immediate from $Z_{G,q}(1)=q^n$; assume $0<\beta<1$.{

\paragraph{Step 1: State-weighted directed-edge message convergence.}
Write $\mu_G=\mu_{G,q,\beta}$ and $\psi(s,t)=\beta^{\one{s=t}}$. For probability vectors $\*x,\*y\in\Delta_{[q]}$, put
\[
 R_{\*x,\*y}(s,t)=\frac{\psi(s,t)x_sy_t}{\sum_{i,j}\psi(i,j)x_i y_j}.
\]
The unique pair $(\*x,\*y)$ for which $R_{\*x,\*y}$ has two uniform marginals is $(\*U_q,\*U_q)$: indeed, uniform marginals imply $x_i(1-(1-\beta)y_i)=y_i(1-(1-\beta)x_i)$ for all $i$, so $x_i=y_i$ and $x_i(1-(1-\beta)x_i)$ is independent of $i$; since $1-(1-\beta)(x_i+x_j)>0$, all $x_i$ are equal.

By color symmetry, $\mu_{G,v}=\*U_q$ for every vertex. Let $\Sigma_G$ be the covariance matrix under $\mu_G$ of the indicators $\one{\sigma_v=a}$, for $v\in[n]$ and $a\in[q]$, and put $r_n(G)=\|\Sigma_G\|_{\mathrm{op}}/n$. Replica symmetry gives $\E{r_n(G)}=o(1)$, as in the proof of \Cref{lem:rs-measure-properties}. Choose $L_n\to\infty$ sufficiently slowly that $(L_n+1)\E{r_n(G)}\to0$. By the integer-degree Potts Bethe-state decomposition~\cite[Theorems~2.5 and~7.4 and eq.~(2.4)]{coja2019spin}, $[q]^n$ is partitioned into an exceptional class $S_0$ with $\E{\mu_G(S_0)}=o(1)$ and at most $L_n$ Bethe states $S_1,\ldots,S_\ell$ ($\ell\le L_n$), each equipped with conditional single-site marginals $\mu_{G,v}(\cdot\mid S_h)$ and conditional cavity messages $\mu_{G,u\to v}(\cdot\mid S_h)\in\Delta_{[q]}$ defined on the directed edges $u\to v$ of $G$, such that the conditional two-point marginals on edges $uv\in E(G)$ satisfy the state-weighted edge approximation
\[
  \E{\sum_{h=1}^\ell\mu_G(S_h)\frac1{|E(G)|}\sum_{uv\in E(G)}\DTV{\mu_{G,u,v}(\cdot\mid S_h)}{R_{\mu_{G,u\to v}(\cdot\mid S_h),\,\mu_{G,v\to u}(\cdot\mid S_h)}}}=o(1).
\]
Applying~\eqref{eq:state-partition-marginals} to the partition $\{S_0,S_1,\ldots,S_\ell\}$, using $\ell+1\le L_n+1$ and Jensen's inequality, gives
\[
  \E{\sum_{h=1}^\ell\mu_G(S_h)\frac1n\sum_{v\in[n]}\DTV{\mu_{G,v}(\cdot\mid S_h)}{\*U_q}}
 \le\frac12\sqrt{q(L_n+1)\E{r_n(G)}}=o(1).
\]
Since $G$ is $d$-regular, the endpoints of a uniform directed edge are uniform vertices, so the same bound holds when averaged over the endpoints of a uniform directed edge. By the triangle inequality, the two single-site marginals of $R_{\mu_{G,u\to v}(\cdot\mid S_h),\,\mu_{G,v\to u}(\cdot\mid S_h)}$ are within $o(1)$ of $\*U_q$ in the state-and-edge average. Because $(\*x,\*y)=(\*U_q,\*U_q)$ is the unique pair in the compact space $\Delta_{[q]}^2$ for which $R_{\*x,\*y}$ has two uniform marginals, compactness implies that the conditional cavity messages converge to $\*U_q$ in the state-weighted directed-edge average:
\begin{equation}\label{eq:potts-directed-edge-message-avg}
  \E{\sum_{h=1}^\ell\mu_G(S_h)\frac1{dn}\sum_{u\to v}\norm{\mu_{G,u\to v}(\cdot\mid S_h)-\*U_q}_1}=o(1).
\end{equation}

\paragraph{Step 2: Transfer to the cavity boundary.}
We now apply the Potts Bethe-state free-energy formula of~\cite[Theorems~2.6 and~7.5 and eqs.~(2.4) and~(2.10)--(2.12)]{coja2019spin} (for which condition \textup{\textbf{POS}} holds by~\cite[Lemma~7.3]{coja2019spin}). In that construction, given $G$ and $\omega>0$, one draws independent $X,Y\sim\mathrm{Poisson}(\omega)$ and, independently, uniformly random vertices $u_1,\ldots,u_X\in V(G)$ and uniformly random undirected edges $e_1=v_1w_1,\ldots,e_Y=v_Yw_Y\in E(G)$. Let $\mathcal C(G)$ be the cavity set consisting of all vertices of degree strictly less than $d$ in the subgraph obtained from $G$ by deleting $u_1,\ldots,u_X$ and $e_1,\ldots,e_Y$. Fix a truncation level $K_0\ge1$ and work on the event $\{1\le X+Y\le K_0\}$. Define the overlap event $\mathcal O_{X,Y}$ on which {two deleted objects with distinct draw indices}, among $U_i\defeq\{u_i\}$ ($1\le i\le X$) and $E_j\defeq\{v_j,w_j\}$ ($1\le j\le Y$), have graph distance at most~$2$ in $G$. Because $G$ is $d$-regular, the distance-$2$ neighborhood of any vertex or edge has size $O_d(1)$, so a union bound over the at most $\binom{K_0}{2}$ pairs of distinct deleted objects gives $\Pr{\mathcal O_{X,Y}\mid G,X,Y}=O_{d,K_0}(n^{-1})$ whenever $1\le X+Y\le K_0$. Outside $\mathcal O_{X,Y}$, the $dX+2Y$ vertices in $\bigcup_{i=1}^X N_G(u_i)\cup\bigcup_{j=1}^Y\{v_j,w_j\}$ are all distinct in $V(G)\setminus\{u_1,\ldots,u_X\}$ and each loses a single incident edge of $G$, so every $c\in\mathcal C(G)$ has degree $d-1$ with a unique missing $G$-neighbor $c'$, $|\mathcal C(G)|=dX+2Y$, and the $dX+2Y$ directed edges $(c\to c')_{c\in\mathcal C(G)}$ are precisely the $dX$ directed edges $v\to u_i$ ($1\le i\le X$, $v\in N_G(u_i)$) together with both orientations $v_j\to w_j$ and $w_j\to v_j$ of each deleted edge $e_j$ ($1\le j\le Y$). For any non-negative function $f$ on the $dn$ directed edges of the $d$-regular graph $G$, a uniform vertex $u_i\in V(G)$ and a uniform undirected edge $e_j=v_jw_j\in E(G)$ satisfy
\[
\begin{aligned}
  \E{\sum_{v\in N_G(u_i)}f(v\to u_i)\;\middle|\;G}&=\frac1n\sum_{u\to w}f(u\to w),\\
  \E{f(v_j\to w_j)+f(w_j\to v_j)\;\middle|\;G}&=\frac2{dn}\sum_{u\to w}f(u\to w).
\end{aligned}
\]
Summing over $1\le i\le X$ and $1\le j\le Y$ on $\{1\le X+Y\le K_0\}\cap\mathcal O_{X,Y}^c$ and applying~\eqref{eq:potts-directed-edge-message-avg} therefore transfers the directed-edge estimate to all cavity boundary pairs $(c,c')$:
\begin{equation}\label{eq:potts-cavity-boundary-message-avg}
\begin{aligned}
  &\E{\one{1\le X+Y\le K_0,\,\mathcal O_{X,Y}^c}\sum_{h=1}^\ell\mu_G(S_h)\sum_{c\in\mathcal C(G)}\norm{\mu_{G,c\to c'}(\cdot\mid S_h)-\*U_q}_1}\\
  &\qquad\le K_0\,\E{\sum_{h=1}^\ell\mu_G(S_h)\frac1n\sum_{u\to w}\norm{\mu_{G,u\to w}(\cdot\mid S_h)-\*U_q}_1}
  =o(1).
\end{aligned}
\end{equation}

\paragraph{Step 3: Reweighted Bethe functional and order of limits.}
In that free-energy formula, the exceptional class $S_0$ is omitted and the non-exceptional states $h\in\{1,\ldots,\ell\}$ are reweighted by $\*z_{G,h}/\*z_G$, where $\*z_G\defeq\sum_{h=1}^\ell\*z_{G,h}$ and
\[
\begin{aligned}
  \*z_{G,h}
  \defeq\mu_G(S_h)
    &\prod_{i=1}^X\left(\sum_{\sigma\in[q]}\prod_{v\in N_G(u_i)}\bigl(1-(1-\beta)\mu_{G,v\to u_i}(\sigma\mid S_h)\bigr)\right)^{-1}\\
    &\times\prod_{j=1}^Y\left(1-(1-\beta)\sum_{\sigma\in[q]}\mu_{G,v_j\to w_j}(\sigma\mid S_h)\mu_{G,w_j\to v_j}(\sigma\mid S_h)\right)^{-1}.
\end{aligned}
\]
For $0<\beta<1$ and any messages in $\Delta_{[q]}$, each vertex factor lies in $[q\beta^d,q]$ and each edge factor lies in $[\beta,1]$. Thus, on $\{X+Y\le K_0\}$, every $h\in\{1,\ldots,\ell\}$ satisfies $q^{-K_0}\mu_G(S_h)\le\*z_{G,h}\le\beta^{-dK_0}\mu_G(S_h)$. On the high-probability event
\[
  \mathcal M_G\defeq\Bigl\{\sum_{h=1}^\ell\mu_G(S_h)=1-\mu_G(S_0)\ge1/2\Bigr\}\qquad\bigl(\Pr{\mathcal M_G^c}\le2\E{\mu_G(S_0)}=o(1)\bigr),
\]
summing the lower bound gives $\*z_G\ge\frac12 q^{-K_0}$, and hence $\*z_{G,h}/\*z_G\le2(q\beta^{-d})^{K_0}\mu_G(S_h)$ for all $1\le h\le\ell$. Together with~\eqref{eq:potts-cavity-boundary-message-avg}, this implies that on $\mathcal E_{K_0}\defeq\{1\le X+Y\le K_0\}\cap\mathcal O_{X,Y}^c\cap\mathcal M_G$,
\[
  \E{\one{\mathcal E_{K_0}}\sum_{h=1}^\ell\frac{\*z_{G,h}}{\*z_G}\sum_{c\in\mathcal C(G)}\norm{\mu_{G,c\to c'}(\cdot\mid S_h)-\*U_q}_1}=o(1).
\]
{Call the construction admissible if $\mathcal C(G)\ne\varnothing$, every $c\in\mathcal C(G)$ has a unique missing $G$-neighbor $c'$, and $\*z_G>0$. On the admissible event, sample $c_1,c_2,\ldots$ uniformly and independently from $\mathcal C(G)$ and use the corresponding missing neighbors $c'_i$. On the complement, define the integrand below to be $0$. Any uniformly bounded completion on this event gives the same outer $\omega\to\infty$ limit: for each fixed $\omega$, the probability of inadmissibility is at most $\e^{-2\omega}+o_n(1)$, as follows from the overlap and $\mathcal M_G$ bounds above by first truncating $X+Y$ and then removing the truncation.}

The free-energy formula of~\cite[Theorems~2.6 and~7.5]{coja2019spin} (after restoring our $q^n$ prior normalization) states that $\lim_{n\to\infty}n^{-1}\E{\log Z_{G,q}(\beta)}=\liminf_{\omega\to\infty}\liminf_{n\to\infty}\E{\mathcal B_{\mathrm{CP}}(G)}$ (and in fact the outer limit exists), where $\mathcal B_{\mathrm{CP}}(G)$ is the conditional expectation given $G$ of {the following expression on the admissible event, with the fallback just specified otherwise:}
\[
\begin{aligned}
  &\log\sum_{h=1}^\ell\frac{\*z_{G,h}}{\*z_G}\sum_{\sigma\in[q]}\prod_{i=1}^d\bigl(1-(1-\beta)\mu_{G,c_i\to c'_i}(\sigma\mid S_h)\bigr)\\
  &\qquad-\frac d2\log\sum_{h=1}^\ell\frac{\*z_{G,h}}{\*z_G}\left(1-(1-\beta)\sum_{\sigma\in[q]}\mu_{G,c_1\to c'_1}(\sigma\mid S_h)\mu_{G,c_2\to c'_2}(\sigma\mid S_h)\right).
\end{aligned}
\]
Since the vertex and edge expressions inside the logarithms lie in $[q\beta^d,q]$ and $[\beta,1]$ and are Lipschitz on $\Delta_{[q]}^d$ and $\Delta_{[q]}^2$, the integrand is uniformly bounded by a constant $M_{d,q,\beta}<\infty$ and converges in $L^1$ on $\mathcal E_{K_0}$ as $n\to\infty$ (for fixed $\omega,K_0$) to $\log(qs^d)-\frac d2\log s=\phi_{d,q}(\beta)$ with $s=(q-1+\beta)/q$. Because $\Pr{1\le X+Y\le K_0,\,\mathcal O_{X,Y}}=O_{d,K_0}(n^{-1})=o(1)$ and $\Pr{\mathcal M_G^c}=o(1)$, we have
\[
  \Pr{\mathcal E_{K_0}^c}\le\Pr{X+Y=0}+\Pr{X+Y>K_0}+o(1)=\e^{-2\omega}+\Pr{\mathrm{Poisson}(2\omega)>K_0}+o(1),
\]
and therefore, for every fixed $\omega>0$ and $K_0\ge1$,
\[
  \limsup_{n\to\infty}\bigl|\E{\mathcal B_{\mathrm{CP}}(G)}-\phi_{d,q}(\beta)\bigr|
  \le \bigl(M_{d,q,\beta}+|\phi_{d,q}(\beta)|\bigr)\bigl(\e^{-2\omega}+\Pr{\mathrm{Poisson}(2\omega)>K_0}\bigr);
\]
letting $K_0\to\infty$ at fixed $\omega$ and then $\omega\to\infty$ yields}
\[
  \lim_{n\to\infty}\frac1n\E{\log Z_{G,q}(\beta)}
  =\log(qs^d)-\frac d2\log s
  =\phi_{d,q}(\beta).
\]
{Finally, concentration~\eqref{eq:regular-logz-concentration}, after conditioning on simplicity as explained there, gives $n^{-1}(\log Z_{G,q}(\beta)-\E{\log Z_{G,q}(\beta)})\xrightarrow{\mathrm p}0$. Combining this with the expectation limit yields $n^{-1}\log Z_{G,q}(\beta)\xrightarrow{\mathrm p}\phi_{d,q}(\beta)$.}
\end{proof}

\subsection{Variational Bridges for the Hardcore and Potts Bethe Functionals}\label{subsec:app-variational-bridges}\label{app:variational-bridges}
{
In this final subsection, we prove \Cref{prop:nu-eta-fractional-correspondence}, which relates the hardcore fractional-moment functional $\mathcal F_m$ to the canonical Bethe functional $\mathcal B_{d,\alpha}$. We then establish the Potts free-energy criterion and use it to prove \Cref{lem:strict-free-energy-gaps}(ii).

\begin{proof}[Proof of~\Cref{prop:nu-eta-fractional-correspondence}]
For independent $Q_1,\ldots,Q_d\sim\eta_m$ written as $Q_i=u_i/a(u_i)$ with $u_i\sim\nu_m$, substituting into $V^{\mathrm{BP}}=1+\lambda\prod_{i=1}^d(1-Q_i)$ and $E^{\mathrm{BP}}=1-Q_1Q_2$ yields
\[
  V^{\mathrm{BP}}=\frac{\prod_{i=1}^d a(u_i)+\theta\prod_{i=1}^d(1-u_i)}
          {\prod_{i=1}^d a(u_i)},\qquad
  E^{\mathrm{BP}}=\frac{a(u_1)a(u_2)-u_1u_2}{a(u_1)a(u_2)}.
\]
For $\*\mu_i=(1-u_i,u_i)^\trans$, the normalized hardcore kernel satisfies $(M_\alpha\*\mu_i)_0=a(u_i)/(1-\alpha)$ and $(M_\alpha\*\mu_i)_1=(1-u_i)/(1-\alpha)$. Substitution into~\eqref{eq:marked-functional} gives
\[
  V_{d,\alpha}
  =\frac{\prod_{i=1}^d a(u_i)+\theta\prod_{i=1}^d(1-u_i)}
         {(1-\alpha)^{d-1}},\qquad
  E_\alpha=\frac{a(u_1)a(u_2)-u_1u_2}{1-2\alpha}.
\]
Consequently, the two pairs of factors satisfy the pointwise identities
\[
  V^{\mathrm{BP}}=(1-\alpha)^{d-1}
       \frac{V_{d,\alpha}}{\prod_{i=1}^d a(u_i)},
  \qquad
  E^{\mathrm{BP}}=(1-2\alpha)\frac{E_\alpha}{a(u_1)a(u_2)}.
\]
Taking $m$th moments under $\eta_m$ cancels the denominators against the tilt densities, giving
\[
  \E[\eta_m]{(V^{\mathrm{BP}})^m}=\frac{(1-\alpha)^{m(d-1)}}{\bigl(\E[\nu]{a(u)^m}\bigr)^d} \cdot \E[\nu]{V_{d,\alpha}^m} \qquad \text{and} \qquad \E[\eta_m]{(E^{\mathrm{BP}})^m}=\frac{(1-2\alpha)^m}{\bigl(\E[\nu]{a(u)^m}\bigr)^2} \cdot \E[\nu]{E_\alpha^m},
\]
where each expectation uses independent inputs with the indicated law. Substitution into $\mathcal F_m(\eta_m)=\frac{1}{m} \log\E[\eta_m]{(V^{\mathrm{BP}})^m}-\frac{d}{2m}\log\E[\eta_m]{(E^{\mathrm{BP}})^m}$ cancels the normalizing factors against $\phi_d(\alpha)=(d-1)\log(1-\alpha)-\frac d2\log(1-2\alpha)$ and proves~\eqref{eq:Phi-normalized-fractional-correspondence};~\eqref{eq:fractional-moment-derivative} then identifies the left derivative at $m=1$ with $\mathcal B_{d,\alpha}(\nu)$.
\end{proof}

For the Potts model at positive interaction parameter, we use the following strict-gap consequence of the general free-energy criterion in~\cite[Theorem~17.2 and Proposition~12.2]{coja2020inference}.
\begin{lemma}[Potts free-energy criterion]\label{lem:local-potts-free-energy-criterion}
Fix integers $d\ge3$ and $q\ge2$, and let $\beta\in(0,1)$. If $\mathcal B_{d,\beta}(\nu)>0$ for some $\nu\in\mathcal P_{\*U_q}$, then
\[
 \limsup_{n\to\infty}\frac1n
        \E[G\sim\+G_{n,d}]{\log Z_{G,q}(\beta)}
        <\phi_{d,q}(\beta).
\]
\end{lemma}

To prove \Cref{lem:local-potts-free-energy-criterion}, we specialize the factor-graph free-energy bounds of~\cite{coja2020inference}. The antiferromagnetic Potts interaction $\psi_\beta(a,b)\defeq\beta^{\one{a=b}}$ satisfies the hypotheses of that bound by~\cite[Lemma~16.1 and Section~16.3]{coja2020inference}.

Fix an integer $d\ge3$ and let $\+G_{\mathrm{bip}}(n,d)$ denote the bipartite factor-graph model of~\cite[Section~2.1]{coja2020inference} with $n$ variable nodes of degree $d$ and $m_{\mathrm{fac}}\sim\mathrm{Poisson}(dn/2)$ binary factor nodes, conditioned on $dn=2m_{\mathrm{fac}}$. The factor graph is uniform conditional on these degrees. We take $n$ along the admissible values for which $dn$ is even. With $\ell(u)\defeq u\log u$ and $\zeta_\beta\defeq(q-1+\beta)/q$, the functional of~\cite[Section~2.3]{coja2020inference} specialized to $\psi_\beta$ and $\nu\in\mathcal P_{\*U_q}$ is
\begin{equation}\label{eq:ckpz-bethe-functional}
  \mathcal B_{\mathrm{bip}}(\nu)
  \defeq\frac1q\E{\zeta_\beta^{-d}\ell\bigl(q\zeta_\beta^d V_{d,\beta}\bigr)}
        -\frac d{2\zeta_\beta}\E{\ell(\zeta_\beta E_\beta)},
\end{equation}
where $V_{d,\beta}$ and $E_\beta$ are the Potts cavity factors from~\eqref{eq:marked-functional} evaluated at independent draws from $\nu$.

\begin{lemma}[Bipartite factor-graph Potts free energy~{\cite[Lemma~16.1 and Theorem~17.2]{coja2020inference}}]\label{lem:ckpz-potts-free-energy}
Fix integers $d\ge3$ and $q\ge2$, and let $\beta\in(0,1)$. Along the admissible values of $n$ defined above,
\[
  \lim_{n\to\infty}\frac1n\E[G\sim\+G_{\mathrm{bip}}(n,d)]{\log Z_{G,q}(\beta)}=\phi_{d,q}(\beta)
  \quad\Longleftrightarrow\quad
  \max_{\nu\in\mathcal P_{\*U_q}}\mathcal B_{\mathrm{bip}}(\nu)=\phi_{d,q}(\beta),
\]
and if $\max_{\nu\in\mathcal P_{\*U_q}}\mathcal B_{\mathrm{bip}}(\nu)>\phi_{d,q}(\beta)$, then $\limsup_{n\to\infty}n^{-1}\E[G\sim\+G_{\mathrm{bip}}(n,d)]{\log Z_{G,q}(\beta)}<\phi_{d,q}(\beta)$.
\end{lemma}

\begin{proof}[Proof of \Cref{lem:local-potts-free-energy-criterion}]
For independent messages with law $\nu\in\mathcal P_{\*U_q}$, using $\ell(u)=u\log u$ and $\E{V_{d,\beta}}=\E{E_\beta}=1$ in~\eqref{eq:ckpz-bethe-functional} gives
\begin{align*}
 \mathcal B_{\mathrm{bip}}(\nu)
 &=\frac1q\E{\zeta_\beta^{-d}\ell(q\zeta_\beta^d V_{d,\beta})}
       -\frac d{2\zeta_\beta}\E{\ell(\zeta_\beta E_\beta)}\\
 &=\log q+\frac d2\log\zeta_\beta+\mathcal B_{d,\beta}(\nu)
       =\phi_{d,q}(\beta)+\mathcal B_{d,\beta}(\nu).
\end{align*}
Thus a positive witness for $\mathcal B_{d,\beta}$ gives $\max_\nu\mathcal B_{\mathrm{bip}}(\nu)>\phi_{d,q}(\beta)$, so the strict-gap alternative in \Cref{lem:ckpz-potts-free-energy} applies.

It remains to transfer the free energy from $\+G_{\mathrm{bip}}(n,d)$ to $\+G_{n,d}$. Write $\mathbb P$ for $\mathrm{CM}_{n,d}$ and $F=\log Z_{G,q}(\beta)$. The pairing-exposure martingale gives $\operatorname{Var}_{\mathbb P}(F)=O(n)$, since one edge edit changes $F$ by at most $|\log\beta|$; see~\eqref{eq:regular-logz-concentration}. The conditioning in the factor-graph model fixes $m_{\mathrm{fac}}=dn/2$, and its projection is a uniform $d$-regular pairing conditioned on having no loops. Both the loop-free and simple events have probability bounded away from zero under $\mathbb P$ by \Cref{lem:simplicity-prob}. Cauchy--Schwarz and the variance bound therefore show that the expectations of $F$ under both the projected factor-graph law and $\+G_{n,d}$ differ from $\E[\mathbb P]{F}$ by $O(\sqrt n)$. This transfers the strict free-energy gap and proves \Cref{lem:local-potts-free-energy-criterion}.
\end{proof}

\begin{proof}[Proof of \Cref{lem:strict-free-energy-gaps}(ii)]
Fix $\beta\in(0,1)$ and a witness $\nu\in\mathcal P_{\*U_q}$ with $\mathcal B_{d,\beta}(\nu)>0$. By \Cref{lem:local-potts-free-energy-criterion}, we may choose $\varepsilon>0$ such that
\[
  \limsup_{n\to\infty}\frac1n\E[G\sim\+G_{n,d}]{\log Z_{G,q}(\beta)}
  \le\phi_{d,q}(\beta)-2\varepsilon.
\]
Concentration~\eqref{eq:regular-logz-concentration} then gives
\[
  \Pr[G\sim\+G_{n,d}]{\frac1n\log Z_{G,q}(\beta)\le\phi_{d,q}(\beta)-\varepsilon}\longrightarrow1,
\]
proving both claimed gaps.
\end{proof}
}

% END INCLUDED FILE: appendices/app-sec4-cavity.tex
% BEGIN INCLUDED FILE: appendices/d-hard-coloring-entropy.tex
\section{Hard-Coloring Entropy under Geometric Replica Symmetry}\label{app:hard-coloring-entropy}
\begingroup

In this appendix, we prove \Cref{lem:coloring-pointwise-rs-free-energy}, the implication from geometric to thermodynamic replica symmetry for proper $q$-colorings in \Cref{subsec:local-coloring-phase}. The argument compares the quenched proper-coloring law $\mathbb P_{\mathrm{col}}$ and a pairing-level lift $\widehat{\mathbb Q}_n$ of the configuration-model planted law $\mathbb Q_n$ (\Cref{def:planted-overlap} at $\beta=0$) in three steps:
\begin{enumerate}[label=(\arabic*),font=\normalfont,leftmargin=2.2em,labelsep=0.5em,topsep=0.2em,itemsep=0.15em,parsep=0pt]
  \item \Cref{lem:global-entropy-coercivity} establishes an entropy-coercivity bound on a product of symmetric groups $\prod_{a=1}^b\mathfrak S_{N_a}$ using the hypercontractive inequality for global functions on $\mathfrak S_N$~\cite{keevash2023sharp,keller2025sharp};
  \item \Cref{lem:hard-coloring-edge-type} shows via color-preserving half-edge switches that under replica symmetry (\Cref{def:replica-symmetry}), the vertex- and edge-color counts under $\mathbb P_{\mathrm{col}}$ concentrate around the uniform profile; and
  \item The proof of \Cref{lem:coloring-pointwise-rs-free-energy} lifts the conditional law given a typical color-count profile to $\prod_{a=1}^q\mathfrak S_{dn_a}$ and applies \Cref{lem:global-entropy-coercivity} to obtain $\DKL(\mathbb P_{\mathrm{col}}\Vert\widehat{\mathbb Q}_n)=o(n)$, identifying the conditional quenched free energy with $\phi_{d,q}(0)$.
\end{enumerate}

\paragraph{Global entropy coercivity on products of symmetric groups.}
For an integer $N\ge1$, let $\mathfrak S_N$ denote the symmetric group of permutations of $[N]$. Fix an integer $b\ge1$ and let $\mathbb U_n$ be the uniform probability measure on the product group $\prod_{a=1}^b\mathfrak S_{N_a}$. A \emph{consistent restriction} $R$ on $\prod_{a=1}^b\mathfrak S_{N_a}$ of size $|R|=k$ is an event of the form $\{\omega=(\omega_1,\ldots,\omega_b):\omega_a(i)=j\text{ for all }(a,i,j)\in S_R\}$, where $S_R$ is a set of $k$ triples assigning distinct inputs $i\in[N_a]$ to distinct outputs $j\in[N_a]$ within each factor $a\in[b]$ {(including the trivial restriction with $S_R=\varnothing$ and $R=\prod_{a=1}^b\mathfrak S_{N_a}$ of size $0$)}. For a function $g:\prod_{a=1}^b\mathfrak S_{N_a}\to\=R$, we write $\|g_R\|_2\defeq\E[\omega\sim\mathbb U_n]{g(\omega)^2\mid\omega\in R}^{1/2}$ for the $L^2$ norm of $g$ under the conditional uniform measure given $\omega\in R$, and define the \emph{transposition Dirichlet form} by
\[
 \mathcal E_n(g,g)\defeq
 \frac12\sum_{a=1}^b
 \E[\omega\sim\mathbb U_n,\,\tau_a]{\bigl(g(\omega\tau_a)-g(\omega)\bigr)^2},
\]
where $\tau_a$ is a uniformly random transposition in the $a$th factor $\mathfrak S_{N_a}$ (acting on the $a$th coordinate of $\omega$). For a probability density $f\ge0$ with $\E[\mathbb U_n]{f}=1$, its relative entropy with respect to $\mathbb U_n$ is $\operatorname{Ent}_{\mathbb U_n}(f)\defeq\E[\mathbb U_n]{f\log f}=\DKL(f\mathbb U_n\Vert\mathbb U_n)$.

{In the following lemma and its proof, $\|\cdot\|_p$ denotes the $L^p$ norm under the uniform probability measure on the relevant finite domain; the conditional norm $\|g_R\|_2$ is as defined above.}

\begin{lemma}[Entropy coercivity for global permutation densities]\label{lem:global-entropy-coercivity}
Fix an integer $b\ge1$ and constants $c_0,C_0,L>0$, $A\ge1$, and $K\ge0$. Let $\mathbb U_n$ be the uniform measure on $\prod_{a=1}^b\mathfrak S_{N_a}$ with $c_0n\le N_a\le C_0n$ for all $a\in[b]$. Suppose $f:\prod_{a=1}^b\mathfrak S_{N_a}\to\=R_{\ge0}$ and $g\defeq\sqrt f$ satisfy
\[
 \E[\mathbb U_n]{f}=1,\qquad 0\le f\le \exp(Ln),\qquad
 \|g_R\|_2\le n^{K/2}A^{|R|}
 \quad\text{for every consistent restriction }R.
\]
Then $\mathcal E_n(g,g)=o(1)$ implies $\operatorname{Ent}_{\mathbb U_n}(f)=o(n)$, uniformly in the sense that for every $\delta>0$ there exist $\eta>0$ and $n_0$ (depending only on $b,c_0,C_0,L,A,K,\delta$) such that $\operatorname{Ent}_{\mathbb U_n}(f)\ge\delta n$ and $n\ge n_0$ imply $\mathcal E_n(g,g)\ge\eta$.
\end{lemma}

\begin{proof}
\emph{Step 1: Product hypercontractive operator.}
For each factor $\mathfrak S_{N_a}$, the orthogonal decomposition of $L^2(\mathfrak S_{N_a})$ under the regular representation into isotypic subspaces $V_{\lambda^{(a)}}$ (the sums of copies of each irreducible representation) is indexed by integer partitions $\lambda^{(a)}=(\lambda^{(a)}_1\ge\lambda^{(a)}_2\ge\cdots)\vdash N_a$ ($\sum_i\lambda^{(a)}_i=N_a$). For $s\ge0$, let
\[
 \mathcal V_{\le s}\defeq\bigoplus_{\sum_{a=1}^b(N_a-\lambda^{(a)}_1)\le s}\;\bigotimes_{a=1}^b V_{\lambda^{(a)}}\subseteq L^2(\mathbb U_n)
\]
denote the direct sum of product isotypic subspaces of total degree $\sum_{a=1}^b(N_a-\lambda^{(a)}_1)\le s$, and write $u^{\le s}$ for the orthogonal projection of $u\in L^2(\mathbb U_n)$ onto $\mathcal V_{\le s}$. We first show, extending~\cite[Theorem~1.10]{keevash2023sharp} to the product group, that there exist constants $\rho\in(0,\log(4)/(128A)]$ and $b_0=\rho/72\in(0,1)$ (depending only on $A$) and a positive self-adjoint operator $T$ on $L^2(\mathbb U_n)$ preserving each product isotypic subspace such that, for every function $0\le h\le g$,
\begin{equation}\label{eq:hard-global-noise-bounds}
 \|Th\|_4^4\le n^K\|h\|_2^2,
 \qquad
 \langle u,Tu\rangle\ge b_0^s\|u\|_2^2
 \quad\text{for all }u\in\mathcal V_{\le s}\text{ and }s\le10^{-5}\min_{a\in[b]} N_a.
\end{equation}
Indeed, let $C_a:L^2(\mathfrak S_{N_a})\to L^2([N_a]^{N_a})$ be the greedy coupling Markov operator of~\cite[Definition~3.1]{keevash2023sharp}, and put $C\defeq\bigotimes_{a=1}^b C_a$. By~\cite[Lemma~3.6]{keevash2023sharp} applied independently in each factor, conditioning on $k$ coordinates in $\prod_{a=1}^b[N_a]^{N_a}$ induces a mixture of consistent permutation restrictions of total size at most $k$; since $0\le h\le g$ implies $\|h_R\|_2\le\|g_R\|_2\le n^{K/2}A^{|R|}$, Jensen's inequality shows that every $k$-coordinate restriction of $Ch$ has $L^2$ norm at most $n^{K/2}A^k$. Let $T_{\rho,a}$ be the standard Beckner--Bonami noise operator on $[N_a]^{N_a}$ (which independently keeps each coordinate with probability $\rho$ and resamples it uniformly from $[N_a]$ with probability $1-\rho$), and put $T_{\mathrm{prod}}\defeq\bigotimes_{a=1}^b T_{\rho,a}$. Embedding each coordinate into the common finite alphabet $\prod_{a=1}^b[N_a]$ (so that each coordinate depends only on its designated factor and restrictions of unused factors have no effect) and applying the global hypercontractive inequality on product spaces~\cite[Theorem~4.10]{keller2025sharp} with exponent $4$ and $\rho=\log(4)/(128A)$ gives $\|T_{\mathrm{prod}}Ch\|_4^4\le n^K\|Ch\|_2^2$. Since the adjoint $C^*$ is a Markov operator (hence an $L^4$ contraction) and $C$ is an $L^2$ contraction, the operator $C^*T_{\mathrm{prod}}C$ satisfies the first bound in~\eqref{eq:hard-global-noise-bounds}.

For $\omega=(\omega_1,\ldots,\omega_b)\in\prod_{a=1}^b\mathfrak S_{N_a}$, let $R_\omega$ denote right translation by $\omega$, and define the right-translation average
\[
 T\defeq\E[\omega\sim\mathbb U_n]{R_\omega^*C^*T_{\mathrm{prod}}CR_\omega}
  =\bigotimes_{a=1}^b
   \E[\omega_a\sim\mathrm{Unif}(\mathfrak S_{N_a})]{R_{\omega_a}^*C_a^*T_{\rho,a}C_aR_{\omega_a}}.
\]
Because right translation maps consistent restrictions of size $k$ to consistent restrictions of size $k$, Minkowski's inequality preserves the first bound in~\eqref{eq:hard-global-noise-bounds} for $T$. Each tensor factor of $T$ is the positive self-adjoint central operator of~\cite[Theorem~1.10]{keevash2023sharp} (positivity follows from $T_{\rho,a}\ge0$, and centrality---commutation with left and right translations of $\mathfrak S_{N_a}$---follows from the left invariance of $C_a$ and the average over right translations $R_{\omega_a}$). Consequently, $T$ preserves each product isotypic component $\bigotimes_{a=1}^b V_{\lambda^{(a)}}$, and multiplying the single-factor spectral lower bounds $(\rho/72)^{s_a}$ of~\cite[Theorem~1.10]{keevash2023sharp} with $s_a\defeq N_a-\lambda^{(a)}_1$ yields the second bound in~\eqref{eq:hard-global-noise-bounds}.

\emph{Step 2: High-entropy truncation and spectral gap.}
Suppose $\operatorname{Ent}_{\mathbb U_n}(f)\ge\delta n$. Since $\log f\le Ln$ and $\E[\mathbb U_n]{f}=1$ (with $u\log u\ge -1/\e$), we have
\[
 \delta n\le\E[\mathbb U_n]{f\log f\one{f\le\e^{\delta n/2}}}+Ln\,\E[\mathbb U_n]{f\one{f>\e^{\delta n/2}}}
 \le\frac{\delta n}{2}+Ln\,\E[\mathbb U_n]{f\one{f>\e^{\delta n/2}}},
\]
so $\E[\mathbb U_n]{f\one{f>\e^{\delta n/2}}}\ge\delta/(2L)$. With $x_+\defeq\max\{x,0\}$, set $h\defeq(g-\e^{\delta n/8})_+$. Then $0\le h\le g\cdot\one{g>\e^{\delta n/8}}$, so Markov's inequality and $\E[\mathbb U_n]{g^2}=1$ give, for all sufficiently large $n$,
\[
 \|h\|_2\le1,\qquad \|h\|_1\le\e^{-\delta n/8}\E[\mathbb U_n]{g^2}=\e^{-\delta n/8},
 \qquad \|h\|_2^2\ge\bigl(1-\e^{-\delta n/8}\bigr)^2\E[\mathbb U_n]{f\one{f>\e^{\delta n/2}}}\ge\frac{\delta}{4L}.
\]
By H\"older's inequality ($\|h\|_{4/3}\le\|h\|_1^{1/2}\|h\|_2^{1/2}$) and~\eqref{eq:hard-global-noise-bounds},
\[
 \langle h,Th\rangle
 \le\|h\|_{4/3}\|Th\|_4
 \le\|h\|_1^{1/2}\|h\|_2\,\e^{(K\log n)/4}
 \le n^{K/4}\e^{-\delta n/16}.
\]
Since $T$ is positive and preserves the orthogonal decomposition into product isotypic components, the second bound in~\eqref{eq:hard-global-noise-bounds} implies
\[
 \|h^{\le s}\|_2^2
 \le b_0^{-s}\langle h^{\le s},Th^{\le s}\rangle
 \le b_0^{-s}\langle h,Th\rangle
 \le b_0^{-s}n^{K/4}\e^{-\delta n/16}.
\]
Choose $s=\lfloor\varepsilon n\rfloor$ with $\varepsilon\in(0,10^{-5}c_0)$ small enough that $b_0^{-\varepsilon}<\e^{\delta/32}$; then $\|h^{\le s}\|_2^2\to0$, so $\|h-h^{\le s}\|_2^2\ge\delta/(8L)$ for all sufficiently large $n$.

On a single factor $L^2(\mathfrak S_N)$, the operator $u(\omega)\mapsto u(\omega)-\E[\tau]{u(\omega\tau)}$ has Dirichlet form
\[
  \frac12\E{(u(\omega\tau)-u(\omega))^2}.
\]
Let $\chi_\lambda$ be the character of the irreducible representation indexed by $\lambda\vdash N$. The operator acts on each isotypic component $V_\lambda$ by the scalar $1-\chi_\lambda(\tau)/\chi_\lambda(\mathrm{id})$, where Frobenius's character formula and $\lambda_i-2i+1\le\lambda_1-1$ ($i\ge1$) give
\[
 \frac{\chi_\lambda(\tau)}{\chi_\lambda(\mathrm{id})}
 =\frac{\sum_i\lambda_i(\lambda_i-2i+1)}{N(N-1)}
 \le\frac{(\lambda_1-1)\sum_i\lambda_i}{N(N-1)}
 =\frac{\lambda_1-1}{N-1}.
\]
Thus its eigenvalue on $V_\lambda$ is at least $(N-\lambda_1)/(N-1)\ge(N-\lambda_1)/N$. Summing over the $b$ factors shows that on every product isotypic component $\bigotimes_{a=1}^b V_{\lambda^{(a)}}$ outside $\mathcal V_{\le s}$, the transposition Dirichlet form has eigenvalue at least $\sum_{a=1}^b(N_a-\lambda^{(a)}_1)/N_a\ge s/(C_0n)$. Since $x\mapsto(x-\e^{\delta n/8})_+$ is $1$-Lipschitz, we conclude that for all sufficiently large $n$,
\[
 \mathcal E_n(g,g)\ge\mathcal E_n(h,h)\ge\frac{s}{C_0n}\|h-h^{\le s}\|_2^2\ge\frac{\varepsilon\delta}{16C_0L}>0,
\]
proving the uniform assertion.
\end{proof}

\paragraph{Planted pairing setup and switch identity.}
For the remainder of this appendix, fix integers $d,q\ge3$, put $S\defeq dn$ and $\xi\defeq 1-1/q$, and assume that $G\sim\+G_{n,d}$ is $q$-colorable with probability $1-o(1)$ and that its uniform proper-coloring measure $\mu_G=\mu_{G,q,0}$ satisfies \Cref{def:replica-symmetry} (all limits below are along admissible $n$). Let $\mathcal M$ be a uniformly random pairing of the $S$ labeled half-edges, so $\Pr{\mathcal M=\mathfrak m}=1/(S-1)!!$ for every pairing $\mathfrak m$, and let $G(\mathcal M)$ be the induced multigraph. Then $G(\mathcal M)\sim\mathrm{CM}_{n,d}$ (\Cref{def:config-model}). Throughout, expectations $\E[\mathcal M]{\cdot}$ are over this uniform pairing; for graph-measurable functions they equal expectations under $\mathrm{CM}_{n,d}$. Write $Z=Z(\mathcal M)\defeq Z_{G(\mathcal M),q}(0)$ and $\mu_{\mathcal M}\defeq\mu_{G(\mathcal M),q,0}$. By \Cref{lem:simplicity-prob} and the colorability hypothesis, $c\defeq\Pr{\mathcal M\text{ is simple and }Z(\mathcal M)>0}$ satisfies $c\ge c_*>0$ for all large admissible $n$. Define the quenched coloring law $\mathbb P_{\mathrm{col}}$ and the planted pairing law $\widehat{\mathbb Q}_n$ on pairs $(\mathcal M,\sigma)$ by
\[
 \mathbb P_{\mathrm{col}}(\mathfrak m,\sigma)\defeq\frac{\Pr{\mathcal M=\mathfrak m}\one{\mathfrak m\text{ simple}}\one{\sigma\text{ proper for }\mathfrak m}}{cZ(\mathfrak m)},\qquad
 \widehat{\mathbb Q}_n(\mathfrak m,\sigma)\defeq\frac{\Pr{\mathcal M=\mathfrak m}\one{\sigma\text{ proper for }\mathfrak m}}{\E[\mathcal M]{Z}},
\]
where the first fraction is defined to be zero when $Z(\mathfrak m)=0$. The pushforward of $\widehat{\mathbb Q}_n$ under $(\mathcal M,\sigma)\mapsto(G(\mathcal M),\sigma)$ is $\mathbb Q_n$ from \Cref{def:planted-overlap} at $\beta=0$. The likelihood ratio $f\defeq\dif\mathbb P_{\mathrm{col}}/\dif\widehat{\mathbb Q}_n$ and relative entropy are, with the first identity holding $\widehat{\mathbb Q}_n$-almost surely,
\begin{equation}\label{eq:hard-coloring-likelihood}
 f(\mathcal M,\sigma)=\frac{\E[\mathcal M]{Z}}{cZ(\mathcal M)}\one{\mathcal M\text{ simple}},\qquad
 \DKL(\mathbb P_{\mathrm{col}}\Vert\widehat{\mathbb Q}_n)=\log\E[\mathcal M]{Z}-\log c-\E[\mathbb P_{\mathrm{col}}]{\log Z}.
\end{equation}
Recall that an \emph{uncolored half-edge switch} chooses two distinct half-edges $i,j$ with tail vertices $u,v$ that are paired to $i'\ne j$ and $j'\ne i$ in $\mathcal M$, and replaces the two edges $\{i,i'\},\{j,j'\}$ by $\{i,j'\},\{j,i'\}$ to form $\mathcal M'$. For two simple colorable pairings $\mathcal M,\mathcal M'$ related by such a switch at tail vertices $u,v$ and any color $a\in[q]$, put $p_a\defeq\mu_{\mathcal M}(\sigma_u=\sigma_v=a)$ and $p'_a\defeq\mu_{\mathcal M'}(\sigma_u=\sigma_v=a)$. On the event $\{\sigma_u=\sigma_v=a\}$, the two old edge constraints and two new edge constraints coincide, so the proper colorings with $\sigma_u=\sigma_v=a$ are identical for $\mathcal M$ and $\mathcal M'$. Consequently,
\begin{equation}\label{eq:hard-coloring-hellinger-switch}
 W_a\defeq Z(\mathcal M)p_a=Z(\mathcal M')p'_a,\qquad
 W_a\left(\frac1{\sqrt{Z(\mathcal M)}}-\frac1{\sqrt{Z(\mathcal M')}}\right)^2
 =\bigl(\sqrt{p_a}-\sqrt{p'_a}\bigr)^2.
\end{equation}
From any simple $d$-regular pairing $\mathcal M$, when an ordered pair of distinct half-edges $(i,j)$ is chosen uniformly at random from the $S(S-1)$ ordered pairs (or with conditional joint weights bounded by $O(n^{-2})$), only $O_d(n)$ pairs have $u=v$, satisfy $\{i,j\}\in\mathcal M$, or create a self-loop or parallel edge (total probability $O_d(n^{-1})$), while each ordered pair of distinct tail vertices $(u,v)$ has conditional probability $O_d(n^{-2})$. Consequently, for $\mathcal M$ conditioned on simplicity (whose graph $G(\mathcal M)\sim\+G_{n,d}$ is colorable with probability $1-o(1)$), color symmetry on colorable graphs ($\mu_{\mathcal M}(\sigma_u=a)=1/q$) and the pair-averaged replica-symmetry estimate (\Cref{def:replica-symmetry}) imply $p_a\to q^{-2}$ in $L^1$ for every $a\in[q]$; moreover, since the simple switch chain (which proposes a uniform ordered pair of distinct half-edges and rejects moves with $\{i,j\}\in\mathcal M$ or creating loops or parallel edges) is reversible for the uniform simple-pairing law and rejects with probability $O_d(n^{-1})$, reversibility transfers the same $L^1$ convergence $p'_a\to q^{-2}$ to the switched graph $\mathcal M'$ on valid simple switches.

\begin{lemma}[Edge-type concentration]\label{lem:hard-coloring-edge-type}
In the setting above, for $(\mathcal M,\sigma)\sim\mathbb P_{\mathrm{col}}$, let $n_a=n\rho_a(\sigma)$ be the number of vertices of color $a\in[q]$ and $m_{ab}=dnQ_{ab}$ be the number of edges joining colors $1\le a<b\le q$ (so $(\*\rho,Q)$ is the profile of $(\mathcal M,\sigma)$ as in \Cref{def:profiles-bethe-rates}, with $Q_{aa}=0$). Then $(\*\rho,Q)\xrightarrow{\mathrm p}(\*U_q,P_0)$:
\[
 n_a=\frac nq+o_{\mathrm p}(n)\quad(a\in[q]),\qquad
 m_{ab}=\frac{dn}{q(q-1)}+o_{\mathrm p}(n)\quad(1\le a<b\le q).
\]
\end{lemma}
\begin{proof}
\emph{Step 1: Vanishing Dirichlet energy under color-preserving switches.}
Consider the transition kernel $\mathcal K$ on pairs $(\mathcal M,\sigma)$ that picks an ordered pair of distinct half-edges $(i,j)$ uniformly at random, proposes the switch $\mathcal M\leftrightarrow\mathcal M'$ when $\{i,j\}$ is not an edge of $\mathcal M$, accepts the switch if $\sigma$ remains a proper coloring of $\mathcal M'$, and otherwise stays at $(\mathcal M,\sigma)$. Because each accepted switch is an involution for fixed $(i,j,\sigma)$ and $\widehat{\mathbb Q}_n(\cdot,\sigma)$ is uniform over pairings compatible with $\sigma$, $\mathcal K$ is reversible with respect to $\widehat{\mathbb Q}_n$. Write $\mathcal E_{\widehat{\mathbb Q}_n}(\sqrt f,\sqrt f)\defeq\tfrac12\E[X\sim\widehat{\mathbb Q}_n,\,X'\sim\mathcal K(X,\cdot)]{(\sqrt{f(X)}-\sqrt{f(X')})^2}$ for the Dirichlet form of $\sqrt f$. For a proposed switch $\mathcal M\leftrightarrow\mathcal M'$ with tail vertices $u,v$, let $W(\mathcal M,\mathcal M')$ be the number of common proper colorings of $\mathcal M$ and $\mathcal M'$, and on simple colorable endpoints put $a_{\mathcal M}\defeq W(\mathcal M,\mathcal M')/Z(\mathcal M)$ and $a_{\mathcal M'}\defeq W(\mathcal M,\mathcal M')/Z(\mathcal M')$. Expanding $\mathcal E_{\widehat{\mathbb Q}_n}(\sqrt f,\sqrt f)$, canceling $\E[\mathcal M]{Z}$, and using the $O_d(n^{-1})$ bound on switches that create or destroy loops or parallel edges gives
\[
 \mathcal E_{\widehat{\mathbb Q}_n}(\sqrt f,\sqrt f)
 =\frac1{2c}\E[\mathcal M,\,(i,j)]{\one{\mathcal M,\mathcal M'\text{ simple and colorable}}\bigl(\sqrt{a_{\mathcal M}}-\sqrt{a_{\mathcal M'}}\bigr)^2}+O_d(n^{-1}).
\]
Here the expectation is over the uniform pairing and the uniform ordered pair of distinct half-edges, with invalid switch proposals contributing zero.
Summing~\eqref{eq:hard-coloring-hellinger-switch} over $a\in[q]$ gives $Z(\mathcal M)\mu_{\mathcal M}(\sigma_u=\sigma_v)=Z(\mathcal M')\mu_{\mathcal M'}(\sigma_u=\sigma_v)$ with $\mu_{\mathcal M}(\sigma_u=\sigma_v),\mu_{\mathcal M'}(\sigma_u=\sigma_v)\xrightarrow{\mathrm p}1/q$. Since $a_{\mathcal M}Z(\mathcal M)=a_{\mathcal M'}Z(\mathcal M')$ and $0\le a_{\mathcal M},a_{\mathcal M'}\le1$, we obtain $Z(\mathcal M')/Z(\mathcal M)\xrightarrow{\mathrm p}1$ and hence $\mathcal E_{\widehat{\mathbb Q}_n}(\sqrt f,\sqrt f)=o(1)$. By reversibility of $\mathcal K$, factoring $f(X')-f(X)=(\sqrt{f(X')}-\sqrt{f(X)})(\sqrt{f(X')}+\sqrt{f(X)})$, and applying Cauchy--Schwarz with $\E[\widehat{\mathbb Q}_n]{f}=1$, any observable $\Psi$ whose increment under a single accepted switch is bounded by $L$ satisfies
\begin{equation}\label{eq:hard-coloring-ibp}
 \bigl|\E[\mathbb P_{\mathrm{col}}]{\mathcal K\Psi-\Psi}\bigr|
 \le 2L\sqrt{2\mathcal E_{\widehat{\mathbb Q}_n}(\sqrt f,\sqrt f)}=o(1).
\end{equation}

\emph{Step 2: Stationarity of the edge-color potential.}
For a proper coloring $\sigma$, the edge marginal $Q$ has entries $Q_{ab}=Q_{ba}=m_{ab}/S$ ($a\ne b$) and $Q_{aa}=0$, with $\sum_{b\ne a}Q_{ab}=n_a/n=\rho_a(\sigma)$. Fix $\delta\in(0,1)$, write $h_\delta(x)\defeq(x+\delta)\log(x+\delta)$, and define $\Psi_\delta(\mathcal M,\sigma)\defeq S\sum_{1\le a<b\le q}h_\delta(Q_{ab})$, whose increment under a single switch is $O_{q,\delta}(1)$. Under $\mathcal K$, the directed colors of the two chosen half-edges and their partners are $(a,b)$ and $(c,e)$ with probability $Q_{ab}Q_{ce}+O(S^{-1})$ uniformly over proper pairings; an accepted switch changes $(Q_{ab})$ only when $a,b,c,e$ are pairwise distinct (switches with $a=c$ or $b=e$ preserve all $Q_{ab}$, while $a=e$ or $b=c$ is rejected as improper), in which case it replaces edge colors $ab,ce$ by $ae,cb$. Pairing each ordered quadruple $(a,b,c,e)$ of distinct colors with $(a,e,c,b)$ and Taylor expanding $h_\delta$ gives, uniformly in $(\mathcal M,\sigma)$ for fixed $\delta$,
\[
 \mathcal K\Psi_\delta-\Psi_\delta
 =-\frac12\sum_{a,b,c,e\,\mathrm{distinct}}
 \bigl(Q_{ab}Q_{ce}-Q_{ae}Q_{cb}\bigr)
 \log\frac{(Q_{ab}+\delta)(Q_{ce}+\delta)}{(Q_{ae}+\delta)(Q_{cb}+\delta)}
 +O_{q,\delta}(S^{-1}).
\]
For $x=Q_{ab}Q_{ce}$, $y=Q_{ae}Q_{cb}$, $u=(Q_{ab}+\delta)(Q_{ce}+\delta)$, and $v=(Q_{ae}+\delta)(Q_{cb}+\delta)$, we have $|u-x|,|v-y|\le3\delta$ and $|\log(u/v)|\le 2\log(2/\delta)$, so $(u-v)\log(u/v)\ge4(\sqrt u-\sqrt v)^2$ yields $\frac12(x-y)\log(u/v)\ge(\sqrt x-\sqrt y)^2-C\delta\log(2/\delta)$. Setting $R(Q)\defeq\sum_{a,b,c,e\,\mathrm{distinct}}(\sqrt{Q_{ab}Q_{ce}}-\sqrt{Q_{ae}Q_{cb}})^2$, combining with~\eqref{eq:hard-coloring-ibp}, and letting first $n\to\infty$ and then $\delta\downarrow0$ gives $\E[\mathbb P_{\mathrm{col}}]{R(Q)}\to0$.

Replica symmetry and color symmetry give $\rho_a(\sigma)=n_a/n\xrightarrow{\mathrm p}1/q$ for each $a\in[q]$. For $q=3$, the three row sums $\rho_a$ uniquely determine the three off-diagonal entries $Q_{ab}$. For $q\ge4$, any symmetric nonnegative matrix $Q$ with zero diagonal, constant row sums $r>0$, and $Q_{ab}Q_{ce}=Q_{ae}Q_{cb}$ for all distinct $a,b,c,e$ satisfies $Q_{ab}=r/(q-1)$ for all $a\ne b$: choosing a maximum entry $Q_{12}=m_*>0$, if $Q_{1c}=0$ for some $c\ge3$, then $Q_{ce}=0$ for all $e\notin\{1,2,c\}$ and hence $Q_{2c}=r$, contradicting $Q_{21}=m_*>0$ and the row sum at $2$; similarly $Q_{2c}>0$, so $Q_{ce}>0$ for all distinct $c,e\ge3$, the ratio $Q_{1c}/Q_{2c}$ is a constant $\lambda$ forced to equal $1$ by the equal row sums at $1,2$, and setting $b_c\defeq Q_{1c}=Q_{2c}$ and $B\defeq\sum_{c\ge3}b_c$ gives $Q_{ce}=b_cb_e/m_*$ and $b_c(2+(B-b_c)/m_*)=r$, whose strict monotonicity in $b_c\in[0,m_*]$ forces $b_c\equiv b$ and $(m_*-b)(m_*+(q-3)b)=0$, hence $m_*=b$. By compactness of the simplex, $\E[\mathbb P_{\mathrm{col}}]{R(Q)}\to0$ therefore implies $Q_{ab}\xrightarrow{\mathrm p}1/[q(q-1)]$ for all $a\ne b$.
\end{proof}

\begin{proof}[Proof of~\Cref{lem:coloring-pointwise-rs-free-energy}]
\emph{Step 1: Color-profile entropy.}
Denote the vertex- and edge-color count profile under $\mathbb P_{\mathrm{col}}$ and $\widehat{\mathbb Q}_n$ by $\mathcal T\defeq((n_a)_{a\in[q]},\allowbreak(m_{ab})_{1\le a<b\le q})$, with marginal laws $\mathbb P_{\mathrm{col},\mathcal T}$ and $\widehat{\mathbb Q}_{n,\mathcal T}$. For a feasible count profile $\tau=((n_a),(m_{ab}))$ with $S_a\defeq dn_a$, $\rho_a\defeq n_a/n$, and $Q_{ab}=Q_{ba}\defeq m_{ab}/S$ ($a\ne b$, $Q_{aa}=0$), the number of vertex colorings with color counts $(n_a)$ is $A_\tau\defeq n!/\prod_a n_a!$ and the number of compatible pairings for each such coloring is $N_\tau\defeq(\prod_a S_a!)/\prod_{a<b}m_{ab}!$ (since choosing the $m_{ab}$ half-edges in each color class $a$ to pair with color $b$ and then choosing the $m_{ab}!$ bijections between them yields $(\prod_a S_a!/\prod_{b\ne a}m_{ab}!)\prod_{a<b}m_{ab}! = N_\tau$ pairings), so
\begin{equation}\label{eq:hard-coloring-type-law}
\begin{aligned}
 \widehat{\mathbb Q}_n(\mathcal T=\tau)&=\frac{A_\tau N_\tau}{(S-1)!!\E[\mathcal M]{Z}},\\
 \log\E[\mathcal M]{Z}&=n\phi_{d,q}(0)+O(\log n)=n\log q+\frac S2\log\xi+O(\log n).
\end{aligned}
\end{equation}
Indeed, at $\beta=0$ the nonconstant eigenvalues are $-1/(q-1)\le0$. The normalization $\mathcal Z_G(0)=\e^{-n\phi_{d,q}(0)}Z_G(0)$ and \Cref{prop:profile-moments}(1)--(2) therefore give, uniformly over feasible profiles,
\[
 \frac{A_\tau N_\tau}{(S-1)!!}
 =\exp\bigl\{n\phi_{d,q}(0)+n\Phi_0(\*\rho,Q)+O(\log n)\bigr\},
 \qquad \Phi_0(\*\rho,Q)\le0.
\]
At the uniform profile $(\*\rho,Q)=(\*U_q,P_0)$, $\Phi_0=0$. Feasible integer counts with $n_a=n/q+O(1)$ and $m_{ab}=S/[q(q-1)]+O(1)$ have $\Phi_0=O(n^{-2})$. To construct them, choose $n_a=n/q+O(1)$ with $\sum_a n_a=n$, and put $b=\lfloor S/[q(q-1)]\rfloor-K$ for a sufficiently large constant $K$. The residual degrees $dn_a-(q-1)b$ form a bounded nonnegative even-sum sequence whose largest entry is at most the sum of the others, so they are realized by a loopless multigraph on $[q]$; adding its edge counts to $b$ gives the required $m_{ab}$. Summing over the polynomially many profiles and using this one profile gives the first-moment estimate in~\eqref{eq:hard-coloring-type-law}. The same estimate yields $-\log\widehat{\mathbb Q}_n(\mathcal T=\tau)=-n\Phi_0(\*\rho,Q)+O(\log n)$ uniformly over feasible $\tau$. The rate $-\Phi_0$ is bounded and continuous on these profiles and vanishes at $(\*U_q,P_0)$, so \Cref{lem:hard-coloring-edge-type} gives
\begin{equation}\label{eq:hard-coloring-type-entropy}
 \DKL(\mathbb P_{\mathrm{col},\mathcal T}\Vert\widehat{\mathbb Q}_{n,\mathcal T})\le\E[\mathbb P_{\mathrm{col}}]{-\log\widehat{\mathbb Q}_n(\mathcal T)}=o(n).
\end{equation}

\emph{Step 2: Permutation lift and restriction bounds.}
Choose $c_0>0$ below the limiting uniform proportions $1/q$ and $d/[q(q-1)]$, and fix $K$ larger by two than the polynomial exponent of the number of feasible count profiles $\tau$. Retain the profiles $\tau$ with $n_a,m_{ab}\ge c_0n$ and $\mathbb P_{\mathrm{col}}(\mathcal T=\tau)\ge n^{-K}$, which have total $\mathbb P_{\mathrm{col}}$-probability $1-o(1)$ by \Cref{lem:hard-coloring-edge-type}. Conditional on $\mathcal T=\tau$, the vertex coloring $\sigma$ is uniform under both $\mathbb P_{\mathrm{col}}$ and $\widehat{\mathbb Q}_n$ by vertex-label symmetry; fix one such coloring $\sigma$ and a canonical pairing matching $m_{ab}$ half-edges between colors $a$ and $b$ for each $a<b$, and independently permute the $S_a$ half-edges of color $a$ via $\omega=(\omega_1,\ldots,\omega_q)\in\prod_{a=1}^q\mathfrak S_{S_a}$ to induce a pairing $\mathcal M=\mathcal M(\omega)$. The uniform measure $\mathbb U_\tau$ on $\prod_{a=1}^q\mathfrak S_{S_a}$ pushes forward to $\widehat{\mathbb Q}_n(\,\cdot\mid\tau,\sigma)$ with constant fiber size $\prod_{a<b}m_{ab}!$. Lifting $\mathbb P_{\mathrm{col}}(\,\cdot\mid\tau,\sigma)$ uniformly over these fibers gives a probability measure $\mathbb P_\tau$ on $\prod_{a=1}^q\mathfrak S_{S_a}$ whose density $f_\tau\defeq\dif\mathbb P_\tau/\dif\mathbb U_\tau$ satisfies $\DKL(\mathbb P_{\mathrm{col}}(\,\cdot\mid\tau)\Vert\widehat{\mathbb Q}_n(\,\cdot\mid\tau))=\operatorname{Ent}_{\mathbb U_\tau}(f_\tau)$. {Since $\{\mathcal T=\tau,\sigma\}$ has probability $\mathbb P_{\mathrm{col}}(\mathcal T=\tau)/A_\tau$ under $\mathbb P_{\mathrm{col}}$ and $\widehat{\mathbb Q}_n(\mathcal T=\tau)/A_\tau$ under $\widehat{\mathbb Q}_n$,~\eqref{eq:hard-coloring-likelihood} and~\eqref{eq:hard-coloring-type-law} give, for every $\omega\in\prod_{a=1}^q\mathfrak S_{S_a}$ with $\mathcal M=\mathcal M(\omega)$ and retained $\tau$,}
\begin{equation}\label{eq:hard-coloring-density-range}
 {f_\tau(\omega)
 =\frac{\widehat{\mathbb Q}_n(\mathcal T=\tau)\E[\mathcal M]{Z}}{c\mathbb P_{\mathrm{col}}(\mathcal T=\tau)}\,\frac{\one{\mathcal M\text{ simple}}}{Z(\mathcal M)}}
 \le\frac{q^n}{c\mathbb P_{\mathrm{col}}(\mathcal T=\tau)}\le\exp(Cn).
\end{equation}
For $g_\tau\defeq\sqrt{f_\tau}$ and a consistent restriction $R$ on $\prod_{a=1}^q\mathfrak S_{S_a}$ fixing $|R|=k=\sum_a k_a$ permutation coordinates ($k_a$ in factor $a$), let $e$ canonical edges have both endpoints fixed by $R$ (prescribing an actual half-edge-pair pattern $F$ on $v\le 2e$ vertices with color counts $v_a$) and $s=k-e$ canonical edges be touched by $R$ (with $s_{ab}$ touched color-$ab$ edges). Given the pairing $\mathcal M$, the fiber probability of $R$ is either zero or $\prod_{a<b}(m_{ab})_{s_{ab}}^{-1}\le(C/n)^{k-e}$, where $(m)_s\defeq m(m-1)\cdots(m-s+1)\ge(m/\e)^s$ on retained profiles. Averaging $F$ under vertex permutations within each color class gives denominator $\prod_a(n_a)_{v_a}$, while the uncolored ordered embedding count of $F$ is vertex-permutation-invariant; dropping the conditioning on $\{\mathcal T=\tau,\sigma\}$ and on simplicity/colorability costs at most $(c\mathbb P_{\mathrm{col}}(\mathcal T=\tau))^{-1}$, and $e$ prescribed disjoint half-edge pairs in the uniform pairing $\mathcal M$ have probability $\prod_{i=0}^{e-1}(S-2i-1)^{-1}\le(C/n)^e$ uniformly for $e\le S/2$. Thus $\mathbb P_{\mathrm{col}}(F\subseteq\mathcal M\mid\tau,\sigma)\le C\mathbb P_{\mathrm{col}}(\mathcal T=\tau)^{-1}n^v(C/n)^e/\prod_a(n_a)_{v_a}\le n^K(C/n)^e$. Since $R$ has probability $\mathbb U_\tau(R)=\prod_a(S_a)_{k_a}^{-1}\ge(C'n)^{-k}$ under $\mathbb U_\tau$, every consistent restriction $R$ of size $|R|=k$ satisfies
\begin{equation}\label{eq:hard-coloring-global-restrictions}
 {\E[\mathbb U_\tau]{f_\tau\mid R}
 =\frac{\mathbb P_\tau(R)}{\mathbb U_\tau(R)}
 \le\frac{\mathbb P_{\mathrm{col}}(F\subseteq\mathcal M\mid\tau,\sigma)\prod_{a<b}(m_{ab})_{s_{ab}}^{-1}}{\prod_{a=1}^q(S_a)_{k_a}^{-1}}}
 \le n^K C^k,
 \qquad \|(g_\tau)_R\|_2\le n^{K/2}C^{k/2}.
\end{equation}

\emph{Step 3: Dirichlet energy and relative entropy.}
A transposition $\tau_a\in\mathfrak S_{S_a}$ in the color-$a$ permutation factor performs an uncolored half-edge switch at two tail vertices $u,v$ of color $a$, replacing $\mathcal M=\mathcal M(\omega)$ by $\mathcal M'\defeq\mathcal M(\omega\tau_a)$ while preserving both $\sigma$ and $\mathcal T=\tau$. {Substituting the explicit formula for $f_\tau$ from~\eqref{eq:hard-coloring-density-range} into $\mathbb P_{\mathrm{col}}(\mathcal T=\tau)\mathcal E_\tau(g_\tau,g_\tau)$ cancels the conditional normalization $\mathbb P_{\mathrm{col}}(\mathcal T=\tau)$:}
\begin{equation}\label{eq:hard-coloring-energy-cancellation}
{\begin{aligned}
 &\sum_{\tau\text{ retained}}\mathbb P_{\mathrm{col}}(\mathcal T=\tau)\mathcal E_\tau(g_\tau,g_\tau)\\
 &\qquad=\sum_{\tau\text{ retained}}\frac{\widehat{\mathbb Q}_n(\mathcal T=\tau)\E[\mathcal M]{Z}}{2c}\sum_{a=1}^q\E[\omega\sim\mathbb U_\tau,\,\tau_a]{\left(\frac{\one{\mathcal M'\text{ simple}}}{\sqrt{Z(\mathcal M')}}-\frac{\one{\mathcal M\text{ simple}}}{\sqrt{Z(\mathcal M)}}\right)^2}.
\end{aligned}}
\end{equation}
{Since $\widehat{\mathbb Q}_n(\mathcal T=\tau)\E[\mathcal M]{Z}$ is the expected number under the uniform pairing $\mathcal M$ of proper colorings of profile $\tau$ and $S_a\ge c_0n$ on retained profiles, the factor transposition rates $\binom{S_a}{2}^{-1}$ are bounded by $O(1)$ times the uniform uncolored half-edge-switch rate $\binom{S}{2}^{-1}$.} Summing~\eqref{eq:hard-coloring-energy-cancellation} over common proper colorings and applying~\eqref{eq:hard-coloring-hellinger-switch}, the $L^1$ convergence $p_a,p'_a\to q^{-2}$, and the $O_d(n^{-1})$ nonsimple switch boundary bound established above gives
\begin{equation}\label{eq:hard-coloring-vanishing-energy}
 \sum_{\tau\text{ retained}}\mathbb P_{\mathrm{col}}(\mathcal T=\tau)\mathcal E_\tau(g_\tau,g_\tau)=o(1).
\end{equation}
By~\eqref{eq:hard-coloring-density-range},~\eqref{eq:hard-coloring-global-restrictions},~\eqref{eq:hard-coloring-vanishing-energy}, and the uniform assertion of \Cref{lem:global-entropy-coercivity},
\[
 \sum_{\tau\text{ retained}}\mathbb P_{\mathrm{col}}(\mathcal T=\tau)\DKL\bigl(\mathbb P_{\mathrm{col}}(\,\cdot\mid\tau)\Vert\widehat{\mathbb Q}_n(\,\cdot\mid\tau)\bigr)
 =\sum_{\tau\text{ retained}}\mathbb P_{\mathrm{col}}(\mathcal T=\tau)\operatorname{Ent}_{\mathbb U_\tau}(f_\tau)=o(n).
\]
Since discarded profiles have total mass $\mathbb P_{\mathrm{col}}(\mathcal T\text{ discarded})=o(1)$ and profile entropy $O(\log n)$, their contribution to $\DKL(\mathbb P_{\mathrm{col}}\Vert\widehat{\mathbb Q}_n)$ is at most 
$$
  (n\log q-\log c)\mathbb P_{\mathrm{col}}(\mathcal T\text{ discarded})+\allowbreak\sum_{\tau\text{ discarded}}\mathbb P_{\mathrm{col}}(\mathcal T=\tau)\log(1/\mathbb P_{\mathrm{col}}(\mathcal T=\tau))=o(n).
  $$
  Combining these bounds with~\eqref{eq:hard-coloring-type-entropy} and the relative-entropy chain rule yields $\DKL(\mathbb P_{\mathrm{col}}\Vert\widehat{\mathbb Q}_n)=o(n)$, so~\eqref{eq:hard-coloring-likelihood} and~\eqref{eq:hard-coloring-type-law} give $\frac1n\E[\mathbb P_{\mathrm{col}}]{\log Z}=\phi_{d,q}(0)+o(1)$. Finally, Markov's inequality and~\eqref{eq:hard-coloring-type-law} give $\Pr[\mathbb P_{\mathrm{col}}]{n^{-1}\log Z>\phi_{d,q}(0)+\varepsilon}=o(1)$ for every $\varepsilon>0$, which together with $0\le n^{-1}\log Z\le\log q$ on $\{Z>0\}$ and $\Pr[\+G_{n,d}]{Z>0}=1-o(1)$ yields $\frac1n\log Z_{G,q}(0)\xrightarrow{\mathrm p}\phi_{d,q}(0)$.
\end{proof}
\endgroup

% END INCLUDED FILE: appendices/d-hard-coloring-entropy.tex
% BEGIN INCLUDED FILE: appendices/b-analytic-tools.tex
\section{Analytic Tools}\label{sec:tools}\label{app:tools}
In this appendix, we collect the asymptotic and combinatorial tools used in our moment calculations.

Because our analysis combines lattice sums over empirical profiles with contour-integral representations of fixed-density weights, we first establish a complex-parameter version of the discrete Laplace method before turning to pairing enumeration and coefficient-extraction identities on the configuration model.

\begin{lemma}[Complex discrete Laplace method]
\label{lem:complex-discrete-laplace}
Let $k,k_{\mathrm{par}}\ge1$ be integers. Let $D\subseteq\=C^k$ and $U\subseteq\=C^{k_{\mathrm{par}}}$
be neighborhoods of $\*x_0\in\=R^k$ and $t\in\=R^{k_{\mathrm{par}}}$,
respectively.
Let $f,g_n$ be holomorphic on $D\times U$, with
$g_n\to g$ locally uniformly.
Suppose that $\*p:U\to D$ is holomorphic and satisfies
\[
  \*p(t)=\*x_0,\qquad
  f(\*p(z),z)=0,\qquad
  \nabla_{\*x} f(\*p(z),z)=0.
\]
Assume that $f(\*x,t)$ is real for real $\*x$ near $\*x_0$ and that
\[
  H(t)\defeq-D_{\*x}^2f(\*x_0,t)
\]
is positive definite.

For every sufficiently small fixed open box
$B\subseteq\=R^k$ containing $\*x_0$, there exists a
neighborhood $U_B\subseteq U$ of $t$ such that
\[
  \frac{1}{(2\pi n)^{k/2}}
  \sum_{\*x\in B\cap n^{-1}\=Z^k}
  g_n(\*x,z)\e^{nf(\*x,z)}
  =
  \frac{g(\*p(z),z)}{\sqrt{\det H(z)}}+o(1),
\]
locally uniformly for $z\in U_B$, where
\[
  H(z)=-D_{\*x}^2f(\*p(z),z).
\]
The square-root branch is chosen to be positive at $z=t$.
\end{lemma}
\begin{proof}
Write $B=\prod_{j=1}^k(a_j,b_j)$.
The assumptions at $(\*x_0,t)$ and Taylor expansion give
$$
f(\*x,t)\le-c\|\*x-\*x_0\|_2^2
$$
on a sufficiently small real neighborhood of $\*x_0$.
Choose $s>0$ sufficiently small and then shrink the parameter
neighborhood so that $T=\overline B+\mathrm i[-s,s]^k$
lies compactly in $D$ and
$$
\begin{aligned}
\Re f(\*x+\mathrm i \*y,z)&\le-\gamma
&& (\*x\in\partial B,\ \*y\in[-s,s]^k),\\
\sup_{\*u\in T}\Re f(\*u,z)&\le\pi s
\end{aligned}
$$
for some $\gamma>0$.
These choices are possible because $\partial B$ stays away
from $\*x_0$ and, at $z=t$, Taylor expansion in the imaginary
directions gives $\sup_T\Re f=O(s^2)$.
The bounds remain valid, with smaller $\gamma$, under small
perturbations of the endpoints of $B$.
We also ensure that $|\Im p_j(z)|<s$ and that $\Re \*p(z)$
stays a positive distance from $\partial B$.

We first compare the lattice sum with an integral.
For $R=[a,b]+\mathrm i[-s,s]$, where
$na,nb\in\=Z+\frac12$, the residue theorem gives
$$
\sum_{j/n\in(a,b)}h(j/n)
=\frac{n}{2\mathrm i}
\int_{\partial R}h(u)\cot(\pi nu)\dif u
$$
whenever $h$ is holomorphic near $R$.
On the horizontal sides,
$\cot(\pi n(x\pm\mathrm i s))
=\mp\mathrm i+O(\e^{-2\pi ns})$,
and the cotangent is bounded on the vertical sides.
Cauchy's theorem therefore yields
$$
\left|
\sum_{j/n\in(a,b)}h(j/n)-n\int_a^b h(x)\dif x
\right|
\le Cn\left(
\sup_{\text{vertical sides}}|h|
+\e^{-2\pi ns}\sup_R|h|
\right).
$$

Move each endpoint of $B$ by $O(n^{-1})$ to a point in
$n^{-1}(\=Z+1/2)$, obtaining a box $B_n$ with
$$
B_n\cap n^{-1}\=Z^k=B\cap n^{-1}\=Z^k.
$$
Since $g_n\to g$ locally uniformly, the functions $g_n$ are uniformly
bounded on the compact sets under consideration.
Applying the preceding estimate coordinatewise to
$h_n(\*x,z)=g_n(\*x,z)\e^{nf(\*x,z)}$ gives
$$
\sum_{\*x\in B\cap n^{-1}\=Z^k}h_n(\*x,z)
=n^k\int_{B_n}h_n(\*x,z)\dif \*x+O(n^k\e^{-cn}).
$$
By Cauchy's theorem and the boundary bound above, shifting
the contour to $B_n+\mathrm i\Im \*p(z)$ introduces another
exponentially small error.

On the shifted contour, write $\*x=\*p(z)+\*v$ with $\*v$ real.
Taylor expansion gives, uniformly in nearby parameters,
$$
f(\*p(z)+\*v,z)
=-\frac12\*v^\trans H(z)\*v+O(\|\*v\|_2^3),
\qquad
\Re f(\*p(z)+\*v,z)\le-c_1\|\*v\|_2^2.
$$
Set $\*v=\*u/\sqrt n$.
The cubic remainder in the exponent is
$O(n^{-1/2}\|\*u\|_2^3)$ and tends to zero for each fixed $\*u$.
Moreover, $\sqrt n(B_n-\Re \*p(z))$ expands to $\=R^k$.
Gaussian domination and local uniform convergence of $g_n$
therefore give
$$
\begin{aligned}
\int_{B_n+\mathrm i\Im \*p(z)}
g_n(\*x,z)\e^{nf(\*x,z)}\dif \*x
&=n^{-k/2}\left(
g(\*p(z),z)\int_{\=R^k}
\e^{-\frac12\*u^\trans H(z)\*u}\dif \*u+o(1)
\right)\\
&=\left(\frac{2\pi}{n}\right)^{k/2}
\left(
\frac{g(\*p(z),z)}{\sqrt{\det H(z)}}+o(1)
\right).
\end{aligned}
$$
All estimates are uniform on compact subsets of a sufficiently
small $U_B$.
Combining the integral asymptotic with the lattice-sum estimate
and dividing by $(2\pi n)^{k/2}$ proves the claim.
\end{proof}

We next record the exact combinatorial count of pairings realizing a given empirical profile $(\*\rho,Q)$, together with the conditional inclusion probabilities of prescribed half-edge pairs.

\begin{lemma}[Pairing enumeration and edge-inclusion probability~{\cite[Section~2]{kemkes2010chromatic}, \cite[(2.7)]{coja2016chromatic}}]\label{lem:pairing-enumeration}
  With $m!!\defeq m(m-2)\cdots 1$ for odd $m\ge1$, $m!!\defeq m(m-2)\cdots 2$ for even $m\ge2$, and $0!=0!!=(-1)!!=1$, the expected number of assignments with feasible profile $(\*\rho,Q)$ (\Cref{def:profiles-bethe-rates}) under $\mathrm{CM}=\mathrm{CM}_{n,d}$ (\Cref{def:config-model}) is
  \begin{equation}\label{eq:pairing-type-count}
    \frac{n!\prod_x(dn\rho_x)!}
    {(dn-1)!!\prod_x(n\rho_x)!\prod_{x<y}(dnQ_{xy})!
      \prod_x(dnQ_{xx})!!}.
  \end{equation}
  For the underlying uniform labeled pairing $\mathcal M$ generating $G\sim\mathrm{CM}_{n,d}$, conditional on a spin assignment and feasible profile $(\*\rho,Q)$, the pairing is uniform among those realizing $(\*\rho,Q)$; for each fixed $e\ge1$, the probability that it contains $e$ prescribed disjoint half-edge pairs of types $(a_r,b_r)_{r=1}^e$ is
  \[
    (dn)^{-e}\left(\prod_{r=1}^e\frac{Q_{a_rb_r}}{\rho_{a_r}\rho_{b_r}}+O(n^{-1})\right)
  \]
  uniformly whenever all $\rho_{a_r}$ and $\rho_{b_r}$ are bounded away from zero (including when some entries of $Q$ vanish).
\end{lemma}
\begin{proof}
  Let $e_{xy}=e_{yx}$ be the number of prescribed pairs of unordered type $\{x,y\}$ and $k_x=\sum_y(1+\one{x=y})e_{xy}$ be the number of prescribed half-edges of spin $x$. By~\eqref{eq:pairing-type-count}, the ratio of the number of compatible pairings containing the $e$ prescribed pairs to the total number of compatible pairings is
  \[
    \frac{\prod_{x<y}(dnQ_{xy})_{e_{xy}}\prod_x\prod_{j=0}^{e_{xx}-1}(dnQ_{xx}-2j)}{\prod_x(dn\rho_x)_{k_x}}
    =(dn)^{-e}\left(\prod_{r=1}^e\frac{Q_{a_rb_r}}{\rho_{a_r}\rho_{b_r}}+O(n^{-1})\right)
  \]
  uniformly whenever all $\rho_{a_r}$ and $\rho_{b_r}$ are bounded away from zero, since expanding the numerator as a polynomial in $dnQ_{xy}$ (with $0\le Q_{xy}\le1$) gives $\prod_{r=1}^e(dnQ_{a_rb_r})+O(n^{e-1})$ and dividing by $\prod_x(dn\rho_x)_{k_x}=(dn)^{2e}\prod_{r=1}^e(\rho_{a_r}\rho_{b_r})(1+O(n^{-1}))$ yields leading term $(dn)^{-e}\prod_{r=1}^e\frac{Q_{a_rb_r}}{\rho_{a_r}\rho_{b_r}}$ and lower-order terms of order $O(n^{-e-1})$ even when some $Q_{xy}$ vanish.
\end{proof}

We recall the standard iterated Cauchy coefficient formula.

\begin{theorem}[Iterated Cauchy coefficient formula]\label{thm:cauchy-coeff}
For a polynomial $F(X_1,\ldots,X_k)$, multi-index $\*\alpha\in\=Z_{\ge0}^k$, and $u_1,\ldots,u_k\ne0$, writing $[\*X^{\*\alpha}]F$ for the coefficient of $\*X^{\*\alpha}\defeq\prod_{j=1}^k X_j^{\alpha_j}$ in $F$ and $\*u^{-\*\alpha}\defeq\prod_{j=1}^k u_j^{-\alpha_j}$,
$$
 [\*X^{\*\alpha}]F=
 \frac{\*u^{-\*\alpha}}{(2\pi)^k}
 \int_{[-\pi,\pi]^k}F(u_1\e^{\mathrm i\theta_1},\ldots,u_k\e^{\mathrm i\theta_k})
                  \e^{-\mathrm i\*\alpha^\trans\*\theta}\dif\*\theta.
$$
\end{theorem}

Finally, the following lemma provides the explicit polynomial generating function to which \Cref{thm:cauchy-coeff} will be applied, expressing the fixed-density weight $W_n(\*\rho,\zeta)$ in terms of a single multivariate coefficient.

\begin{lemma}[Polynomial representation of the fixed-density weight]
\label{lem:polynomial-representation}
Let ${S}=dn$ be even, and let $\*\rho\in\Delta_{[m]}$ be a vertex marginal distribution satisfying
$n\*\rho\in\=Z_{\ge0}^m$.
Given $\*\pi=\*\pi(\zeta)\in(\=C\setminus\{0\})^m$, a symmetric matrix $P=P(\zeta)\in\=C^{m\times m}$, and $M=\diag(\*\pi)^{-1}P\diag(\*\pi)^{-1}$, the fixed-density weight under $\mathrm{CM}=\mathrm{CM}_{n,d}$ (\Cref{def:config-model}) satisfies
$$
W_n(\*\rho,\zeta)
\defeq\E[\mathrm{CM}]{
      \sum_{\sigma:\,\*\rho(\sigma)=\*\rho}\prod_{v\in[n]}\pi_{\sigma_v}
                        \prod_{uv\in E}M_{\sigma_u\sigma_v}}
=\frac{n!}{{S}!}
\prod_{i=1}^m
\frac{({S}\rho_i)!}{(n\rho_i)!}\pi_i^{-({S}-n)\rho_i}
[\*X^{{S}\*\rho}](\*X^\trans P\*X)^{{S}/2},
$$
where $[\*X^{{S}\*\rho}]$ extracts the coefficient of
$\prod_{i=1}^m X_i^{{S}\rho_i}$.
\end{lemma}
\begin{proof}
Fix an assignment with vertex density $\*\rho$, so that there
are ${S}\rho_i$ half-edges of spin $i$.
Each pairing has $2^{{S}/2}({S}/2)!$ representations obtained
by ordering its pairs and orienting each pair.
Expanding $(\*X^\trans P\*X)^{{S}/2}$ and assigning the labeled
half-edges to their spin positions therefore gives
$$
\sum_{\mathcal M}
\prod_{\{u,v\}\in\mathcal M}P_{\sigma_u\sigma_v}
=
\frac{\prod_{i=1}^m({S}\rho_i)!}{2^{{S}/2}({S}/2)!}
[\*X^{{S}\*\rho}](\*X^\trans P\*X)^{{S}/2},
$$
where $\mathcal M$ ranges over all pairings and $\sigma_u$
denotes the spin inherited by half-edge $u$.
Divide by
$({S}-1)!!={S}!/[2^{{S}/2}({S}/2)!]$,
multiply by the number $n!/\prod_i(n\rho_i)!$ of assignments,
and include the normalization
$\prod_i\pi_i^{-({S}-n)\rho_i}$.
This yields the asserted identity.
\end{proof}

% END INCLUDED FILE: appendices/b-analytic-tools.tex
% BEGIN INCLUDED FILE: appendices/a-moment-bounds.tex
\section{Proofs of the Moment Bounds}\label{sec:moment_bounds}\label{app:appendix-moments}
{We prove \Cref{prop:profile-moments} and \Cref{lem:local-profile-moments} by reducing the Bethe rate functions to the vertex density and applying coefficient extraction and Laplace summation.}

\subsection{{Profile Rate Function} and Derivatives}\label{subsec:common-local-calc}
We treat the hardcore and Potts models and their tensor products
at two possibly different parameter values in a common framework.
For the first-moment model, the parameter is $\zeta=z$ and the vertex and edge marginals
are $(\*\pi(z),P(z))$; for the paired (second-moment) model, $\zeta=(z,w)$ and they are $\bigl(\*\pi(z)\otimes\*\pi(w),\,P(z)\otimes P(w)\bigr)$.
In either case, we denote these marginals by
$\*\pi(\zeta)$ and $P(\zeta)$.

Following the notation in \Cref{def:bethe-canonical}, \Cref{def:profiles-bethe-rates}, and \Cref{def:local-second-moment}, assume that
$\*\pi(\zeta)$ and $P(\zeta)$ are holomorphic near a real base point
$t_0$, take real values at real parameters, and satisfy
\[
\allone^\trans\*\pi=1,\qquad P=P^\trans,\qquad P\allone=\*\pi,
\qquad \*\pi(t_0)>0,\quad P(t_0)\ge0,
\]
where $\zeta\in\=C^{k_\zeta}$ and $k_\zeta=1$ for the first-moment model and $k_\zeta=2$ for the paired model.
Let $m$ denote the dimension of $\*\pi$, and write
\[
A=\diag(\*\pi),\qquad
K=A^{-1}P,\qquad
R=K-\allone\*\pi^\trans,\qquad
M=A^{-1}PA^{-1}.
\]
For the local calculations, assume that the nonconstant eigenvalues
$\kappa_a$ of $K(t_0)$ satisfy
\begin{equation}\label{eq:local-spectrum}
-1<\kappa_a<\frac1{d-1}\qquad(a\ge2).
\end{equation}

For the hardcore and Potts models considered here,
\eqref{eq:local-spectrum} holds for both $K_t$
at $t$ and the paired kernel $K_t\otimes K_t$ at $(t,t)$ under
\Cref{cond:strict-KS}.
Indeed, the nonconstant eigenvalues of $K_t$ are nonpositive
by \eqref{eq:hardcore_kappa} and \eqref{eq:Potts_kappa},
while the paired nonconstant eigenvalues are
\[
\kappa_a(t),\qquad \kappa_b(t),\qquad
\kappa_a(t)\kappa_b(t),\qquad a,b\ge2.
\]
Nonpositivity and \Cref{cond:strict-KS} place all these
eigenvalues in $(-1,1/(d-1))$.

\paragraph{Decomposition of $\Phi$.}
At a real parameter $\zeta$ with $P(\zeta)\ge0$, fix the vertex marginal distribution $\*\rho$. Maximizing
\[
 \Phi_\zeta(\*\rho,Q)
 =(d-1)\DKL(\*\rho\Vert\*\pi)-\frac d2\DKL(Q\Vert P)
\]
over symmetric nonnegative $Q$ with $Q\allone=\*\rho$ amounts to
minimizing $\DKL(Q\Vert P)$. Near $\*\pi(t_0)$, define
\begin{equation}\label{eq:def-of-Q-star}
 L=\sum_{i,j}P_{ij}\e^{s_i+s_j},\qquad
 Q^*_{ij}=\frac{P_{ij}\e^{s_i+s_j}}L,\qquad
 Q^*\allone=\*\rho,\quad s_m=0.
\end{equation}
The row equation $Q^*\allone = \*\rho$ determines holomorphic multipliers locally, as
justified below.

At real parameters with $P\ge0$, direct substitution gives
\begin{equation}\label{eq:decomposition}
 \Phi_\zeta(\*\rho,Q)=f(\*\rho,\zeta)-\frac d2\DKL(Q\Vert Q^*),
 \qquad
 f\defeq(d-1)\DKL(\*\rho\Vert\*\pi)-d\*\rho^\trans\*s+\frac d2\log L.
\end{equation}
Thus $Q^*$ is the unique optimizer at fixed $\*\rho$, and the remaining
dependence on $\*\rho$ is contained in $f$.
Since \(L=1\) at \((\*\rho,\zeta)=(\*\pi(t_0),t_0)\), we may choose local holomorphic logarithms so that \(f\) is holomorphic near \((\*\pi(t_0),t_0)\).

\begin{lemma}[Derivatives of {the profile rate function} $f$]\label{lem:gradient}
The multipliers $\*s(\*\rho,\zeta)$ in \eqref{eq:def-of-Q-star} exist and depend holomorphically on $(\*\rho,\zeta)$
near $(\*\pi(t_0),t_0)$.
For $\*h\in\allone^\perp\defeq\{\*h\in\=R^m:\allone^\trans\*h=0\}$ (with $\log(\*\rho/\*\pi)\defeq(\log(\rho_i/\pi_i))_{i=1}^m$ taken entrywise),
\begin{align}
 D_{\*\rho} f(\*\rho,\zeta)[\*h]
 &=\left((d-1)\log\frac{\*\rho}{\*\pi}
                  -d\*s\right)^\trans \*h,
 \label{eq:grad_f}\\
 D_{\*\rho}^2f(\*\pi,\zeta)[\*h,\*h]
 &=-\*h^\trans H\*h,\qquad
 H=d(A+P)^{-1}-(d-1)A^{-1}. \label{eq:hessian}
\end{align}
Moreover, $f(\*\pi(\zeta),\zeta)=0$, its constrained first derivative vanishes
there, and $H$ is positive definite on $\allone^\perp$ at $\zeta=t_0$.
\end{lemma}

\begin{proof}
The Jacobian of the first $m-1$ row equations in
\eqref{eq:def-of-Q-star} is the leading principal minor $V'$ of
\[
 V=\diag(\*\rho)+Q^*-2\*\rho\*\rho^\trans.
\]
At $(\*\rho,\zeta) = (\*\pi(t_0),t_0)$, $V$ has kernel
$\operatorname{span}\{\allone\}$ and is positive definite on
$\allone^\perp$, since $\kappa_a>-1$. Hence $V'\succ0$, and the
implicit function theorem gives the existence of multipliers.

At fixed $\zeta$, differentiation along $\*h\in\allone^\perp$ gives $D_{\*\rho}\log L[\*h]=2\*\rho^\trans D_{\*\rho}\*s[\*h]$, which yields \eqref{eq:grad_f}. At $\*\rho=\*\pi(\zeta)$, we have
$\*s=\*0$, $Q^*=P$, and differentiation of the row equations gives
\[
 D_{\*\rho}\*s(\*\pi,\zeta)[\*h]=(A+P)^{-1}\*h+c_{\*h}\allone.
\]
This proves \eqref{eq:hessian}.
At $(\*\rho,\zeta) = (\*\pi(t_0),t_0)$, with $\sqrt{\*\pi}\defeq(\sqrt{\pi_i})_{i=1}^m$, the eigenvalues of $A^{1/2}HA^{1/2}$ on $\sqrt{\*\pi}^{\,\perp}$ are $[1-(d-1)\kappa_a]/(1+\kappa_a)>0$ by~\eqref{eq:local-spectrum}.
Since $A^{-1/2}\*h\perp\sqrt{\*\pi}$ whenever $\*h\perp\allone$, the claimed positivity follows.
\end{proof}

\subsection{{Local Quadratic Gap and Laplace Asymptotics}}

\begin{lemma}[Local quadratic gap]\label{lem:local-gap}
Under \eqref{eq:local-spectrum}, there exist neighborhoods $B_0$ of
$\*\pi(t_0)$ and $U$ of $t_0$, and $c>0$, such that for real
$\zeta\in U$ with $P(\zeta)\ge0$, $\*\rho\in B_0$, and every edge marginal $Q$ with vertex marginal $\*\rho$,
\[
 \Phi_\zeta(\*\rho,Q)\le-c\left(
 \|\*\rho-\*\pi(\zeta)\|_2^2+\|Q-P(\zeta)\|_{\mathrm F}^2\right).
\]
\end{lemma}

\begin{proof}
Writing $\*h=\*\rho-\*\pi(\zeta)$, \Cref{lem:gradient},
Taylor expansion, and positive definiteness {of $H$ on $\allone^\perp$} and continuity of $H$ give that on sufficiently small
real neighborhoods, the following holds
\begin{align*}
  f(\*\rho,\zeta) &=f(\*\pi(\zeta),\zeta)+D_{\*\rho} f(\*\pi(\zeta),\zeta)[\*h]
  +\frac12D_{\*\rho}^2f(\*\pi(\zeta),\zeta)[\*h,\*h]
  +O(\|\*h\|_2^3)\\
  &=-\frac12\*h^\trans H \*h+O(\|\*h\|_2^3) \le-c_1\|\*\rho-\*\pi(\zeta)\|_2^2.
\end{align*}
Since $Q^*(\*\pi(\zeta),\zeta)=P(\zeta)$ and $Q^*$ is smooth
in $\*\rho$, Pinsker's inequality and the triangle inequality give
\begin{align*}
\DKL(Q\Vert Q^*)
&\ge \frac12\left(\sum_{i,j}|Q_{ij}-Q^*_{ij}|\right)^2
 \ge \frac12\|Q-Q^*\|_{\mathrm F}^2,\\
\|Q-P\|_{\mathrm F}
&\le \|Q-Q^*\|_{\mathrm F} 
 +\|Q^*(\*\rho,\zeta)-Q^*(\*\pi(\zeta),\zeta)\|_{\mathrm F}\\
&\le \|Q-Q^*\|_{\mathrm F}
 +C\|\*\rho-\*\pi(\zeta)\|_2,
\end{align*}
where $C$ is the uniform local Lipschitz constant.
Combining these inequalities with \eqref{eq:decomposition}, we complete the proof of the desired bound.
\end{proof}

Write $\mathrm{CM}=\mathrm{CM}_{n,d}$ for the degree-$d$ configuration model (\Cref{def:config-model}). For $G\sim\mathrm{CM}$ and empirical vertex marginal $\*\rho(\sigma)$ (\Cref{def:profiles-bethe-rates}), recall from \Cref{lem:polynomial-representation} in \Cref{app:tools} the fixed-density weight
\[
 W_n(\*\rho,\zeta)=\E[\mathrm{CM}]{
       \sum_{\sigma:\*\rho(\sigma)=\*\rho}\prod_v\pi_{\sigma_v}
                         \prod_{uv\in E}M_{\sigma_u\sigma_v}}.
\]
Use independent coordinates $\*\rho'=(\rho_1,\ldots,\rho_{m-1})^\trans$,
with $\rho_m=1-\sum_{i<m}\rho_i$, and call the preimage in $\Delta_{[m]}$ of an open axis-aligned box in these $m-1$ coordinates (or its image in $\=R^{m-1}$) a \emph{density box}. Let $V=\diag(\*\rho)+Q^*(\*\rho,\zeta)-2\*\rho\*\rho^\trans$, and let $V'$ be the submatrix formed by its first $m-1$ rows and columns.

\begin{lemma}[Local moment asymptotics]\label{lem:local-moments}
Under \eqref{eq:local-spectrum}, for every sufficiently small fixed
density box $B$ containing $\*\pi(t_0)$, there is a complex neighborhood
$U$ of $t_0$ such that
\begin{align}
 W_n(\*\rho,\zeta)
 &=\frac{\e^{nf(\*\rho,\zeta)}}
 {(2\pi n)^{(m-1)/2}\sqrt{\det V'(\*\rho,\zeta)}}(1+o(1)),
 \label{eq:asym-W-rho}\\
 \sum_{\*\rho'\in B\cap n^{-1}\=Z^{m-1}}W_n(\*\rho,\zeta)
 &\longrightarrow\det(I-(d-1)R_\zeta)^{-1/2}.
 \label{eq:asym-W-sum}
\end{align}
The first estimate is uniform over density lattice points in $B$;
both are locally uniform in $\zeta$. Square roots are positive at
the real base point.
\end{lemma}

\begin{proof}
Put ${S}=dn$ and use the coefficient representation of $W_n(\*\rho,\zeta)$ in \Cref{lem:polynomial-representation}. The coefficient notation $[\*X^{{S}\*\rho}]$ is defined there.
Set $\vartheta_m=0$ and $\chi_{Q^*}(\*\vartheta)=\sum_{i,j}Q^*_{ij}\e^{\mathrm i(\vartheta_i+\vartheta_j)}$.
Applying the iterated Cauchy coefficient formula (\Cref{thm:cauchy-coeff} in \Cref{app:tools}) with $u_i=\e^{s_i}$ gives
\begin{equation}\label{eq:Cauchy_coeff_formula}
[\*X^{{S}\*\rho}](\*X^\trans P\*X)^{{S}/2}
=\frac{L^{{S}/2}\e^{-{S}\*\rho^\trans\*s}}{(2\pi)^{m-1}}
\int_{[-\pi,\pi]^{m-1}}
\chi_{Q^*}(\*\vartheta)^{{S}/2}
\e^{-\mathrm i{S}\*\rho'^\trans\*\vartheta}\dif\*\vartheta.
\end{equation}
$Q^*\allone=\*\rho$ implies $\chi_{Q^*}(\*0)=1$ and
$\nabla\chi_{Q^*}(\*0)=2\mathrm i\*\rho'$. Therefore,
\begin{equation}\label{eq:taylor_chi}
\frac12\log\chi_{Q^*}(\*\vartheta)
-\mathrm i\*\rho'^\trans\*\vartheta
=-\frac12\*\vartheta^\trans V'\*\vartheta
+O(\|\*\vartheta\|_2^3).
\end{equation}
The graph with vertex set $[m]$ and an edge between $i$ and $j$ whenever $P_{ij}(t_0)>0$ (including a loop when $i=j$) is connected and nonbipartite by~\eqref{eq:local-spectrum}.
$|\chi_{P(t_0)}(\*\vartheta)|=1$ then forces the phases
$\vartheta_i+\vartheta_j$ for all pairs with $P_{ij}(t_0)>0$ to agree modulo $2\pi$; an odd cycle and connectedness then force $\*\vartheta=\*0$ modulo $2\pi$.
Therefore, $|\chi_{P(t_0)}(\*\vartheta)|\le1$, with equality
if and only if $\*\vartheta=\*0$ in $[-\pi,\pi]^{m-1}$.

Fix a sufficiently small $\varepsilon>0$.
By compactness and continuity, after shrinking the complex
neighborhood of $(\*\pi(t_0),t_0)$, there exists
$q_\varepsilon\in(0,1)$ such that
\[
|\chi_{Q^*}(\*\vartheta)|
\e^{2\Im(\*\rho')^\trans\*\vartheta}
\le q_\varepsilon
\qquad
\bigl(\*\vartheta\in[-\pi,\pi]^{m-1},\
\|\*\vartheta\|_2\ge\varepsilon\bigr).
\]
Therefore, uniformly on this neighborhood,
\[
\int_{\substack{
\*\vartheta\in[-\pi,\pi]^{m-1}\\
\|\*\vartheta\|_2\ge\varepsilon}}
|\chi_{Q^*}(\*\vartheta)|^{{S}/2}
\e^{{S}\Im(\*\rho')^\trans\*\vartheta}\dif\*\vartheta
\le(2\pi)^{m-1}q_\varepsilon^{{S}/2}
=o({S}^{-(m-1)/2}).
\]

Set $\*\vartheta=\*u/\sqrt {S}$.
For sufficiently small fixed $\varepsilon>0$, $\Re V'\succ0$
gives Gaussian domination on $\|\*u\|_2\le\varepsilon\sqrt {S}$.
Then by \eqref{eq:taylor_chi}, dominated convergence
and the negligible contribution away from zero give
\begin{align*}
\int_{[-\pi,\pi]^{m-1}}
\chi_{Q^*}(\*\vartheta)^{{S}/2}
\e^{-\mathrm i{S}\*\rho'^\trans\*\vartheta}\dif\*\vartheta
&={S}^{-(m-1)/2}\left(
(1+o(1))
\int_{\=R^{m-1}}\e^{-\frac12\*u^\trans V'\*u}\dif\*u
+o(1)
\right)\\
&=\frac{(2\pi/{S})^{(m-1)/2}}{\sqrt{\det V'}}(1+o(1)).
\end{align*}
Stirling's formula then yields \eqref{eq:asym-W-rho}.

For summation, extend $W_n$ holomorphically in $\*\rho$ using the
fixed-domain integral above and Gamma functions in place of factorials.
The same estimates hold on a fixed complex neighborhood, so
\[
g_n(\*\rho',\zeta)
\defeq(2\pi n)^{(m-1)/2}\e^{-nf(\*\rho,\zeta)}W_n(\*\rho,\zeta)
\longrightarrow(\det V')^{-1/2}
\]
locally uniformly there.
\Cref{lem:gradient} supplies a stationary point
$\*\pi(\zeta)$ with value zero and constrained negative Hessian
$H'(\zeta)=-D_{\*\rho'}^2f(\*\pi(\zeta),\zeta)$,
with $H'(t_0)\succ0$.
The complex discrete Laplace lemma
(\Cref{lem:complex-discrete-laplace} in \Cref{app:tools}) gives the limit
$(\det V'(\*\pi(\zeta),\zeta)\det H'(\zeta))^{-1/2}$.
To evaluate this product at $\*\rho=\*\pi(\zeta)$, let $E=\begin{pmatrix}I_{m-1}\\ \*0^\trans\end{pmatrix}$, $J=\begin{pmatrix}I_{m-1}\\ -\allone_{m-1}^\trans\end{pmatrix}$, $C=(E\;\;\allone)$, and $L=(J\;\;\*\pi)$, so $\det C=\det L=1$.
Since $V=A+P-2\*\pi\*\pi^\trans$ satisfies $V\allone=\*0$ and $V^\trans=V$, we have
\[
  C^\trans(A+P)C=C^\trans(V+2\*\pi\*\pi^\trans)C
  =\begin{pmatrix}V'+2\*\pi'\*\pi'^\trans & 2\*\pi'\\ 2\*\pi'^\trans & 2\end{pmatrix},
\]
whence $\det(A+P)=2\det V'$.

Meanwhile, $(A+P)\allone=2\*\pi$ and $A\allone=\*\pi$ give $H\*\pi=\frac{2-d}{2}\allone$; since $J^\trans\allone=\*0$, $\*\pi^\trans\allone=1$, and $H'=J^\trans HJ$, we have $L^\trans HL=\diag(H',\frac{2-d}{2})$ and $\det H=\frac{2-d}{2}\det H'$.
Finally, $H(A+P)=dI-(d-1)A^{-1}(A+P)=I-(d-1)K_\zeta=(I-(d-1)\allone\*\pi^\trans)(I-(d-1)R_\zeta)$ (as $\*\pi^\trans R_\zeta=\*0^\trans$), so taking determinants gives $(2-d)\det V'\det H'=\det(A+P)\det H=(2-d)\det(I-(d-1)R_\zeta)$ and hence $\det V'\det H'=\det(I-(d-1)R_\zeta)$.
\end{proof}

\subsection{Proofs of \Cref{prop:profile-moments} and \Cref{lem:local-profile-moments}}\label{subsec:common-moment-proof}

\begin{proof}[Proof of~\Cref{prop:profile-moments}]
  Each assignment with feasible profile $(\*\rho,Q)$ has weight $\prod_i\pi_i^{n\rho_i}\prod_{i,j}M_{ij}^{dnQ_{ij}/2}$.
  Therefore, multiplying the weight by~\eqref{eq:pairing-type-count} of \Cref{lem:pairing-enumeration} in \Cref{app:tools} and applying uniform Stirling bounds (including when entries vanish) yields $\Phi_s(\*\rho,Q)$ in~\eqref{eq:planted-profile-gap}.

  {For result~(2), for every real $s$ near $t$ with $\*\pi(s)>0$ and $P(s)\ge0$, the nonconstant eigenvalues of $K_s$ remain nonpositive (by~\eqref{eq:hardcore_kappa} and~\eqref{eq:Potts_kappa}, or by continuity when $\kappa_a(t)<0$), which gives $\*\rho^\trans M(s)\*\rho\le1$ for all $\*\rho\in\Delta_\Omega$. If $Q$ charges a zero entry of $P(s)$, then $\Phi_s(\*\rho,Q)=-\infty$; otherwise, since $Q$ has row and column marginals $\*\rho$, applying the entropy variational inequality to $(Q_{ij})_{i,j}$ and $(\pi_i(s)^{-1}\rho_i\,\pi_j(s)^{-1}\rho_j)_{i,j}$ gives
  \[
    \DKL(Q\Vert P(s))
    \ge 2\DKL(\*\rho\Vert\*\pi(s))-\log\bigl(\*\rho^\trans M(s)\*\rho\bigr)
    \ge 2\DKL(\*\rho\Vert\*\pi(s)).
  \]
  Since $d\ge3$, substituting this lower bound into $\Phi_s(\*\rho,Q)=(d-1)\DKL(\*\rho\Vert\*\pi(s))-\frac d2\DKL(Q\Vert P(s))$ yields both
  \[
    \Phi_s(\*\rho,Q)\le-\DKL(\*\rho\Vert\*\pi(s))
    \qquad\text{and}\qquad
    \Phi_s(\*\rho,Q)\le-\frac12\DKL(Q\Vert P(s)).
  \]
  Averaging the two bounds and applying Pinsker's inequality gives
  \[
    \Phi_s(\*\rho,Q)
    \le-\frac12\DKL(\*\rho\Vert\*\pi(s))-\frac14\DKL(Q\Vert P(s))
    \le-\frac18\left(
      \|\*\rho-\*\pi(s)\|_2^2+
      \|Q-P(s)\|_{\mathrm F}^2
    \right),
  \]
  establishing result~(2) with $c=1/8$ (including at the endpoint $q=2,\beta=0,\kappa_2=-1$).}
  Summing over polynomially many profiles then gives
  $\log\E{\mathcal Z_G(s)}\le O(\log n)$.
\end{proof}

\begin{proof}[Proof of~\Cref{lem:local-profile-moments}]
Under~\Cref{cond:strict-KS}, the local spectrum condition~\eqref{eq:local-spectrum} holds.
Take the intersection of the neighborhoods $U$ in~\Cref{lem:local-gap} and~\Cref{lem:local-moments}.
Now the two common lemmas apply on a common disk $D_0$
and product neighborhood $D_0\times D_0$.

\paragraph{Proof of result (2).}
Apply \Cref{lem:local-gap} to the paired marginals with $\*\rho\in B$.

\paragraph{Proof of result (1).}
Choose a small density box $B_1$ around $\*\pi(t)$.
At complex $z$, each assignment and compatible pairing with feasible profile $(\*\rho,Q)$ has complex weight of modulus $\bigl(\prod_i|\pi_i(z)|^{n\rho_i}\bigr)\allowbreak\bigl(\prod_{i,j}|M_{ij}(z)|^{dnQ_{ij}/2}\bigr)$. Multiplying this modulus by the pairing enumeration formula~\eqref{eq:pairing-type-count} of \Cref{lem:pairing-enumeration} and applying uniform Stirling bounds (including when entries of $\*\rho$ or $Q$ vanish) shows that the expected absolute contribution of each feasible profile $(\*\rho,Q)$ is at most $\exp\{n\Psi_z(\*\rho,Q)+O(\log n)\}$ uniformly in $(\*\rho,Q)$ and $z$, where
\[
 \Psi_z(\*\rho,Q)
 =(d-1)\sum_i\rho_i\log\rho_i
 -\frac d2\sum_{ij}Q_{ij}\log Q_{ij}
 +\sum_i\rho_i\log|\pi_i(z)|
 +\frac d2\sum_{ij}Q_{ij}\log|M_{ij}(z)|.
\]
With $0\log 0=0$ and $x\log 0=-\infty$ for $x>0$, each term $Q_{ij}\log|M_{ij}(z)|$ is jointly upper semicontinuous on $Q_{ij}\ge0$ (being continuous where $M_{ij}(z)\ne0$ and nonpositive near zeros of $M_{ij}$), so $\Psi_z(\*\rho,Q)$ is jointly upper semicontinuous in $(\*\rho,Q,z)$ and equals $\Phi_t(\*\rho,Q)$ at $z=t$.
By \Cref{prop:profile-moments}(2), $\Phi_t(\*\rho,Q)\le-2\eta<0$ on the compact set of profiles with $\*\rho\notin B_1$; upper semicontinuity and compactness therefore give $\Psi_z(\*\rho,Q)\le-\eta<0$
outside $B_1$, uniformly on a smaller disk $D_0$ around $t$.
Since there are only polynomially many feasible profiles $(\*\rho,Q)$, summing $\exp\{n\Psi_z(\*\rho,Q)+O(\log n)\}$ over all feasible profiles $(\*\rho,Q)$ with $\*\rho\notin B_1$ bounds the omitted contribution to $\sum_{\*\rho} W_n(\*\rho,z)$ uniformly by $\operatorname{poly}(n)\cdot\e^{-\eta n} = o(1)$.
\Cref{lem:local-moments} therefore yields
\[
 a_n(z) = \sum_{\*\rho^\trans\allone = 1} W_n(\*\rho,z) = o(1) + \sum_{\*\rho \in B_1\cap n^{-1}\=Z^{m},\,\*\rho^\trans\allone = 1} W_n(\*\rho,z)\longrightarrow
 a(z)=\det(I-(d-1)R_z)^{-1/2}.
\]
For the paired model the sum is already restricted to $B$.
Its centered kernel is
\[
 R^{(2)}
 =R_z\otimes(\allone\*\pi(w)^\trans)+(\allone\*\pi(z)^\trans)\otimes R_w+R_z\otimes R_w.
\]
The determinant factors over the subspaces spanned by $\allone$ and consisting of vectors $\*v$ with $\*\pi(z)^\trans\*v=0$, respectively, for each tensor factor (using $\*\pi(w)$ for the second factor).
The same lemma consequently gives
\begin{equation}\label{eq:paired-limit}
 \frac{\E{\mathcal Z^{(2)}_{G,B}(z,w)}}{a_n(z)a_n(w)}
 \longrightarrow
 \det(I-(d-1)(R_z\otimes R_w))^{-1/2}.
\end{equation}
All convergence is locally uniform. After shrinking $D_0$, $a$ is
nonvanishing, and $a_n$ is nonvanishing there for sufficiently large $n$.

\paragraph{Proof of result (3).}
Recalling the short-cycle counts $X_{\ell,n}$, Poisson means $\lambda_\ell$, and short-cycle factors $\delta_\ell(z)$ from \Cref{def:short-cycle-factors} and the trace identity $\delta_\ell(z)=\operatorname{tr}K_z^\ell-1=\operatorname{tr}R_z^\ell$ from~\eqref{eq:cycle-factor-trace}, for real $s\in D_0$ with $P(s)\ge0$, the log-determinant series in
\eqref{eq:paired-limit} gives
\[
 -\tfrac12\log\det(I-(d-1)(R_s\otimes R_s))
 =\sum_{\ell\ge1}\frac{(d-1)^\ell}{2\ell}
                  (\operatorname{tr}R_s^\ell)^2
 =\sum_{\ell\ge1}\lambda_\ell\delta_\ell(s)^2.
\]
For the cycle statements, fix $j_1,\ldots,j_L$ and put
$C_{L,n}=\prod_{\ell\le L}(X_{\ell,n})_{j_\ell}$.
This counts ordered tuples of distinct cycles, which may intersect.
A union of intersecting distinct cycles has $e>v$ and contributes
$O(n^vn^{-e})=O(n^{-1})$, where $e$ and $v$ are the number of edges and vertices of the union of cycles, respectively.
Use the underlying uniform labeled pairing $\mathcal M$ that generates $G\sim\mathrm{CM}_{n,d}$. A cycle candidate is specified by an ordered list of distinct vertices and an ordered pair of distinct labeled half-edges at each vertex, with two specifications identified if they differ only by a change of starting vertex or a reversal of direction (with the order of each half-edge pair reversed). It occurs when the first half-edge at each vertex is paired to the second half-edge at the next vertex in the cyclic order. For lengths $1$ and $2$, these candidates describe loops and pairs of parallel edges, respectively. Put $T=\sum_{\ell=1}^L\ell j_\ell$. Summing over ordered families of cycle candidates gives
\begin{align*}
  \E{C_{L,n}} &= \sum_{\gamma_{1,1},\ldots,\gamma_{L,j_L}}\Pr{\forall i\in[L], k\in [j_i], \gamma_{i,k}\text{ is a cycle of length }i}
  \\&= O(n^{-1}) + \sum_{\text{disjoint }\gamma_{1,1},\ldots,\gamma_{L,j_L}}\Pr{\forall i\in[L], k\in [j_i], \gamma_{i,k}\text{ is a cycle of length }i}
  \\ &=O(n^{-1}) + \underbrace{\frac{(n)_T[d(d-1)]^T} {\prod_{\ell=1}^L(2\ell)^{j_\ell}}}_{\substack{\text{number of vertex-disjoint}\\
             \text{cycle candidates}}}
 \times \underbrace{\prod_{r=0}^{T-1}\frac{1}{dn-2r-1}}_{\substack{\text{probability that all}\\
             \text{prescribed pairs occur}}}\\
&=\prod_{\ell=1}^L
  \left(\frac{(d-1)^\ell}{2\ell}\right)^{j_\ell}
  +O(n^{-1})
=\prod_{\ell=1}^L\lambda_\ell^{j_\ell}+O(n^{-1}).
\end{align*}
The joint factorial-moment criterion gives the Poisson limits.

For weighted moments, work under the planted law $\mathbb Q_n$ (\Cref{def:planted-overlap}) at a real parameter $s$, retaining the underlying pairing used to generate the graph. Its weight depends on the pairing only through the induced graph, so conditional on a spin assignment and its profile $(\*\rho,Q)$, compatible pairings are uniform.
For a cycle with spins $a_1,\ldots,a_\ell$, the preceding count
gives
$$
\lambda_\ell\prod_{r=1}^{\ell}
\frac{Q_{a_ra_{r+1}}}{\rho_{a_r}}+O(n^{-1}),
\qquad a_{\ell+1}=a_1.
$$
Summing over the spins yields
$\lambda_\ell\operatorname{tr}(K_{\*\rho,Q}^{\ell})+O(n^{-1})$,
where $K_{\*\rho,Q}=\diag(\*\rho)^{-1}Q$.
Counting vertex-disjoint cycles jointly and bounding overlapping
configurations as above therefore gives
$$
  \E[\mathbb Q_n]{C_{L,n}\mid\*\rho,Q}=
\prod_{\ell=1}^L
\left[\lambda_\ell
\operatorname{tr}(K_{\*\rho,Q}^{\ell})\right]^{j_\ell}
+O(n^{-1}),
$$
uniformly near $(\*\pi(s),P(s))$.
The pairing enumeration gives additive error bounds, so this estimate remains valid when some entries of $Q$ vanish.

\Cref{prop:profile-moments} and $a_n(s)\to a(s)>0$ imply exponential
concentration of $(\*\rho,Q)$ at $(\*\pi(s),P(s))$ under the
planted law. Since $C_{L,n}$ is polynomially bounded, profiles
outside a fixed neighborhood contribute $o(1)$.
Averaging the conditional estimate and using continuity,
$K_{\*\pi(s),P(s)}=K_s$, and
$\operatorname{tr}(K_s^\ell)=1+\delta_\ell(s)$, we obtain
$$
\frac{\E{\+Z_G(s)C_{L,n}}}{a_n(s)}
=\E[\mathbb Q_n]{C_{L,n}}
\longrightarrow
\prod_{\ell=1}^L
\left[\lambda_\ell(1+\delta_\ell(s))\right]^{j_\ell}.
$$

\paragraph{Proof of result (4).}
Nonnegativity of $K_t$ and nonpositivity of its nonconstant eigenvalues give
$$
  0\leq \operatorname{tr}(K_t) = 1 + \sum_{a\ge2}\kappa_a(t)\implies \sum_{a\ge2}|\kappa_a(t)|\le1 .
$$
Thus, with $b_0=\max_{a\ge2}|\kappa_a(t)|<(d-1)^{-1/2}<1$ by \Cref{cond:strict-KS},
\begin{equation}\label{eq:1+delta_nonzero}
  1+\delta_\ell(t)
 \ge1-\sum_{a\ge2}|\kappa_a(t)|^\ell
 \ge1-b_0^{\ell-1}>0\qquad(\ell\ge2).
\end{equation}
\Cref{cond:strict-KS} and continuity give
$|\delta_\ell(z)|\le(|\Omega|-1)b^\ell$ uniformly on a small $D_0$,
for some $b<(d-1)^{-1/2}$.
This proves
$$
  \sum_{\ell\ge1}\lambda_\ell\sup_{D_0}|\delta_\ell|^2 \leq \frac {(|\Omega|-1)^2}2\sum_{\ell \ge 1}\frac{\tp{b^2(d-1)}^\ell}{\ell}<\infty
$$
and $|\delta_\ell|<1$ for sufficiently large $\ell$.
By continuity and~\eqref{eq:1+delta_nonzero}, shrinking $D_0$ handles $1 + \delta_\ell(z)\neq 0$ for the remaining finitely many lengths $\ell\ge3$. \qedhere
\end{proof}

% END INCLUDED FILE: appendices/a-moment-bounds.tex
% BEGIN INCLUDED FILE: appendices/c-sec5-deduplicated.tex
\section{Planted Concentration, Pinning, and Truncation Proofs}\label{app:planted-converse-proofs}

{We establish the planted concentration, pinning, and truncation bounds used in \Cref{subsec:local-variational-geometric} and \Cref{sec:variational-rs}, including the quantitative error exponent near the coloring endpoint: \Cref{app:profile-conc} proves concentration of edit-Lipschitz observables on a planted profile (\Cref{lem:profile-concentration}); \Cref{app:partial-pinning} establishes planted partial-pairing profile concentration (\Cref{lem:rsb-planted-partial-profiles}) and finite-sample pinning (\Cref{lem:rsb-finite-pinning}); \Cref{app:revealed-entropy} proves the pinned log-likelihood and planted free-energy bound (\Cref{lem:revealed-entropy-bound}); and \Cref{app:coloring-endpoint} constructs the typical-graph truncation event (\Cref{lem:graph-truncation}).}

\subsection{Concentration of Edit-Lipschitz Observables on a Planted Profile}\label{subsec:app-profile-conc}\label{app:profile-conc}
{
Fix ${d}\ge1$ and a finite spin set ${\Omega}$. Consider a random triple $(G,\*\tau,\*\xi)$ where $\*\tau=(\tau_v)_{v\in[n]}\in{\Omega}^n$ assigns each vertex its spin, $G=G(\mathcal M)$ is the multigraph induced by a partial pairing $\mathcal M$ of the $dn$ labeled half-edges, and $\*\xi\in\{0,1\}^n$ is a vector of vertex \emph{marks} (or reveal indicators). Define the \emph{profile} $p=(\*\rho,W,\*u)$ of $(G,\*\tau)$ by
\begin{align*}
  \rho_\tau&\defeq\frac1n\sum_{v=1}^n\one{\tau_v=\tau},\\
  W_{\tau\tau'}&\defeq\frac1n\bigl(\text{number of directed paired edges of type }(\tau,\tau')\bigr),\\
  u_\tau&\defeq\frac1n\bigl(\text{number of unmatched half-edges at vertices of type }\tau\bigr),
\end{align*}
with loops contributing twice to $W_{\tau\tau}$. Thus $p$ determines the vertex-type counts, the paired-edge type counts, and the unmatched-half-edge counts; for a full pairing, $\*u=\*0$ and $W=dQ$, where $Q$ is the empirical distribution of the two endpoint types of a uniform directed edge. Assume that, conditional on $p$, the vertex-type assignment $\*\tau$ and compatible partial pairing are jointly uniform, while $\*\xi\in\{0,1\}^n$ has a product law independent of $(G,\*\tau)$. Assume also that there exist deterministic profiles $p_n^*$ and constants $C,c,A>0$ such that
\begin{equation}\label{eq:bounded-degree-profile-tail}
  \Pr{\|p-p_n^*\|_1>r}\le Cn^A\e^{-cnr^2}\qquad(r>0).
\end{equation}
For two realizations $x=(G,\*\tau,\*\xi)$ and $x'=(G',\*\tau',\*\xi')$, define their \emph{edit distance} by
\[
  d_{\mathrm{edit}}(x,x')
  \defeq |E(G)\triangle E(G')| + d_{\mathrm{Ham}}(\*\tau,\*\tau') + d_{\mathrm{Ham}}(\*\xi,\*\xi'),
\]
where $|E(G)\triangle E(G')|\defeq\sum_{\{u,v\}}\bigl|\operatorname{mult}_G(\{u,v\})-\operatorname{mult}_{G'}(\{u,v\})\bigr|$ counts edge-multiset discrepancies with multiplicity and $d_{\mathrm{Ham}}$ denotes Hamming distance. Call $H_n(G,\*\tau,\*\xi)$ \emph{$L$-edit-Lipschitz} if
\begin{equation}\label{eq:edit-lipschitz-def}
  \bigl|H_n(x)-H_n(x')\bigr|\le L\,d_{\mathrm{edit}}(x,x')
\end{equation}
for all supported realizations $x,x'$ of $(G,\*\tau,\*\xi)$ (i.e., realizations with positive probability).

\begin{lemma}[Concentration {of edit-Lipschitz observables}]\label{lem:profile-concentration}
Under the assumptions above, there exist constants $C',c',K>0$, depending only on ${d},|\Omega|,C,c,A$, such that every real-valued $L$-edit-Lipschitz observable $H_n$ with $L>0$ satisfies $\operatorname{Var}(H_n)=O(L^2n\log(n+1))$ and, whenever $a>0$ and $na^2/L^2\ge K\log(n+1)$,
\begin{equation}\label{eq:bounded-degree-functional-tail}
  \Pr{|H_n-\E{H_n}|>an}\le C'\exp\!\left(-c'\frac{na^2}{L^2}\right).
\end{equation}
\end{lemma}

\begin{proof}[Proof of~\Cref{lem:profile-concentration}]
Put $m_n(p)\defeq\E{H_n\mid p}$.
\emph{Conditional concentration.}
There exists $c_1>0$, depending only on ${d}$ and $|\Omega|$,
such that, for every supported profile $p$ and every $r>0$,
\begin{equation}\label{eq:rsb-profile-conditional-concentration}
  \Pr{|H_n-m_n(p)|>r\mid p}
  \le 2\exp\!\left(-\frac{c_1r^2}{nL^2}\right).
\end{equation}
Fix a vertex-type list and an abstract partial matching realizing $p$.
The conditional law is generated by a uniform vertex permutation,
independent uniform type-preserving bijections from vertex half-edges
to abstract half-edges, and independent marks on the final labels.
Indeed, each compatible pairing has the same number of representations.

Reveal these permutations, bijections, and marks one image at a time;
there are at most $({d}+2)n$ revelations. Two possible values at a
permutation revelation can be coupled by transposing two images in the
completion, producing two supported realizations of profile $p$. A vertex transposition changes at most $4{d}$ edge
occurrences and two types; a half-edge transposition changes at most
four edge occurrences; changing a mark changes one coordinate. Thus the two coupled completions have edit distance $d_{\mathrm{edit}}\le 4{d}+2$, so by~\eqref{eq:edit-lipschitz-def}
each difference in the Doob martingale is at most $C_{d} L$ with $C_{d}\defeq 4{d}+2$.
Azuma--Hoeffding gives~\eqref{eq:rsb-profile-conditional-concentration}
with $c_1=[2({d}+2)C_{d}^2]^{-1}$.

\emph{Comparison of conditional means.}
We use the type-preserving bijection coupling of~\cite[Section~12.7]{coja2020inference}.
There exists $C_1>0$, depending only on ${d}$ and $|\Omega|$, such that, for any two supported profiles $p,p'$,
\begin{equation}\label{eq:rsb-profile-mean-comparison}
  |m_n(p)-m_n(p')|
  \le C_1Ln\|p-p'\|_1.
\end{equation}
Write $p=(\*\rho,W,\*u)$ and $p'=(\*\rho',W',\*u')$. {Fix a total order on ${\Omega}$. For a profile $p$, fix a canonical vertex-type list of length $n$ with empirical distribution $\*\rho$ ($dn\rho_\tau$ labeled half-edges of type $\tau$), and create $nW_{\tau\tau'}$ labeled binary factors of type $(\tau,\tau')$ for each $\tau<\tau'$ (with one slot of type $\tau$ and one of type $\tau'$), $\frac12nW_{\tau\tau}$ labeled binary factors of type $(\tau,\tau)$ (with two ordered slots of type $\tau$), and $nu_\tau$ labeled unary factors of type $\tau$ (with one slot of type $\tau$). The total number of factor slots of type $\tau$ is $\sum_{\tau'}nW_{\tau\tau'}+nu_\tau=dn\rho_\tau$. Under the map that pairs the two half-edges assigned to each binary factor, leaves the half-edges assigned to unary factors unmatched, and forgets factor labels and internal slot orderings, every compatible partial pairing (including loops and parallel edges) has the same fiber size
\[
  \Bigl(\prod_{\tau<\tau'}(nW_{\tau\tau'})!\Bigr)\Bigl(\prod_{\tau\in{\Omega}}\bigl(\tfrac12nW_{\tau\tau}\bigr)!\,2^{nW_{\tau\tau}/2}(nu_\tau)!\Bigr).
\]
Consequently, a uniform type-preserving bijection from vertex half-edges to factor slots pushes forward to the uniform compatible partial pairing.

Choose the canonical vertex-type lists for $p$ and $p'$ to agree on $n\sum_\tau(\rho_\tau\wedge\rho'_\tau)=n(1-\frac12\|\*\rho-\*\rho'\|_1)$ vertices, and identify a maximum common subset of binary factors of each type $(\tau,\tau')$ and unary factors of each type $\tau$. This leaves $H\defeq dn\|\*\rho-\*\rho'\|_1$ half-edges at uncoupled vertices and
\[
  B\defeq\sum_{\tau<\tau'}n|W_{\tau\tau'}-W'_{\tau\tau'}|+\sum_{\tau\in{\Omega}}\tfrac12n|W_{\tau\tau}-W'_{\tau\tau}|+\sum_{\tau\in{\Omega}}n|u_\tau-u'_\tau|
  \le\frac n2\|W-W'\|_1+n\|\*u-\*u'\|_1
\]
uncoupled binary and unary factors (at most $2B$ uncoupled factor slots). The colored-bijection switching construction of~\cite[Section~12.7]{coja2020inference} couples the two uniform type-preserving bijections across each type $\tau\in{\Omega}$ so that they agree on all common half-edges except at most $H+2B$ displaced half-edges, and hence at most $H+2B$ binary factors disagree between the two realizations. Each disagreeing binary factor changes at most two paired-edge occurrences. Applying the same random vertex permutation and common independent marks $\*\xi$ to both canonical realizations produces coupled supported realizations $x=(G,\*\tau,\*\xi)$ and $x'=(G',\*\tau',\*\xi')$ whose vertex types differ at $\frac12n\|\*\rho-\*\rho'\|_1$ vertices and whose edge multisets satisfy $|E(G)\triangle E(G')|\le 2(H+2B)$. Thus $d_{\mathrm{edit}}(x,x')\le 2(H+2B)+\frac12n\|\*\rho-\*\rho'\|_1=O(n\|p-p'\|_1)$, so the $L$-edit-Lipschitz bound~\eqref{eq:edit-lipschitz-def} on supported realizations proves~\eqref{eq:rsb-profile-mean-comparison}.}

\emph{Combining the bounds.}
Put $D=\|p-p_n^*\|_1$.
Integrating the assumed profile tail~\eqref{eq:bounded-degree-profile-tail} gives
\[
  \E{D}\le C_0\sqrt{\frac{\log(n+1)}{n}},
\]
where $C_0$ depends only on $C,c,A$.
By~\eqref{eq:rsb-profile-mean-comparison}, averaging over an
independent copy $p'$ of the profile yields
\[
  |m_n(p)-\E{H_n}|
  \le C_1Ln\,\E[p']{\|p-p'\|_1}
  \le C_1Ln\bigl(D+\E{D}\bigr).
\]
Let $q=na^2/L^2$ and choose
$K\ge\max\bigl\{1,\,16C_1^2C_0^2,\,32AC_1^2/c\bigr\}$.
If $q\ge K\log(n+1)$, then
$\E{D}\le a/(4C_1L)$, so on $\{D\le a/(4C_1L)\}$ we have
$|m_n(p)-\E{H_n}|\le an/2$.
\eqref{eq:rsb-profile-conditional-concentration} and~\eqref{eq:bounded-degree-profile-tail} therefore give
\[
\begin{aligned}
  \Pr{|H_n-\E{H_n}|>an}
  &\le \Pr{|H_n-m_n(p)|>an/2}+\Pr{D>a/(4C_1L)}\\
  &\le 2\exp\!\left(-\frac{c_1q}{4}\right)+Cn^A\exp\!\left(-\frac{cq}{16C_1^2}\right)
   \le (2+C)\exp(-c'q),
\end{aligned}
\]
where $c'=\min\{c_1/4,c/(32C_1^2)\}$ and the last inequality uses $A\log n\le cq/(32C_1^2)$.
Taking $C'=2+C$ proves~\eqref{eq:bounded-degree-functional-tail}; integrating the tail yields $\operatorname{Var}(H_n)=O\bigl(L^2n\log(n+1)\bigr)$.
\end{proof}
}

\subsection{Planted Partial-Pairing Profiles and Finite-Sample Pinning}\label{subsec:app-partial-pinning}\label{app:partial-pinning}
{Fix an integer $d\ge3$ and one of the two models at a positive parameter $t$ ($\alpha\in(0,1/2)$ for hardcore or $\beta\in(0,1)$ for Potts), and suppress $t$ in $M=M(t)$ and $\*\pi=\*\pi(t)$. Put $S=dn$. For $0\le j\le\lfloor S/2\rfloor$, let $\mathrm{CM}_{n,d,j}$ be the law of a uniform partial pairing $\mathcal M$ with $j$ edges, obtained by selecting $2j$ of the $S$ labeled half-edges uniformly and pairing them uniformly. The subscript $j$ distinguishes this pairing law from the graph law $\mathrm{CM}_{n,d}$. Write $G=G(\mathcal M)$ for its induced multigraph and define the first moment $a_{n,j}$ and the planted law $\mathbb Q_{n,j}$ on $(\mathcal M,\sigma)$ by
\[
 a_{n,j}\defeq\E[\mathrm{CM}_{n,d,j}]{\mathcal Z_G},\qquad
 \frac{\dif\mathbb Q_{n,j}}{\dif(\mathrm{CM}_{n,d,j}\otimes\*\pi^{\otimes n})}(\mathcal M,\sigma)
 \defeq a_{n,j}^{-1}\prod_{uv\in E(G)}M_{\sigma_u,\sigma_v}.
\]
The conditional law of $\sigma$ given $\mathcal M$ or $G$ is the canonical law $\mu^{\mathrm{can}}_{G,t}$. At maximum pairing $j=\lfloor S/2\rfloor$, the graph pushforward of $\mathrm{CM}_{n,d,j}$ is $\mathrm{CM}_{n,d}$, $a_{n,j}=a_n$, and the pushforward of $\mathbb Q_{n,j}$ under $(\mathcal M,\sigma)\mapsto(G(\mathcal M),\sigma)$ is $\mathbb Q_n$. When $S$ is odd, one half-edge remains unmatched.

Put $\theta_j=2j/S$. Use the spin types $a\in\Omega$ from \Cref{lem:profile-concentration} and write $p=(\*\rho,W,\*u)$ for the resulting profile. For $j>0$, let $Q$ be the empirical type distribution of the endpoints of a uniformly directed paired edge; for $j=0$, set $Q=0$. When $2j<S$, let $\*\chi$ be the empirical spin distribution of the unmatched half-edges; when $2j=S$, set $\*\chi=\*\rho$. Then
\[
 \*\rho=\theta_j Q\allone+(1-\theta_j)\*\chi,\qquad
 W=d\theta_j Q,\qquad \*u=d(1-\theta_j)\*\chi.
\]
Under $\mathbb Q_{n,j}$, conditional on $p$, the spin assignment and compatible pairing are jointly uniform. With $P_{ab}=\pi_a\pi_bM_{ab}$, the profile center is
\[
 p^*=\bigl(\*\pi,\,d\theta_j P,\,d(1-\theta_j)\*\pi\bigr).
\]
At maximum pairing with odd $S$, the unmatched-half-edge coordinate has total mass $1/n$.

To analyze partial pairings, let $S\ge1$ labeled half-edges $u\in[S]$ have types $\tau_u$ in a finite set $\mathcal T_0$, with type counts $T=(T_\tau)_{\tau\in\mathcal T_0}$ and empirical law $\widehat{\*\rho}_0=T/S$. Let $\*p_0\in\Delta_{\mathcal T_0}$ be positive on the types present, let $M_0\in\=R_{\ge0}^{\mathcal T_0\times\mathcal T_0}$ be symmetric with $\gamma_0\defeq\*p_0^\trans M_0\*p_0>0$, and for $0\le j\le\lfloor S/2\rfloor$, let $\mathcal M$ be a uniform partial pairing with $j$ pairs. In the corresponding \emph{independent-pair experiment}, sample independently $j$ ordered type pairs with law $\widetilde Q_{\tau\tau'}\defeq(\*p_0)_\tau(\*p_0)_{\tau'}(M_0)_{\tau\tau'}/\gamma_0$ and $S-2j$ single types with law $\*p_0$, and let $\widetilde T$ be their total type counts.

\begin{lemma}[Independent-pair representation]\label{lem:independent-pair-representation}
For every event $A$ determined by the unordered pair-type and single-type counts,
\[
 \E{\one A\prod_{uv\in\mathcal M}(M_0)_{\tau_u,\tau_v}}
 =\gamma_0^j\frac{\Pr{A,\,\widetilde T=T}}
 {\Pr{\operatorname{Mult}(S,\*p_0)=T}}.
\]
If $\*p_0=\widehat{\*\rho}_0$, this is at most $(S+1)^{|\mathcal T_0|}\gamma_0^j\Pr{A}$.
\end{lemma}
\begin{proof}
Take a uniform ordering of the half-edges and pair the first $2j$
consecutively. Its type sequence has the law of independent $\*p_0$-types
conditioned on counts $T$. Multiplication by the pair weights changes
each pair law from $\*p_0\otimes\*p_0$ to $\widetilde Q$, with factor $\gamma_0$.
This proves the identity. When $\*p_0=\widehat{\*\rho}_0$, the denominator is at least
$(S+1)^{-|\mathcal T_0|}$ because $T=S\widehat{\*\rho}_0$ is a multinomial mode.
\end{proof}
Applying \Cref{lem:independent-pair-representation} conditional on the vertex spins yields both a polynomial bound on the partial-pairing first moment $a_{n,j}$ and a Gaussian tail for the empirical profile $p$ around $p^*$.
}

\begin{lemma}[Planted partial-pairing profiles]\label{lem:rsb-planted-partial-profiles}
{Fix $\theta_0\in(0,1)$. There exist $C,c,K>0$ such that for all sufficiently large $n$, uniformly over all $j$ with $\theta_j=2j/S\in[\theta_0,1]$,
\[
 |\log a_{n,j}|\le C\log(n+1),\qquad
 \Pr[\mathbb Q_{n,j}]{\|p-p^*\|_1>v}\le C(n+1)^K\exp(-cnv^2)\quad(v>0).
\]
}
\end{lemma}
\begin{proof}
{
Fix the vertex spins $\sigma$, and write $T_\tau$ for the number of
half-edges of spin $\tau$ and $\widehat{\*\rho}=T/S=\*\rho$. Set
$\gamma=\widehat{\*\rho}^{\trans}M\widehat{\*\rho}$.
Since $M\*\pi=\allone$ and the nonconstant eigenvalues of
$M\diag(\*\pi)$ are nonpositive, $0\le\gamma\le1$. If $\gamma=0$, every
partial pairing has zero weight. Otherwise let
\[
 \widetilde Q_{\tau\tau'}
 =\frac{\widehat\rho_\tau\widehat\rho_{\tau'}M_{\tau\tau'}}{\gamma},
\]
and consider the independent-pair experiment of
\Cref{lem:independent-pair-representation} with $\*p_0=\widehat{\*\rho}$.
For any event $A$
determined by the profile, that lemma bounds its weighted probability
by a polynomial factor times its probability in the independent-pair
experiment. The pair contribution to $Q$ counts both orientations.
Taking $A$ to be the whole space also gives
$a_{n,j}\le(n+1)^{O(1)}$.

Under independent vertex spins with law $\*\pi$, Hoeffding's inequality
gives a Gaussian tail on scale $n$ for $\|\*\rho-\*\pi\|_1$.
If this distance is small, then
$(\widehat{\*\rho},\gamma,\widetilde Q)$ is within $O(\|\*\rho-\*\pi\|_1)$
of $(\*\pi,1,P)$. The empirical pair law in the independent-pair
experiment has mean $\widetilde Q$ and, since $j\ge\theta_0S/2$,
also has Gaussian tails on scale $n$. Together with
$W=d\theta_j Q$ and
$u_a=d\rho_a-\sum_bW_{ab}$, this proves the
stated profile tail up to a polynomial factor divided by $a_{n,j}$.
The estimates are uniform over the stated pairing fractions, including maximum pairing with odd $S$.

For the lower bound on $a_{n,j}$, it is enough to count spin totals.
Fix an assignment $\sigma$ and let $T^{\mathrm{sp}}_a$ be the number of its half-edges
of spin $a$. Put $P_{ab}=\pi_aM_{ab}\pi_b$, and let $U^{\mathrm{sp}}$
be the vector of spin counts in $j$ independent pairs with law $P$ and
$S-2j$ independent single spins with law $\*\pi$. Since
$\*\pi^\trans M\*\pi=1$, \Cref{lem:independent-pair-representation} gives
\[
 \E[\mathrm{CM}_{n,d,j}]{\prod_{uv\in E(G)}M_{\sigma_u,\sigma_v}}
 =\frac{\Pr{U^{\mathrm{sp}}=T^{\mathrm{sp}}}}{\Pr{\operatorname{Mult}(S,\*\pi)=T^{\mathrm{sp}}}}.
\]

For $\sigma\sim\*\pi^{\otimes n}$, let
$N(\sigma)=(|\sigma^{-1}(a)|)_{a\in\Omega}$ be its vertex-spin counts,
and choose a mode $N^{\mathrm{sp}}$ of
$\operatorname{Mult}(n,\*\pi)$. Then
\[
  \Pr[\sigma\sim\*\pi^{\otimes n}]{N(\sigma)=N^{\mathrm{sp}}}
  \ge (n+1)^{-|\Omega|}.
\]
Set $T^{\mathrm{sp}}=dN^{\mathrm{sp}}=S\*\pi+O(1)$ coordinatewise.
These half-edge counts are realizable by vertex assignments.
Moreover, the denominator in the displayed ratio is at most $1$,
so restricting the outer expectation to
$N(\sigma)=N^{\mathrm{sp}}$ gives
\[
  a_{n,j}\ge (n+1)^{-|\Omega|}
  \Pr{U^{\mathrm{sp}}=T^{\mathrm{sp}}}.
\]
Next choose modes of the multinomial counts of unordered
pair types in $U^{\mathrm{sp}}$ and of the single-spin counts. Their joint probability
is at least $(n+1)^{-O(1)}$, and their spin totals differ from
$S\*\pi$ by $O(1)$. Here a multinomial mode differs from its mean by
$O(1)$ in every coordinate, as follows by comparing the probabilities
of counts differing by one exchange.

We correct the resulting $O(1)$ discrepancy from $T^{\mathrm{sp}}$ by changing
$O(1)$ pair types: for hardcore, change $00$ to $01$ or conversely; for Potts with $q=2$ and $\beta>0$, change $22$ to $12$ or conversely; and for Potts with $q\ge3$, switch an unordered off-diagonal pair $\{a,b\}$ to $\{a,c\}$ for $a\in[q]\setminus\{b,c\}$ (which adjusts the counts of $b,c$ by $\mp1$ while using only off-diagonal probabilities $2P_{ab}(\beta)=2/[q(q-1+\beta)]\ge 2/[q^2]$, uniformly for $\beta\in[0,1]$). Since each required pair type has probability bounded away from zero and $j=\Omega(n)$, these switches yield $\Pr{U^{\mathrm{sp}}=T^{\mathrm{sp}}}\ge C(n+1)^{-K}$ and hence the same polynomial lower bound for $a_{n,j}$ (uniformly in $\beta\in(0,1]$ when $q\ge3$), completing the tail estimate.
}
\end{proof}

{Recall from \Cref{def:replica-symmetry} and \Cref{subsec:preliminaries} that for a probability measure $\mu$ on $\Omega^{[n]}$ and indices $v_1,\ldots,v_k\in[n]$, we write $\mu_{v_1,\ldots,v_k}\in\Delta_{\Omega^k}$ and $\mu_{v_i}\in\Delta_\Omega$ for the joint and single-site marginals of $\mu$. Independently of $\sigma\sim\mu$, choose a random integer $J\in\{0,\ldots,n\}$ and then a uniform $J$-subset $U\subseteq[n]$, and set $p_*\defeq\max_j\Pr{J=j}$ and $\mu^U\defeq\mu(\cdot\mid\sigma_U)$. We next record a finite-sample version of the pinning lemma, adapted from~\cite[Lemma~3.15]{coja2019spin}.

\begin{lemma}[Finite-sample pinning]\label{lem:rsb-finite-pinning}
For any probability law $\mu$ on $\Omega^{[n]}$ and independent uniform indices $v_1,\ldots,v_k\in[n]$ (independent of $(\sigma,J,U)$),
\[
 \E{\DTV{\mu^U_{v_1,\ldots,v_k}}{\bigotimes_{i=1}^k\mu^U_{v_i}}}
 =O_{k,|\Omega|}\bigl(\sqrt{p_*}+n^{-1}\bigr).
\]
\end{lemma}
\begin{proof}
Let $h_j$ be the expected Shannon entropy of a uniform unpinned coordinate after $j$ distinct uniform coordinates and their spin values have been revealed, and set $h_j=0$ for $j\ge n$. The exchangeability of the random order and reduction of entropy by conditioning give $\log|\Omega|\ge h_j\ge h_{j+1}\ge0$. For $r\le n-j$ distinct uniform unpinned coordinates, the chain rule expresses the expected relative entropy of their joint law from the product of its marginals as
\[
 \sum_{i=0}^{r-1}(h_j-h_{j+i}).
\]
For distinct $v_1,\ldots,v_k$, revealed coordinates are constant under $\mu^U$ and the remaining coordinates are uniform distinct unpinned coordinates. Their expected relative entropy is therefore at most $\sum_{i=0}^{k-1}(h_j-h_{j+i})$ when $J=j$. Telescoping yields
\[
 \sum_{j=0}^{n}\sum_{i=0}^{k-1}(h_j-h_{j+i})
 \le\frac{k(k-1)}2\log|\Omega|.
\]
Thus Pinsker's and Jensen's inequalities give an $O_{k,|\Omega|}(\sqrt{p_*})$ bound on the expected distance for distinct indices. Repetitions have probability $O_k(n^{-1})$ and contribute at most one to the distance.
\end{proof}
Since the maximum point probability of $\operatorname{Bin}(n,s)$ is $O_s(n^{-1/2})$ and that of a uniform $J\in\{0,\ldots,T-1\}$ ($T\le n+1$) is $1/T$, \Cref{lem:rsb-finite-pinning} gives a bound of $O_{k,|\Omega|,s}(n^{-1/4})$ when $U$ is a Bernoulli-$s$ subset for fixed $0<s<1$, and $O_{k,|\Omega|}(T^{-1/2}+n^{-1})$ when $J$ is uniform on $\{0,\ldots,T-1\}$. The estimate holds uniformly in $\mu$, including for a conditional law provided $J$ and $U$ are sampled independently of the spin configuration after conditioning.
}

\subsection{Pinned Log-Likelihood and Planted Free-Energy Bound}\label{subsec:app-revealed-entropy}\label{app:revealed-entropy}
\begin{proof}[Proof of~\Cref{lem:revealed-entropy-bound}]
Fix the parameter $t$ {and reveal density $s\in[0,1]$}. All implicit constants below may depend on $t,s$ and $d$.

\paragraph{Step 1: Concentration and edit-Lipschitz bounds.}
{For a supported realization $(G,\*\tau,\*\xi)$ of the planted graph with vertex types $\tau_v=\sigma_v$ and reveal indicators $\xi_v=\one{v\in U}$, write}
\[
 F=\log\mathcal Z_G+\ell_s
   =\log\left(\frac{\mathcal Z_G(\sigma_U)}
                         {\prod_{v\in U}\pi_{\sigma_v}}\right),
\]
where the pinned partition function is
\[
 \mathcal Z_G(\sigma_U)\defeq\mathcal Z_G\mu^{\mathrm{can}}_{G,t}(\sigma_U)
 =\sum_{\substack{{\varsigma}\in\Omega^{V(G)}\\{\varsigma}_U=\sigma_U}}
  \prod_{v\in V(G)}\pi_{{\varsigma}_v}\prod_{uv\in E(G)}M_{{\varsigma}_u,{\varsigma}_v}.
\]
{We verify that $\log\mathcal Z_G$, $F$, and $\ell_s=F-\log\mathcal Z_G$ satisfy the global edit-Lipschitz bound~\eqref{eq:edit-lipschitz-def} of \Cref{lem:profile-concentration} across all supported realizations of maximum degree at most ${d}$. Given any two such supported realizations $x=(G,\*\tau,\*\xi)$ and $x'=(G',\*\tau',\*\xi')$, write $U'=\{v:\xi'_v=1\}$ and {put $\sigma'_v=\tau'_v$}. If they have edit distance
\[
  k \defeq d_{\mathrm{edit}}(x,x') = |E(G)\triangle E(G')| + d_{\mathrm{Ham}}(\*\tau,\*\tau') + d_{\mathrm{Ham}}(\*\xi,\*\xi'),
\]
let $S_{\mathrm{diff}}\subseteq V(G)$ denote the set of discrepancy vertices,
\[
  S_{\mathrm{diff}} \defeq \bigl\{v\in V(G):\tau_v\ne\tau'_v\text{ or }\xi_v\ne\xi'_v\bigr\} \cup \bigl\{u,v:\operatorname{mult}_G(\{u,v\})\ne\operatorname{mult}_{G'}(\{u,v\})\bigr\},
\]
so that $|S_{\mathrm{diff}}|\le 2k$.

Both models have a spin $a_0$ with $M_{a_0a}>0$ for every $a\in\Omega$:
take $a_0=0$ for hardcore and any spin for Potts. On a graph of maximum
degree ${d}$, changing an unpinned spin to $a_0$ preserves positive
weight. This map has at most $|\Omega|$ preimages and changes the weight
by at most a constant factor depending on ${d}$ and $t$. Consequently,
forcing any set of unpinned vertices to $a_0$ decreases a pinned partition
function by at most $\e^{O_{{d},t}(1)}$ per vertex.

For any supported pinned configuration on $G$ and any vertex set $S$,
compare its pinned partition function with that on $G-S$, retaining the
pins outside $S$. For the lower bound, first force the unpinned neighbors
of $S$ outside $S$ to $a_0$, then extend each configuration by the
prescribed spins on $S\cap U$ and $a_0$ on $S\setminus U$. Every new edge
has positive weight: edges with both endpoints pinned are allowed by
supportedness, and every other new edge touches $a_0$. The forcing and
extension change the weight by at most $\e^{O_{{d},t}(|S|)}$. For the
upper bound, there are at most $|\Omega|^{|S|}$ extensions and only
$O_{d}(|S|)$ new vertex and edge factors. Hence
\[
  \bigl|\log\mathcal Z_G(\sigma_U)
       -\log\mathcal Z_{G-S}(\sigma_{U\setminus S})\bigr|
  \le O_{{d},t}(|S|).
\]
The argument includes loops and parallel edges, counted with multiplicity.
Apply it to $S=S_{\mathrm{diff}}$ in both realizations. The resulting
graphs and pins coincide, so
\[
  \bigl|\log\mathcal Z_G(\sigma_U)-\log\mathcal Z_{G'}(\sigma'_{U'})\bigr|
  \le O_{{d},t}(k).
\]
Taking $U=U'=\varnothing$ gives $\bigl|\log\mathcal Z_G-\log\mathcal Z_{G'}\bigr|\le O_{{d},t}(|E(G)\triangle E(G')|)\le O_{{d},t}(k)$, and
the pin normalization $\sum_{v\in U}\log\pi_{\sigma_v}$ in $F$ also changes by at most $O_t\bigl(d_{\mathrm{Ham}}(\*\tau,\*\tau')+d_{\mathrm{Ham}}(\*\xi,\*\xi')\bigr)=O_t(k)$. Thus $F$,
$\log\mathcal Z_G$, and $\ell_s=F-\log\mathcal Z_G$ satisfy~\eqref{eq:edit-lipschitz-def} with $L=O_{{d},t}(1)$ across all supported realizations of maximum degree at most ${d}$,
as required by \Cref{lem:profile-concentration}. For Potts, since $\pi_a=1/q$ and $\min_{a,b}M_{ab}(\beta)=q\beta/(q-1+\beta)$, the constant
is $O_{{d},q}(1+|\log\beta|)$ as $\beta\downarrow0$.
}

\smallskip
{At maximum pairing, \Cref{lem:rsb-planted-partial-profiles} supplies the profile-tail hypothesis of \Cref{lem:profile-concentration} with a deterministic center. Applying that lemma to $\ell_s$ with independent Bernoulli-$s$ reveal indicators $\xi_v=\one{v\in U}$ proves the asserted exponential upper tail.}

{
\paragraph{Step 2: Cavity increment under partial pairing.}
Assume first that $0<s<1$, fix $\theta\in(0,1)$, set $\theta_0\defeq\min\{\theta,1-\theta\}/2>0$, and recall $F=\log\mathcal Z_G+\ell_s$ from Step~1.
For any integer $m\ge1$, define the binomial probabilities and central window
\[
 b_m(j)\defeq\Pr{\operatorname{Bin}(m,\theta/2)=j},\qquad
 \mathcal M_m\defeq\bigl\{j\in\=Z:|j-\theta m/2|\le\sqrt n\log n\bigr\}.
\]
Since $S=dn$ with $d\ge3$, for all sufficiently large $n$, every $m\in{\{S,\ldots,S+d\}}$, every $j\in\mathcal M_m$, and every integer shift $i$ with $|i|\le 2{d}$ satisfy
\[
 \frac{\theta_0}{2}S \le j+i \le \frac{1-\theta_0}{2}S < \lfloor S/2\rfloor \le \lfloor m/2\rfloor.
\]
Hence, for all sufficiently large $n$, the partial-pairing laws $\mathrm{CM}_{n,d,j+i}$ and $\mathbb Q_{n,j+i}$ are well-defined with pairing fraction in $[\theta_0,1-\theta_0]\subset(0,1)$ whenever $j\in\mathcal M_S$ (or $j\in\mathcal M_{S+d}$) and $|i|\le 2{d}$.

For each fixed graph $G$ and such $j$, when $k\le{d}$ distinct unmatched half-edges $e_1,\ldots,e_k$ are sampled uniformly without replacement, their endpoint tuple $(u_1,\ldots,u_k)\in[n]^k$ has point probability bounded by
\[
  \prod_{i=0}^{k-1}\frac{{d}}{S-2j-i} \le \left(\frac{C_\theta}{n}\right)^k,
  \qquad
  C_\theta \asymp \frac{1}{1-\theta},
\]
where $C_\theta$ is slightly enlarged to absorb the without-replacement denominators $S-2j-i$. Comparing $(u_1,\ldots,u_k)$ with $k$ i.i.d.\ uniform vertices amplifies the estimate of \Cref{lem:rsb-finite-pinning} by at most $C_\theta^k$. Thus, uniformly over these graphs and pairing counts,
\[
  \E{\DTV{\mu^{\mathrm{can}}_{G,t}(\sigma_{u_1}=\cdot,\ldots,\sigma_{u_k}=\cdot\mid\sigma_U)}
  {\bigotimes_{i=1}^k\mu^{\mathrm{can}}_{G,t}(\sigma_{u_i}=\cdot\mid\sigma_U)}}
  =o(1).
\]

For admissible $j$, let $g_{n,j}\defeq\E[\mathbb Q_{n,j}]{F}$ (averaging over the planted graph, $\sigma$, and $U$), and define the binomially smoothed average
\[
 f_n\defeq\sum_{j\in\mathcal M_S}b_S(j)g_{n,j}.
\]
By Step~1 and \Cref{lem:rsb-planted-partial-profiles} (with parameter $\theta_0$), $|g_{n,j+i}|=O(n)$ uniformly over $j\in\mathcal M_S\cup\mathcal M_{S+d}$ and $|i|\le 2{d}$. Moreover, since $b_S(j-i)/b_S(j)=1+O_\theta(n^{-1/2}\log n)$ on $\mathcal M_S$ and $\Pr{\operatorname{Bin}(S,\theta/2)\notin\mathcal M_S}\le\exp\{-\Omega((\log n)^2)\}$ by Hoeffding's inequality, every integer shift $|i|\le 2{d}$ and uniformly bounded family $H(n,j)$ on the admissible range satisfy
\begin{equation}\label{eq:rsb-binomial-reindex}
 \sum_{j\in\mathcal M_S}b_S(j)H(n,j+i)
 =\sum_{j\in\mathcal M_S}b_S(j)H(n,j)+o(1)
\end{equation}
uniformly over $n$.

\smallskip
\emph{(a) Elementary reweighting and edge-insertion increment ($g_{n,j+1}-g_{n,j}$).}
Under $\mathbb Q_{n,j}$ for $j\in\mathcal M_S$ (or within $2{d}$ of $\mathcal M_S$), consider either inserting one new edge between two uniformly chosen distinct unmatched half-edges with endpoint vertices $u_1,u_2$, or adding a vertex $n+1$ with spin $a\sim\*\pi$ and $l$ edges connected to uniformly chosen distinct unmatched half-edges with endpoint vertices $u_1,\ldots,u_l$. Write $G'$ for the graph after the chosen operation. When inserting an edge, set $\sigma'=\sigma$ and $U'=U$. When adding vertex $n+1$, extend $\sigma$ by setting $\sigma'_{n+1}=a$, independently reveal that vertex with probability $s$, and let $U'$ be the resulting set of revealed vertices. This gives independent Bernoulli-$s$ marks on $[n+1]$. The respective reweighting factors relative to $\mathbb Q_{n,j}$ are
\[
 \Psi_e\defeq M_{\sigma_{u_1},\sigma_{u_2}},\qquad
 \Psi_l\defeq\prod_{i=1}^lM_{a,\sigma_{u_i}}\quad(a\sim\*\pi).
\]
Let $p$ and $p^*$ be the partial-pairing profile and its center in the coordinates of \Cref{lem:profile-concentration}, and put $D=\|p-p^*\|_1$. Integrating the tail in \Cref{lem:rsb-planted-partial-profiles} gives
\[
 \E{D^k}=O\bigl((\log n/n)^{k/2}\bigr)\qquad(k=1,2,4).
\]
Recall $\*\chi$ for the spin marginal of the unmatched half-edge law. Since $\*u=d(1-2j/S)\*\chi$ and $2j/S=\theta+O(n^{-1/2}\log n)$, the weighted profile bound gives $\|\*\chi-\*\pi\|_1\le C_\theta D$ uniformly over the stated pairing counts. Conditional on all old data $(G,\*\tau,U)$, averaging the fresh endpoint choices and the new spin gives a function $H_\Psi(p)$ of the profile alone. Replacing sampling without replacement by independent sampling costs $O(n^{-1})$, so
\[
 H_{\Psi_e}(p)=\*\chi^\trans M\*\chi+O(n^{-1}),\qquad
 H_{\Psi_l}(p)=\sum_a\pi_a(M\*\chi)_a^l+O(n^{-1}).
\]
Both polynomials have value $1$ and zero tangent derivative at $\*\chi=\*\pi$, by $M\*\pi=\allone$. Thus $|H_\Psi(p)-1|\le C(D^2+n^{-1})$ for $\Psi\in\{\Psi_e,\Psi_l\}$. \Cref{lem:profile-concentration} and Step~1 give $\operatorname{Var}(F)=O(n\log n)$, uniformly over $j\in\mathcal M_S$ and their bounded shifts. Since $\E{\Psi\mid G,\*\tau,U}=H_\Psi(p)$, Cauchy--Schwarz gives
\begin{equation}\label{eq:rsb-reweighting-covariance}
 \begin{aligned}
 \E{\Psi}&=1+O(\log n/n),\\
 |\operatorname{Cov}(\Psi,F)|
 &\le C\sqrt{n\log n}\,\sqrt{\E{(D^2+n^{-1})^2}}
 =O(\log^{3/2}n/\sqrt n)=o(1).
 \end{aligned}
\end{equation}

\smallskip
Let $\*u_{\mathrm{end}}=(u_1,\ldots,u_k)$ be the sampled endpoint tuple, with $k=2$ for edge insertion and $k=l$ for adding vertex $n+1$. Take a symbol $\dagger\notin\Omega$, and let $\omega_{\mathrm{new}}=a$ if vertex $n+1$ is added and revealed with spin $a$, and $\omega_{\mathrm{new}}=\dagger$ otherwise. Set $x\defeq\E{\Psi\mid G,\sigma_U,\*u_{\mathrm{end}},\omega_{\mathrm{new}}}$. Conditional on $(G,\sigma_U)$, the planted spin assignment $\sigma$ has law $\mu^{\mathrm{can}}_{G,t}(\cdot\mid\sigma_U)$. Thus $x=\mathcal Z_{G'}(\sigma'_{U'})/\mathcal Z_G(\sigma_U)$ when $\omega_{\mathrm{new}}=\dagger$, and $x=\mathcal Z_{G'}(\sigma'_{U'})/(\pi_a\mathcal Z_G(\sigma_U))$ when $\omega_{\mathrm{new}}=a$. Hence the observable on $G'$ equals $F+\log x$. If $x=0$, then $\Psi=0$ almost surely under this conditioning. Using $\E{\Psi F}=\operatorname{Cov}(\Psi,F)+\E{\Psi}\E{F}$ and $\E{\Psi\log x}=\E{x\log x}$ (by conditioning on $G,\sigma_U,\*u_{\mathrm{end}},\omega_{\mathrm{new}}$, with $0\log0=0$), we obtain for every constant $c\in\=R$ the identity
\begin{equation}\label{eq:rsb-sizebiased-increment}
 \E{\Psi(F+\log x-c)}
 =\operatorname{Cov}(\Psi,F)+(\E{F}-c)\E{\Psi}+\E{x\log x}.
\end{equation}
Write $A_e(n,j)$ and $A_l(n,j)$ for $\E[\mathbb Q_{n,j}]{x\log x}$ when $\Psi=\Psi_e$ and $\Psi=\Psi_l$, respectively; both are uniformly bounded since $x\ge0$ is bounded. Because adding a uniform edge to $\mathrm{CM}_{n,d,j}$ produces $\mathrm{CM}_{n,d,j+1}$, reweighting by $\Psi_e$ gives $a_{n,j+1}/a_{n,j}=\E[\mathbb Q_{n,j}]{\Psi_e}$ and, by~\eqref{eq:rsb-sizebiased-increment} with $c=g_{n,j}$ and~\eqref{eq:rsb-reweighting-covariance},
\[
 \frac{a_{n,j+1}}{a_{n,j}}=\E[\mathbb Q_{n,j}]{\Psi_e}=1+O(\log n/n),
 \qquad g_{n,j+1}-g_{n,j}=A_e(n,j)+o(1).
\]

\smallskip
\emph{(b) Vertex-addition increment ($g_{n+1,K}-g_{n,K}$).}
To pass from $n$ to $n+1$ vertices, fix any $K\in\mathcal M_{S+d}$. For all sufficiently large $n$, every pair $b=(k,l)$ of nonnegative integers with $2k+l\le {d}$ has $j_b\defeq K-k-l$ satisfying $\theta_0 S/2\le j_b\le K\le(1-\theta_0)S/2<\lfloor S/2\rfloor$, so $\mathrm{CM}_{n+1,d,K}$, $\mathbb Q_{n+1,K}$, $\mathrm{CM}_{n,d,K}$, $\mathbb Q_{n,K}$, and $\mathbb Q_{n,j_b}$ are all well-defined with pairing fractions in $[\theta_0,1-\theta_0]$. For a partial pairing drawn from $\mathrm{CM}_{n+1,d,K}$, let $K_{\mathrm{loop}}$ and $L_{\mathrm{ext}}$ be the numbers of self-loops and edges to $[n]$, respectively, incident to vertex $n+1$. Writing $q_b\defeq\Pr[\mathrm{CM}_{n+1,d,K}]{K_{\mathrm{loop}}=k,\ L_{\mathrm{ext}}=l}$ and $\Psi_b\defeq M_{aa}^k\Psi_l$, the probability of $b=(k,l)$ under $\mathbb Q_{n+1,K}$ is $q_b(a_{n,j_b}/a_{n+1,K})\E[\mathbb Q_{n,j_b}]{\Psi_b}$. Since $a_{n,j_b}/a_{n,K}=1+O(\log n/n)$ by the edge ratio above, realizations with self-loops ($k\ge1$) have total unplanted probability $\sum_{\substack{k\ge1,\ l\ge0\\2k+l\le d}}q_{(k,l)}=O(1/n)$, and loop-free realizations ($k=0$) satisfy $\E[\mathbb Q_{n,j_b}]{\Psi_{(0,l)}}=1+O(\log n/n)$ by~\eqref{eq:rsb-reweighting-covariance}, summing over $b=(k,l)$ yields $a_{n+1,K}/a_{n,K}=1+O(\log n/n)$ uniformly for $K\in\mathcal M_{S+d}$.

For this conditional coupling, the reverse construction starts with a uniform old partial pairing with $j_b$ edges, chooses uniformly the $2k+l$ half-edges of the new vertex used by its incident edges, pairs $2k$ of these into $k$ loops, and matches the other $l$ to a uniformly selected ordered $l$-tuple of unmatched old half-edges. The choices are uniform among constructions giving the same $(k,l)$, and each full partial pairing has the same number of such representations; hence the resulting law is the unplanted conditional law. Relative to this reverse construction, the planted density is precisely $(a_{n,j_b}/a_{n+1,K})\Psi_b$: the old planted density supplies the old edge weights, and $\Psi_b$ supplies exactly the new loop and external-edge weights.

\smallskip
Applying~\eqref{eq:rsb-sizebiased-increment} with $c=g_{n,K}$ for each possible pair $b=(k,l)$ (with expectations under $\mathbb Q_{n,j_b}$ and the independent spin $a\sim\*\pi$ of vertex $n+1$ and sampled unmatched half-edges, and $x_b$ the corresponding conditional mean of $\Psi_b$) gives
\begin{align*}
 &\frac{a_{n+1,K}}{a_{n,K}}
       \bigl(g_{n+1,K}-g_{n,K}\bigr)\\
       ={}&\sum_b q_b\frac{a_{n,j_b}}{a_{n,K}}
       \left\{\operatorname{Cov}(\Psi_b,F)
       +(g_{n,j_b}-g_{n,K})\E{\Psi_b}
       +\E{x_b\log x_b}\right\}.
\end{align*}
For realizations with self-loops ($k\ge1$), the corresponding profile function satisfies $H_{\Psi_b}(p)=c_b+O(D+n^{-1})$, where $c_b=\sum_a\pi_aM_{aa}^k$. The same Cauchy--Schwarz bound gives $|\operatorname{Cov}(\Psi_b,F)|=O(\log n)$ under each old law $\mathbb Q_{n,j_b}$. Since their total unplanted probability is $O(n^{-1})$, $\Psi_b$ and $x_b\log x_b$ are bounded, and $|g_{n,j_b}-g_{n,K}|=O(1)$ (a sum of $k+l$ edge increments), their total contribution is $O(\log n/n)=o(1)$. For loop-free realizations ($k=0$), $j_{(0,l)}=K-l$, so the edge increment estimate gives $g_{n,K-l}-g_{n,K}=-\sum_{i=1}^l A_e(n,K-i)+o(1)$. The number $l$ of external edges of vertex $n+1$ under $\mathrm{CM}_{n+1,d,K}$ has law $\operatorname{Bin}(d,\theta)+o(1)$ in total variation uniformly for $K\in\mathcal M_{S+d}$. Hence, uniformly over $K\in\mathcal M_{S+d}$,
\[
 g_{n+1,K}-g_{n,K}
 =\E[L_d\sim\operatorname{Bin}(d,\theta)]{A_{L_d}(n,K-L_d)-\sum_{i=1}^{L_d}A_e(n,K-i)}+o(1).
\]

\emph{(c) Binomial smoothing and $f_{n+1}-f_n$.}
To compare $f_{n+1}$ with $f_n$, couple $K=J_{\mathrm{edges}}+B_d$ with independent $J_{\mathrm{edges}}\sim\operatorname{Bin}(S,\theta/2)$ and $B_d\sim\operatorname{Bin}(d,\theta/2)$, so that $K\sim\operatorname{Bin}(S+d,\theta/2)$ marginally. Let $\mathcal W_{S,d}\defeq\{J_{\mathrm{edges}}\in\mathcal M_S,\ K\in\mathcal M_{S+d}\}$ be the joint central-window event. Since $0\le B_d\le {d}$, the symmetric difference $\{J_{\mathrm{edges}}\in\mathcal M_S\}\triangle\{K\in\mathcal M_{S+d}\}$ is contained in $\{|J_{\mathrm{edges}}-\theta S/2|>\sqrt n\log n-{d}\}$, which has probability $\exp\{-\Omega((\log n)^2)\}$ by Hoeffding's inequality. The same tail bound applies to replacing the window of width $\sqrt{n+1}\log(n+1)$ in $f_{n+1}$ by $\mathcal M_{S+d}$. Because $|g_{n,j}|=O(n)$ on $\mathcal M_S$ and $|g_{n+1,K}|=O(n)$ on $\mathcal M_{S+d}$, replacing both window indicators by $\one{\mathcal W_{S,d}}$ changes each expectation by $o(1)$. On $\mathcal W_{S,d}$, every intermediate edge count between $J_{\mathrm{edges}}$ and $K=J_{\mathrm{edges}}+B_d$ is admissible, so
\[
 g_{n,K}-g_{n,J_{\mathrm{edges}}}
 =\sum_{i=0}^{B_d-1}A_e(n,J_{\mathrm{edges}}+i)+o(1).
\]
All $A_l$ and $A_e$ are uniformly bounded. In averaging the preceding display, take the auxiliary $L_d\sim\operatorname{Bin}(d,\theta)$ independently of $(J_{\mathrm{edges}},B_d)$. Applying~\eqref{eq:rsb-binomial-reindex} to each bounded shift in the two displays, and then using $\E{B_d-L_d}=\theta d/2-\theta d=-\theta d/2$, gives
\begin{equation}\label{eq:rsb-smoothed-cavity-increment}
 f_{n+1}-f_n
 =\sum_{j\in\mathcal M_S}b_S(j)
       \left[\E[L_d\sim\operatorname{Bin}(d,\theta)]{A_{L_d}(n,j)}
                    -\frac{\theta d}{2}A_e(n,j)\right]+o(1).
\end{equation}

\paragraph{Step 3: Evaluation via a common cavity law.}
For a fixed old graph $G$ and pinned spin configuration $\sigma_U$, write $\*\mu_e\defeq\mu^{\mathrm{can}}_{G,t}(\sigma_{v(e)}=\cdot\mid\sigma_U)\in\Delta_\Omega$ for the pinned single-site marginal at the endpoint $v(e)$ of an unmatched half-edge $e$, and let $\nu_0\in\mathcal P(\Delta_\Omega)$ be the empirical distribution of $\*\mu_e$ over the $S-2j$ unmatched half-edges $e$. Since conditional on $(G,\sigma_U)$ the planted configuration $\sigma$ has law $\mu^{\mathrm{can}}_{G,t}(\cdot\mid\sigma_U)$, the mean of $\nu_0$ is
\[
 \int\*\mu\,\nu_0(\dif\*\mu)
 =\E[\mathbb Q_{n,j}]{\*\chi\mid G,\sigma_U},
\]
which converges to $\*\pi$ in $L^1(\mathbb Q_{n,j})$ uniformly in $j\in\mathcal M_S$, by Jensen's inequality and the bound on $\*\chi$ in Step~2. For $k\in\{2,l\}$ uniformly sampled distinct unmatched half-edges $e_1,\ldots,e_k$ with endpoint vertices $u_1,\ldots,u_k$, the pinning estimate in Step~2 gives
\[
 \E{\DTV{\mu^{\mathrm{can}}_{G,t}(\sigma_{u_1}=\cdot,\ldots,\sigma_{u_k}=\cdot\mid\sigma_U)}
 {\bigotimes_{i=1}^k\*\mu_{e_i}}}=o(1).
\]
Moreover, the tuple $(\*\mu_{e_1},\ldots,\*\mu_{e_k})$ is within total variation $O(1/n)$ of $k$ independent draws from $\nu_0$. The conditional mean $x$ of $\Psi$ is therefore within $o(1)$ in $L^1$ of the corresponding multilinear functional of these marginal distributions. Uniform continuity of $x\log x$ on bounded intervals allows us to evaluate $A_{L_d}(n,j)$ and $A_e(n,j)$ using independent draws from $\nu_0$, at cost $o(1)$.

Let $H(\nu,s)$ be the resulting bounded continuous functional of the cavity message distribution and reveal density, and put $b(\nu)=\int\*\mu\,\nu(\dif\*\mu)$. The profile tail makes the $L^1$ convergence of $b(\nu_0)$ to $\*\pi$ uniform over pairing counts. Markov's inequality, boundedness, and compactness therefore give
\[
 \limsup_{n\to\infty}\E{H(\nu_0,s)}
 \le\max_{\nu\in\mathcal P_{\*\pi}}H(\nu,s).
\]

For i.i.d.\ $\*\mu_i\sim\nu\in\mathcal P_{\*\pi}$, put $T_a\defeq\prod_{i=1}^l(M\*\mu_i)_a$ and $V_l=\sum_{a\in\Omega}\pi_aT_a$. When the new vertex $n+1$ is unrevealed ($n+1\notin U'$, probability $1-s$), summing over its spin gives $x=V_l$; when $n+1$ is revealed with spin $a$ ($n+1\in U'$, probability $s\pi_a$), its vertex weight $\pi_a$ in $\mathcal Z_{G'}(\sigma'_{U'})$ cancels the factor $\pi_a$ in the denominator $\prod_{v\in U'}\pi_{\sigma'_v}$ of $F$, leaving $x=T_a$. Since $\E{V_l}=1$ gives $\E{V_l\log V_l}\ge0$ and independence with $\E{(M\*\mu)_a}=1$ and $u\log u\le u-1+(u-1)^2$ ($u\ge0$) gives
\[
 \E{T_a\log T_a}
 =l\E{(M\*\mu)_a\log(M\*\mu)_a}
 \le l\E{((M\*\mu)_a-1)^2},
\]
the contribution from adding vertex $n+1$, namely $(1-s)\E{V_l\log V_l}+s\sum_{a\in\Omega}\pi_a\E{T_a\log T_a}$, is at most
\[
 \E{V_l\log V_l}
 +sl\E[\nu]{\sum_{a\in\Omega}\pi_a((M\*\mu)_a-1)^2},
\]
while the edge term is $\E{E\log E}$ for $E=\sum_{a,b\in\Omega}(\*\mu_1)_aM_{ab}(\*\mu_2)_b$. Substituting these into~\eqref{eq:rsb-smoothed-cavity-increment} and telescoping over $n$ yields
\[
 \limsup_{n\to\infty}\frac{f_n}{n}
 \le\max_{\nu\in\mathcal P_{\*\pi}}
 \E[L_d\sim\operatorname{Bin}(d,\theta)]{
    \mathcal B_{L_d,t}(\nu)
    + sL_d\E[\nu]{\sum_{a\in\Omega}\pi_a((M\*\mu)_a-1)^2}}
.
\]

\paragraph{Step 4: Completing the pairing.}
For each $j\in\mathcal M_S$, \Cref{lem:rsb-planted-partial-profiles} gives the centers $p_j^*=(\*\pi,d\theta_jP,d(1-\theta_j)\*\pi)$ and $p_{\max,S}^*=(\*\pi,d\theta_{\lfloor S/2\rfloor}P,d(1-\theta_{\lfloor S/2\rfloor})\*\pi)$ for $\mathbb Q_{n,j}$ and $\mathbb Q_{n,\lfloor S/2\rfloor}$, respectively. Their $\ell_1$ distance is $2d(1-\theta)+O(n^{-1/2}\log n)$ uniformly in $j\in\mathcal M_S$. For the two laws, define the conditional mean of $F$ on the union of their supported profiles using the uniform conditional spin-assignment and pairing law and the common independent marks. Step~1 gives the required edit-Lipschitz bound for $F$ across this union, so the coupling in~\eqref{eq:rsb-profile-mean-comparison} compares these means even though the two laws have different numbers of edges. This comparison, combined with the uniform profile concentration of \Cref{lem:rsb-planted-partial-profiles} and the $O(1)$-edit-Lipschitz bound for $F$, gives
\[
 \left|f_n-\E[\mathbb Q_n]{\log\mathcal Z_G+\ell_s}\right|
 \le C(1-\theta)n+o(n).
\]
Conditional on $G$, $\E{\ell_u\mid G}$ is the expected relative entropy of the revealed marginal from $\*\pi^{\otimes U}$. Data processing and the chain rule give, for $0\le u\le v\le1$,
\[
 0\le\E{\ell_v}-\E{\ell_u}
 \le(v-u)n\log(1/\min_a\pi_a).
\]
Letting $n\to\infty$ and then $\theta\uparrow1$ proves the second inequality in~\eqref{eq:planted-free-energy-bound} for $0<s<1$: $L_d\sim\operatorname{Bin}(d,\theta)\to d$, and the functionals converge uniformly on $\mathcal P_{\*\pi}$. The displayed continuity bound, together with continuity of the right-hand side in $s$, extends the inequality to $s=0,1$.
For the first assertion of~\Cref{lem:revealed-entropy-bound} and the first inequality in~\eqref{eq:planted-free-energy-bound}, \Cref{lem:rsb-planted-partial-profiles} at maximum pairing gives $|\log a_n|=O(\log(n+1))$. Let $\mathbb Q_n^G$ be the graph marginal of $\mathbb Q_n$. By the definition of $\mathbb Q_n$,
\[
 \E[\mathbb Q_n]{\log\mathcal Z_G}
 =\DKL(\mathbb Q_n^G\Vert\mathrm{CM})+\log a_n\ge-o(n),
\]
completing the proof.}
\end{proof}

\subsection{{Typical-Graph Truncation and Proof of \Cref{lem:graph-truncation}}}\label{subsec:app-coloring-endpoint}\label{app:coloring-endpoint}
\begin{proof}[Proof of~\Cref{lem:graph-truncation}]
  \leavevmode\par
\begingroup
\paragraph{Step 1: Construction of $E_n$ and real-parameter truncation errors at $x>0$.}
We first work at a positive parameter $x$ satisfying
\Cref{cond:unique-variational-maximizer} and \Cref{cond:strict-KS}.
Write $\mathcal Z=\mathcal Z_G(x)$, $a_n=a_n(x)$, and
$\mathcal Z_B^{(2)}=\mathcal Z_{G,B}^{(2)}(x,x)$. We construct $E_n$ such that, for some constant $c>0$,
\begin{equation}\label{eq:three-truncation-errors}
 \frac{\E{\mathcal Z\one{E_n^c}}}{a_n}
 +\frac{\E{\mathcal Z_B^{(2)}\one{E_n^c}}}{a_n^2}
 +\frac{\E{(\mathcal Z^2-\mathcal Z_B^{(2)})\one{E_n}}}{a_n^2}
 \le\e^{-cn}.
\end{equation}
Choose $\eta>0$ such that $\{\*\rho\in\Delta_{\Omega\times\Omega}:\|\*\rho-\*\pi(x)\otimes\*\pi(x)\|_\infty\le\eta\}\subseteq B$, so $\{\*\rho\notin B\}\subseteq\{\|\*\rho-\*\pi(x)\otimes\*\pi(x)\|_\infty>\eta\}$. By \Cref{lem:exponential-overlap}, after choosing a constant $\gamma_0>0$,
\[
 \eta_n\defeq\E[\mathbb Q_n]{\mu_G^{\otimes2}(\*\rho\notin B)}
 \le\e^{-2\gamma_0 n}.
\]
Put $C_n=\E{\mathcal Z_B^{(2)}}/a_n^2$. By
\Cref{lem:local-profile-moments}, $C_n$ converges to a positive finite
constant. Let $\mathbb Q_B$ have density
$\mathcal Z_B^{(2)}/\E{\mathcal Z_B^{(2)}}$ relative to $\mathrm{CM}$.
Since $\mathcal Z>0$ at $x>0$, the random variable $V=a_n/\mathcal Z$
under $\mathbb Q_B$ satisfies
\[
 \E[\mathbb Q_B]{V}=\frac{1-\eta_n}{C_n},\qquad
 \E[\mathbb Q_B]{V^2}
 =\frac{\E{\mu_G^{\otimes2}(\*\rho\in B)}}{C_n}
 \le\frac1{C_n}.
\]
The Paley--Zygmund inequality therefore gives fixed constants $M,c_x>0$
such that, for all large $n$,
\begin{equation}\label{eq:central-pair-cutoff-mass}
 \Pr[\mathbb Q_B]{\log(\mathcal Z/a_n)\le\log M}\ge c_x>0.
\end{equation}

\paragraph{Step 2: A partition-function cutoff.}
Here $n$ is admissible, so $dn$ is even and the maximum pairing is a full pairing.
Extend $\mathbb Q_B$ by sampling a uniform labeled pairing $\mathcal M$ inducing $G$ and, independently given $G$, sampling the spin pair
proportionally to its paired weight restricted to
$\{\*\rho(\sigma,\sigma')\in B\}$. Its paired profile is
$p=(\*\rho,Q)$, with center
$p^*=(\*\pi(x)\otimes\*\pi(x),P(x)\otimes P(x))$. The profile count in
\Cref{prop:profile-moments}(1), applied to the paired model, and
\Cref{lem:local-profile-moments}(1)--(2) give
\[
 \Pr[\mathbb Q_B]{\|p-p^*\|_1>v}
 \le C_x'n^K\exp(-\gamma_xnv^2).
\]
Conditional on the profile, the spin-pair assignment and compatible
pairing are uniform. Apply the profile concentration bound (\Cref{lem:profile-concentration}) on the alphabet
$\Omega\times\Omega$, with full-pairing profile $(\*\rho,dQ,\*0)$, to the observable
$\log(\mathcal Z/a_n)$, which satisfies the global edit-Lipschitz bound~\eqref{eq:edit-lipschitz-def} with $L=O_{d,x}(1)$ across all $d$-regular multigraphs $G$ (independently of the paired spins) by Step~1 of the proof of \Cref{lem:revealed-entropy-bound}.
For every fixed $v>0$, its deviation by
$vn$ from its $\mathbb Q_B$-mean thus has probability at most
$\e^{-c'_xv^2n}$ for all large $n$ and some constant $c'_x>0$. Together with
\eqref{eq:central-pair-cutoff-mass}, this shows that
\[
 \limsup_{n\to\infty}\frac1n\E[\mathbb Q_B]{\log(\mathcal Z/a_n)}\le0.
\]
Applying the same concentration bound again with $c''_x\defeq c'_x/4>0$, for every fixed
$\varepsilon>0$,
\[
 \Pr[\mathbb Q_B]{\log(\mathcal Z/a_n)>\varepsilon n}
 \le\exp(-c''_x\varepsilon^2n)
\]
for large $n$. Choose $\varepsilon=\gamma_0/2$ and put
$E_n=\{\mathcal Z\le a_n\e^{\varepsilon n}\}$. The displayed bound controls
the second term of \eqref{eq:three-truncation-errors}. Splitting the first term
according to whether the replica overlap lies in
$B$, and using $\mathcal Z/a_n>\e^{\varepsilon n}$ on $E_n^c$, gives
\[
 \Pr[\mathbb Q_n]{E_n^c}\le\eta_n+C_n\e^{-\varepsilon n}.
\]
The third term is at most $\e^{\varepsilon n}\eta_n\le\e^{-(2\gamma_0-\varepsilon)n}$.
Since $C_n$ is bounded, the sum in~\eqref{eq:three-truncation-errors} is at most
$$
  C_x\exp\{-n\min(\varepsilon,2\gamma_0-\varepsilon,c''_x\varepsilon^2)\}\qquad \text{for some $C_x<\infty$.}
$$
Thus, after absorbing this prefactor into the exponential, the bound holds for all sufficiently large $n$ with
\[
 c(x)\defeq\tfrac12\min\bigl\{\varepsilon,\,2\gamma_0-\varepsilon,\,c''_x\varepsilon^2\bigr\}>0.
\]

\paragraph{Step 3: Quantitative uniformity near the coloring endpoint.}
When $t=0$ (proper $q$-colorings, $q\ge3$), the strict endpoint assumptions imply $\mathcal B_{d,0}\le0$; the Potts contraction argument in \Cref{lem:model-variational-conditions} therefore gives unique maximization at $x>0$, while $|\kappa_x|<|\kappa_0|$ preserves strict Kesten--Stigum stability. We construct $E_n$ by applying Steps~1 and~2 at a sufficiently small positive Potts parameter $x\in(0,x_*]$, where $[0,x_*]\subset D_0$ for the complex neighborhood $D_0$ of $t=0$ supplied by \Cref{lem:local-profile-moments}. For $q\ge3$, the proof of \Cref{lem:rsb-planted-partial-profiles} gives $|\log a_{n,j}(x)|\le C\log(n+1)$ and the Gaussian profile tail~\eqref{eq:bounded-degree-profile-tail} under $\mathbb Q_n$ with constants $C,c,K>0$ independent of $x\in(0,x_*]$. Likewise, under $\mathbb Q_B$, \Cref{lem:local-profile-moments}(1)--(2) and \Cref{prop:profile-moments}(1) give $0<c_*\le a_n(x),C_n(x)\le C_*<\infty$ and a single quadratic-gap constant on $B$ across $x\in[0,x_*]$ for all $n\ge n_1$, so the paired profile tail in Step~2 holds with constants $C'_x=C'$, $K$, and $\gamma_x=\gamma>0$ independent of $x\in(0,x_*]$, and Paley--Zygmund in Step~1 gives~\eqref{eq:central-pair-cutoff-mass} with $M=4C_*$ and $c_x=1/(16C_*)$ whenever $\eta_n\le 1/2$.

The only $x$-dependent degradation in the concentration bounds therefore comes from the edit-Lipschitz constant of $\ell_s$ and $\log(\mathcal Z/a_n)$ in Step~1 of the proof of \Cref{lem:revealed-entropy-bound}, which is $O_{d,q}(L_x)$ with $L_x\defeq 1+\log(1/x)$. Moreover, the functional inside the maximum in~\eqref{eq:planted-free-energy-bound} converges uniformly on $\mathcal P_{\*U_q}$ as $x\downarrow0$ to its $\beta=0$ limit. By compactness and unique maximization of $\mathcal B_{d,0}$ at $\delta_{\*U_q}$, the maximum of this limiting functional over $\mathcal P_{\*U_q}$ is $o(s)$ as $s\downarrow0$. Fix $\eta>0$ so that $\{\*\rho:\|\*\rho-\*U_q\otimes\*U_q\|_\infty\le\eta\}\subseteq B$, then fix $s\in(0,1]$ and $b=s\eta^2/12$ in the proof of \Cref{lem:exponential-overlap} so that the $\beta=0$ limit of the right-hand side of~\eqref{eq:planted-free-energy-bound} is strictly less than $b/16$; by uniform convergence in $x$, the right-hand side of~\eqref{eq:planted-free-energy-bound} is at most $b/12$ for all $x\in(0,x_*]$ after shrinking $x_*>0$ if needed. Applying \Cref{lem:profile-concentration} with $L=O_{d,q}(L_x)$ and these $x$-uniform profile-tail constants to $\ell_s$ under $\mathbb Q_n$ and to $\log(\mathcal Z/a_n)$ under $\mathbb Q_B$ gives constants that can be chosen as $c_0\asymp L_x^{-2}$ in the proof of \Cref{lem:exponential-overlap} (hence $\gamma_0\asymp L_x^{-2}$ in Step~1) and $c''_x\asymp L_x^{-2}$ in Step~2. With $\varepsilon=\gamma_0/2\asymp L_x^{-2}$, the exponent $c(x)$ chosen in Step~2 satisfies
\[
 c(x)=\tfrac12\min\bigl\{\varepsilon,\,2\gamma_0-\varepsilon,\,c''_x\varepsilon^2\bigr\}
 \ge c_2\bigl(1+\log(1/x)\bigr)^{-6},
\]
where $c_2>0$ is independent of $x$. As $x\downarrow0$, we have $c_2(1+\log(1/x))^{-6}\gg x$, which allows us to choose a fixed sufficiently small $x\in(0,x_*]$ in Step~4 below so that the truncation decay at $x$ dominates the $O(xn)$ weight-perturbation cost on $|z|,|w|\le x$.
\par\endgroup

  \paragraph{Step 4: Analytic continuation to a complex disk $D_0$.}
  Take $x=t$ when $t>0$, and a fixed $x\in(0,x_*]$ when $t=0$. For each of the three restricted configuration sums in~\eqref{eq:three-truncation-errors}, comparing each configuration weight $\prod_{v\in V(G)}\pi_{\sigma_v}(z)\prod_{uv\in E(G)}M_{\sigma_u,\sigma_v}(z)$ to its positive weight at $x$ bounds the modulus at $(z,w)$ by its value at $(x,x)$ times $\exp\{O_x(n(|z-x|+|w-x|))\}$ when $t=x>0$, and times $C_0\e^{C_0xn}$ when $t=0$ and $|z|,|w|\le x$ {(since for Potts, $\pi_a(z)\equiv 1/q$, $|M_{ab}(z)|=q/|q-1+z|\le M_{ab}(x)\e^{O(x)}$ for $a\ne b$, and $|M_{aa}(z)|=q|z|/|q-1+z|\le qx/|q-1+z|\le M_{aa}(x)\e^{O(x)}$ for $|z|\le x$, with $C_0>0$ independent of $x\in(0,x_*]$)}. Since $a_n$ is bounded away from zero on a smaller disk by \Cref{lem:local-profile-moments}(1),~\eqref{eq:three-truncation-errors} implies~\eqref{eq:holomorphic-truncation} {and the stated bound on $\E{(\mathcal Z_G(t)^2-\mathcal Z_{G,B}^{(2)}(t,t))\one{E_n}}$} after shrinking $D_0$ (fixing a single $x\in(0,x_*]$ satisfying {$c_2(1+\log(1/x))^{-6}>4C_0x$} when $t=0$, and taking $n\ge n_0(x)$).
\end{proof}

% END INCLUDED FILE: appendices/c-sec5-deduplicated.tex

\begin{thebibliography}{KMRT{\etalchar{+}}07}

\bibitem[AASV21]{alimohammadi2021fraction}
Yeganeh Alimohammadi, Nima Anari, Kirankumar Shiragur, and Thuy-Duong Vuong.
\newblock Fractionally log-concave and sector-stable polynomials: counting planar matchings and more.
\newblock In {\em STOC}, pages 433--446. ACM, 2021.

\bibitem[AC08]{achlioptas2008algorithmic}
Dimitris Achlioptas and Amin {Coja-Oghlan}.
\newblock Algorithmic barriers from phase transitions.
\newblock In {\em FOCS}, pages 793--802, 2008.

\bibitem[Ach26]{achlioptas2026sampling}
Dimitris Achlioptas.
\newblock A sampling {L}ov{\'a}sz local lemma.
\newblock {\em arXiv preprint arXiv:2609.18712}, 2026.

\bibitem[ACOG22]{ayre2022lower}
Peter Ayre, Amin Coja-Oghlan, and Catherine Greenhill.
\newblock Lower bounds on the chromatic number of random graphs.
\newblock {\em Combinatorica}, 42(5):617--658, 2022.

\bibitem[AIS19]{AIS19}
Dimitris Achlioptas, Fotis Iliopoulos, and Alistair Sinclair.
\newblock Beyond the {L}ovász local lemma: Point to set correlations and their algorithmic applications.
\newblock In {\em FOCS}, pages 725--744, 2019.

\bibitem[AJK{\etalchar{+}}22]{anari2022entropic}
Nima Anari, Vishesh Jain, Frederic Koehler, Huy~Tuan Pham, and Thuy-Duong Vuong.
\newblock Entropic independence: Optimal mixing of down-up random walks.
\newblock In {\em STOC}, pages 1418--1430. ACM, 2022.

\bibitem[AKV24]{anari2024trickle}
Nima Anari, Frederic Koehler, and Thuy-Duong Vuong.
\newblock Trickle-down in localization schemes and applications.
\newblock In {\em STOC}, pages 1094--1105. ACM, 2024.

\bibitem[AM97]{achlioptas1997analysis}
Dimitris Achlioptas and Michael Molloy.
\newblock The analysis of a list-coloring algorithm on a random graph.
\newblock In {\em FOCS}, pages 204--212, 1997.

\bibitem[AM04]{achlioptas2004chromatic}
Dimitris Achlioptas and Cristopher Moore.
\newblock The chromatic number of random regular graphs.
\newblock In {\em RANDOM}, pages 219--228. Springer, 2004.

\bibitem[AN05]{achlioptas2005two}
Dimitris Achlioptas and Assaf Naor.
\newblock The two possible values of the chromatic number of a random graph.
\newblock {\em Ann. Math.}, 162(3):1335--1351, 2005.
\newblock (Conference version in \emph{STOC}'04).

\bibitem[ANP05]{achlioptas2005rigorous}
Dimitris Achlioptas, Assaf Naor, and Yuval Peres.
\newblock Rigorous location of phase transitions in hard optimization problems.
\newblock {\em Nature}, 435:759--764, 2005.

\bibitem[Asa70]{asano1970lee}
Taro Asano.
\newblock {Lee-Yang} theorem and the {Griffiths} inequality for the anisotropic {Heisenberg} ferromagnet.
\newblock {\em Phys. Rev. Lett.}, 24(25):1409--1411, 1970.

\bibitem[Bar16]{barvinok2016combinatorics}
Alexander Barvinok.
\newblock {\em Combinatorics and complexity of partition functions}, volume~30 of {\em Algorithms and Combinatorics}.
\newblock Springer, 2016.

\bibitem[BB19]{bauerschmidt2019simple}
Roland Bauerschmidt and Thierry Bodineau.
\newblock A very simple proof of the {LSI} for high temperature spin systems.
\newblock {\em J. Funct. Anal.}, 276(8):2582--2588, 2019.

\bibitem[BBR26]{bencs2026deterministic}
Ferenc Bencs, Khallil Berrekkal, and Guus Regts.
\newblock Deterministic approximate counting of colorings with fewer than {$2\Delta$} colors via absence of zeros.
\newblock {\em TheoretiCS}, 5:1--41, 2026.
\newblock Article 1.

\bibitem[BCG22]{bordenave2022convergence}
Charles Bordenave, Djalil Chafa{\"\i}, and David {Garc{\'\i}a-Zelada}.
\newblock Convergence of the spectral radius of a random matrix through its characteristic polynomial.
\newblock {\em Probab. Theory Related Fields}, 182(3--4):1163--1181, 2022.

\bibitem[BCH{\etalchar{+}}16]{bapst2014condensation}
Victor Bapst, Amin {Coja-Oghlan}, Samuel Hetterich, Felicia Ra{\ss}mann, and Dan Vilenchik.
\newblock The condensation phase transition in random graph coloring.
\newblock {\em Commun. Math. Phys.}, 341(2):543--606, 2016.
\newblock (Conference version in \emph{RANDOM}'14).

\bibitem[Bet35]{bethe1935statistical}
Hans~A. Bethe.
\newblock Statistical theory of superlattices.
\newblock {\em Proc. R. Soc. Lond. A}, 150(871):552--575, 1935.

\bibitem[BGG{\etalchar{+}}20]{blanca2020sampling}
Antonio Blanca, Andreas Galanis, Leslie~Ann Goldberg, Daniel {\v{S}tefankovi\v{c}}, Eric Vigoda, and Kuan Yang.
\newblock Sampling in uniqueness from the {Potts} and random-cluster models on random regular graphs.
\newblock {\em SIAM J. Discrete Math.}, 34(1):742--793, 2020.
\newblock (Conference version in \emph{RANDOM}'18).

\bibitem[BHL{\etalchar{+}}26a]{bencs2026ising}
Ferenc Bencs, Brice Huang, Daniel~Z. Lee, Kuikui Liu, and Guus Regts.
\newblock On zeros and algorithms for disordered systems {II}: The antiferromagnetic {Ising} model up to {Kesten--Stigum}.
\newblock {\em arXiv preprint arXiv:2610.40598}, 2026.

\bibitem[BHL{\etalchar{+}}26b]{bencs2026zeros}
Ferenc Bencs, Brice Huang, Daniel~Z. Lee, Kuikui Liu, and Guus Regts.
\newblock On zeros and algorithms for disordered systems: Mean-field spin glasses.
\newblock In {\em STOC}, pages 508--517. ACM, 2026.

\bibitem[Bil99]{billingsley2013convergence}
Patrick Billingsley.
\newblock {\em Convergence of Probability Measures}.
\newblock John Wiley \& Sons, 2nd edition, 1999.

\bibitem[BKZZ13]{barbier2013hardcore}
Jean Barbier, Florent Krzakala, Lenka Zdeborov{\'a}, and Pan Zhang.
\newblock The hard-core model on random graphs revisited.
\newblock {\em J. Phys. Conf. Ser.}, 473(1):012021, 2013.

\bibitem[BLO26]{boban2026rank}
Mathews Boban, Anqi Li, and Shayan {Oveis Gharan}.
\newblock Rank-1-perturbed trickledown theorems: Mixing time of {Glauber} dynamics for the {Sherrington--Kirkpatrick} model up to $\beta\leq\frac{1}{2}+\varepsilon$.
\newblock {\em arXiv preprint arXiv:2609.13138}, 2026.

\bibitem[Bol80]{bollobas1980probabilistic}
B{\'e}la Bollob{\'a}s.
\newblock A probabilistic proof of an asymptotic formula for the number of labelled regular graphs.
\newblock {\em European J. Combin.}, 1(4):311--316, 1980.

\bibitem[BS17]{barvinok2017computing}
Alexander Barvinok and Pablo Sober{\'{o}}n.
\newblock Computing the partition function for graph homomorphisms.
\newblock {\em Combinatorica}, 37(4):633--650, 2017.

\bibitem[BST16]{Bhatnagar2016decay}
Nayantara Bhatnagar, Allan Sly, and Prasad Tetali.
\newblock Decay of correlations for the hardcore model on the $d$-regular random graph.
\newblock {\em Electron. J. Probab.}, 21:Paper No. 9, 42 pp., 2016.

\bibitem[Bur60]{burley1960lattice}
David~M. Burley.
\newblock A lattice model of a classical hard sphere gas.
\newblock {\em Proc. Phys. Soc.}, 75(2):262--274, 1960.

\bibitem[CCC{\etalchar{+}}25]{chen2025rapidrandom}
Xiaoyu Chen, Zejia Chen, Zongchen Chen, Yitong Yin, and Xinyuan Zhang.
\newblock Rapid mixing on random regular graphs beyond uniqueness.
\newblock In {\em FOCS}, pages 2170--2193. IEEE, 2025.

\bibitem[CDM{\etalchar{+}}19]{chen2019improved}
Sitan Chen, Michelle Delcourt, Ankur Moitra, Guillem Perarnau, and Luke Postle.
\newblock Improved bounds for randomly sampling colorings via linear programming.
\newblock In {\em SODA}, pages 2216--2234. SIAM, 2019.

\bibitem[CE25]{CE22}
Yuansi Chen and Ronen Eldan.
\newblock Localization schemes: A framework for proving mixing bounds for {Markov} chains.
\newblock {\em Duke Math. J.}, 174(8):1431--1510, 2025.
\newblock (Conference version in \emph{FOCS}'22).

\bibitem[CEH16]{coja2016chromatic}
Amin {Coja{-}Oghlan}, Charilaos Efthymiou, and Samuel Hetterich.
\newblock On the chromatic number of random regular graphs.
\newblock {\em J. Combin. Theory Ser. B}, 116:367--439, 2016.

\bibitem[CEJ{\etalchar{+}}18]{coja2018charting}
Amin {Coja-Oghlan}, Charilaos Efthymiou, Nor Jaafari, Mihyun Kang, and Tobias Kapetanopoulos.
\newblock Charting the replica symmetric phase.
\newblock {\em Commun. Math. Phys.}, 359(2):603--698, 2018.
\newblock (Conference version in \emph{RANDOM}'17).

\bibitem[Cel24]{celetano2024sudakov}
Michael Celentano.
\newblock {Sudakov–Fernique post-AMP, and a new proof of the local convexity of the TAP free energy}.
\newblock {\em Ann. Probab.}, 52(3):923--954, 2024.

\bibitem[CGG{\etalchar{+}}16]{cai2016hardness}
Jin-Yi Cai, Andreas Galanis, Leslie~Ann Goldberg, Heng Guo, Mark Jerrum, Daniel {\v{S}}tefankovi\v{c}, and Eric Vigoda.
\newblock {\#BIS}-hardness for 2-spin systems on bipartite bounded degree graphs in the tree non-uniqueness region.
\newblock {\em J. Comput. Syst. Sci.}, 82(5):690--711, 2016.
\newblock (Conference version in \emph{RANDOM}'14).

\bibitem[CG{\v{S}}V22]{chen2022sampling}
Zongchen Chen, Andreas Galanis, Daniel {\v{S}}tefankovi{\v{c}}, and Eric Vigoda.
\newblock Sampling colorings and independent sets of random regular bipartite graphs in the non-uniqueness region.
\newblock In {\em SODA}, pages 2198--2207. SIAM, 2022.

\bibitem[CHKL{\etalchar{+}}20]{coja2020inference}
Amin {Coja-Oghlan}, Max Hahn-Klimroth, Philipp Loick, No{\"e}la M{\"u}ller, Konstantinos Panagiotou, and Matija Pasch.
\newblock Inference and mutual information on random factor graphs.
\newblock {\em arXiv preprint arXiv:2007.07494}, 2020.
\newblock (Conference version in \emph{STACS}'21).

\bibitem[CKM20]{coja2020replica}
Amin {Coja-Oghlan}, Tobias Kapetanopoulos, and No{\"e}la M{\"u}ller.
\newblock The replica symmetric phase of random constraint satisfaction problems.
\newblock {\em Combin. Probab. Comput.}, 29(3):346--422, 2020.

\bibitem[CKPZ18]{coja2017information}
Amin {Coja-Oghlan}, Florent Krzakala, Will Perkins, and Lenka Zdeborov{\'a}.
\newblock Information-theoretic thresholds from the cavity method.
\newblock {\em Adv. Math.}, 333:694--795, 2018.
\newblock (Conference version in \emph{STOC}'17).

\bibitem[CLMS22]{coja2022ising}
Amin {Coja-Oghlan}, Philipp Loick, Bal{\'a}zs~F. Mezei, and Gregory~B. Sorkin.
\newblock The {Ising} antiferromagnet and {MAX CUT} on random regular graphs.
\newblock {\em SIAM J. Discrete Math.}, 36(2):1306--1342, 2022.

\bibitem[CLV24]{chen2022spectral}
Zongchen Chen, Kuikui Liu, and Eric Vigoda.
\newblock Spectral independence via stability and applications to {Holant}-type problems.
\newblock {\em TheoretiCS}, 3:Article 16, 1--49, 2024.
\newblock (Conference version in \emph{FOCS}'21).

\bibitem[CMR22]{coja2022belief}
Amin {Coja-Oghlan}, No{\"e}la M{\"u}ller, and Jean~B. Ravelomanana.
\newblock Belief propagation on the random $k$-{SAT} model.
\newblock {\em Ann. Appl. Probab.}, 32(5):3718--3796, 2022.

\bibitem[CP16]{coja2014asymptotic}
Amin {Coja-Oghlan} and Konstantinos Panagiotou.
\newblock The asymptotic $k$-{SAT} threshold.
\newblock {\em Adv. Math.}, 288:985--1068, 2016.
\newblock (Conference versions in \emph{STOC}'13 and \emph{STOC}'14, under different titles).

\bibitem[CP18]{coja2016belief}
Amin {Coja-Oghlan} and Will Perkins.
\newblock Belief propagation on replica symmetric random factor graph models.
\newblock {\em Ann. Inst. Henri Poincar\'e Comb. Phys. Interact.}, 5(2):211--249, 2018.
\newblock (Conference version in \emph{RANDOM}'16).

\bibitem[CP19]{coja2019spin}
Amin {Coja-Oghlan} and Will Perkins.
\newblock Spin systems on {Bethe} lattices.
\newblock {\em Commun. Math. Phys.}, 372(2):441--523, 2019.

\bibitem[CS92]{crisanti1992spherical}
Andrea Crisanti and Hans-J{\"u}rgen Sommers.
\newblock The spherical $p$-spin interaction spin glass model: the statics.
\newblock {\em Z. Phys. B: Condens. Matter}, 87(3):341--354, 1992.

\bibitem[CSS26]{chen2026connective}
Yuan Chen, Shuai Shao, and Ke~Shi.
\newblock Zero-freeness of the hard-core model with bounded connective constant.
\newblock {\em arXiv preprint arXiv:2604.02746}, 2026.

\bibitem[CV16]{coja2013chasing}
Amin {Coja-Oghlan} and Dan Vilenchik.
\newblock The chromatic number of random graphs for most average degrees.
\newblock {\em Int. Math. Res. Not.}, 2016(19):5801--5859, 2016.
\newblock (Conference version, ``Chasing the $k$-Colorability Threshold,'' in \emph{FOCS}'13).

\bibitem[CV25]{carlson2025flip}
Charlie Carlson and Eric Vigoda.
\newblock Flip dynamics for sampling colorings: Improving $(11/6-\varepsilon)$ using a simple metric.
\newblock In {\em SODA}, pages 2194--2212. SIAM, 2025.

\bibitem[CW26]{chen2026high}
Zongchen Chen and Sihan Wang.
\newblock High-dimensional ultra-log-concave distributions.
\newblock {\em arXiv preprint arXiv:2609.23994}, 2026.

\bibitem[DFJ02]{dyer2002counting}
Martin Dyer, Alan Frieze, and Mark Jerrum.
\newblock On counting independent sets in sparse graphs.
\newblock {\em SIAM J. Comput.}, 31(5):1527--1541, 2002.
\newblock (Conference version in \emph{FOCS}'99).

\bibitem[DHMR26]{du2026sampling}
Hang Du, Brice Huang, Sidhanth Mohanty, and Amit Rajaraman.
\newblock Sampling from the {Sherrington-Kirkpatrick} model up to criticality.
\newblock {\em arXiv preprint arXiv:2610.17717}, 2026.

\bibitem[dJKT78]{deAlmeida1978infinite}
Jairo R.~L. {de Almeida}, Raymund~C. Jones, John~M. Kosterlitz, and David~J. Thouless.
\newblock The infinite-ranged spin glass with $m$-component spins.
\newblock {\em J. Phys. C: Solid State Phys.}, 11(21):L871--L875, 1978.

\bibitem[DKMZ11]{decelle2011asymptotic}
Aur{\'e}lien Decelle, Florent Krzakala, Cristopher Moore, and Lenka Zdeborov{\'a}.
\newblock Asymptotic analysis of the stochastic block model for modular networks and its algorithmic applications.
\newblock {\em Phys. Rev. E}, 84(6):066106, 2011.

\bibitem[DLSS26a]{davies2026potential}
Ewan Davies, Holden Lee, Juspreet~Singh Sandhu, and Jonathan Shi.
\newblock {Potential Hessian Ascent} {III}: Sampling the {Sherrington--Kirkpatrick} model at $\beta<1/2$.
\newblock {\em arXiv preprint arXiv:2605.03718}, 2026.

\bibitem[DLSS26b]{davies2026weak}
Ewan Davies, Holden Lee, Juspreet~Singh Sandhu, and Jonathan Shi.
\newblock Weak {Poincar\'e} inequalities via approximate stochastic localization: Application to sampling the {Sherrington--Kirkpatrick} model.
\newblock {\em arXiv preprint arXiv:2607.08160}, 2026.

\bibitem[DM10]{dembo2010gibbs}
Amir Dembo and Andrea Montanari.
\newblock Gibbs measures and phase transitions on sparse random graphs.
\newblock {\em Braz. J. Probab. Stat.}, 24(2):137--211, 2010.

\bibitem[DMS13]{dembo2013factor}
Amir Dembo, Andrea Montanari, and Nike Sun.
\newblock Factor models on locally tree-like graphs.
\newblock {\em Ann. Probab.}, 41(6):4162--4213, 2013.

\bibitem[DMSS14]{dembo2014replica}
Amir Dembo, Andrea Montanari, Allan Sly, and Nike Sun.
\newblock The replica symmetric solution for {Potts} models on $d$-regular graphs.
\newblock {\em Commun. Math. Phys.}, 327(2):551--575, 2014.

\bibitem[Dom58]{domb1958some}
Cyril Domb.
\newblock Some theoretical aspects of melting.
\newblock {\em Nuovo Cim.}, 9(Suppl. 1):9--26, 1958.

\bibitem[DSS22]{ding2022satisfiability}
Jian Ding, Allan Sly, and Nike Sun.
\newblock Proof of the satisfiability conjecture for large $k$.
\newblock {\em Ann. Math.}, 196(1):1--388, 2022.
\newblock (Conference version in \emph{STOC}'15).

\bibitem[EF23]{efthymiou2023mixing}
Charilaos Efthymiou and Weiming Feng.
\newblock On the mixing time of {Glauber} dynamics for the hard-core and related models on {$G(n,d/n)$}.
\newblock In {\em ICALP}, pages 54:1--54:17. LZI, 2023.

\bibitem[Eft22]{efthymiou2022sampling}
Charilaos Efthymiou.
\newblock On sampling symmetric {Gibbs} distributions on sparse random graphs and hypergraphs.
\newblock In {\em ICALP}, volume 229, pages 57:1--57:16. LZI, 2022.

\bibitem[Ehr33]{ehrenfest1933phasenumwandlungen}
Paul Ehrenfest.
\newblock Phasenumwandlungen im {\"u}blichen und erweiterten {Sinn}, classifiziert nach den entsprechenden {Singularit{\"a}ten} des thermodynamischen {Potentiales}.
\newblock {\em Proc. K. Ned. Akad. Wet.}, 36(2):153--157, 1933.

\bibitem[EKZ22]{eldan2022spectral}
Ronen Eldan, Frederic Koehler, and Ofer Zeitouni.
\newblock A spectral condition for spectral gap: fast mixing in high-temperature {Ising} models.
\newblock {\em Probab. Theory Related Fields}, 182(3--4):1035--1051, 2022.

\bibitem[EMS25]{alaoui2025sampling}
Ahmed {El Alaoui}, Andrea Montanari, and Mark Sellke.
\newblock Sampling from mean-field {Gibbs} measures via diffusion processes.
\newblock {\em Probab. Math. Phys.}, 6(3):961--1022, 07 2025.
\newblock (Preliminary conference version for $\beta<1/2$ appeared as ``Sampling from the {Sherrington-Kirkpatrick} {Gibbs} measure via algorithmic stochastic localization'' in \emph{FOCS}'22).

\bibitem[Fis65]{fisher1965nature}
Michael~E. Fisher.
\newblock The nature of critical points.
\newblock In {\em Lectures in Theoretical Physics}, volume VII C, pages 1--159. University of Colorado Press, Boulder, CO, 1965.

\bibitem[F{\L}92]{frieze1992independence}
Alan~M. Frieze and Tomasz {\L}uczak.
\newblock On the independence and chromatic numbers of random regular graphs.
\newblock {\em J. Combin. Theory Ser. B}, 54(1):123--132, 1992.

\bibitem[FL21]{fabian2021ising}
Christian Fabian and Philipp Loick.
\newblock The {Ising} antiferromagnet in the replica symmetric phase.
\newblock In {\em EUROCOMB}, pages 299--306. Birkh{\"a}user, 2021.

\bibitem[Fri90]{frieze1990independence}
Alan~M. Frieze.
\newblock On the independence number of random graphs.
\newblock {\em Discrete Math.}, 81(2):171--175, 1990.

\bibitem[Gam23]{gamarnik2023correlation}
David Gamarnik.
\newblock Correlation decay and the absence of zeros property of partition functions.
\newblock {\em Random Struct. Algorithms}, 62(1):155--180, 2023.

\bibitem[Gar85]{gardner1985spin}
Elizabeth Gardner.
\newblock Spin glasses with $p$-spin interactions.
\newblock {\em Nucl. Phys. B}, 257:747--765, 1985.

\bibitem[GGS21]{galanis2021fast}
Andreas Galanis, Leslie~Ann Goldberg, and James Stewart.
\newblock Fast algorithms for general spin systems on bipartite expanders.
\newblock {\em ACM Trans. Comput. Theory}, 13(4):25:1--25:18, 2021.
\newblock (Conference version in \emph{MFCS}'20).

\bibitem[GH25]{gerencser2025boosted}
Bal{\'a}zs Gerencs{\'e}r and Viktor Harangi.
\newblock Boosted second moment method in random regular graphs.
\newblock {\em arXiv preprint arXiv:2510.12600}, 2025.
\newblock Version 3, revised March 12, 2026.

\bibitem[GJM{\etalchar{+}}26]{gobel2026simple}
Andreas G{\"o}bel, Matthew Jenssen, Marcus Michelen, Marcus Pappik, Will Perkins, and Leon Schiller.
\newblock A simple proof of rapid mixing on random regular graphs beyond uniqueness.
\newblock {\em arXiv preprint arXiv:2606.27545}, 2026.

\bibitem[GKS85]{gross1985mean}
David~J. Gross, Ido Kanter, and Haim Sompolinsky.
\newblock Mean-field theory of the {Potts} glass.
\newblock {\em Phys. Rev. Lett.}, 55(3):304--307, 1985.

\bibitem[GM84]{gross1984simplest}
David~J. Gross and Marc M{\'e}zard.
\newblock The simplest spin glass.
\newblock {\em Nucl. Phys. B}, 240(4):431--452, 1984.

\bibitem[GPY25]{gheissari2025fixed}
Reza Gheissari, Will Perkins, and Corrine Yap.
\newblock Fixed-magnetization {Ising} on random graphs up to reconstruction.
\newblock {\em arXiv preprint arXiv:2511.16620}, 2025.

\bibitem[GS17]{gamarnik2017limits}
David Gamarnik and Madhu Sudan.
\newblock Limits of local algorithms over sparse random graphs.
\newblock {\em Ann. Probab.}, 45(4):2353--2376, 2017.
\newblock (Conference version in \emph{ITCS}'14).

\bibitem[G{\v{S}}V15]{galanis2015inapproximability}
Andreas Galanis, Daniel {\v{S}}tefankovi\v{c}, and Eric Vigoda.
\newblock Inapproximability for antiferromagnetic spin systems in the tree nonuniqueness region.
\newblock {\em J. ACM}, 62(6):50:1--50:60, 2015.
\newblock (Conference version in \emph{STOC}'14).

\bibitem[GT81]{gabay1981coexistence}
Marc Gabay and G{\'e}rard Toulouse.
\newblock Coexistence of spin-glass and ferromagnetic orderings.
\newblock {\em Phys. Rev. Lett.}, 47(3):201--204, Jul 1981.

\bibitem[HL72]{heilmann1972theory}
Ole~J. Heilmann and Elliott~H. Lieb.
\newblock Theory of monomer-dimer systems.
\newblock {\em Commun. Math. Phys.}, 25(3):190--232, 1972.

\bibitem[HLMM26]{harbuzova2026color}
Alina Harbuzova, Saba Lepsveridze, Mahbod Majid, and Ankur Moitra.
\newblock Color coding for the {Sherrington--Kirkpatrick} model.
\newblock {\em arXiv preprint arXiv:2610.10443}, 2026.

\bibitem[HLY26]{hou2026simulated}
Zhe Hou, Jingcheng Liu, and Yixiao Yu.
\newblock Simulated annealing and weak {Poincar\'e} inequalities with inverse-polynomial accuracy for the {Sherrington--Kirkpatrick} model.
\newblock {\em arXiv preprint arXiv:2610.09882}, 2026.

\bibitem[Jan95]{janson1995random}
Svante Janson.
\newblock Random regular graphs: Asymptotic distributions and contiguity.
\newblock {\em Combin. Probab. Comput.}, 4(4):369--405, 1995.

\bibitem[Jan14]{janson2014simple}
Svante Janson.
\newblock The probability that a random multigraph is simple, {II}.
\newblock {\em J. Appl. Probab.}, 51A:123--137, 2014.

\bibitem[Jer95]{jerrum1995simple}
Mark Jerrum.
\newblock A very simple algorithm for estimating the number of {$k$}-colorings of a low-degree graph.
\newblock {\em Random Struct. Algorithms}, 7(2):157--165, 1995.

\bibitem[JKP20]{jenssen2020algorithms}
Matthew Jenssen, Peter Keevash, and Will Perkins.
\newblock Algorithms for \#{BIS}-hard problems on expander graphs.
\newblock {\em SIAM J. Comput.}, 49(4):681--710, 2020.
\newblock (Conference version in \emph{SODA}'19).

\bibitem[JMP26]{jain2026general}
Vishesh Jain, Clayton Mizgerd, and Huy~Tuan Pham.
\newblock A general counting and sampling {L}ov{\'a}sz local lemma.
\newblock {\em arXiv preprint arXiv:2609.23671}, 2026.

\bibitem[JMPV23]{vishesh2023optimal}
Vishesh Jain, Marcus Michelen, Huy~Tuan Pham, and Thuy-Duong Vuong.
\newblock {Optimal mixing of the down-up walk on independent sets of a given size}.
\newblock In {\em FOCS}, pages 1665--1681. IEEE, 2023.

\bibitem[JPPS26]{jenssen2026replica}
Matthew Jenssen, Will Perkins, Aditya Potukuchi, and Michael Simkin.
\newblock The replica symmetric solution for hypergraph independent sets in the critical regime.
\newblock {\em arXiv preprint arXiv:2610.10433}, 2026.

\bibitem[JPSS22]{vishesh2022approximate}
Vishesh Jain, Will Perkins, Ashwin Sah, and Mehtaab Sawhney.
\newblock Approximate counting and sampling via local central limit theorems.
\newblock In {\em {STOC}}, pages 1473--1486. ACM, 2022.

\bibitem[JVV86]{jerrum1986random}
Mark~R. Jerrum, Leslie~G. Valiant, and Vijay~V. Vazirani.
\newblock Random generation of combinatorial structures from a uniform distribution.
\newblock {\em Theoret. Comput. Sci.}, 43(2-3):169--188, 1986.
\newblock (Earlier extended abstract by Jerrum in \emph{ICALP}'85).

\bibitem[Kah00]{kahn2000normal}
Jeff Kahn.
\newblock A normal law for matchings.
\newblock {\em Combinatorica}, 20(3):339--391, 2000.

\bibitem[Kau48]{kauzmann1948nature}
Walter Kauzmann.
\newblock The nature of the glassy state and the behavior of liquids at low temperatures.
\newblock {\em Chem. Rev.}, 43(2):219--256, 10 1948.

\bibitem[KKO24]{karlin2021improved}
Anna~R. Karlin, Nathan Klein, and Shayan {Oveis Gharan}.
\newblock A (slightly) improved approximation algorithm for metric {TSP}.
\newblock {\em Oper. Res.}, 72(6):2543--2594, 2024.
\newblock (Conference version in \emph{STOC}'21).

\bibitem[KL23]{keevash2023sharp}
Peter Keevash and Noam Lifshitz.
\newblock Sharp hypercontractivity for symmetric groups and its applications.
\newblock {\em arXiv preprint arXiv:2307.15030}, 2023.
\newblock Version 1.

\bibitem[KLM26]{keller2025sharp}
Nathan Keller, Noam Lifshitz, and Omri Marcus.
\newblock Sharp hypercontractivity for global functions.
\newblock {\em J. Eur. Math. Soc.}, 2026.
\newblock Published online first; theorem numbering refers to arXiv:2307.01356v2.

\bibitem[KLR22]{koehler2022sampling}
Frederic Koehler, Holden Lee, and Andrej Risteski.
\newblock Sampling approximately low-rank {Ising} models: {MCMC} meets variational methods.
\newblock In {\em COLT}, volume 178, pages 4945--4988. PMLR, 2022.

\bibitem[KMRT{\etalchar{+}}07]{florent2007gibbs}
Florent Krzakala, Andrea Montanari, Federico Ricci-Tersenghi, Guilhem Semerjian, and Lenka Zdeborov{\'a}.
\newblock Gibbs states and the set of solutions of random constraint satisfaction problems.
\newblock {\em Proc. Natl. Acad. Sci. U.S.A.}, 104(25):10318--10323, 2007.

\bibitem[KOT26]{kocurek2026sampling}
Nicholas Kocurek, Shayan {Oveis Gharan}, and Dante Tjowasi.
\newblock Sampling from the hardcore model on random regular bipartite graphs above the uniqueness threshold.
\newblock {\em arXiv preprint arXiv:2604.21847}, 2026.

\bibitem[KPPY25]{aiya2025fast}
Aiya Kuchukova, Marcus Pappik, Will Perkins, and Corrine Yap.
\newblock Fast and slow mixing of the {Kawasaki} dynamics on bounded-degree graphs.
\newblock {\em Random Struct. Algorithms}, 67(4):e70038, 2025.
\newblock (Conference version in \emph{RANDOM}'24).

\bibitem[KPW10]{kemkes2010chromatic}
Graeme Kemkes, Xavier {P{\'e}rez-Gim{\'e}nez}, and Nicholas Wormald.
\newblock On the chromatic number of random $d$-regular graphs.
\newblock {\em Adv. Math.}, 223(1):300--328, 2010.

\bibitem[KS66]{kesten1966additional}
Harry Kesten and Bernt~P. Stigum.
\newblock Additional limit theorems for indecomposable multidimensional {Galton-Watson} processes.
\newblock {\em Ann. Math. Statist.}, 37(6):1463--1481, 1966.

\bibitem[KS78]{kirpatrick1978infinite}
Scott Kirkpatrick and David Sherrington.
\newblock Infinite-ranged models of spin-glasses.
\newblock {\em Phys. Rev. B}, 17(11):4384--4403, 1978.

\bibitem[KS18]{kyng2018matrix}
Rasmus Kyng and Zhao Song.
\newblock A matrix {Chernoff} bound for strongly {Rayleigh} distributions and spectral sparsifiers from a few random spanning trees.
\newblock In {\em FOCS}, pages 373--384. IEEE, 2018.

\bibitem[KZ08]{florent2008Potts}
Florent Krzakala and Lenka Zdeborov\'a.
\newblock {Potts} glass on random graphs.
\newblock {\em Europhys. Lett.}, 81(5):57005, 2008.

\bibitem[LLLM22]{liao2019counting}
Chao Liao, Jiabao Lin, Pinyan Lu, and Zhenyu Mao.
\newblock An {FPTAS} for the hardcore model on random regular bipartite graphs.
\newblock {\em Theoret. Comput. Sci.}, 929:174--190, 2022.

\bibitem[LLY26]{li2026counting}
Zhidan Li, Siyu Liu, and Kuan Yang.
\newblock Counting and sampling anti-ferromagnetic {Potts} models on random regular bipartite graphs in the non-uniqueness regime.
\newblock {\em arXiv preprint arXiv:2606.21250}, 2026.

\bibitem[LO18]{lelarge2018replica}
Marc Lelarge and Mendes Oulamara.
\newblock Replica bounds by combinatorial interpolation for diluted spin systems.
\newblock {\em J. Stat. Phys.}, 173(3--4):917--940, 2018.

\bibitem[LPRS16]{lebowitz2016central}
Joel~L. Lebowitz, Boris Pittel, David Ruelle, and Eugene~R. Speer.
\newblock Central limit theorems, {Lee--Yang} zeros, and graph-counting polynomials.
\newblock {\em J. Combin. Theory Ser. A}, 141:147--183, 2016.

\bibitem[LS81]{lieb1981general}
Elliott~H. Lieb and Alan~D. Sokal.
\newblock A general {Lee-Yang} theorem for one-component and multicomponent ferromagnets.
\newblock {\em Commun. Math. Phys.}, 80(2):153--179, 1981.

\bibitem[LSS08]{laumann2008griffiths}
Christopher~R. Laumann, Antonello Scardicchio, and Shivaji~L. Sondhi.
\newblock {Griffiths-McCoy} singularities, {Lee-Yang} zeros, and the cavity method in a solvable diluted ferromagnet.
\newblock {\em Phys. Rev. E}, 77(6):061139, 2008.

\bibitem[LSS25]{liu2025correlation}
Jingcheng Liu, Alistair Sinclair, and Piyush Srivastava.
\newblock Correlation decay and partition function zeros: Algorithms and phase transitions.
\newblock {\em SIAM J. Comput.}, 54(4):FOCS19{\char45}200--FOCS19{\char45}252, 2025.
\newblock (Conference version in \emph{FOCS}'19).

\bibitem[LSS26]{lee2026potentialhessianascentiv}
Holden Lee, Juspreet~Singh Sandhu, and Jonathan Shi.
\newblock {Potential Hessian Ascent} {IV}: Sampling the {Sherrington--Kirkpatrick} model at $\beta<1$.
\newblock {\em arXiv preprint arXiv:2609.30590}, 2026.

\bibitem[LWY{\etalchar{+}}26]{liu2026counting}
Hongyang Liu, Chunyang Wang, Yitong Yin, Yiyao Zhang, and Can Zhou.
\newblock A counting and sampling {L}ov{\'a}sz local lemma.
\newblock {\em arXiv preprint arXiv:2608.08616}, 2026.

\bibitem[LWYY25]{liu2025phase}
Jingcheng Liu, Chunyang Wang, Yitong Yin, and Yixiao Yu.
\newblock Phase transitions via complex extensions of {Markov} chains.
\newblock In {\em STOC}, pages 903--914. ACM, 2025.

\bibitem[LY52]{lee1952lattice}
Tsung-Dao Lee and Chen-Ning Yang.
\newblock Statistical theory of equations of state and phase transitions. {II}. {L}attice gas and {Ising} model.
\newblock {\em Phys. Rev.}, 87(3):410--419, 1952.

\bibitem[LY26a]{li2026fpras}
Zhidan Li and Kuan Yang.
\newblock An {FPRAS} for antiferromagnetic {Ising} models on random regular bipartite graphs.
\newblock {\em arXiv preprint arXiv:2608.18612}, 2026.

\bibitem[LY26b]{li2026rapid}
Zhidan Li and Kuan Yang.
\newblock Rapid mixing for sampling colorings of random graphs up to tree uniqueness threshold.
\newblock {\em arXiv preprint arXiv:2610.25391}, 2026.

\bibitem[LY26c]{liu2026zero}
Jingcheng Liu and Yixiao Yu.
\newblock Zero-free regions and concentration inequalities for hypergraph colorings in the local lemma regime.
\newblock In {\em STOC}, pages 991--1002. ACM, 2026.

\bibitem[McD84]{mcdiarmid1984colouring}
Colin McDiarmid.
\newblock Colouring random graphs.
\newblock {\em Ann. Oper. Res.}, 1(3):183--200, 1984.

\bibitem[MM06]{mezard2006reconstruction}
Marc M{\'e}zard and Andrea Montanari.
\newblock Reconstruction on trees and spin glass transition.
\newblock {\em J. Stat. Phys.}, 124(6):1317--1350, 2006.

\bibitem[MM09]{mezard2009information}
Marc M{\'e}zard and Andrea Montanari.
\newblock {\em Information, Physics, and Computation}.
\newblock Oxford Graduate Texts. Oxford University Press, 2009.

\bibitem[MMN{\etalchar{+}}10]{matsuda2010distribution}
Yoshiki Matsuda, Markus M{\"u}ller, Hidetoshi Nishimori, Tomoyuki Obuchi, and Antonello Scardicchio.
\newblock Distribution of partition function zeros of the {$\pm J$} model on the {Bethe} lattice.
\newblock {\em J. Phys. A: Math. Theor.}, 43(28):285002, 2010.

\bibitem[Mol19]{molloy2019list}
Michael Molloy.
\newblock The list chromatic number of graphs with small clique number.
\newblock {\em J. Combin. Theory Ser. B}, 134:264--284, 2019.

\bibitem[Mon98]{monasson1998otpimization}
R{\'e}mi Monasson.
\newblock Optimization problems and replica symmetry breaking in finite connectivity spin glasses.
\newblock {\em J. Phys. A: Math. Gen.}, 31(2):513--529, 1998.

\bibitem[Mon01]{montanari2001glassy}
Andrea Montanari.
\newblock The glassy phase of {Gallager} codes.
\newblock {\em Eur. Phys. J. B}, 23(1):121--136, 2001.

\bibitem[MP01]{mezard2001bethe}
Marc M{\'e}zard and Giorgio Parisi.
\newblock The {Bethe} lattice spin glass revisited.
\newblock {\em Eur. Phys. J. B}, 20(2):217--233, 2001.

\bibitem[MP03a]{mezard2003cavity}
Marc M{\'e}zard and Giorgio Parisi.
\newblock {The cavity method at zero temperature}.
\newblock {\em {J. Stat. Phys.}}, 111(1--2):1--34, 2003.

\bibitem[MP03b]{mossel2003information}
Elchanan Mossel and Yuval Peres.
\newblock Information flow on trees.
\newblock {\em Ann. Appl. Probab.}, 13(3):817--844, 2003.

\bibitem[MPV86]{mezard1986SK}
Marc M{\'e}zard, Giorgio Parisi, and Miguel~A. Virasoro.
\newblock {SK} model: The replica solution without replicas.
\newblock {\em Europhys. Lett.}, 1(2):77--82, 1986.

\bibitem[MPV87]{mezard1987spin}
Marc M{\'e}zard, Giorgio Parisi, and Miguel~Angel Virasoro.
\newblock {\em Spin Glass Theory and Beyond: An Introduction to the Replica Method and Its Applications}, volume~9 of {\em World Scientific Lecture Notes in Physics}.
\newblock World Scientific Publishing Co., 1987.

\bibitem[MPZ02]{mezard2002analytic}
Marc M{\'e}zard, Giorgio Parisi, and Riccardo Zecchina.
\newblock Analytic and algorithmic solution of random satisfiability problems.
\newblock {\em Science}, 297(5582):812--815, 2002.

\bibitem[MR25]{mohanty2025eigenvalue}
Sidhanth Mohanty and Amit Rajaraman.
\newblock Eigenvalue bounds for random matrices via zerofreeness.
\newblock {\em arXiv preprint arXiv:2509.25471}, 2025.

\bibitem[MS26]{marcus2024central}
Marcus Michelen and Julian Sahasrabudhe.
\newblock Central limit theorems and the geometry of polynomials.
\newblock {\em J. Eur. Math. Soc.}, 28(5):2261--2305, 2026.

\bibitem[MSW04]{martinelli2004glauber}
Fabio Martinelli, Alistair Sinclair, and Dror Weitz.
\newblock Glauber dynamics on trees: Boundary conditions and mixing time.
\newblock {\em Commun. Math. Phys.}, 250:301--334, 09 2004.
\newblock (Conference version, ``The Ising Model on Trees: Boundary Conditions and Mixing Time,'' in \emph{FOCS}'03).

\bibitem[MSW07]{martinelli2007fast}
Fabio Martinelli, Alistair Sinclair, and Dror Weitz.
\newblock Fast mixing for independent sets, colorings, and other models on trees.
\newblock {\em Random Struct. Algorithms}, 31(2):134--172, 09 2007.
\newblock (Conference version in \emph{SODA}'04).

\bibitem[MWW09]{mossel2007hardness}
Elchanan Mossel, Dror Weitz, and Nicholas~C. Wormald.
\newblock On the hardness of sampling independent sets beyond the tree threshold.
\newblock {\em Probab. Theory Related Fields}, 143(3--4):401--439, 2009.
\newblock arXiv:math/0701471.

\bibitem[New74]{newman1974zeros}
Charles~M. Newman.
\newblock Zeros of the partition function for generalized {Ising} systems.
\newblock {\em Commun. Pure Appl. Math.}, 27(2):143--159, 1974.

\bibitem[{Ope}26a]{OpenAISKCriticalMixing26}
{OpenAI}.
\newblock Critical mixing in the {Sherrington--Kirkpatrick} model.
\newblock \href{https://github.com/openai/math/blob/1581d4cbdd69a0faf34b2089614fe2a814316a6a/preprints/Critical-mixing-in-the-Sherrington-Kirkpatrick-model-September-25-2026/paper.pdf}{OpenAI Math Release preprint}, 2026.

\bibitem[{Ope}26b]{OpenAISKCriticalSlowing26}
{OpenAI}.
\newblock Critical slowing down in the {Sherrington--Kirkpatrick} model.
\newblock \href{https://github.com/openai/math/blob/1581d4cbdd69a0faf34b2089614fe2a814316a6a/preprints/Critical-slowing-down-in-the-Sherrington-Kirkpatrick-model-September-24-2026/paper.pdf}{OpenAI Math Release preprint}, 2026.

\bibitem[{Ope}26c]{OpenAISKCutoff26}
{OpenAI}.
\newblock Cutoff throughout the high-temperature {Sherrington--Kirkpatrick} phase.
\newblock \href{https://github.com/openai/math/blob/1581d4cbdd69a0faf34b2089614fe2a814316a6a/preprints/Cutoff-throughout-the-high-temperature-Sherrington-Kirkpatrick-phase-September-24-2026/paper.pdf}{OpenAI Math Release preprint}, 2026.

\bibitem[{Ope}26d]{OpenAISKSpectralGap26}
{OpenAI}.
\newblock A spectral gap throughout the high-temperature {Sherrington--Kirkpatrick} phase.
\newblock \href{https://github.com/openai/math/blob/1581d4cbdd69a0faf34b2089614fe2a814316a6a/preprints/A-spectral-gap-throughout-the-high-temperature-Sherrington-Kirkpatrick-phase-September-24-2026/main.pdf}{OpenAI Math Release preprint}, 2026.

\bibitem[OSS11]{shayan2011randomized}
Shayan {Oveis Gharan}, Amin Saberi, and Mohit Singh.
\newblock A randomized rounding approach to the {Traveling Salesman Problem}.
\newblock In {\em FOCS}, pages 550--559. IEEE, 2011.

\bibitem[Pan13]{panchenko2013Sherrington}
Dmitry Panchenko.
\newblock {\em The Sherrington-Kirkpatrick Model}.
\newblock Springer, 2013.

\bibitem[Par79]{parisi1979infinite}
Giorgio Parisi.
\newblock Infinite number of order parameters for spin-glasses.
\newblock {\em Phys. Rev. Lett.}, 43(23):1754--1756, 1979.

\bibitem[Par83]{parisi1983order}
Giorgio Parisi.
\newblock Order parameter for spin-glasses.
\newblock {\em Phys. Rev. Lett.}, 50:1946--1948, 1983.

\bibitem[Pot52]{Potts1952generalized}
Renfrey~B. Potts.
\newblock Some generalized order-disorder transformations.
\newblock {\em Math. Proc. Camb. Philos. Soc.}, 48(1):106--109, 1952.

\bibitem[PP14]{pemantle2014concentration}
Robin Pemantle and Yuval Peres.
\newblock Concentration of {Lipschitz} functionals of determinantal and other strong {Rayleigh} measures.
\newblock {\em Combin. Probab. Comput.}, 23(1):140--160, 2014.

\bibitem[PR17]{patel2017deterministic}
Viresh Patel and Guus Regts.
\newblock Deterministic polynomial-time approximation algorithms for partition functions and graph polynomials.
\newblock {\em SIAM J. Comput.}, 46(6):1893--1919, 2017.
\newblock (Conference version in \emph{EuroComb}'17).

\bibitem[PS26]{prodromidis2026distribution}
Kyprianos-Iason Prodromidis and Allan Sly.
\newblock Distribution of the magnetization of the critical {Ising} model on sparse random graphs.
\newblock {\em arXiv preprint arXiv:2603.28702}, 2026.

\bibitem[Reg23]{regts2023absence}
Guus Regts.
\newblock Absence of zeros implies strong spatial mixing.
\newblock {\em Probab. Theory Related Fields}, 186(1--2):621--641, 2023.

\bibitem[RU08]{richardson2008modern}
Tom Richardson and R{\"u}diger Urbanke.
\newblock {\em Modern Coding Theory}.
\newblock Cambridge University Press, Cambridge, UK, 2008.

\bibitem[Rue71]{ruelle1971extension}
David Ruelle.
\newblock Extension of the {Lee-Yang} circle theorem.
\newblock {\em Phys. Rev. Lett.}, 26(6):303--304, 1971.

\bibitem[RV17]{rahman2017local}
Mustazee Rahman and B{\'a}lint Vir{\'a}g.
\newblock Local algorithms for independent sets are half-optimal.
\newblock {\em Ann. Probab.}, 45(3):1543--1577, 2017.

\bibitem[RW92]{robinson1992almost}
Robert~W. Robinson and Nicholas~C. Wormald.
\newblock Almost all cubic graphs are {Hamiltonian}.
\newblock {\em Random Struct. Algorithms}, 3(2):117--125, 1992.

\bibitem[RW94]{robinson1994almost}
Robert~W. Robinson and Nicholas~C. Wormald.
\newblock Almost all regular graphs are {Hamiltonian}.
\newblock {\em Random Struct. Algorithms}, 5(2):363--374, 1994.

\bibitem[She85]{shearer85}
James~B. Shearer.
\newblock On a problem of {S}pencer.
\newblock {\em Combinatorica}, 5(3):241--245, 1985.

\bibitem[SK75]{kirkpatrick1975solvable}
David Sherrington and Scott Kirkpatrick.
\newblock Solvable model of a spin-glass.
\newblock {\em Phys. Rev. Lett.}, 35(26):1792--1796, 1975.

\bibitem[SS05]{scott2005repulsive}
Alexander~D. Scott and Alan~D. Sokal.
\newblock The repulsive lattice gas, the independent-set polynomial, and the {L}ov{\'a}sz local lemma.
\newblock {\em J. Stat. Phys.}, 118(5--6):1151--1261, 2005.

\bibitem[SS14]{sly2014computational}
Allan Sly and Nike Sun.
\newblock Counting in two-spin models on $d$-regular graphs.
\newblock {\em Ann. Probab.}, 42(6):2383--2416, 2014.
\newblock (Conference version in \emph{FOCS}'12).

\bibitem[SS21]{shao2020contraction}
Shuai Shao and Yuxin Sun.
\newblock Contraction: A unified perspective of correlation decay and zero-freeness of 2-spin systems.
\newblock {\em J. Stat. Phys.}, 185(2):12, 2021.
\newblock (Conference version in \emph{ICALP}'20).

\bibitem[SS26]{shao2026coupling}
Shuai Shao and Ke~Shi.
\newblock Coupling independence implies zero-freeness.
\newblock {\em arXiv preprint arXiv:2610.12456}, 2026.

\bibitem[SS{\v{S}}Y17]{sinclair2017spatial}
Alistair Sinclair, Piyush Srivastava, Daniel {\v{S}}tefankovi\v{c}, and Yitong Yin.
\newblock Spatial mixing and the connective constant: Optimal bounds.
\newblock {\em Probab. Theory Related Fields}, 168(1--2):153--197, 2017.
\newblock (Conference version in \emph{SODA}'15).

\bibitem[SST14]{sinclair2014approximation}
Alistair {Sinclair}, Piyush {Srivastava}, and Marc {Thurley}.
\newblock Approximation algorithms for two-state anti-ferromagnetic spin systems on bounded degree graphs.
\newblock {\em J. Stat. Phys.}, 155(4):666--686, 2014.
\newblock (Conference version in \emph{SODA}'12).

\bibitem[SZ17]{sly2017glauber}
Allan Sly and Yumeng Zhang.
\newblock The {Glauber} dynamics of colorings on trees is rapidly mixing throughout the nonreconstruction regime.
\newblock {\em Ann. Appl. Probab.}, 27(5):2646--2674, 2017.

\bibitem[SZ26]{schmidhuber2026polynomial}
Alexander Schmidhuber and Alexander Zlokapa.
\newblock Polynomial-time classical algorithms for mean-field models up to the glass transition.
\newblock {\em arXiv preprint arXiv:2610.06807}, 2026.

\bibitem[Tal10]{talagrand2010mean}
Michel Talagrand.
\newblock {\em Mean Field Models for Spin Glasses: Volume {I}: Basic Examples}.
\newblock Springer, 2010.

\bibitem[VB85]{viana1985phase}
Laura Viana and Alan~J. Bray.
\newblock Phase diagrams for dilute spin glasses.
\newblock {\em J. Phys. C: Solid State Phys.}, 18(15):3037--3051, 1985.

\bibitem[Vig00]{vigoda1999improved}
Eric Vigoda.
\newblock Improved bounds for sampling colorings.
\newblock {\em J. Math. Phys.}, 41(3):1555--1569, 2000.
\newblock (Conference version in \emph{FOCS}'99).

\bibitem[Wag09]{wagner2009weighted}
David~G. Wagner.
\newblock Weighted enumeration of spanning subgraphs with degree constraints.
\newblock {\em J. Combin. Theory Ser. B}, 99(2):347--357, 2009.

\bibitem[Wan26]{wang2026optimal}
Sihan Wang.
\newblock Optimal mixing of {Glauber} dynamics for the {Sherrington--Kirkpatrick} model at $\beta<1/2$.
\newblock {\em arXiv preprint arXiv:2608.22159}, 2026.

\bibitem[Wei06]{weitz06counting}
Dror Weitz.
\newblock Counting independent sets up to the tree threshold.
\newblock In {\em {STOC}}, pages 140--149. ACM, 2006.

\bibitem[Wor99]{wormald1999models}
Nicholas~C. Wormald.
\newblock Models of random regular graphs.
\newblock In J.~D. Lamb and D.~A. Preece, editors, {\em Surveys in Combinatorics, 1999}, volume 267 of {\em London Mathematical Society Lecture Note Series}, pages 239--298. Cambridge University Press, 1999.

\bibitem[YL52]{lee1952statistical}
Chen-Ning Yang and Tsung-Dao Lee.
\newblock Statistical theory of equations of state and phase transitions. {I}. {T}heory of condensation.
\newblock {\em Phys. Rev.}, 87(3):404--409, 1952.

\bibitem[YZ16]{yin2016sampling}
Yitong Yin and Chihao Zhang.
\newblock Sampling in {Potts} model on sparse random graphs.
\newblock In {\em RANDOM}, pages 47:1--47:22. LZI, 2016.

\bibitem[ZK07]{zdeborova2007phase}
Lenka Zdeborová and Florent Krzakala.
\newblock Phase transitions in the coloring of random graphs.
\newblock {\em Phys. Rev. E}, 76(3):031131, 2007.

\bibitem[ZK16]{Zdeborova2016statistical}
Lenka Zdeborov{\'a} and Florent Krzakala.
\newblock Statistical physics of inference: thresholds and algorithms.
\newblock {\em Adv. Phys.}, 65(5):453--552, 2016.

\bibitem[ZZZ09]{zhang2009stability}
Pan Zhang, Ying Zeng, and Haijun Zhou.
\newblock Stability analysis on the finite-temperature replica-symmetric and first-step replica-symmetry-broken cavity solutions of the random vertex cover problem.
\newblock {\em Phys. Rev. E}, 80:021122, 2009.

\end{thebibliography}
\end{document}